\documentclass[12pt,letterpaper]{article}
\usepackage{amsmath,amssymb,mathrsfs,amsthm}
\usepackage{microtype}
\usepackage{booktabs}
\usepackage{graphicx}
\usepackage{bbm}
\usepackage{blkarray}
\usepackage{enumerate}

\usepackage{bm}
\usepackage{listings}
\usepackage{caption}
\usepackage{subcaption}
\usepackage{econometrics}

\usepackage{setspace}
\usepackage{multirow}

\usepackage[authoryear]{natbib}

\usepackage[colorlinks=true,linkcolor=blue,citecolor=blue, pagebackref=false, backref=false]{hyperref}

\usepackage[toc,title,titletoc,header]{appendix}

\usepackage{hyperref}
\hypersetup{
colorlinks=true,
linkcolor=blue,
filecolor=magenta,      
urlcolor=magenta,
pdftitle={Overleaf Example},
pdfpagemode=FullScreen,
}

\usepackage[left=1in,right=1in,top=0.8in,bottom=1in]{geometry}

\theoremstyle{definition}

\newtheorem{example}{Example}[section]
\newtheorem{definition}{Definition} 
\newtheorem{theorem}{Theorem}[section]

\newtheorem{lemma}{Lemma}[section]
\newtheorem{proposition}{Proposition}[section]
\newtheorem{remark}{Remark}[section]

\newtheorem{assumption}{Assumption}

\newcommand{\indep}{\perp \!\!\! \perp}

\allowdisplaybreaks

\usepackage{xr}
\usepackage{lineno}
\setpagewiselinenumbers
\title{Causal inference in network experiments: \\
regression-based analysis and design-based properties\thanks{Peng Ding is supported by the NSF DMS (grant \# 1945136). We thank Michael Leung for his helpful comments.}}
\author{Mengsi Gao \\
Department of Economics\\
UC Berkeley\\
\url{mengsi.gao@berkeley.edu}
\and 
Peng Ding \\
Department of Statistics\\
UC Berkeley \\
\url{pengdingpku@berkeley.edu} \\
}
\date{\today 
}

\begin{document}

% \begin{frontmatter}

% \title{Causal inference in network experiments: \\
% regression-based analysis and design-based properties}
% % \runtitle{A Sample Running Head Title}

% \begin{aug}
% %
% \author[id=au1,addressref={add1}]{\fnms{Mengsi}~\snm{Gao}\ead[label=e1]{mengsi.gao@berkeley.edu}}
% \author[id=au2,addressref={add2}]{\fnms{Peng}~\snm{Ding}\ead[label=e2]{pengdingpku@berkeley.edu}}
% %%%%%%%%%%%%%%%%%%%%%%%%%%%%%%%%%%%%%%%%%%%%%%
% %% Addresses                                %%
% %%%%%%%%%%%%%%%%%%%%%%%%%%%%%%%%%%%%%%%%%%%%%%
% \address[id=add1]{%
% \orgdiv{Department of Economics},
% \orgname{UC Berkeley}}

% \address[id=add2]{%
% \orgdiv{Department of Statistics},
% \orgname{UC Berkeley}}
% \end{aug}

% %% Put support info here. Reminder: do not thank the handling coeditor anonymously or by name
% \support{Peng Ding is supported by the NSF DMS (grant \# 1945136).}
% %
% \coeditor{\fnm{[Name} \snm{Surname}; will be inserted later]}
\onehalfspacing
% % % \begin{spacing}{1.2}
\maketitle

\vspace{-0.3in}
\thispagestyle{empty} 

%what is the point of this paper?
% what are the implications of our findings?
% \begin{spacing}{1.2}

% \vspace{-1cm}
\begin{abstract}
% Investigating interference or spillover effects among units is a central task in many social science problems. 
Network experiments are powerful tools for studying spillover effects, which avoid endogeneity by randomly assigning treatments to units over networks. However, it is non-trivial to analyze network experiments properly without imposing strong modeling assumptions. 
% Previously, many researchers have proposed sophisticated point estimators and standard errors for causal effects under network experiments. 
We show that regression-based point estimators and standard errors can have strong theoretical guarantees if the regression functions and robust standard errors are carefully specified to accommodate the interference patterns under network experiments. We first recall a well-known result that the Hájek estimator is numerically identical to the coefficient from the weighted-least-squares fit based on the inverse probability of the exposure mapping. Moreover, we demonstrate that the regression-based approach offers three notable advantages: its ease of implementation, the ability to derive standard errors through the same regression fit, and the potential to integrate covariates into the analysis to improve efficiency.
% , thereby enhancing estimation efficiency.
% Furthermore, we analyze the asymptotic bias of the regression-based network-robust standard errors.  
Recognizing that the regression-based network-robust covariance estimator can be anti-conservative under nonconstant effects, we propose an adjusted covariance estimator to improve the empirical coverage rates. 
% Although we focus on regression-based point estimators and standard errors, our theory holds under the design-based framework, which assumes that the randomness comes solely from the design of network experiments and allows for arbitrary misspecification of the regression models. 
\end{abstract}

% it is easy to implement, can provide standard errors through the same weighted-least-squares fit, and it allows for the integration of covariates into the analysis to improve estimation efficiency. 

% \medskip

% \noindent JEL classification codes: C13, C21
% \end{frontmatter}

% \newpage
% \onehalfspacing

\section{Introduction}
% \cite{Savje2023}

Network experiments have gained growing interest across various fields, including economics, social science, public health, and tech companies \citep{Jackson2008, Valente2010, BlakeCoey2014, AngelucciDiMaro2016, Aral2016, Breza2016, AtheyImbens2017, AtheyEcklesImbens2018a, AronowEcklesSamii2021}.  
% [PD: cannot see Aral2016a when i complied.]
They present an exceptional avenue to delve into the intricacies of interactions among units. 
Important examples of such experiments include 
\cite{Sacerdote2001},
\cite{MiguelKremer2004}, 
\cite{BandieraRasul2006},
\cite{BakshyRosennMarlow2012}, 
\cite{BanerjeeChandrasekharDuflo2013}, 
\cite{BursztynEdererFerman2014},
\cite{CaiJanvrySadoulet2015}, 
\cite{PaluckShepherdAronow2016}, 
\cite{BeamanDillon2018}, 
\cite{HaushoferShapiro2018}, 
and
\cite{CarterLaajajYang2021}.
These experiments transcend the conventional framework of individual-level randomization by exploring the effects of treatments not only on the treated individuals but also on their peers. This introduces the concept of ``interference,'' which challenges the ``stable unit treatment value assumption'' (SUTVA) that rules out interference in classic causal inference. 
% (\citet{Rubin1980, ImbensRubin2015}). 

% These experiments go beyond the traditional individual-level randomization and consider the impact of treatments or interventions on a network of connected individuals, known as ``interference".
% It violates the traditional assumption ``stable unit treatment value assumption” (SUTVA) posited by \cite{Rubin1990}.
% In economics, network experiments have proven to be valuable in investigating various phenomena where social interactions play a crucial role. 

% COMMENT ON \cite{Manski1993} and \cite{Angrist2014}

Over the last decade, the study of social interactions and peer effects through structural models has gained considerable attention \citep{Manski1993, Graham2008, BramoulleDjebbariFortin2009a,  Goldsmith-PinkhamImbens2013}. Distinguishing between the influence of peers' outcomes (endogenous peer effects) and the influence of peers' characteristics (contextual peer effects) can become challenging due to the simultaneous behavior of interacting agents. This challenge is known as the ``reflection problem'' \citep{Manski1993}. 
\cite{Angrist2014} criticized several econometric approaches to estimating peer effects. 
Without covariates, outcome-on-outcome regressions either reflect a tautological identity or capture group-level clustering without behavioral meaning. 
With covariates, the resulting estimates may be biased by measurement error and other factors, leading to spurious evidence of peer effects.
\cite{Paula2017} and \cite{BramoulleDjebbariFortin2020} relate the counterexample proposed in \citet[Section 6]{Angrist2014} to a well-known instance of non-generic identification failure, initially noted by \cite{Manski1993} and also demonstrated by \cite{BramoulleDjebbariFortin2009a}.

% An expanding volume of literature focuses on interference without imposing strong structural assumptions \citep{HalloranStruchiner1995, TchetgenTchetgenVanderWeele2012, Manski2013, HuLiWager2022, Viviano2023a}. 
An expanding volume of literature explores scenarios with interference of arbitrary but known forms which in turn requires researchers to make specific assumptions about the extent of interference. 
Many papers assume correctly specified \emph{exposure mappings} for inference \citep{AronowSamii2017, BairdBohrenMcIntosh2018, Vazquez-Bare2022, Owusu2023}.
These mappings impose assumptions on the interference structure in the experiment, where the treatment assignment vector affects potential outcomes through a low-dimensional function \citep{Manski2013, AronowSamii2017}.
This approach can be critiqued for typically ruling out endogenous peer effects.
% [PD: vector of sufficient statistics not appropriate because sufficient statistics has unique meaning in statistics.]
Some other papers assume ``partial interference'' \citep{Sobel2006, HudgensHalloran2008, UganderKarrerBackstrom2013, KangImbens2016, LiuHudgensBecker-Dreps2016, BasseFeller2018, QuXiongLiu2021, AlzubaidiHiggins2023}, where units are partitioned into separate clusters, and interference is restricted to occur exclusively among units within the same cluster.
Conversely, more recent literature further relaxes the partial interference assumption and studies interference of general forms \citep{SAVJEARONOWHUDGENS2021, Viviano2023}. 

% \citep{UganderKarrerBackstrom2013, LiuHudgensBecker-Dreps2016, AlzubaidiHiggins2023}. 

% \cite{Savje2023} 
\cite{Leung2022} proposed to estimate exposure effects under ``approximate neighborhood interference'' (ANI) while allowing for misspecification of exposure mappings. 
ANI refers to the situation where treatments assigned to individuals further from the focal unit have a smaller, but potentially nonzero, effect on the focal unit's response. 
\cite{Leung2022} verified that ANI is applicable to well-known models of social interactions, such as the network version of the linear-in-means model \citep{Manski1993} and the complex contagion model \citep{Granovetter1978}, both of which allow for endogenous peer effects. 
% This further motivates studying under the framework of ANI.
He considered the Horvitz--Thompson estimator and studied its consistency and asymptotic normality. For inference, he proposed a network Heteroskedasticity and Autocorrelation Consistent (HAC) covariance estimator, and studied its asymptotic bias for estimating the true covariance. 
However, he did not derive the point and covariance estimator directly from regression-based analysis, which is our focus.

% \cite{AronowEcklesSamii2021}
Our paper builds upon \cite{Leung2022}, which accommodates a single large network.
We enrich the discussion of the regression estimators from the design-based perspective, with a special emphasis on network experiments.
The design-based inference makes weak distributional assumptions about outcome models and relies solely on the randomization mechanism.
We focus on the Hájek estimator, which is numerically identical to the coefficient derived from the weighted-least-squares (WLS) fit involving unit data that relies on the inverse probability of exposure mappings \citep{AronowSamii2017}.
The regression-based approach offers three notable advantages. First, it is easy to implement without too much additional programming. Second,
it can provide standard errors through the same WLS fit. Third, it allows for incorporating covariates into the analysis, which can potentially increase the estimation precision if the covariates are predictive of the outcome.
Moreover, we examine the asymptotic performance of the regression-based network HAC estimator and prove results that justify the regression-based inference for network experiments from the design-based perspective. 
% This constitutes our first contribution.
% Although we focus on regression-based point estimators and standard errors, our theory holds under the design-based framework, which assumes that the randomness comes solely from the design of network experiments and allows for arbitrary misspecification of the regression models. 
%, drawing inspiration from the literature of time series \citep{NeweyWest1987}, spatial econometrics and Generalized Method
%of Moments estimation. 
%Our results justify the regression-based inference for network experiments from the design-based perspective. This constitutes our first contribution.

Unlike their spatial or time-series counterparts, network HAC estimators lack a theoretical guarantee of positive semi-definiteness \citep{Kojevnikov2021}. Moreover, they are known to have poor finite-sample properties \citep{Matyas1999}.
% The well-established outcome regarding the asymptotic bias of the standard covariance estimator for the estimation of the average treatment effect through difference-in-means \citep{ImbensRubin2015}
% Additionally, in contrast to the well-established outcome regarding the asymptotic bias of the standard covariance estimator for the estimation of the average treatment effect through difference-in-means \citep{ImbensRubin2015}, 
In network experiments, the asymptotic bias of the HAC estimator can be negative under interference, resulting in undercoverage of the associated confidence interval.
To address these concerns, we propose a modified HAC estimator that ensures positive semi-definiteness and asymptotic conservativeness, which also performs well in finite-sample simulation. 
% This constitutes our second contribution.

% In addition, \cite{ZhaoDingMukerjee2018}  made a heuristic link between the design-based inference and the regression-based inference in the context of uniform split-plot designs, motivating
% with the idea of the derived linear model. 

% \textcolor{magenta}{TBA.}

% In particular, we focus on how the estimator behaves when applied in regression models to estimate these parameters. 
% Additionally, we will explore the topic of covariate adjustment, which is a statistical technique used to improve the precision of estimates in survey sampling. By examining the properties of the Hájek estimator and discussing covariate adjustment techniques, we hope to provide insights into the usefulness and limitations of these methods in practical applications.

% The application of statistical methods in the analysis of network experiments has become increasingly popular in recent years. In this context, regression analysis plays a vital role in understanding the relationship between network structures and the outcomes of interest. This paper aims to contribute to this area of research by studying the regression analysis of network experiments.
% The main objective of this paper is to investigate the regression-based properties of the Hájek estimator, particularly how it behaves when used in regression models to estimate population parameters in network experiments. 

Furthermore, we delve into the subject of covariate adjustment. Proper covariate adjustment can enhance the accuracy of estimators in randomized experiments by accounting for the imbalance in pretreatment covariates. 
Recall the results in the classical completely randomized treatment-control experiment.
The regression framework offers a versatile approach to incorporating covariate information with a potential of enhancing
asymptotic efficiency by including the interactions of the treatment and covariates \citep{Fisher1935, Lin2013, NegiWooldridge2021}.
An expanding body of literature explores the design-based justification of regression-based covariate adjustment with different types of experimental data \citep{Fogarty2018, ChangMiddletonAronow2021, SuDing2021, ZhaoDing2022, WangSusukidaMojtabai2023, ZhaoDingLi2024}.
% \cite{NegiWooldridge2021} studied the efficiency gain of covariate adjustment in the superpopulation model.
% \cite{WangSamiiChang2023} did not discuss how to increase efficiency by incorporating covariate information. 
Our paper studies the theoretical properties of covariate adjustment in network experiments and demonstrates the potential efficiency gain in simulation and empirical application under reasonable data-generating processes. 

\paragraph*{Organization of the paper}
Section \ref{sec:setup} sets up the framework for the design-based
inference in network experiments, reviews the Horvitz--Thompson and Hájek estimators, and introduces the main assumptions from \cite{Leung2022}. Section \ref{sec:WLS} reviews the Hájek estimator recovered from the WLS fit \citep{AronowSamii2017}, proposes the regression-based HAC covariance estimator, and analyzes its asymptotic bias. Because the covariance estimator can be anti-conservative, we propose a modified, positive semi-definite covariance estimator. Section \ref{sec:covadju} considers additive and fully-interacted covariate adjustment to the WLS fit, describes associated asymptotic properties, proposes modified covariance estimators, and studies their asymptotic properties. Section \ref{sec:Numerical Illustration} studies the finite-sample performance of our point and covariance 
% [PD: i changed ``variance'' to ``covariance'' in the Abstract. be coherent?] 
estimators based on simulation and illustrates the practical relevance of our results by re-analyzing the network experiment in \cite{PaluckShepherdAronow2016}. Section \ref{sec:conclusions} discusses the extension to continuous exposure mappings.  The appendix includes all the proofs and intermediate results. 

\paragraph*{Notation}
Let $\mathbb{N}$ denote the set of all non-negative integers.
Let ${I}_{m}$ be an $m \times m$ identity matrix and $\iota_{m}$ be an $m \times 1$ vector of ones. We suppress the dimension $m$ when it is clear from the context. 
Unless stated otherwise, all vectors are column vectors.
Let $1(\cdot)$ be the indicator function.
Let $\|\cdot\|$ denote the Euclidean norm, i.e., $\|w\| = \sqrt{w^\top w}$ for $w\in\mathbb{R}^v$. 
% Let $\text{lm}(Y_i \sim x_i)$ denote the least-squares
% regression of $Y_i$ on $x_i$ and focus on the associated HAC covariance estimator. 
The terms ``regression'' and ``HAC covariance'' refer to the numerical outputs of the WLS fit without any modeling assumptions; we evaluate their properties under the design-based framework.
We use ``IID'' and ``CLT'' to denote ``independent and identically distributed'' and ``central limit theorem,'' respectively.

\section{Framework, estimators and assumptions}\label{sec:setup}
\subsection{Setup of network experiments}
% \label{sec:setup}
We consider a finite population model that conditions on potential outcomes and networks while viewing the treatment assignment as the sole source of randomness. This approach follows the design-based framework \citep{ImbensRubin2015, AronowSamii2017, LiDing2020a, AbadieAtheyImbens2020, Leung2022, Chang2023}.
Let $\mathcal{N}_n = \{1,\ldots,n\}$ denote the set of units.
% We define $\mathcal{N}_n$ as the set of units $\{1,\ldots,n\}$. 
The network structure is undirected, unweighted, has no self-links, and can be described using an adjacency matrix $ {A}=(A_{ij})_{i,j=1}^n$ with the $(i, j)$th entry $A_{ij} \in \{0, 1\}$  indicating the connection between units $i$ and $j$. 
Let $\mathcal{A}_n$ denote the set of all possible networks with $n$ units.
% We denote the set of all possible networks with $n$ units as $\mathcal{A}_n$. 
The assignment of treatments is represented by a binary vector $ {D} = (D_i)^n_{i=1}$, where each $D_i$ is a binary variable indicating whether unit $i$ has been assigned to the treatment. 

We define the potential outcome for each unit $i$ as $Y_i( {d})$, which represents the outcome of unit $i$ under the hypothetical scenario in which the units on the entire network are assigned the treatment vector ${d} = (d_i)^n_{i=1}\in
\{0,1\}^n$. From the notation, $Y_i( {d})$ depends not only on $d_i$, the treatment assignment of unit $i$, but also on the treatment assignments of all other units. This results in ``interference'' or ``spillover'' between units, which is not accounted for in the standard potential outcomes model under SUTVA.
We adopt the design-based framework in which the potential outcomes $Y_i(d)$'s and network $A$ are fixed, whereas the distribution of $D$ is known and does not depend on $Y_i(d)$'s and $A$.
%  (\citet{Rubin1980, ImbensRubin2015}).
% \begin{remark}

% \end{remark}

% We define the potential outcome for each unit $i$ as $Y_i( {d})$, which represents the hypothetical outcome of unit $i$ if a treatment vector $ {d} = (d_i)^n_{i=1}\in
% \{0,1\}^n$ is assigned to the entire network. In contrast to the standard potential outcomes model, the dependence of $Y_i( {d})$ on the entire vector of assignments introduces interference or ``spillovers" between units.
% The focus of our study, as defined by \cite{Leung2022}, is on the "exposure effects", i.e., the changes in outcomes resulting from the manipulation of $ {d}$.

% The potential outcome $Y_i(d)$ is very high-dimensional, and we introduce the exposre mapping for dimension reduction. 
% The \emph{exposure mapping} serves two central roles in the literature. The first role is to encode assumptions about the interference structure in the experiment, where the treatment assignment vector affects the potential outcomes through a low-dimensional vector of sufficient statistics \citep{Manski2013,AronowSamii2017}.
% See \cite{Savje2023} for a more general discussion of
% inference with misspecified exposure mappings.
With binary treatment $D_i$'s, we have $2^n$ potential outcomes for each unit. We utilize the exposure mapping as defined by \cite{AronowSamii2017} or the effective treatment mapping introduced by \cite{Manski2013} for dimensionality reduction.
% Let $\mathcal{T} \subseteq \mathbb{R}^{d_T}$ be a discrete set of dimension ${d_T}$, where $|\mathcal{T}|$ is finite and fixed. 
For any $n$, an exposure mapping is a function $T: \mathcal{N}_n \times\{0,1\}^n \times \mathcal{A}_n \rightarrow \mathcal{T}$, 
which maps the units, the treatment assignment vector, and the network structure to exposures received by a unit. We focus on the regime in which $|\mathcal{T}|$ is finite and fixed. 
With correctly specified exposure mapping, we can simplify the potential outcomes as $Y_i(d) = {Y}_i(t)$ because $Y_i(D)$ depends on $D$ only through $T_i = T(i, D, A)$.
We follow \cite{Leung2022}'s framework and allow for misspecified exposure mappings: $T_i$ need not correctly capture how others affect an individual's potential outcome.
With the potential outcomes $Y_i(d)$, we can define the unit $i$'s expected response under exposure mapping value $t$ as 
\begin{equation}
\mu_i(t)=\sum_{{d} \in\{0,1\}^n} Y_i({d}) \mathbb{P}\left({D}={d} \mid T_i=t\right), 
\label{eq:mu}
\end{equation}
which equals the expected potential outcome of unit $i$ over all possible treatment assignment vectors given the exposure mapping value at $t$.
Let $\mu(t) = n^{-1} \sum_{i=1}^n \mu_i(t)$ be the finite-population average and $ {\mu} = (\mu(t):t\in\mathcal{T})$ be the $|\mathcal{T}|\times 1$ vector containing all the $\mu(t)$'s corresponding to exposure mapping values $t\in\mathcal{T}$.
% The unit-level exposure effect is defined as the contrast between two exposure mapping values $t$ and $t^{\prime}$, i.e., $\tau_i(t,t')=\mu_i(t)-\mu_i(t^{\prime})$. 
% % We denote the finite-population average by $\mu(t)=\frac{1}{n} \sum_{i=1}^n \mu_i(t)$.
% We define the estimand of interest as the average effect
% \[
% \tau(t, t^{\prime})=\mu(t)-\mu(t^{\prime}).
% \]
% Define causal effects as linear combinations of the expected responses. 
We will discuss inference of the general estimand $\tau = G {\mu}$, where $G$ is an arbitrary contrast matrix, and the key lies in estimating $ {\mu}$. 
We focus on estimators of the form $\hat{\tau} = G \hat{Y}$, where $\hat{Y}$ is some regression estimator of $\mu$. 
Although we focus on regression-based point estimators and standard errors, our theory holds under the design-based framework, which assumes that the randomness comes solely from the design of network experiments and allows for misspecification of the regression models. 

While the theory can accommodate misspecified exposure mappings, this flexibility comes at the cost of complicating the causal interpretation.
When the exposure mapping is correctly specified, we have $Y_i(d) = {Y}_i(t)$ for $t \in \mathcal{T}$, allowing the average expected response to simplify to $\mu(t) = n^{-1}\sum_{i=1}^n {Y}_i(t)$. 
In this case, the estimand ${\tau} = G \mu$ becomes independent of the treatment assignment and has a clear causal interpretation.
However, when the exposure mapping is misspecified, $\mu(t)$ represents a weighted average of all potential outcomes, where the weights correspond to $\mathbb{P}\left({D}={d} \mid T_i=t\right)$ that depends on both the treatment assignment and the definition of the exposure mapping.  
Consequently, any change in the treatment assignment alters the estimand. 
As a result, ${\tau} = G \mu$ may lack a causal interpretation.
One scenario in which the estimand $\tau$ can still be interpreted causally is when treatments are assigned independently and $T_i$ depends only on unit $i$'s local group of neighbors \citep{LeungLoupos2023}. 
In this case, \( \tau \) represents a weighted average of unit-level treatment or spillover effects, comparing outcomes across different treatment assignments within this local group.
% When the exposure mapping is correctly specified, the estimand becomes independent of the treatment assignment. 
% However, when the exposure mapping is misspecified, the estimand depends on the treatment assignment, as shown in \eqref{eq:mu}, particularly through $\mathbb{P}\left({D}={d} \mid T_i=t\right)$. 
For a more general discussion of causal inference with misspecified exposure mappings, see \cite{Savje2023}.

To conclude this subsection, we present three examples of exposure mappings and interpret the corresponding estimands with some choices of $G$, in the context of \cite{PaluckShepherdAronow2016}, which we will revisit in Section \ref{app:paluck}.
\cite{PaluckShepherdAronow2016} conducted a randomized experiment to study how an anti-conflict intervention influences teenagers' social norms regarding hostile behaviors such as bullying, social exclusion, harassment, and rumor-spreading.
The treatment indicator $D_i$ corresponds to whether student $i$ was randomly assigned to participate in bi-weekly meetings that incorporated an anti-conflict curriculum.
The outcome $Y_i$ is self-reported data on wristband wearing--a public signal of anti-conflict behavior and participation in the program. 
The network $A$ is measured by asking students to name up to ten students at the school they spent time with in the last few weeks.

\begin{example}\label{ex:direct}
Setting $T_{1i} =D_i$ is a special case of exposure mapping. 
With $G = (-1,1)$, the estimand $\tau$ compares the average number of self-reported wristband wearing if a student were assigned to participate in the bi-weekly anti-conflict meetings versus if they were not.
We refer to this difference as the direct effect of the treatment on students' visible engagement in anti-conflict behavior.
\end{example}

\begin{example}\label{ex:one-dim}
For researchers interested in the spillover effect of having at least one friend assigned to the treatment versus none such friends, a natural choice of one-dimensional exposure mapping is $T_{2i} = {1}(\sum_{j=1}^n A_{ij}D_j > 0 ) \in \{0,1\}$.
With $G = (-1,1)$, the estimand $\tau$ compares the aaverage number of self-reported wristband wearing if a student has at least one treated friend versus when they have none.
We refer to this difference as the spillover effect.
\end{example}

\begin{example}\label{ex:two-dim}
For researchers interested in both the direct effect and the spillover effect, they can employ the following two-dimensional exposure mapping: $T_i = (T_{1i}, T_{2i})\in \{(0,0), (0,1), (1,0), (1,1)\}$.
% \begin{equation*}
% T_i = \left(D_i, {1}\left(\sum_{j=1}^n A_{ij}D_j > 0 \right) \right) \in \{(0,0), (0,1), (1,0), (1,1)\},
% \end{equation*}
% where the first factor captures the direct effect, and the second factor captures the spillover effect.  
In  this case, we have a $2 \times 2$ factorial exposure mapping.
Setting $G=\left(g_{1}, g_{2}, g_{12}\right)^\top$ with $g_{1}=2^{-1}(-1,-1,1,1)^\top, g_{2}=2^{-1}(-1,1,-1,1)^\top$, and $g_{12}=2^{-1}(1,-1,-1,1)^\top$, then the estimand $\tau$ recovers the direct effect, spillover effect, and interaction effect of two factors. 
\end{example}

% \begin{remark}
Different specifications of the exposure mapping may change the estimand. 
For instance, the estimand defined using $T_i = T_{1i}$ or $T_i = T_{2i}$ alone differs from that obtained using a two-dimensional exposure mapping, $T_i = (T_{1i}, T_{2i})$, unless $T_{1i}$ and $T_{2i}$ are orthogonal.  
With independent $D_i$'s, the components $T_{1i}$ and $T_{2i}$ of the exposure mapping in Example \ref{ex:two-dim} are orthogonal. 
Therefore, the exposure mappings in Examples \ref{ex:direct} and \ref{ex:one-dim} respectively capture the direct and spillover effects in Example \ref{ex:two-dim}. 
We examine all exposure mappings from Examples \ref{ex:direct}–\ref{ex:two-dim} when revisiting the empirical applications of \cite{PaluckShepherdAronow2016} and \cite{CaiJanvrySadoulet2015} to assess the robustness of our results to variations in the number of exposures; see Section \ref{app:paluck} and Appendix \ref{sec:Cai}.

\subsection{Horvitz--Thompson and Hájek estimators}
% [Discuss the Horvitz--Thompson estimator and the Hájek estimator.]
Inverse probability weighting is a general estimation strategy in survey sampling and causal inference.
%In the context of observational studies, the sampling probabilities are unkown. 
In the context of observational studies with interference,  \cite{TchetgenTchetgenVanderWeele2012}, \cite{LiuHudgensBecker-Dreps2016} and \cite{JacksonLinYu2020} studied inverse probability-weighted estimators of causal effects under different assumptions on the interference pattern.  In this subsection, we will review the Horvitz--Thompson and Hájek estimators for estimating population parameters based on the observed data in network experiments.

% based on group-level propensity scores under the assumption of partial interference. 
%Additionally, \cite{LiuHudgensBecker-Dreps2016} proposed a generalized inverse probability-weighted estimator that allows for any form of interference.
% in experimental setting: the Horvitz--Thompson estimator \citep{HorvitzThompson1952} and the Hájek estimator \citep{Hajek1971}. 
%Both methods rely on sampling probability, where each unit in the population has a known and nonzero probability of being selected in the sample. 

The Horvitz--Thompson estimator is a weighted estimator that assigns each unit a weight equal to the inverse of its selection probability. 
Recall $T_i = T(i, {D}, {A})$, and define the generalized propensity score \citep{Imbens2000} as $\pi_i(t) = \mathbb{P}( T_i=t )$.
% define ${1}_i(t)={1}\{T_i=t\}$ and the generalized propensity score \citep{imbens2000}:
% \[
% \pi_i(t)=\mathbb{E}\left[{1}_i(t)\right]  =  \mathbb{P}( T_i=t ).
% \]
The value of the propensity score is known by design and can be determined through exact calculation or approximation using Monte Carlo \citep{AronowSamii2017}. 
The Horvitz--Thompson estimator for $\mu(t)$ equals $\hat{Y}_{\text{ht}}(t) = n^{-1} \sum_{i=1}^n {1}(T_i=t) Y_i/\pi_i(t)$.
% \[
% \hat{Y}_{\text{ht}}(t) = \frac{1}{n} \sum_{i=1}^n \frac{{1}(T_i=t)}{\pi_i(t)} Y_i.
% \]
The Horvitz--Thompson estimator is unbiased if the propensity score $\pi_i(t)$'s are non-zero and is consistent under additional regularity conditions.
% The Horvitz--Thompson estimator is inefficient compared to the Hájek estimator and is inefficient compared to the covariate-adjusted estimator when some units have a very small selection probability.
\cite{Leung2022} focused on $\tau(t,t')$ and examined the asymptotic properties of the Horvitz--Thompson estimator 
$\hat{\tau}_{\text{ht}}(t, t^{\prime})
= \hat{Y}_{\text{ht}}(t) - \hat{Y}_{\text{ht}}(t')$.
% \[
% \hat{\tau}_{\text{ht}}(t, t^{\prime})
% = \hat{Y}_{\text{ht}}(t) - \hat{Y}_{\text{ht}}(t'). 
% \]
% The Hájek estimator is a modification of the Horvitz--Thompson estimator by incorporating a correction factor that accounts for variability in selection probabilities. 

The Hájek estimator refines the Horvitz--Thompson estimator by normalizing the Horvitz--Thompson estimator by dividing it by the sum of the individual weights involved in its definition: $\hat{Y}_{\textup{haj}}(t) 
% = \frac{\frac{1}{n}\sum_{i=1}^n \frac{{1}(T_i=t)Y_i}{\pi_i(t)}}{\frac{1}{n}\sum_{i=1}^n \frac{{1}(T_i=t)}{\pi_i(t)}}
= \hat{Y}_{\text{ht}}(t)/\hat{1}_{\text{ht}}(t)$,
% \[
% \hat{Y}_{\textup{haj}}(t) 
% % = \frac{\frac{1}{n}\sum_{i=1}^n \frac{{1}(T_i=t)Y_i}{\pi_i(t)}}{\frac{1}{n}\sum_{i=1}^n \frac{{1}(T_i=t)}{\pi_i(t)}}
% = \hat{Y}_{\text{ht}}(t)/\hat{1}_{\text{ht}}(t),
% % \frac{\hat{Y}_{\text{ht}}(t)}{\hat{1}_{\text{ht}}(t)},
% \]
where $\hat{1}_{\text{ht}}(t) 
= n^{-1} \sum_{i=1}^n {1}(T_i=t)/\pi_i(t)$ is the Horvitz--Thompson estimator for constant potential outcome $1$. 
The Hájek estimator is biased in general since $\hat{1}_{\text{ht}}(t)$ is random, but it is consistent since $\hat{1}_{\text{ht}}(t)$ is consistent for $1$ under regularity conditions. 
% The Hájek estimator may significantly reduce variance than the Horvitz--Thompson estimator with little cost in terms of bias \citep[pages 181-184]{SarndalSwenssonWretman2003}.

The existing literature provides two motivations for using the Hájek estimator.
First, it ensures invariance under the location shift of the outcome \citep{Fuller2011}. Second, empirical evidence suggests that the Hájek estimator is more stable and efficient with little cost of bias in most reasonable scenarios \citep{LuncefordDavidian2004, Fuller2011, Ding2024}.
%  (\citet{SarndalSwenssonWretman2003}). 
\cite{Leung2022} mentioned the Hájek estimator in the footnote of his paper without detailed theory. 
Moreover, the Hájek estimator is more natural from the regression perspective. 
Numerically, the Hájek estimator is identical to the coefficient from the WLS fit based on the inverse probability of the
exposure mapping \citep{AronowSamii2017}.
The regression-based approach offers three notable advantages.
First, WLS is easy to implement without too much additional programming.
Second, WLS can provide network-robust standard errors. 
Third, WLS can incorporate covariates to improve efficiency when covariates are predictive of the outcome. 
The main focus of our paper is to explore the design-based properties of the Hájek estimators obtained through the regression-based method and associated HAC covariance estimator.

\begin{remark}\label{re:HT}
The Horvitz--Thompson estimator can also be implemented via WLS. However, it requires transformations of both the weights and the outcome, making it a less natural option via regression.
More importantly, the corresponding regression-based variance estimator is not guaranteed to be exact for inference even if the individual effects are constant. For further discussion, see Appendix \ref{app:HT}.
\end{remark}

% \cite{AronowSamii2017} discussed two refinements to the Horvitz--Thompson estimator. The first refinement involves incorporating auxiliary covariate information to improve efficiency, which requires prior knowledge of how outcomes relate to covariates. The second refinement is the use of the Hájek estimator, which can be computed using WLS fit, with covariance adjustment through WLS residualization. They mention that variance estimation proceeds via Taylor linearization \citep[pages 172-176]{SarndalSwenssonWretman2003}. Our paper demonstrates that we can incorporate covariate information into WLS fit without any modeling assumptions. We justify the regression-based analysis with network HAC variance estimator, which is straightforward to implement for applied researchers.

% The primary emphasis of this paper revolves around the Hájek estimator for two distinct reasons. Firstly, the Hájek estimator is widely acknowledged for its superior performance. Secondly, it demonstrates a more inherent connection with regression, which motivates our exploration of its regression-based properties. In particular, we investigate the behavior of the estimator when employed in regression models to estimate population parameters.

\subsection{Main assumptions}
We consider \cite{Leung2022}'s framework of ANI. ANI refers to a situation where treatments assigned to individuals who are farther away from the focal unit have a diminishing effect on the focal unit's response, although the effect is not necessarily zero.

In this subsection, we provide an overview of the key assumptions outlined in \cite{Leung2022}, which serve as the foundation for our analysis. These conditions ensure the theoretical properties of the regression-based point and covariance estimators.
For readers more interested in practical applications, they have the option to skip this subsection during their initial reading and focus on the procedures and properties presented in Sections \ref{sec:WLS} and \ref{sec:covadju}.

% To facilitate the presentation, we begin by introducing some essential definitions and notations from \cite{Leung2022}. 
Let $\ell_{ {A}}(i, j)$ denote the path distance between units $i$ and $j$ within network $A$, representing the length of the shortest path connecting them. The path distance refers to the smallest number of edges that must be crossed to journey from unit $i$ to unit $j$ within the network. 
% That is, for $L \in \mathcal{N}_n$, 
% \[
% \ell_A(i, j):=\inf _{\substack{\left\{t_1, \ldots, t_L\right\} \\ \text { s.t. } t_1=i, t_L=j}} \Pi_{s=1}^{L-1} A_{t_s t_{s+1}}.
% \]
Furthermore,
$\ell_A(i, j)$ is defined as $\infty$ if $i\ne j$ and no path exists between units $i$ and $j$ and defined as $0$ if $i = j$.
% \begin{remark}
% We can extend the definition of path distance to a weighted and directed network following \cite{AuerbachTabord-Meehan2023a} by defining it as the smallest sum of values in matrix $A$ along any sequence from unit $i$ to unit $j$.
% % Define the weight of a link from agent $i$ to $j$ as $A_{ij} \in \mathbb{Z} \cup \{\infty\}$. We assume $A_{ij} = 0$ if and only if $i = j$, and $A_{ij} = \infty$ denotes no link from $i$ to $j$. Then the path distance can be defined as follows for $L \in \{2,\ldots,n\}$.
% % \[
% % \ell_A(i, j):=\inf _{\substack{\left\{t_1, \ldots, t_L\right\} \\ \text { s.t. } t_1=i, t_L=j}} \sum_{s=1}^{L-1} A_{t_s t_{s+1}}.
% % \]
% With such a network, modifications to the assumptions of the theory are necessary, given that the distance is now measured on a different scale.
% \end{remark}
For a specific unit $i$, its $K$-neighborhood, denoted by $\mathcal{N}(i, K;  {A})= \{j \in \mathcal{N}_n: \ell_{ {A}}(i, j) \leq K \}$, includes the set of units within network $A$ that are at most at a path distance of $K$ from unit $i$. 
% Mathematically, $\mathcal{N}(i, K;  {A})=\left\{j \in \mathcal{N}_n: \ell_{ {A}}(i, j) \leq K \right\}$.
% , where the path distance between units $i$ and $j$, $\ell_{ {A}}(i, j)$, is defined as the length of the shortest path between them. 
% The path distance is the minimum number of edges that need to be traversed in order to go from unit $i$ to unit $j$ in the network.
Define ${d}_{{\mathcal{N}(i, K; {A})}}=(d_j: j \in \mathcal{N}(i, K;  {A}))$ and $ {A}_{\mathcal{N}(i, K; {A})}=(A_{kl}: k, l \in \mathcal{N}(i, K;  {A}))$ as the subvector of $ {d}$ and subnetwork of $ {A}$ on $\mathcal{N}(i, K;  {A})$, respectively.

% Unit $i$’s $K$-neighborhood, denote by $\mathcal{N}(i, K; A)$, is the set of units that are at most path distance $K$ apart from unit $i$.
% \[
% \mathcal{N}(i, K; A)=\left\{j \in \mathcal{N}_n: \ell_{ {A}}(i, j) \leq K \right\}
% \]
% where the path distance between units $i$ and $j$, $\ell_{ {A}}(i, j)$, is defined as the length of the shortest path between them. 
% The path distance is the minimum number of edges that need to be traversed in order to go from unit $i$ to unit $j$ in the network.
% Define $ {d}_{{\mathcal{N}(i, K; A)}}=\left(d_j: j \in \mathcal{N}(i, K; A)\right)$ and $ {A}_{\mathcal{N}(i, K; A)}=\left(A_{k l}: k, l \in \mathcal{N}(i, K; A)\right)$, respectively the subvector of $ {d}$ and subnetwork of $ {A}$ on $\mathcal{N}(i, K; A)$.

\begin{assumption}[Exposure Mapping]\label{asu1}
% [PD: i move something up] 
There exists a $K \in \mathbb{N}$ not dependent on the sample size $n$ such that for any $n \in \mathbb{N}$ and $i \in \mathcal{N}_n$, if $\mathcal{N}(i, K; {A})=\mathcal{N}(i, K; {A}^{\prime}), A_{\mathcal{N}(i, K;  {A})}= {A}_{\mathcal{N}(i, K; {A}^{\prime})}^{\prime}$, and ${d}_{\mathcal{N}(i, K; {A})}= {d}_{\mathcal{N}(i, K; {A}^{\prime})}^{\prime}$, then $T(i, {d}, {A})=T(i, {d}^{\prime}, {A}^{\prime})$ for all ${d}, {d}^{\prime} \in\{0,1\}^n$ and ${A}, {A}^{\prime} \in \mathcal{A}_n$.
\end{assumption}

\begin{assumption}[Overlap]\label{asu2}
$\pi_i(t) \in[\underline{\pi}, \bar{\pi}] \subset(0,1)$, for all $n \in \mathbb{N}, i \in \mathcal{N}_n, t \in \mathcal{T}$, where $\underline{\pi} $ and $\bar{\pi}$ are some absolute constant values.
\end{assumption}
\begin{assumption}[Bounded Potential Outcomes]\label{asu3}
$\left|Y_i({d})\right|<c_Y<\infty$, for all $ n \in \mathbb{N}, i \in \mathcal{N}_n, {d} \in\{0,1\}^n$, where $c_Y$ is an absolute constant.
\end{assumption}
Assumption \ref{asu1} requires the interference pattern of interest to be local, implying that the exposure mapping indicators are weakly dependent. Specifically, ${1}(T_i=t) \indep {1}(T_j=t)$ if $\ell_{ {A}}(i, j)>2 K$ for some $K$. % that is not dependent on the sample size $n$. 
For instance, $K=0$ for the exposure mapping in Example \ref{ex:direct} and  $K=1$ for both in Examples \ref{ex:one-dim} and \ref{ex:two-dim}.
% [PD: also the other example?] 
% is correctly specified, then $K=1$.  
Assumption \ref{asu2} requires the generalized propensity scores to be uniformly bounded between 0 and 1.
Assumption \ref{asu3} imposes uniform boundedness on the potential outcomes.

Let ${D}^{\prime}$ be an IID copy of $ {D}$. Define ${D}^{(i, s)}= ({D}_{\mathcal{N}(i, s; {A})}, {D}_{\mathcal{N}_n \backslash \mathcal{N}(i, s; {A})}^{\prime})$ as the concatenation of the subvector of $ {D}$ on $\mathcal{N}(i, s; {A})$ and the subvector of ${D}^{\prime}$ on $\mathcal{N}_n \backslash \mathcal{N}(i, s; {A})$. 
Define
\begin{equation}
\theta_{n, s}
= \max _{i \in \mathcal{N}_n} \mathbb{E}\left[\left|Y_i( {D})-Y_i( {D}^{(i, s)})\right|\right], \label{eq:theta}   
\end{equation}
where the expectation is over the randomness of $D$ and $D'$
% [PD: and $D'$?] 
with all potential outcomes fixed.
The interference, caused by distant individuals with a distance of more than $s$ from the subject, is measured as the largest expected change in any individual's potential outcome when altering the treatment assignments of those distant individuals. 
% [PD: this is a dense sentence.]
% \begin{definition}
% A triangular array $\left\{Z_i\right\}_{i=1}^n$ is $\psi$-dependent if there exist (a) uniformly bounded constants $\left\{\tilde{\theta}_{n, s}\right\}_{s, n \in \mathbb{N}}$ with $\theta_{n, 0}=1 \forall n$ such that $\sup _n \tilde{\theta}_{n, s} \rightarrow 0$ as $s \rightarrow \infty$, and (b) functionals $\left\{\psi_{h, h^{\prime}}(\cdot, \cdot)\right\}_{h, h^{\prime} \in \mathbb{N}}$ with $\psi_{h, h^{\prime}}: \mathcal{L}_h \times \mathcal{L}_{h^{\prime}} \rightarrow[0, \infty)$ such that
% $$
% \left|\operatorname{Cov}\left(f\left( {Z}_H\right), f^{\prime}\left( {Z}_{H^{\prime}}\right)\right)\right| \leq \psi_{h, h^{\prime}}\left(f, f^{\prime}\right) \tilde{\theta}_{n, s}
% $$
% for all $n, h, h^{\prime} \in \mathbb{N} ; s>0 ; f \in \mathcal{L}_h ; f^{\prime} \in \mathcal{L}_{h^{\prime}}$; and $\left(H, H^{\prime}\right) \in \mathcal{P}_n\left(h, h^{\prime} ; s\right)$.
% \end{definition}
Mathematically, ANI assumes that as the distance $s$ approaches infinity, the largest value of $\theta_{n, s}$, taken over all feasible networks, converges to zero, which is formalized in Assumption \ref{asu4} below. 
\begin{assumption}[ANI]\label{asu4}
The $\theta_{n,s}$ defined in \eqref{eq:theta} satisfies 
$\sup _n \theta_{n, s} \rightarrow 0 \text { as } s \rightarrow \infty$. 
\end{assumption}
In simpler terms, Assumption \ref{asu4} stipulates that interference from distant individuals should vanish as the distance becomes large.
We skip Assumption 5 in \cite{Leung2022}, which is for showing consistency of the Horvitz--Thompson estimator, and proceed to Assumption \ref{asu6} below for the asymptotic normality of the Hájek estimator.   
% For consistency, we impose Assumption \ref{asu7}(a) below which is stronger than Assumption 5 in \cite{Leung2022}.  
% [PD: i rewrote the sentence.  also it is very strange here. we have not commented on assumption 5 yet. why Assumption \ref{asu7}(a) direct?]
%We skip Assumption 5 in \cite{Leung2022} for consistency of the estimator, and instead, we utilize its stronger version, Assumption \ref{asu7}(a) below.
Define 
\begin{equation}
M_n(m, k)= n^{-1} \sum_{i=1}^n|\mathcal{N}(i, m; {A})|^k  
\label{eq:Mn}
\end{equation}
as the $k$-th moment of the $m$-neighborhood size within network $A$.
%Denote by $\mathcal{H}_n(s, m)$ the set of paired couples $(i, k)$ and $(j, l)$ such that the units within each couple are at most path distance $m$ apart from each other, and the two pairs are exactly path distance $s$ apart:
For any $H, H' \subseteq \mathcal{N}_n$, define $\ell_A(H, H') = \min\{\ell_A(i, j) : i \in H, j \in H'\}$.
Define
\begin{equation}
\mathcal{H}_n(s, m)=\left\{(i, j, k, l) \in \mathcal{N}_n^4: k \in \mathcal{N}(i, m; {A}), l \in \mathcal{N}(j, m; {A}), \ell_{ {A}}(\{i, k\},\{j, l\})=s\right\}  \label{eq:Hn}
\end{equation}
as the set of paired couples $(i, k)$ and $(j, l)$ such that the units within each couple are at most path distance $m$ apart from each other, and the two pairs are exactly path distance $s$ apart. 
%Similar to $\mathcal{H}_n(s, m)$, we denote by $\mathcal{J}_n(s, m)$ the set of paired couples $(i, k)$ and $(j, l)$ such that the units within each couple are at most path distance $m$ apart from each other, and $i$ and $j$ are exactly path distance $s$ apart: 
Similarly, define 
\begin{equation}
\mathcal{J}_n(s, m)=\left\{(i, j, k, l) \in 
\mathcal{N}_n^4: k \in \mathcal{N}(i, m; {A}), l \in \mathcal{N}(j, m; {A}), \ell_{ {A}}(i, j)=s\right\}
\label{eq:Jn}
\end{equation}
as the set of paired couples $(i, k)$ and $(j, l)$ such that the units within each couple are at most path distance $m$ apart from each other, and $i$ and $j$ are exactly path distance $s$ apart. 
% Denote the variance of the Horvitz--Thompson estimator defined on the centered outcome by
In Assumption \ref{asu6}, we replace $\sigma_n^2$ from \citet[Assumption 6]{Leung2022} with the matrix ${\Sigma}_\textup{haj}$:
\begin{equation}
{\Sigma}_\textup{haj}  = \operatorname{Var}\left(  n^{-1 / 2} \sum_{i=1}^n  \frac{1(T_i=t)}{\pi_i(t)} (Y_i-\mu(t)): {t\in\mathcal{T}} \right).  \label{eq:Sigma_haj}  
\end{equation}
Theorem \ref{thm:Hájek_asym_n} below will show that ${\Sigma}_\textup{haj}$ defined in \eqref{eq:Sigma_haj} is the asymptotic covariance of the Hájek estimator of $\mu$.
% which we will show in Theorem  \ref{thm:Hájek_asym_n} below is the asymptotic variance of the Hájek etimator without covariate adjustment.
% which represents the variance of the Horvitz--Thompson estimator defined on the demeaned outcome $Y_i - \sum_{t\in\mathcal{T}}{1}_i(t) \mu(t)$, scaled by $\sqrt{n}$.
% $ \operatorname{Cov}\left( \left\{ n^{-1 / 2} \sum_{i=1}^n (Y_i-\mu(t)) \frac{{1}i(t)}{\pi_i(t)} \right\}{t\in\mathcal{T}} \right)$, which serves as the analog to $\Sigma_n^2$ and is defined on the demeaned outcome $Y_i - \sum_{t\in\mathcal{T}}{1}_i(t) \mu(t)$.
Based on the definition of $\theta_{n, s}$ in \eqref{eq:theta} and \citet[][Theorem 1]{Leung2022},  
% [PD:  not sure what does ``from Theorem 1 in \cite{Leung2022}'' mean ]
we define
\begin{equation}
\tilde{\theta}_{n, s} = \theta_{n,\lfloor s / 2\rfloor} 1(s>2 \max \{K, 1\})+{1}(s \leq 2 \max \{K, 1\}) 
\label{eq:tildetheta}
\end{equation} 
%from Theorem 1 in \cite{Leung2022}, 
where $K$ is the constant from Assumption \ref{asu1} and $\lfloor s \rfloor$ is $s$ rounded down to the nearest integer.
Assumptions \ref{asu4} and \ref{asu6} both posit that interference diminishes with path distance. 
Additionally, Assumption \ref{asu6} imposes further that for some sequence $m_n$, $\tilde{\theta}_{n, s}$ diminishes to zero at a sufficiently rapid rate relative to the size of the $m_n$-neighborhood, moreover, constraints on the growth of $m_n$-neighborhoods, and ensures that $\tilde{\theta}_{n, m_n}$ decays at an adequately fast pace.
Moreover, Assumption \ref{asu6} is closely related to the conditions proposed in \cite{ChandrasekharJacksonMcCormick2024} to achieve the asymptotic normality of sums of dependent random variables.
% \cite{ChandrasekharJacksonMcCormick2024} introduce affinity sets to capture collections of random variables that may exhibit high correlation. 
The three components of Assumption \ref{asu6} below are analogous to their Assumptions 1--3.
% , which bound: (i) the total weighted covariance within affinity sets, (ii) the total weighted covariance across affinity sets, and (iii) the total weighted covariance contributed by variables outside the affinity sets.

\begin{assumption}[Weak Dependence for CLT]\label{asu6}
Recall $M_n(m,k)$, $\mathcal{H}_n(s, m)$ and ${\Sigma}_\textup{haj}$ defined in \eqref{eq:Mn}, \eqref{eq:Hn} and \eqref{eq:Sigma_haj}, respectively.
Define $\lambda_{\min}(\Sigma_{\textup{haj}})$ as the smallest eigenvalue of $\Sigma_{\textup{haj}}$.
There exist $\epsilon>0$ and a positive sequence $\{m_n\}_{n \in \mathbb{N}}$ such that as $n\rightarrow \infty$ we have $m_n \rightarrow \infty$ and 
\begin{align*}
\frac{n^{-2} \sum_{s=0}^n\left|\mathcal{H}_n(s, m_n)\right| \tilde{\theta}_{n, s}^{1-\epsilon}}{({\lambda_{\min}(\Sigma_{\textup{haj}})
})^2} \rightarrow 0,
\quad
\frac{n^{-1 / 2} M_n(m_n, 2)}{ (\lambda_{\min}(\Sigma_{\textup{haj}}))^{3/2}} \rightarrow 0, 
\quad 
\frac{ n^{3 / 2} \tilde{\theta}_{n, m_n}^{1-\epsilon}}{\sqrt{\lambda_{\min}(\Sigma_{\textup{haj}})
}} \rightarrow 0.
\end{align*}
% where the convergence of the matrices is element-wise. 
% $$
% \max \left\{
% {\Sigma}_\textup{haj}^{-2} n^{-2} \sum_{s=0}^n\left|\mathcal{H}_n\left(s, m_n\right)\right| \tilde{\theta}_{n, s}^{1-\epsilon}, 
% {\Sigma}_\textup{haj}^{-3/2} n^{-1 / 2} M_n\left(m_n, 2\right), {\Sigma}_\textup{haj}^{-1/2} n^{3 / 2} \tilde{\theta}_{n, m_n}^{1-\epsilon}
% \right\} \rightarrow 0.
% $$
\end{assumption}
Assumption \ref{asu6} corresponds to Assumption 3.4 of \cite{KojevnikovMarmerSong2019}, which limits the extent
of dependence across units of $1(T_i=t)\pi_i(t)^{-1}(Y_i-\mu(t))$'s through restrictions on the network. 
\citet[Section A.1]{Leung2022a} verifies Assumption \ref{asu6} for networks with polynomial or exponential neighborhood
growth rates.
We impose Assumption \ref{asu6} to ensure the asymptotic normality of the Hájek estimator of $\mu$. % without covariate adjustment. 
We defer Assumption \ref{asu7}, which ensures the consistency of covariance estimation, to Section \ref{sec:wls_cov}.

\section{Hájek estimator in network experiments}\label{sec:WLS}
% We define $\hat{Y}'_{\text{ht}}(t) = \frac{1}{n} \sum_{i=1}^n (Y_i-\mu(t)) \frac{{1}_i(t)}{\pi_i(t)}$ as the Horvitz--Thompson estimator defined on the centered outcome $Y_i-\mu(t)$. Let $\hat{Y}'_{\text{ht}}$ be the $|\mathcal{T}|\times 1$ vectorization of $\hat{Y}'_{\text{ht}}(t)$ and $\mu$ be the $|\mathcal{T}|\times 1$ vectorization of $\mu(t)$.
% The difference between the Hájek estimator and the true finite-population average equals 
% \[
% \hat{Y}_{\textup{haj}}(t) - \mu(t)
% = \frac{\hat{Y}_{\text{ht}}(t) - \hat{1}_{\text{ht}}(t)\mu(t)}{\hat{1}_{\text{ht}}(t)}  
% = \frac{\hat{Y}'_{\text{ht}}(t)}{\hat{1}_{\text{ht}}(t)}
% \]

\subsection{WLS-based point and covariance estimation}\label{sec:wls_cov}
Let $z_i = ({1}(T_i=t): t\in\mathcal{T})$ be the vector of exposure mapping indicators.
Motivated by the inverse probability weighting in the Hájek estimator, we consider the WLS fit:
% \[
% Y_i \sim  {1}_i(t_1) + \ldots + {1}_i(t_{|\mathcal{T}|})
% \]
% $Y_i \sim \sum_{t\in \mathcal{T}}{1}_i(t)$ 
% with inverse propensity score as weight $w_i = \sum_{t\in \mathcal{T}} \frac{{1}_i(t)}{\pi_i(t)}$. 
\begin{equation}
\text{regress } Y_i \text{ on } z_i \text{ with weights } w_i = 1/\pi_i(T_i).
% \text{lm}(Y_i \sim z_i) \text{ with weights } w_i = 1/\pi_i(T_i).
\label{eq:WLS}
\end{equation}
Let $\hat{\beta}_{\textup{haj}}$ denote the estimtors of coefficients for $z_i$ in \eqref{eq:WLS}.
%  resulting coefficient of $z_i$. 
Define the concatenated Hájek estimator vector as $\hat{Y}_{\textup{haj}} = (\hat{Y}_{\textup{haj}}(t): t\in\mathcal{T})$. 
% Proposition \ref{prop:haj} shows the numerical equivalence between $\hat{\beta}_{\textup{haj}}$ and $\hat{Y}_{\textup{haj}}$.  
% Let ${W}=\text{diag}\{ {w}_i: i = 1,\ldots, n\}$, ${Y}$ and $e_{\textup{haj}}$ be the vectorization of $\{{Y}_i: i=1,\ldots, n \}$ and $\{e_{i}: i=1,\ldots, n \}$ where $e_i$ is the residual from the above weighted least square regression.
% \begin{proposition} \label{prop:haj}
% $\hat{\beta}_{\textup{haj}} = \hat{Y}_{\textup{haj}}$. 
% % and $\hat{ {V}}_{\textup{haj}} = ( {D}^{\top}  {W}  {D})^{-1}
% % (D^{\top} {W} e_{\textup{haj}} W e_{\textup{haj}}^{\top} {W} D)
% % ( {D}^{\top}  {W}  {D})^{-1}$. 
% \end{proposition}
The numerical equivalence $\hat{\beta}_{\textup{haj}} = \hat{Y}_{\textup{haj}}$ is a well known result and shows the utility of WLS in reproducing the Hájek estimators \citep{AronowSamii2017, Ding2024}.
Theorem \ref{thm:Hájek_asym_n} below states the asymptotic normality of $\hat{\beta}_{\textup{haj}}$. 
% Let ${1}(t)$ be a vector of size $|\mathcal{T}|\times 1$, where all elements are zero except for the $t$th element, which is equal to one.
% Proposition \ref{ref: haj} is numeric and shows we can recover the Hájek estimator with $\hat{\tau}_{\textup{haj}}(t,t') = \hat{\beta}_{\textup{haj}}(t) - \hat{\beta}_{\textup{haj}}(t') = (d(t) - d(t'))^{\top}\hat{\beta}_{\textup{haj}}$.
% The regression-based approach delivers the estimator of the standard errors via the weighted least squares fit.
% Denote by $\hat{ {V}}_{\textup{haj}}$ the network robust variance for $\hat{\beta}_{\textup{haj}}$ from the above weighted least squares fit.
% We define $\hat{\Sigma}^2_{\textup{haj}}(t,t')$ as the network-robust variance for $\hat{\tau}_{\textup{haj}}(t,t')$ based on $\hat{ {V}}_{\textup{haj}}$, i.e., $\hat{\Sigma}^2_{\textup{haj}}(t, t^{\prime})) 
% =  (d(t) - d(t'))^{\top} 
% \hat{ {V}}_{\textup{haj}}
% (d(t) - d(t'))$. 
% Theorem \ref{thm:WLS_Hájek} below establishes the conservative of $\hat{\Sigma}^2_{\textup{haj}}(t,t')$ for the asymptotic variance of $\hat{\tau}_{\textup{haj}}(t,t')$.
% ${\Sigma}^2_\text{n,haj} = \operatorname{Var}\left( \frac{1}{\sqrt{n}} \sum_{i=1}^n \frac{{1}_i(t)}{\pi_i(t)} (Y_i-\mu(t)) \right)$.

\begin{theorem}\label{thm:Hájek_asym_n}
% [PD: $K_n$ is introduced later. i am not sure whether this is the right place to introduce $\hat{ {\Sigma}}_{*,\textup{haj}}$. theorem 3.1 is about clt. move all other results to theorem 3.2? also the theory depends on $b_n$. here we do not even have assumptions on $b_n$.]
Under Assumptions \ref{asu1}--\ref{asu6}, we have ${\Sigma}_{\textup{haj}}^{-1/2} 
\sqrt{n}
( \hat{\beta}_{\textup{haj}} -  {\mu} ) 
\stackrel{\textup{d}}{\rightarrow} \mathcal{N}(0,{I})$.
% \begin{align*}
% {\Sigma}_{\textup{haj}}^{-1/2} 
% \sqrt{n}
% \left( \hat{\beta}_{\textup{haj}} -  {\mu} \right) 
% \stackrel{\textup{d}}{\rightarrow} \mathcal{N}(0,{I}).
% \end{align*}
\end{theorem}
Theorem \ref{thm:Hájek_asym_n} ensures the consistency of $\hat{\beta}_{\textup{haj}}$ for estimating $\mu$ and establishes ${\Sigma}_{\textup{haj}}$
as the asymptotic sampling covariance of $\sqrt{n}
(\hat{\beta}_{\textup{haj}} -  {\mu})$.

% For researchers with a specific interest in the contrast between exposure mapping values $t$ and $t'$, set the coefficient matrix $G$ as follows: $G = (0, \ldots, 1, \ldots, -1, \ldots, 0)$.
% Then the Hájek estimator of $\tau(t,t')$ can be calculated as $\hat{\beta}_{\textup{haj}}(t) - \hat{\beta}_{\textup{haj}}(t') = G\hat{\beta}_{\textup{haj}}$. 
% Define 
% $ 
% \tilde{\Delta}_i(t,t') = \left( \frac{{1}_i(t)}{\pi_i(t)} (Y_i-\mu(t) ) - \frac{{1}_i(t')}{\pi_i(t')} (Y_i-\mu(t') ) \right) 
% $
% and ${\Sigma}_{\textup{haj}}(t,t') = \operatorname{Var}\left( \frac{1}{\sqrt{n}} \sum_{i=1}^n \tilde{\Delta}_i(t,t') \right)$.
% \begin{corollary}\label{cor: haj}
% Under Assumptions \ref{asu1}-\ref{asu4} and \ref{asu6},
%  we have
% \begin{align*}
% {\Sigma}_{\textup{haj}}^{-1}(t,t') 
% \sqrt{n}
% \left( \hat{\tau}_{\textup{haj}}(t,t') - \tau(t,t') \right) 
% \stackrel{\textup{d}}{\rightarrow}  \mathcal{N}(0,1 ). 
% \end{align*}
% \end{corollary}
% \subsection{WLS-based covariance estimation}\label{sec:WLS_Cov}
The regression-based approach provides an estimator for the standard error via the same WLS fit. 
Denote the design matrix of the WLS fit in \eqref{eq:WLS} by an $n\times |\mathcal{T}|$ matrix ${Z} = (z_1,\ldots,z_n)^\top$, where its rows are the vectors $z_i$ for each unit $i\in\mathcal{N}_n$.  Construct the weight matrix ${W}=\text{diag}\{ {w}_i: i = 1,\ldots, n\}$ by placing the weights $w_i$ along the diagonal. Let $ {Y} = (Y_1,\ldots, Y_n)$ denote the vector of the observed outcomes.
% The residual from the above WLS fit is $e_i  
% = {Y}_i - \sum_{t\in \mathcal{T}} {1}_i(t) \hat{\beta}_{\textup{haj}}(t)$.
Diagonalize the residual $e_i$'s from the same WLS fit to form the matrix $ {e}_{\textup{haj}} = \text{diag}\{e_i: i=1,\ldots,n\}$.
%We denote by $\hat{ {V}}_{\textup{haj}}$ the network-robust variance estimator of $\hat{\beta}_{\textup{haj}}$:
Define 
\begin{equation}
\hat{ {V}}_{\textup{haj}}
= ( {Z}^{\top} {W} {Z})^{-1}
( {Z}^{\top} {W} 
{e}_{\textup{haj}} 
{K}_n 
{e}_{\textup{haj}}
{W}  {Z})
( {Z}^{\top} {W} {Z})^{-1}
\label{eq:V_hat}
\end{equation}
as the network-robust covariance estimator of $\hat{\beta}_{\textup{haj}}$, where $ {K}_n$ is a uniform kernel matrix with $(i,j)$th entry $K_{n,ij}={1}(\ell_{ {A}}(i, j)\le b_n)$. 
Here, choosing $b_n > 0$ places nonzero weight on pairs at most path distance $b_n$ apart from each other in the network $A$, which accounts for the network correlation. 
While \eqref{eq:V_hat} adopts the form of an HAC estimator commonly used in spatial econometrics literature, our paper first discusses its design-based properties under the regression-based analysis for network experiments. 
% Clustered standard errors are frequently used to account for network dependence \citep{EcklesKizilcecBakshy2016, AralZhao2019, Zacchia2020, AbadieAtheyImbens2023}. However, determining the optimal way to partition a network into clusters for inference can be challenging. \cite{Leung2023} establishes conditions for the validity of cluster-robust methods under network dependence. 
% \cite{KojevnikovMarmerSong2021} provides a law of large numbers and a central limit theorem for network dependent variables. Additionally, they introduce a technique for computing standard errors that remains robust when confronted with various types of network dependencies. Their approach utilizes a network-based variance estimator and demonstrates the consistency of the variance estimator to the true sampling variance. \cite{Leung2022} proposes a variance estimator for the Horvitz--Thompson estimator of exposure effects. While their variance estimators share a resemblance to the HAC estimator in terms of structure, neither of them directly originates from a regression approach.

We follow the discussion in \cite{Leung2022} regarding the choice of the bandwidth $b_n$. 
Define the average path length, $\mathcal{L}({A})$, as the average value of $\ell_{{A}}(i, j)$ over all pairs in the largest component of ${A}$.
Here, a component of a network refers to a connected subnetwork where all units within the subnetwork are disconnected from those outside of it.
Let $\delta( {A})=$ $n^{-1} \sum_{i=1}^n \sum_{j=1}^n A_{i j}$ be the average degree. \cite{Leung2022} suggests choosing the bandwidth $b_n$ as follows:
\begin{equation}
b_n=\left\lfloor\max \left\{\tilde{b}_n, 2 K\right\}\right\rceil \quad \text { where } \tilde{b}_n= \begin{cases}\frac{1}{2} \mathcal{L}( {A}) & \text { if } \mathcal{L}( {A})<2 \frac{\log n}{\log \delta( {A})}, \\ \mathcal{L}( {A})^{1 / 3} & \text { otherwise, }\end{cases} 
\label{eq:bandwidth}    
\end{equation}
where $\lfloor\cdot\rceil$ means rounding to the nearest integer.  
The choice of bandwidth $b_n$ is based on the following two reasons. 
First, $b_n$ is set to be at least equal to $2K$ to account for the correlation in $\{1(T_i=t)\}_{i=1}^n$ as per Assumption \ref{asu1}.
If the exposure mapping is correctly specified, we can simply choose $b_n=2K$. 
Second, \eqref{eq:bandwidth} chooses a bandwidth of logarithmic or polynomial order depending on the growth rates of the average $K$-neighborhood size. The logarithmic order in $b_n$ applies when the growth rate is approximately exponential in $K$ and polynomial order applies when the growth rate is approximately polynomial in $K$.
% a bandwidth of logarithmic (polynomial) order when neighborhood growth rates are approximately exponential (polynomial).
Furthermore, \cite{Leung2022} justifies that the bandwidth in \eqref{eq:bandwidth} satisfies Assumption \ref{asu7}(b)–(d) under polynomial and exponential neighborhood growth rates.
Since $K$ is researcher-defined, and $\mathcal{L}({A})$ and $\delta( {A})$ can be computed from the observed network data, $b_n$ in \eqref{eq:bandwidth} can be determined accordingly. 
To align with \cite{Leung2022}, we also recommend that researchers report results for multiple bandwidths in a neighborhood of \eqref{eq:bandwidth} as a robustness check. 
We use the empirical application in Section \ref{app:paluck} as an illustrative example to demonstrate how to select the bandwidth.

% \begin{remark}
% % [PD: this paragraph is somewhat distracting. include it in a Remark?]
% Clustered standard errors are also frequently used to account for network dependence \citep{EcklesKizilcecBakshy2016, AralZhao2019, Zacchia2020, AbadieAtheyImbens2023}. However, determining the optimal way to partition a network into clusters for inference can be challenging. 
% \cite{Leung2023} establishes the conditions that validate cluster-robust methods under network dependence, with a focus on a small number of clusters.
% \cite{VivianoLeiImbens2024} studies the design of cluster experiments to estimate the global treatment effect under interference.
% \end{remark}

We impose Assumption \ref{asu7}, as introduced in \citet[Assumption 7]{Leung2022}, to ensure the consistency of the covariance estimator, where $b_n$ is the bandwidth defined in \eqref{eq:bandwidth}.  
Denote by
\[
\mathcal{N}^{\partial}(i, s; {A})=\{j \in \mathcal{N}_n: \ell_{ {A}}(i, j)=s\}
\] 
the $s$-neighborhood boundary of unit $i$, which is the set of units exactly at a distance of $s$ from $i$, and
\[
M_n^{\partial}(s)=n^{-1} \sum_{i=1}^n |\mathcal{N}^{\partial}(i, s; {A})|,
\]
its average size across units. 
% [PD: not sure whether this is the right place to introduce the assumption. nobody knows what is $b_n$ at this point.]
\begin{assumption}
% [For Consistency of Covariance Estimator]
\label{asu7}
(a) $\sum_{s=0}^n M_n^{\partial}(s) \tilde{\theta}_{n, s}^{1-\epsilon}=O(1)$ for some $\epsilon>0$, (b) $M_n(b_n, 1)=o(n^{1/2})$, (c) $M_n(b_n, 2)=o(n)$, (d) $\sum_{s=0}^n|\mathcal{J}_n(s, b_n)| \tilde{\theta}_{n, s}=o(n^2)$. 
\end{assumption}
Assumption \ref{asu7}(a) demonstrates the trade-off between restrictions on the network topology through $M_n^{\partial}(s)$ and the degree of interference through $\tilde{\theta}_{n, s}$.
Assumption \ref{asu7}(b) and (d) regulate the bandwidth $b_n$ by imposing conditions on the first and second moments of the $b_n$-neighborhood size within network $A$.
Assumption \ref{asu7}(d) is used to derive the asymptotic bias, which closely mirrors  Assumption \ref{asu6} with $b_n$ and $\mathcal{J}_n(s, \cdot)$ in place of $m_n$ and $\mathcal{H}_n(s, \cdot)$, respectively. 
Assumption \ref{asu7} strongly depends on the structure of the underlying network. 
\citet[Appendix A.2]{Leung2022} uses a mixture of formal and heuristic arguments to show that the bandwidth $b_n$ in \eqref{eq:bandwidth} satisfies Assumption \ref{asu7}(b)--(d) for networks with polynomial or exponential neighborhood growth rates.

Define ${\Delta}_{\textup{haj}}$ as an $n \times |\mathcal{T}|$ matrix
with $(i,t)$th element ${\Delta}_{\textup{haj},it} = 1(T_i=t)\pi_i(t)^{-1}(Y_i - \mu(t)) - (\mu_i(t) - \mu(t) )$, and $M$ as an $n \times |\mathcal{T}|$ matrix with $(i,t)$th element $M_{it} = \mu_i(t)-\mu(t)$.
Of interest is how this regression-based covariance
estimator approximates the true sampling covariance from the design-based perspective.
\begin{theorem}\label{thm:Hájek_bias}
Define ${ {\Sigma}}_{*,\textup{haj}} = n^{-1} {\Delta}_{\textup{haj}}^\top  {K}_n   {\Delta}_{\textup{haj}}$ and $R_{\textup{haj}} = n^{-1} M^\top  {K}_n M$. 
Under Assumptions \ref{asu1}--\ref{asu4} and \ref{asu7}, we have ${\Sigma}_{*, \textup{haj}}
= {\Sigma}_{\textup{haj}}
+ o_\mathbb{P}(1)$ and $n \hat{ {V}}_{\textup{haj}}
= {\Sigma}_{*, \textup{haj}} + R_{\textup{haj}} + o_\mathbb{P}(1)$.
% \begin{align}
% {\Sigma}_{*, \textup{haj}}
% =&~ {\Sigma}_{\textup{haj}}
% + o_\mathbb{P}(1), \label{eq:oracle} \\
% n \hat{ {V}}_{\textup{haj}}
% =&~  {\Sigma}_{*, \textup{haj}} + R_{\textup{haj}} + o_\mathbb{P}(1) \label{eq:bias}. 
% \end{align}
% where $\hat{ {\Sigma}}_{*,\textup{haj}} = n^{-1}  {\Delta}_{\textup{haj}}^\top  {K}_n   {\Delta}_{\textup{haj}}$ and $R_{\textup{haj}} = n^{-1} M^\top  {K}_n M$. 
\end{theorem}
% It is computed by taking the average of the product of the difference between the observed treatment effect for each individual and their individual treatment effect, and then weighting it by an indicator function that accounts for the distance between individuals.
% In this case, the estimator is demeaned at the individual level mean, meaning that it takes into account the average treatment effect for each individual.
We use $_*$ to indicate that ${\Sigma}_{*, \textup{haj}}$ is the ``oracle'' version of covariance estimator, which takes the form of a HAC estimator.
% ${\Sigma}_{*, \textup{haj}}$ centers around $\mu_i(t) - \mu(t)$, the individual-level deviation of the expected response from the population average under exposure mapping value $t$. 
Theorem \ref{thm:Hájek_bias} first demonstrates that ${\Sigma}_{*, \textup{haj}}$ closely approximates the asymptotic covariance ${\Sigma}_{\textup{haj}}$ and then presents the asymptotic bias of the network-robust covariance estimator in estimating ${\Sigma}_{*, \textup{haj}}$. 
The bias term $R_{\textup{haj}}$ adopts the form of an HAC covariance estimator of the individual-level expected response. The covariance estimation is asymptotically exact with constant individual-level expected response under any exposure mapping value $t\in\mathcal{T}$, which is similar to the canonical results of \cite{Neyman1923} without interference. 
% [PD: citation]. 
In some cases, the uniform kernel used in the network-robust covariance estimator $\hat{{V}}_{\textup{haj}}$ may not be positive semi-definite. This issue can result in an anti-conservative covariance estimator, which can in turn affect the accuracy of hypothesis testing and confidence intervals.  
We will address this issue in the next subsection. Now we end this subsection with a remark on the literature of HAC covariance estimators for network and spatial data. 

\begin{remark}
\cite{AronowSamii2017} studied under the assumption of correctly specified exposure mappings and focused on the Horvitz--Thompson estimator for causal effects. They also discussed the Hájek estimator and its WLS formulation. However, they did not establish the result that justifies the corresponding network HAC estimator from WLS fits, which is easy to implement for applied researchers. 
\citet[Appendix B]{Leung2022} compares his variance estimator to that of \citet{AronowSamii2017}, showing that while the bias terms are not generally ordered, his estimator has a smaller bias in the special case of no interference and homogeneous unit-level exposure effects. He also provides simulation evidence that \citet{AronowSamii2017}’s estimator can exhibit larger bias under a simple model of interference.
\end{remark}

\begin{remark} 
Another related literature strand pertains to the application of HAC estimator in spatial econometrics \citep{Andrews1991, Conley1999, Matyas1999, KelejianPrucha2007, KimSun2011}. 
\cite{WangSamiiChang2025} discussed the usage of regression estimators for causal effects from the design-based perspective and showed that the spatial HAC estimator provided asymptotically conservative inference under certain assumptions. 
Neither \cite{AronowSamii2017} nor \cite{WangSamiiChang2025} discussed how to increase efficiency by incorporating covariate information, which will be our focus in Section \ref{sec:covadju}.  
% [PD: polish this remark. the technical comments only make sense after we introduce the technical results.]
\cite{XuWooldridge2022}  
recommended using spatial HAC standard errors to account for spatial correlation.
% in two cases when the sampling probability is non-negligible: 
% (i) assignment variables exhibit spatial correlation, 
% or (ii) spillover effects are estimated in the model. 
Because the exposure mappings are not independent across units in network experiments, we use network HAC standard errors to take care of dependence when estimating exposure effects, which is the estimand of interest. 
% [PD: i am sure what you want to say about their paper? they did.. .. and what's new in our paper?]
\end{remark}
% \begin{remark}
% \cite{AronowSamii2017} proposes a covariance estimator that is only valid under correctly specified exposure mappings.
% \end{remark}
\subsection{Improvement on covariance estimation} \label{sec:modification}

% The bias term $R_{\textup{haj}}$ cannot be guaranteed non-negative. 
% Notice that $\frac{1}{n} \sum_{i=1}^n \sum_{j=1}^n
% (\tau_i(t,t') - \tau(t,t'))
% (\tau_j(t,t') - \tau(t,t')) 
% = 0$, which can be decomposed as $0 = R_n + R_n^C$ with $R_n^C \equiv \frac{1}{n} \sum_{i=1}^n \sum_{j=1}^n
% (\tau_i(t,t') - \tau(t,t'))
% (\tau_j(t,t') - \tau(t,t'))
% {1}\left\{\ell_{ {A}}(i, j) > b_n \right\}$. Hence, $\hat{\Sigma}^2_{\textup{haj}}(t,t')$ is guaranteed to be conservative when $R_n^C < 0$. 

% In some cases, the truncated kernel used in network-robust variance estimation may not be positive semi-definite. This can lead to a variance estimator that is non-positive and anti-conservative, which can affect the accuracy of hypothesis testing and confidence interval construction.
There are four main concerns regarding the properties of the HAC variance estimator. 
First, it should ideally be non-negative in finite-samples, despite the kernel not always being positive semi-definite. 
Second, the HAC estimator is biased in a design-based setting, and it is desirable for the bias term to be asymptotically non-negative to ensure conservative inference. 
Third, HAC estimators often yield values that are too small in finite-samples compared with the true variance, leading to false discoveries. 
Finally, a computationally feasible bandwidth sequence is necessary for ensuring the consistency of the HAC estimator.

In this subsection, we tackle these issues by proposing a modification to the uniform kernel.
Our proposed modification preserves the network-robustness of the covariance estimator while ensuring that it remains positive semi-definite and conservative.
Let $Q_n \Lambda_n Q_n^{\top}$ be the eigendecomposition of $ {K}_n$. As ${K}_n$ is symmetric, all its eigenvalues are real. 
% Also let $\underline{\lambda}(A)$ denote the smallest eigenvalue of $A$, e.g., $\underline{\lambda}\left(\hat{\Sigma}_n\right)=\min _{1 \leq k \leq v} \Lambda_n$. 
% Consider a sequence of small positive real numbers $c_n \searrow 0$. 
We define the adjusted kernel matrix by truncating the negative eigenvalues at $0$ as ${K}_n^{+}:= Q_n \max\{ \Lambda_n, 0\} Q_n^{\top}$,
% \[
% {K}_n^{+}:= Q_n \max\{ \Lambda_n, 0\} Q_n^{\top},
% %= Q_n \Lambda_n Q_n^{\top}
% %+ Q_n |\min \{ \Lambda_n, 0\}| Q_n^{\top}
% %=  {K}_n +  {K}_n^{-},
% \]
where the maximum is taken element-wise.  Letting ${K}_n^{-} := Q_n |\min \{ \Lambda_n, 0\}| Q_n^{\top}$ with the minimum taken element-wise, we can also write $ {K}_n^{+} 
=  {K}_n +  {K}_n^{-}$.  
% [PD: check this part. the original presentation was too dense.]
By construction, the matrix $ {K}_n^{\diamond}$ ($\diamond=+,-$) is positive semi-definite, and we denote the $(i,j)$th entry of $ {K}_n^{\diamond}$ as ${K}_{n,ij}^{\diamond}$. 
If \( K_n \) were positive semi-definite, then $K_n = {K}_n^+ $. 
% [PD: this seems a confusing notation. why not just use ${K}_{n,ij}^{\diamond}$?] 
We propose the adjusted HAC covariance estimator as
\begin{equation}
\hat{ {V}}_{\textup{haj}}^+
= ( {Z}^{\top}  {W}  {Z})^{-1}
( {Z}^{\top}  {W} 
{e}_{\textup{haj}} 
{K}_n^+ 
{e}_{\textup{haj}}   {W}  {Z})
( {Z}^{\top}  {W}  {Z})^{-1}.  
\label{V+}
\end{equation}
% Moreover, in the case when the smallest eigenvalue of $\Sigma_n$ is bounded from below
% We propose the adjusted variance estimator as
% \[
% \hat{\Sigma}^2_{{\textup{haj}},+}(t,t') 
% = \frac{1}{n} \hat{\mathbf{\Delta}}_{\textup{haj}}(t,t')^{\top}  {K}_n^{+} \hat{\mathbf{\Delta}}_{\textup{haj}}(t,t').
% \]
% The adjusted variance estimator, denoted by $\hat{\Sigma}_{\textup{haj},+}^{2}(t,t')$, is defined as
% \[
% \hat{\Sigma}_{\textup{haj},+}^{2}(t,t') 
% = \frac{1}{n} \left(\mathbf{\Delta}(t,t')-\hat{\tau}(t, t^{\prime}))\right)^{\top}  {K}_n^{+} \left(\mathbf{\Delta}(t,t')-\hat{\tau}(t, t^{\prime}))\right).
% \]
To guarantee the asymptotic conservativeness of $\hat{ {V}}_{\textup{haj}}^+$, we impose Assumption \ref{asu8} below, which pertains to the properties of $K_n^{-}$.
Recall that $K_{n,ij} = {1}(\ell_{ {A}}(i, j)\le b_n)$ and write $M_n(m,k)$ and $\mathcal{J}_n(s, m)$ in \eqref{eq:Mn} and \eqref{eq:Jn} with $m=b_n$ as:
\begin{eqnarray*}
M_n(b_n,k) 
&=& \frac{1}{n} \sum_{i=1}^n \left( \sum_{j=1}^n K_{n,ij} \right)^k \\
\mathcal{J}_n(s, b_n)
&=& \sum_{i=1}^n \sum_{j=1}^n 
{1}(\ell_{ {A}}(i, j)=s) 
\cdot
\sum_{k=1}^n K_{n,ik}
\cdot 
\sum_{l=1}^n K_{n,jl}.
\end{eqnarray*}
% \begin{align*}
% M_n(b_n,k) 
% =& \frac{1}{n} \sum_{i=1}^n \left( \sum_{j=1}^n K_{n,ij} \right)^k
% \end{align*} 
% and
% \begin{align*}
% \mathcal{J}_n(s, b_n)
% =& \sum_{i=1}^n \sum_{j=1}^n 
% {1}(\ell_{ {A}}(i, j)=s) 
% \cdot
% \sum_{k=1}^n K_{n,ik}
% \cdot 
% \sum_{l=1}^n K_{n,jl}.
% \end{align*}
Define $M_n^{-}(b_n,k)$ and $\mathcal{J}^{-}_n(s, b_n)$ as the counterparts of $M_n(b_n,k)$ and $\mathcal{J}_n(s, b_n)$ on $|K_n^-|$, respectively:
\begin{eqnarray*}
M_n^{-}(b_n,k) 
&=& \frac{1}{n} \sum_{i=1}^n \left( \sum_{j=1}^n 
\left| {K}_{n,ij}^- \right| \right)^k \\
\mathcal{J}^{-}_n(s, b_n)
&=& \sum_{i=1}^n \sum_{j=1}^n 
{1}(\ell_{ {A}}(i, j)=s) 
\cdot
\sum_{k=1}^n \left| {K}_{n,ik}^- \right|
\cdot 
\sum_{l=1}^n \left| {K}_{n,jl}^- \right|.
\end{eqnarray*}
% \[
% M_n^{-}(b_n,k) = \frac{1}{n} \sum_{i=1}^n \left( \sum_{j=1}^n 
% \left| {K}_{n,ij}^- \right| \right)^k
% \]
% and
% \[
% \mathcal{J}^{-}_n(s, b_n)
% = \sum_{i=1}^n \sum_{j=1}^n 
% {1}(\ell_{ {A}}(i, j)=s) 
% \cdot
% \sum_{k=1}^n \left| {K}_{n,ik}^- \right|
% \cdot 
% \sum_{l=1}^n \left| {K}_{n,jl}^- \right|.
% \]
% [PD: the formula of J seems different from before. maybe we need to give a similar formula before to make the presentation more symmetric?]
Assumption \ref{asu8} is the analogue of Assumption \ref{asu7}, but specifically tailored to the quantity $|K_n^-|$, with Assumption \ref{asu8}(a) identical to Assumption \ref{asu7}(a).

\begin{assumption}
\label{asu8}
(a) $\sum_{s=0}^n M_n^{\partial}(s) \tilde{\theta}_{n, s}^{1-\epsilon}=O(1)$ for some $\epsilon>0$,
(b) $M_n^{-}(b_n, 1)=o(n^{1 / 2})$,
(c) $M_n^{-}(b_n, 2)=o(n)$, 
(d) $\sum_{s=0}^n|\mathcal{J}^{-}_n(s, b_n)| \tilde{\theta}_{n, s}=o(n^2)$.
\end{assumption}

% Define $ {\Delta}_{\textup{haj}}(t, t^{\prime})$  as the stacked vector of $\left( \tilde{\Delta}_i(t,t') - \left(\tau_i(t, t^{\prime})-\tau(t, t^{\prime})\right) \right)$. 
% Define $\overline{ \tau}(t,t')$ as the stacked vector of $\tau_i(t, t^{\prime}))-\tau(t, t^{\prime})$ and $\hat{\mathbf{\Delta}}_{\textup{haj}}(t,t')$ as the stacked vector of
% \[
% \frac{\frac{{1}_i(t)}{\pi_i(t)}}{\hat 1_{\text{ht}}(t)} 
% (Y_i - \hat{\beta}_{\textup{haj}}(t))   
% - \frac{\frac{{1}_i(t')}{\pi_i(t')}}{\hat 1_{\text{ht}}(t')}   
% (Y_i- \hat{\beta}_{\textup{haj}}(t')). 
% \]
% Let $R_{\textup{haj}}^+ = \frac{1}{n} M^\top  {K}_n^{+}  M + \frac{1}{n}  {\Delta}_{\textup{haj}}^\top  {K}_n^{-}  {\Delta}_{\textup{haj}}$.

\begin{theorem} \label{thm:Hájek_adj}
Define $R_{\textup{haj}}^+ = n^{-1} M^\top  {K}_n^{+}  M 
+ n^{-1} {\Delta}_{\textup{haj}}^\top  {K}_n^{-}  {\Delta}_{\textup{haj}} \geq 0$.
Under Assumptions \ref{asu1}--\ref{asu4} and \ref{asu8}, we have $n \hat{ {V}}_{\textup{haj}}^+
= \hat{\Sigma}_{*,{\textup{haj}}} + R_{\textup{haj}}^+ + o_\mathbb{P}(1)$,
% \[
% n \hat{ {V}}_{\textup{haj}}^+
% = \hat{\Sigma}_{*,{\textup{haj}}} + R_{\textup{haj}}^+ + o_\mathbb{P}(1),
% %\ge \hat{\Sigma}_{*,{\textup{haj}}} + o_\mathbb{P}(1),
% \]
% [PD: i modified the presentation here. you simply want to say that R is nonnegative? using inequalities with the o-p notation can be confusing.]
where $\hat{\Sigma}_{*,{\textup{haj}}}$ is defined in Theorem \ref{thm:Hájek_bias}.
% where 
% \begin{align*}
% R_{\textup{haj}}^+(t,t')
% =& \frac{1}{n} 
% \overline{ \tau}(t,t')^{\top} 
%  {K}_n^{+} 
% \overline{ \tau}(t,t') 
% + \frac{1}{n}
%  {\Delta}_{\textup{haj}}(t, t^{\prime})^{\top}
%  {K}_n^{-} 
%  {\Delta}_{\textup{haj}}(t, t^{\prime})    
% \end{align*}
\end{theorem}
Theorem \ref{thm:Hájek_adj} delineates two key advantages stemming from the construction of the adjusted covariance estimator. First, it ensures that the covariance estimator $\hat{ {V}}_{\textup{haj}}^+$ is positive definite. Second, it produces a positively adjusted bias term $R_{\textup{haj}}^+$, leading to the conservativeness of $\hat{ {V}}_{\textup{haj}}^+$ for estimating the true sampling covariance.
Theorems \ref{thm:Hájek_asym_n} and \ref{thm:Hájek_adj} together justify the regression-based inference of $\tau = G\mu$ from the WLS fit \eqref{eq:WLS} with the point estimator $\hat{\tau} = G\hat{\beta}_{\textup{haj}}$ and the adjusted regression-based HAC covariance estimator $G \hat{ {V}}_{\textup{haj}}^+ G^\top$. 

\begin{remark}\label{justification}
It remains unclear what restrictions on the network topology would ensure that Assumption \ref{asu8} holds when using the bandwidth choice $b_n$ in \eqref{eq:bandwidth}. We leave this as an open question, including whether alternative bandwidth choices could satisfy Assumption \ref{asu8} for certain classes of network structures. In Appendix \ref{app:asu8}, we provide some numerical justification that Assumption \ref{asu8} holds under the choice $b_n$ in \eqref{eq:bandwidth} for two network models.
% When the exposure mapping is correct, then we can choose $b_n = 2K$. So this technical issue only exists with misspecified exposure mapping.
\end{remark}

\subsection{Discussion on other covariance estimation strategies}
% [To be edited]
In this subsection, we briefly discuss other covariance estimation strategies. 
% \cite{WangSamiiChang2023} justified the usage of regression estimators in the spatial setting from the design-based perspective and provide causal interpretations for the coefficients. 
\cite{KojevnikovMarmerSong2021} provides a law of large numbers and a central limit theorem for network dependent variables. Additionally, they introduce a technique for computing standard errors that remains robust under various types of network dependencies. 
Their approach relies on a network HAC covariance estimator for a broad class of kernel functions, which they show consistently estimates the true sampling covariance.
As demonstrated in \citet[Remark 1]{Leung2022}, the uniform kernel provides better size control, especially in cases with smaller samples, compared with alternative kernels that diminish with distance. 
Considering these reasons, we opt for the uniform kernel. 
% In Section \ref{Simulation}, we report the HAC covariance estimators using the kernels in \cite{Leung2019d} and \cite{Kojevnikov2021} and illustrate their poor finite-sample performance. 

\citet{Leung2022} proposes a covariance estimator for the Horvitz--Thompson estimator of exposure effects, while \citet{Kojevnikov2021} develops bootstrap-based alternatives to network HAC estimation. Although both estimators share similarities with the HAC framework, neither is derived from a regression-based approach. \citet{Kojevnikov2021} ensures that the resulting estimator is positive semi-definite, and \citet{Leung2019e} refines this approach by showing that, under a specific bandwidth choice, the variance estimator exhibits non-negative asymptotic bias. However, both methods suffer from substantial overrejection in finite-sample simulations. 
We compare the finite-sample performance of our estimator with those of \cite{Leung2019e} and \cite{Kojevnikov2021} in Section \ref{Simulation}.
 
\citet{Leung2022} and our regression-based HAC estimator \( \hat{V}_{\textup{haj}} \) both use the uniform kernel, which helps mitigate overrejection in finite-samples. 
As shown in \citet[Appendix A]{Leung2022a}, \citet{Leung2022}'s variance estimator is asymptotically conservative under mild weak dependence conditions on the super-population. 
However, the non-positive semi-definiteness of the uniform kernel can lead both estimators to produce negative variance estimates in finite-samples, resulting in potential anti-conservativeness in both asymptotic theory and simulations.
The idea of replacing the negative eigenvalues of $K_n$ with non-negative values appeared in \citet[Appendix B]{Kojevnikov2021}, which can be traced back to the literature on approximating a symmetric matrix by a positive definite matrix \citep{Higham1988,Politis2009}. The key distinction is that \cite{Kojevnikov2021} applied this technique to the final HAC covariance estimator, while we apply it to the kernel matrix. There are two limitations of \cite{Kojevnikov2021}'s approach. First, it is not suitable for estimating a single causal effect, as when the HAC estimator is scalar, it merely involves replacing a negative variance estimate with zero.
% [PD: just 0?]
In contrast, our approach is applicable to joint causal effects. Second, \cite{Kojevnikov2021}'s approach does not address the issue of anti-conservativeness, as the crucial factor for positive bias is the positive semi-definiteness of $K_n$.
\cite{WangSamiiChang2025} recently applied our strategy to the HAC variance estimator in the spatial experiments and found better finite-sample properties. 

\section{Regression-based covariate adjustment}
\label{sec:covadju}
% However, they did not employ the regression coefficients as point estimates or use regression-associated standard errors for inference. Moreover, they did not discuss the design-based properties of the network HAC estimator with covariate adjustment.  

\subsection{Background: covariate adjustment without interference}  
Regression-based methods offer a natural framework for incorporating covariates and can lead to efficiency gains under appropriate conditions.\footnote{\citet{AronowSamii2017} discussed the use of covariates to improve efficiency via difference estimators, although they did not implement this approach in their analysis.} 
To set the stage for our discussion, we briefly review the theory of covariate adjustment under complete randomization without interference.

Consider an experimental setup involving a binary intervention and a population of $n$ units with potential outcomes denoted by $Y_i(0)$ and $Y_i(1)$ for each unit $i = 1, \ldots, n$. The average treatment effect within the finite population is denoted by $\tau(1,0) = \bar{Y}(1) - \bar{Y}(0)$, where $\bar{Y}(z) = n^{-1}\sum_{i=1}^n Y_i(z)$ for $z = 0, 1$.
% [PD: little z to be coherent]
Denote by $z_i$ the treatment indicator of unit $i$ under complete randomization. 
The difference-in-means estimator is unbiased for $\tau(1,0)$, and equals the coefficient of $z_i$ from the Ordinary Least Squares (OLS) regression of $Y_i$ on $(1, z_i)$.
%  $\text{lm}(Y_i \sim 1+z_i)$. 
Given the covariate vector $x_i = (x_{i1}, \ldots, x_{iJ})$ for $i = 1, \ldots, n$, 
\cite{Fisher1935} proposed to use the coefficient of $z_i$ from the OLS fit of regressing $Y_i$ on $(1,z_i,x_i)$ 
% $\text{lm}(Y_i \sim 1 + z_i + x_i)$ 
to estimate $\tau(1,0)$. 
\cite{Freedman2008} criticized this approach, highlighting its potential for efficiency loss compared to the difference-in-means estimator.
\cite{Lin2013} introduced an improved estimator, defined as the coefficient of $z_i$ obtained from the OLS regressing of $Y_i$ on $(1,z_i,(x_i-\bar{x}), z_i(x_i-\bar{x}))$.  This specification includes covariates as well as treatment-covariate interactions.
He proved that this estimator is at least as efficient as the difference-in-means and \cite{Fisher1935}'s estimators in the asymptotic sense.
% $\text{lm}(Y_i \sim 1 + z_i + (x_i-\bar{x}) + z_i(x_i-\bar{x}))$ 

We refer to the regression proposed by \cite{Fisher1935} as the additive specification, and \cite{Lin2013}'s regression as the fully-interacted specification to avoid any ambiguity.
We expand upon their findings in the context of network experiments, which incorporate interference, through the utilization of WLS fits.
% We will focus on the additive and fully-interacted specifications, study the design-based properties of the estimators, and compare their efficiency gains over the unadjusted counterparts.
To simplify the presentation, we center the covariates at $\bar{x} = n^{-1} \sum_{i=1}^n x_i = 0$.
%  which does not complicate the theory in the design-based framework with fixed covariates. 

\subsection{Additive regression in network experiments}\label{sec:add}
% Define $d_i = \{{1}_i(t)\}_{t\in\mathcal{T}}$ and $D$ being the matrix with rows $d_i$. Let $W=diag\{{w}_i, i = 1,\ldots, n\}$. 
Recall $z_i = (1(T_i=t): t\in \mathcal{T})$ as the dummies for the exposure mapping in the network experiment.  
% [PD: i added this sentence because $z_i$ appears as the binary indicator in the section above.  see my next comment.]
Consider the WLS fit  
\begin{equation}
\text{regress } Y_i \text{ on } (z_i, x_i) \text{ with weights }    
w_i = 1/\pi_i(T_i) 
\label{eq:add}.
\end{equation}
% $Y_i \sim z_i + x_i$ with weight $w_i = \sum_{t\in \mathcal{T}} \frac{{1}_i(t)}{\pi_i(t)}$. 
% To simplify the presentation, we center the covariates to have $\bar{x} = \frac{1}{n}\sum_{i=1}^n x_i = 0$. 
% Let $\hat{x}(t)$ be the Horvitz–-Thompson estimator of $\bar{x}$ based on units under treatment $t$.
Let $\hat{\beta}_{\textup{haj},\textsc{f}}$ denote the estimtors of coefficients for $z_i$ from the above WLS fit and $\hat{\beta}_{\textup{haj},\textsc{f}}(t)$ denote the element in $\hat{\beta}_{\textup{haj},\textsc{f}}$ corresponding to $1(T_i=t)$. 
We use the subscript ``F'' to signify \cite{Fisher1935}.
Assumption \ref{asu9} below imposes the uniform boundedness of $x_i$ and adapts Assumption \ref{asu6} to its version with covariate adjustment. 
% Specifically, we denote by $\hat{ {V}}_{\textup{haj},\textsc{f}}$ the network-robust variance for $\hat{\beta}_{\textup{haj},\textsc{f}}$ obtained from this procedure. We define $\hat{\Sigma}^2_{\textup{haj},\textsc{f}}(t, t^{\prime}))$ as the network-robust variance for $\hat{\tau}_\textup{haj}\left(t, t^{\prime},\hat{\gamma}_\textsc{f}\right)$ based on $\hat{ {V}}_{\textup{haj},\textsc{f}}$, calculated as $(d(t) - d(t'))^{\top} \hat{ {V}}_{\textup{haj},\textsc{f}} (d(t) - d(t'))$. Theorem \ref{thm:add} below establishes the conservative of $\hat{\Sigma}^2_{\textup{haj},\textsc{f}}(\hat{\gamma}_\textsc{f})$ as the variance estimator for $\hat{\tau}_\textup{haj}\left(t, t^{\prime},\hat{\gamma}_\textsc{f}\right)$.
% Define $\Sigma_n^2(x) = \operatorname{Var}\left( \left\{ n^{-1 / 2} \sum_{i=1}^n (Y_i-x_i^\top \gamma(t)-\mu(t)) \frac{{1}_i(t)}{\pi_i(t)}  \right\}_{t\in\mathcal{T}} \right)$.
\begin{assumption}\label{asu9}
% The covariate $x_i$ satisfies
(a) $||x_i||<c_x<\infty$, where $c_x$ is an absolute constant. \\ 
% [PD: why vector? what does $|...|$ mean? just use the eucliean norm?]. \\
(b) For the covariance matrix
\[
\Sigma_n(\gamma) = \operatorname{Var}\left( n^{-1 / 2} \sum_{i=1}^n \frac{1(T_i=t)}{\pi_i(t)} (Y_i-x_i^\top \gamma(t)-\mu(t))   : {t\in\mathcal{T}} \right)
\]
with finite and fixed vector $(\gamma(t):t\in\mathcal{T})$,
define $\lambda_{\min}(\Sigma_n(\gamma))$ as the smallest eigenvalue of $\Sigma_n(\gamma)$.
There exist $\epsilon>0$ and a positive sequence $\{m_n\}_{n \in \mathbb{N}}$ such that as $n\rightarrow \infty$ we have $m_n \rightarrow \infty$ and 
% \begin{align*}
% \Sigma_n^{-2}(\gamma) n^{-2} \sum_{s=0}^n\left|\mathcal{H}_n(s, m_n)\right| \tilde{\theta}_{n, s}^{1-\epsilon} \rightarrow 0,
% \text{ }
% \Sigma_n^{-\frac{3}{2}}(\gamma) n^{-\frac{1}{2}} M_n(m_n, 2) \rightarrow 0, 
% \text{ }
% \Sigma_n^{-\frac{1}{2}}(\gamma) n^{\frac{3}{2}} \tilde{\theta}_{n, m_n}^{1-\epsilon} \rightarrow 0.
% \end{align*}
\begin{align*}
\frac{n^{-2} \sum_{s=0}^n\left|\mathcal{H}_n(s, m_n)\right| \tilde{\theta}_{n, s}^{1-\epsilon}}{({\lambda_{\min}(\Sigma_n(\gamma))
})^2} \rightarrow 0,
\quad
\frac{n^{-1 / 2} M_n(m_n, 2)}{ (\lambda_{\min}(\Sigma_n(\gamma)))^{3/2}} \rightarrow 0, 
\quad 
\frac{ n^{3 / 2} \tilde{\theta}_{n, m_n}^{1-\epsilon}}{\sqrt{\lambda_{\min}(\Sigma_n(\gamma))
}} \rightarrow 0.
\end{align*}
% \begin{itemize}
% \item[(i)] $|x_i|<c_x<\infty$, where $c_x$ is some absolute constant value. 
% \item[(ii)] For
% \[
% \Sigma_n(\gamma) = \operatorname{Var}\left( n^{-1 / 2} \sum_{i=1}^n \frac{{1}_i(t)}{\pi_i(t)} (Y_i-x_i^\top \gamma(t)-\mu(t))   : {t\in\mathcal{T}} \right)
% \]
% with some finite and fixed $\gamma(t)$,
% there exist $\epsilon>0$ and a sequence of positive constants $\left\{m_n\right\}_{n \in \mathbb{N}}$ such that as $n\rightarrow \infty$ we have $m_n \rightarrow \infty$ and 
% \begin{align*}
% \Sigma_n^{-2}(\gamma) n^{-2} \sum_{s=0}^n\left|\mathcal{H}_n(s, m_n)\right| \tilde{\theta}_{n, s}^{1-\epsilon} \rightarrow 0, \quad
% \Sigma_n^{-3/2}(\gamma) n^{-1 / 2} M_n(m_n, 2) \rightarrow 0, \quad
% \Sigma_n^{-1/2}(\gamma) n^{3 / 2} \tilde{\theta}_{n, m_n}^{1-\epsilon} \rightarrow 0.
% \end{align*}
% $$
% \max \left\{
% \Sigma_n^{-2}(\gamma) n^{-2} \sum_{s=0}^n\left|\mathcal{H}_n(s, m_n)\right| \tilde{\theta}_{n, s}^{1-\epsilon}, \Sigma_n^{-3/2}(\gamma) n^{-1 / 2} M_n(m_n, 2), \Sigma_n^{-1/2}(\gamma) n^{3 / 2} \tilde{\theta}_{n, m_n}^{1-\epsilon}
% \right\} \rightarrow 0.
% $$
% \end{itemize}
\end{assumption}

Let ${\gamma}_\textsc{f}$ denote the probability limit of $\hat{\gamma}_\textsc{f}$, where $\hat{\gamma}_\textsc{f}$ is the coefficient vector of $x_i$ from the WLS fit in \eqref{eq:add}.  Let ${\Sigma}_{\textup{haj},\textsc{f}}$ denote the analog of ${\Sigma}_{\textup{haj}}$ in \eqref{eq:Sigma_haj} defined on the covariate-adjusted outcome $Y_i-x_i^\top {\gamma}_\textsc{f} $. 
% Let ${\Delta}_{\textup{haj}, \textsc{f}}$ be the analog of ${\Delta}_{\textup{haj}}$ defined on the adjusted outcomes $Y_i-x_i^\top {\gamma}_\textsc{f} $.  
Theorem \ref{thm:add_asym_n} below states the asymptotic normality of $\hat{\beta}_{\textup{haj},\textsc{f}}$.
\begin{theorem} \label{thm:add_asym_n}
% [PD: see my comment above. maybe we just need CLT in this theorem.]
Under Assumptions \ref{asu1}--\ref{asu4} and \ref{asu9}, 
we have ${\Sigma}_{\textup{haj}, \textsc{f}}^{-1/2}
\sqrt{n}( \hat{\beta}_{\textup{haj},\textsc{f}} -  {\mu} )
\stackrel{\textup{d}}{\rightarrow}  \mathcal{N}(0,{I})$.
% \[
% {\Sigma}_{\textup{haj}, \textsc{f}}^{-1/2}
% \sqrt{n}\left( \hat{\beta}_{\textup{haj},\textsc{f}} -  {\mu} \right)
% \stackrel{\textup{d}}{\rightarrow}  \mathcal{N}(0,{I}).
% \]
\end{theorem}

The design matrix of the WLS fit in \eqref{eq:add} equals $ {C}_\textsc{f} = ( {Z},  {X})$ where $Z$ is an $n\times |\mathcal{T}|$ matrix and  $ {X} = (x_i:i=1,\ldots,n)$ is an $n \times J$ matrix.  
Diagonalize the residual ${e_{\textsc{f},i}}$'s from the WLS fit in \eqref{eq:add} to form the matrix $ {e}_{\textup{haj},\textsc{f}} = \text{diag}\{e_{\textsc{f},i}: i=1,\ldots, n\}$.
Let $[\cdot]_{(1:|\mathcal{T}|,1:|\mathcal{T}|)}$ denote the upper-left $|\mathcal{T}|\times |\mathcal{T}|$ submatrix. 
Let $\hat{{V}}_{\textup{haj},\textsc{f}}$ denote the HAC estimator for $\hat{\beta}_{\textup{haj},\textsc{f}}$, which is a submatrix of the covariance estimator obtained from the WLS fit in \eqref{eq:add}:
\[
\hat{ {V}}_{\text{haj,\textsc{f}}}
= \left[( {C}_\textsc{f}^{\top}  {W}  {C}_\textsc{f})^{-1}
( {C}_\textsc{f}^{\top}  {W} 
{e}_{\textup{haj},\textsc{f}}
{K}_n 
{e}_{\textup{haj},\textsc{f}}  
{W}  {C}_\textsc{f})
( {C}_\textsc{f}^{\top}  {W}  {C}_\textsc{f})^{-1} \right]_{(1:|\mathcal{T}|,1:|\mathcal{T}|)}.
\]
% Define $\hat{\Sigma}^2_{\textup{haj},\textsc{f}}\left(t, t^{\prime}, \hat{\gamma}_\textsc{f}\right)$ as the network-robust variance for $\hat{\tau}_\textup{haj}\left(t, t^{\prime},\hat{\gamma}_\textsc{f}\right)$ based on $\hat{ {V}}_{\textup{haj},\textsc{f}}$. 
% Let $\hat{ {\Sigma}}_{*,\textup{haj},\textsc{f}}$ be the analog of $\hat{ {\Sigma}}_{*,\textup{haj}}$ defined on the adjusted outcomes $Y_i-x_i^\top {\gamma}_\textsc{f} $.
% The first matrix, denoted as $ {\Delta}_{\textup{haj}}$, has its $(i,t)$th element being $\left( \frac{{1}_i(t)}{\pi_i(t)}(Y_i - \mu(t))  - (\mu_i(t) - \mu(t) ) \right)$.
Let ${\Delta}_{\textup{haj}, \textsc{f}}$ denote the analog of ${\Delta}_{\textup{haj}}$ defined on the covariate-adjusted outcome $Y_i-x_i^\top {\gamma}_\textsc{f} $. 
Define $M_\textsc{f}$ as an $n\times |\mathcal{T}|$ matrix with $(i,t)$th element $M_{\textsc{f}, it} = \mu_i(t)-\mu(t) - x_i^\top \gamma_\textsc{f}$.  
% [PD: be coherent. whether you want to use $M_{\textsc{f}}(i,t)$ or $M_{\textsc{f}, it}$? we also need to use the elementwise notation for other matrices.]
Theorem \ref{thm:add_bias} below establishes the asymptotic bias of $\hat{{V}}_{\text{haj,\textsc{f}}}$ as an estimator for the asymptotic covariance of $\hat{\beta}_{\text{haj,\textsc{f}}}$. 
\begin{theorem} \label{thm:add_bias}
Define ${ {\Sigma}}_{*,\textup{haj}, \textsc{f}} = n^{-1}  {\Delta}_{\textup{haj}, \textsc{f}}^\top  {K}_n   {\Delta}_{\textup{haj}, \textsc{f}}$ and $R_{\textup{haj},\textsc{f}} = n^{-1}  M_\textsc{f}^\top  {K}_n M_\textsc{f}$.
Under Assumptions \ref{asu1}--\ref{asu4}, \ref{asu7} and \ref{asu9}, we have ${ {\Sigma}}_{*,\textup{haj}, \textsc{f}} 
= {\Sigma}_{\textup{haj}, \textsc{f}} +  o_\mathbb{P}(1)$ and $n \hat{ {V}}_{\textup{haj},\textsc{f}} 
= { {\Sigma}}_{*,\textup{haj}, \textsc{f}} + R_{\textup{haj},\textsc{f}} + o_\mathbb{P}(1)$.
% \begin{align}
% { {\Sigma}}_{*,\textup{haj}, \textsc{f}} 
% =&~ {\Sigma}_{\textup{haj}, \textsc{f}} +  o_\mathbb{P}(1), 
% \label{eq:oracle_F} \\
% n \hat{ {V}}_{\textup{haj},\textsc{f}} 
% =&~  { {\Sigma}}_{*,\textup{haj}, \textsc{f}} + R_{\textup{haj},\textsc{f}} + o_\mathbb{P}(1) \label{eq:bias_F}.
% \end{align}
% where ${ {\Sigma}}_{*,\textup{haj}, \textsc{f}}$ is defined in Theorem \ref{thm:add_asym_n}.
% where ${ {\Sigma}}_{*,\textup{haj}, \textsc{f}} = n^{-1}  {\Delta}_{\textup{haj}, \textsc{f}}^\top  {K}_n   {\Delta}_{\textup{haj}, \textsc{f}}$ and $R_{\textup{haj},\textsc{f}} = n^{-1}  M_\textsc{f}^\top  {K}_n M_\textsc{f}$.
% where
% \begin{align*}
% \hat{\Sigma}^2_{\textup{haj},\textsc{f}}(t,t',\hat{\gamma}_\textsc{f}) 
% =& ({1}(t) - {1}(t'))^{\top} \hat{ {V}}_{\textup{haj},\textsc{f}} ({1}(t) - {1}(t')) 
% \hat{{\Sigma}}_{*,\textup{haj}}(t,t',{\gamma}_\textsc{f}) 
% =& \frac{1}{n} \sum_{i=1}^n \sum_{j=1}^n
% % \left[
% % \begin{array}{c}
% \left( \tilde{\Delta}_i(t,t',{\gamma}_\textsc{f}) - (\tau_i(t, t^{\prime}) - \tau(t, t^{\prime}) )  \right) 
% \left( \tilde{\Delta}_j(t,t',{\gamma}_\textsc{f}) - (\tau_j(t, t^{\prime}) - \tau(t, t^{\prime}))  \right) 
% % \end{array}
% % \right]
% {1}(\ell_{ {A}}(i, j) \leq b_n)
%  R_{\textup{haj}}(t,t') 
% =& \frac{1}{n} \sum_{i=1}^n \sum_{j=1}^n
% (\tau_i(t,t') - \tau(t,t'))
% (\tau_j(t,t') - \tau(t,t')) 
% {1}(\ell_{ {A}}(i, j) \leq b_n)
% \end{align*}
\end{theorem}
% Theorem \ref{thm:add_bias} justifies the regression-based inference of $\tau = G\mu$ from the additive regressions with point estimator $\hat{\tau} = G\hat{\beta}_{\text{haj,\textsc{f}}}$ and regression-based HAC variance estimator $G \hat{ {V}}_{\text{haj,\textsc{f}}} G^\top$. 
The bias term $R_{\textup{haj},\textsc{f}}$ is an analog of $R_{\textup{haj}}$ defined on the adjusted outcome $Y_i-x_i^\top {\gamma}_\textsc{f} $. Given that $K_n$ may not be positive semi-definite, we cannot ensure the asymptotic conservativeness of $\hat{ {V}}_{{\textup{haj},\textsc{f}}}$ for estimating ${ {\Sigma}}_{*,\textup{haj}, \textsc{f}}$.
Similar to \eqref{V+},
we propose the adjusted covariance estimator as
\[
\hat{ {V}}_{\textup{haj},\textsc{f}}^+
= \left[( {C}_\textsc{f}^{\top}  {W}  {C}_\textsc{f})^{-1}
( {C}_\textsc{f}^{\top}  {W} 
{e}_{\textup{haj},\textsc{f}}
{K}_n^+ 
{e}_{\textup{haj},\textsc{f}}
{W}  {C}_\textsc{f})
( {C}_\textsc{f}^{\top}  {W}  {C}_\textsc{f})^{-1} \right]_{(1:|\mathcal{T}|,1:|\mathcal{T}|)}.
\]
% Define $ {\Delta}_{\textup{haj}, \textsc{f}}$ as the analog of $ {\Delta}_{\textup{haj}}$ defined on the adjusted outcomes $Y_i-x_i^\top {\gamma}_\textsc{f} $. 
% we propose an adjusted variance estimator for $\hat{ {V}}_{\textup{haj},\textsc{f}} $ to guarantee conservative. 
% Define $\hat{\mathbf{\Delta}}_{\textup{haj},\textsc{f}}(t,t',\hat{\gamma}_\textsc{f})$ as the stacked vector of
% \[
% \frac{\frac{{1}_i(t)}{\pi_i(t)} (Y_i - x_i^\top \hat{\gamma}_\textsc{f}   - \hat{\beta}_{\textup{haj},\textsc{f}}(t)) }{\hat 1_{\text{ht}}(t)}   - \frac{\frac{{1}_i(t')}{\pi_i(t')} (Y_i - x_i^\top \hat{\gamma}_\textsc{f}   - \hat{\beta}_{\textup{haj},\textsc{f}}(t')) }{\hat 1_{\text{ht}}(t')}. 
% \]
% Define $ {\Delta}_{\textup{haj}}\left(t, t^{\prime},\gamma_{\textsc{f}}\right)$ as the analog to $ {\Delta}_{\textup{haj}}(t, t^{\prime})$ for adjusted outcome $Y_i-\gamma^\top_{\textsc{f}} x_i$.  
% the stacked vector of
% $$
% \left(\left(\frac{{1}_i(t)\left(Y_i-\gamma^\top_{\textsc{f}} x_i-\mu(t)\right)}{\pi_i(t)}-\frac{{1}_i\left(t^{\prime}\right)\left(Y_i-\gamma^\top_{\textsc{f}} x_i-\mu(t')\right)}{\pi_i(t')}\right)-\left(\tau_i(t, t^{\prime}))-\tau(t, t^{\prime}))\right)\right).
% $$
\begin{theorem} \label{thm:add_adj}
Define $R_{\textup{haj},\textsc{f}}^+ = n^{-1} M_\textsc{f}^\top  {K}_n^{+}  M_\textsc{f} 
+ n^{-1}  {\Delta}_{\textup{haj},\textsc{f}}^\top  {K}_n^{-}  {\Delta}_{\textup{haj},\textsc{f}} \geq  0 $. 
% [PD: again i modified the presentation.]
Under Assumptions \ref{asu1}--\ref{asu4} and \ref{asu8}--\ref{asu9}, we have $n \hat{ {V}}_{\textup{haj},\textsc{f}}^+  
= { {\Sigma}}_{*,\textup{haj}, \textsc{f}}
+ R_{{\textup{haj},\textsc{f}}}^+ + o_\mathbb{P}(1)$,
% \[
% n \hat{ {V}}_{\textup{haj},\textsc{f}}^+  
% = { {\Sigma}}_{*,\textup{haj}, \textsc{f}}
% + R_{{\textup{haj},\textsc{f}}}^+ + o_\mathbb{P}(1) ,
% %\ge { {\Sigma}}_{*,\textup{haj}, \textsc{f}} + o_\mathbb{P}(1),
% \]
where ${ {\Sigma}}_{*,\textup{haj}, \textsc{f}}$ is defined in Theorem \ref{thm:add_bias}.
% where
% \begin{align*}
% R_{{\textup{haj},\textsc{f}}}^+(t,t',\gamma_{\textsc{f}})
% =& \frac{1}{n} 
% \overline{ \tau}(t,t')^{\top} 
% %  {K}_n^{+}
%  {K}_n^{+}
% \overline{ \tau}(t,t') 
% + \frac{1}{n}  {\Delta}_{\textup{haj}}\left(t, t^{\prime},\gamma_{\textsc{f}}\right)^{\top} 
%  {K}_n^{-} 
%  {\Delta}_{\textup{haj}}\left(t, t^{\prime},\gamma_{\textsc{f}}\right)    
% \end{align*}
\end{theorem}
Theorem \ref{thm:add_adj}
ensures the asymptotic conservativeness of $\hat{ {V}}_{\text{haj,\textsc{f}}}^+$ for estimating the true sampling covariance. This, together with Theorem \ref{thm:add_asym_n}, justify the regression-based inference of $\tau = G\mu$ from the additive WLS fit in \eqref{eq:add} with the point estimator $\hat{\tau} = G\hat{\beta}_{\text{haj,\textsc{f}}}$ and the adjusted regression-based HAC covariance estimator $G \hat{ {V}}_{\text{haj,\textsc{f}}}^+ G^\top$. 
% By construction, $K_n^+$ and $K_n^-$ are positive semi-definite, which in turn guarantee $\hat{ {V}}_{\textup{haj},\textsc{f}}^+$ is conservative for estimating the variance of $\hat{\beta}_{\text{haj,\textsc{f}}}$. 

\subsection{Fully-interacted regression in network experiments}\label{sec:full}
With full interactions between the exposure mapping indicators and covariates, we consider the WLS fit
\begin{equation}
\text{regress } Y_i \text{ on } (z_i, z_i \otimes x_i) \text{ with weights }  
{w}_i = 1/\pi_i(T_i) 
% \text{lm}(Y_i \sim z_i + z_i \otimes x_i) \text{ with weights }  
% {w}_i = 1/\pi_i(T_i) 
\label{eq:full},  
\end{equation}
where $\otimes$ denotes the Kronecker product. The specification \eqref{eq:full} simply means WLS fit of $Y_i$ on the dummy $1(T_i=t)$'s and the interaction $1(T_i=t)x_i$'s.  
% [PD: i just want a symmetric presentation of the regression function...]
Let $\hat{\beta}_{\textup{haj},\textsc{l}}$ denote the estimtors of coefficients for $z_i$ from the above WLS fit and $\hat{\beta}_{\textup{haj},\textsc{l}}(t)$ denote the element in $\hat{\beta}_{\textup{haj},\textsc{l}}$ corresponding to $1(T_i=t)$. 
We use the subscript ``L'' to signify \cite{Lin2013}. 
% \cite{lin2013} proposed an improved estimator as the coefficient of $z_i$ from the OLS
% fit of $Y_i \sim z_i + \sum_{t\in\mathcal{T}} {1}_i(t)x_i$ with centered covariates and treatment-covariates interactions, and proved that it is at least as efficient as the difference-in-means and \cite{fisher1935}’s estimators asymptotically.
% Let $\hat{\gamma}_\textsc{l}(t)$ denote the coefficient vector of $1(T_i=t)x_i$.
% \begin{proposition} \label{prop:full}
% $\hat{\beta}_{\textup{haj},\textsc{l}}(t) = \hat{Y}_{\textup{haj}}(t) - \hat{x}_{\textup{haj}}(t)^{\top} \hat{\gamma}_\textsc{l}(t)$ for all $t\in\mathcal{T}$.
% \end{proposition}
% % Recall Proposition \ref{prop:add} under the additive WLS fit implies that $\hat{\beta}_{\textup{haj},\textsc{f}}(t) = \hat{Y}_{\textup{haj}}(t) - \hat{x}_{\textup{haj}}(t)^{\top} \hat{\gamma}_\textsc{f}$. 
% Proposition \ref{prop:full} parallels Proposition \ref{prop:add}, and establishes that $\hat{\beta}_{\textup{haj},\textsc{l}}(t)$ is the Hájek estimator based on the covariate-adjusted outcome $Y_i - x_i^\top \hat{\gamma}_\textsc{l}(t) $. 
% A key distinction is that the adjustment is now based on coefficients specific to exposure mapping values.
Let ${\gamma}_\textsc{l}(t)$ be the probability limit of $\hat{\gamma}_\textsc{l}(t)$, where $\hat{\gamma}_\textsc{l}(t)$ is the coefficient vector of $1(T_i=t)x_i$ from the WLS fit in \eqref{eq:full}. 
% See proof in Lemma \ref{lemma: gamma_L}. 
% and denote by $\hat{\Sigma}^2_{\textup{haj},\textsc{l}}(\hat{\gamma}_\textsc{l})$ the network-robust variance for $\hat{\tau}_\textsc{l}\left(t, t^{\prime},\hat{\gamma}_\textsc{l}\right)$ from the weighted least squares fit. 
% Theorem \ref{thm:full} below established the conservative of $\hat{\Sigma}^2_{\textup{haj},\textsc{l}}(\hat{\gamma}_\textsc{l})$ for the asymptotic variance of $\hat{\tau}_\textsc{l}\left(t, t^{\prime},\hat{\gamma}_\textsc{l}\right)$.
% % $\hat\tau_\textsc{l}(t,t')$
% Define 
% $ 
% \tilde{\Delta}_i(t,t',{\gamma}_\textsc{l}) 
% = \left( \frac{{1}_i(t)(Y_i-{\gamma}_\textsc{l}(t)^\top x_i -\mu(t) )}{\pi_i(t)}  - \frac{{1}_i(t')(Y_i- {\gamma}_\textsc{l}(t')^\top x_i -\mu(t') )}{\pi_i(t')}\right). 
% $
Let ${\Sigma}_{\textup{haj},\textsc{l}}$ be the analog of ${\Sigma}_{\textup{haj}}$ in \eqref{eq:Sigma_haj} defined on the adjusted outcome $Y_i-  x_i^\top {\gamma}_\textsc{l}(T_i)$.
Theorem \ref{thm:full} states the asymptotic normality of $\hat{\beta}_{\textup{haj},\textsc{l}}$.
\begin{theorem} \label{thm:full}
% [PD: check my comment above.]
Under Assumptions \ref{asu1}--\ref{asu4} and \ref{asu9}, 
we have ${\Sigma}_{\textup{haj}, \textsc{l}}^{-1/2}
\sqrt{n} ( \hat{\beta}_{\textup{haj},\textsc{l}} -  {\mu} )
\stackrel{\textup{d}}{\rightarrow}  \mathcal{N}(0,{I})$.
\end{theorem}
% Let ${\gamma}_\textsc{l}$ be the probability limit of $\hat{\gamma}_\textsc{l}$ under Assumptions xxx.
% Let $\hat{ {\Sigma}}_{*,\textup{haj}, \textsc{l}}$ be the analog of $\hat{ {\Sigma}}_{*,\textup{haj}}$ defined on the adjusted outcomes $Y_i-  x_i^\top (\sum_{t\in \mathcal{T}} {1}_i(t) {\gamma}_\textsc{l}(t))$.
% Let $\chi$ be the concatenation of $\chi_i = z_i \otimes x_i$ over $i\in\mathcal{N}$. 
Let $ {C}_\textsc{l}$ be the design matrix of the WLS fit in \eqref{eq:full}, with row vectors $(z_i^\top, (z_i \otimes x_i)^\top)$. 
% [PD: check my modification.]
%The design matrix of the WLS fit in \eqref{eq:full} equals $ {C}_\textsc{l} = ( {Z}, \{z_i \otimes x_i\}_{i=1}^n)$. 
Diagonalize the residual $e_{\textsc{l},i}$'s from the same WLS fit to form the matrix $ {e}_{\textup{haj},\textsc{l}} = \text{diag}\{e_{\textsc{l},i}:i=1,\ldots,n\}$.
Let $\hat{ {V}}_{\textup{haj},\textsc{l}}$ denote the HAC covariance estimator for $\hat{\beta}_{\textup{haj},\textsc{l}}$, which is a submatrix of the covariance estimator obtained from the WLS fit in \eqref{eq:full}:
\[
\hat{ {V}}_{\text{haj,\textsc{l}}}
= \left[( {C}_\textsc{l}^{\top}  {W}  {C}_\textsc{l})^{-1}
( {C}_\textsc{l}^{\top}  {W} 
{e}_{\textup{haj},\textsc{l}} 
{K}_n 
{e}_{\textup{haj},\textsc{l}}  
{W}  {C}_\textsc{l})
( {C}_\textsc{l}^{\top}  {W}  {C}_\textsc{l})^{-1} \right]_{(1:|\mathcal{T}|,1:|\mathcal{T}|)}.
\]
Let $ {\Delta}_{\textup{haj}, \textsc{l}}$ be the analog of $ {\Delta}_{\textup{haj}}$ defined on the adjusted outcome $Y_i- x_i^\top {\gamma}_\textsc{l}(T_i)$.
Define $M_\textsc{l}$ as an $n\times |\mathcal{T}|$ matrix with $(i,t)$th element $M_{\textsc{l}, it} = \mu_i(t)-\mu(t) - x_i^\top \gamma_\textsc{l}(t)$.
% Theorem \ref{thm:full_bias} below establishes the asymptotic bias of $\hat{{V}}_{\text{haj,\textsc{l}}}$ for estimating the asymptotic variance of $\hat{\beta}_{\text{haj,\textsc{l}}}$. 
\begin{theorem} \label{thm:full_bias}
Define ${ {\Sigma}}_{*,\textup{haj},\textsc{l}} = n^{-1}  {\Delta}_{\textup{haj},\textsc{l}}^\top {K}_n {\Delta}_{\textup{haj},\textsc{l}}$ and $R_{\textup{haj},\textsc{l}} = n^{-1} M_\textsc{l}^\top {K}_n  M_\textsc{l}$.
Under Assumptions \ref{asu1}--\ref{asu4}, \ref{asu7} and \ref{asu9}, we have $\hat{ {\Sigma}}_{*,\textup{haj}, \textsc{l}} 
= {\Sigma}_{\textup{haj}, \textsc{l}}
+ o_\mathbb{P}(1)$ and $n \hat{ {V}}_{\textup{haj},\textsc{l}} 
= \hat{ {\Sigma}}_{*,\textup{haj}, \textsc{l}}
+ R_{\textup{haj},\textsc{l}} + o_\mathbb{P}(1)$.
\end{theorem}
Theorem \ref{thm:full_bias} establishes the asymptotic bias of $\hat{ {V}}_{{\textup{haj},\textsc{l}}}$ as an estimator for the asymptotic covariance of $\hat{\beta}_{\text{haj,\textsc{f}}}$. Given that $K_n$ may not be positive semi-definite, we cannot ensure the asymptotic conservativeness of $\hat{ {V}}_{{\textup{haj},\textsc{l}}}$ for estimating $\hat{ {\Sigma}}_{*,\textup{haj}, \textsc{l}}$.
Similar to \eqref{V+}, we propose the adjusted HAC covariance estimator as
\[
\hat{ {V}}_{\textup{haj},\textsc{l}}^+
= \left[( {C}_\textsc{l}^{\top}  {W}  {C}_\textsc{l})^{-1}
( {C}_\textsc{l}^{\top}  {W} 
{e}_{\textup{haj},\textsc{l}} 
{K}_n^+ 
{e}_{\textup{haj},\textsc{l}} 
{W}  {C}_\textsc{l})
( {C}_\textsc{l}^{\top}  {W}  {C}_\textsc{l})^{-1} \right]_{(1:|\mathcal{T}|,1:|\mathcal{T}|)}.
\]
\begin{theorem} \label{thm:full_adj}
Define $R_{\textup{haj},\textsc{l}}^+ = n^{-1} M_\textsc{l}^\top  {K}_n^{+}  M_\textsc{l} + n^{-1}  {\Delta}_{\textup{haj},\textsc{l}}^\top  {K}_n^{-}  {\Delta}_{\textup{haj},\textsc{l}} \geq 0$. 
Under Assumptions \ref{asu1}--\ref{asu4} and \ref{asu8}--\ref{asu9}, we have $n \hat{ {V}}_{\textup{haj},\textsc{l}}^+  
= { {\Sigma}}_{*,\textup{haj},\textsc{l}} 
+ R_{\textup{haj},\textsc{l}}^+ + o_\mathbb{P}(1)$,
% \begin{align*}
% n \hat{ {V}}_{\textup{haj},\textsc{l}}^+  
% =&~ { {\Sigma}}_{*,\textup{haj},\textsc{l}} 
% + R_{\textup{haj},\textsc{l}}^+ + o_\mathbb{P}(1),
% \end{align*}
where ${ {\Sigma}}_{*,\textup{haj},\textsc{l}}$ is defined in Theorem \ref{thm:full_bias}.
% where
% \begin{align*}
% R_\text{haj,+}(t,t',{\gamma}_\textsc{l})
% =& \frac{1}{n} \bar{ \tau}(t,t',{\gamma}_\textsc{l})^{\top} 
%  {K}_n^{+} 
% \bar{ \tau}(t,t',{\gamma}_\textsc{l})
% + \frac{1}{n} 
%  {\Delta}_\textup{haj}\left(t, t^{\prime},{\gamma}_\textsc{l}\right)^{\top}  {K}_n^{-}  {\Delta}_\textup{haj}\left(t, t^{\prime},{\gamma}_\textsc{l}\right) 
% \end{align*}
\end{theorem}
Echoing the comment after Theorem \ref{thm:add_adj}, Theorems \ref{thm:full} and \ref{thm:full_adj} together justify the regression-based inference of $\tau = G\mu$ from the fully-interacted WLS fit in \eqref{eq:full} with point estimator $\hat{\tau} = G\hat{\beta}_{\text{haj,\textsc{l}}}$ and adjusted regression-based HAC covariance estimator $G \hat{ {V}}_{\text{haj,\textsc{l}}}^+ G^\top$. 
% By construction, 
% Theorem \ref{thm:full_adj} ensures the asymptotic conservativeness of $\hat{ {V}}_{{\textup{haj},\textsc{l}}}$ for estimating the true sampling variance.

\subsection{Final remarks on the efficiency gain via covariate adjustment}\label{sec:efficiency}

Regression adjustment can improve efficiency under reasonable data-generating processes.
\cite{Lin2013} demonstrated the efficiency gain from including fully interacted covariates when the propensity score is constant and there is no interference.
However, this strategy does not always improve efficiency, especially in the presence of heterogeneous propensity scores or interference.
In Appendix \ref{Counterexample}, we present simulation results demonstrating that including fully-interacted covariates can exhibit higher asymptotic variance than the unadjusted Hájek estimator in scenarios with either heterogeneous propensity scores or interference. 
This lack of guarantee has also been documented in settings without interference, such as cluster experiments with varying sizes \citep{SuDing2021}, split-plot experiments \citep{ZhaoDing2022}, and scenarios where outcomes are not missing completely at random \citep{ZhaoDingLi2024}. 
Despite the lack of theoretical guarantees for efficiency gain, we do observe that covariate adjustment improves efficiency in the simulation studies and empirical examples in Section \ref{sec:Numerical Illustration}.

We focus on regression-based covariate-adjusted estimators for ease of implementation. 
% Although there is no theoretical guarantee that covariate adjustment improves efficiency, our simulation and empirical studies show that the estimated standard errors are reduced. 
There are alternative methods for enhancing efficiency via covariate adjustment. 
One strategy is to find the optimal linearly adjusted estimator by minimizing the true or estimated standard error; see, e.g., \cite{LiDing2020a} and \citet{LuShiFang2025}.  
For example, in our setting, we consider regressions such as regressing \( Y_i - x_i^\top \gamma \text{ on } z_i \) or regressing \( Y_i - 1(T_i = t) x_i^\top \gamma(t) \text{ on } 1(T_i = t) \) for each $t\in \mathcal{T}$. 
We can compute the variance of the resulting estimator and then minimize it with respect to the coefficients $\gamma$ or \( \gamma(t) \)'s.  
This procedure is different from WLS, but the minimization ensures variance reduction. 
% We omit the technical details here for brevity.
Another strategy is to first compute the Hájek estimators of \( Y_i \) and \( x_i \), denoted by \( \hat{\beta}_{\textup{haj}} \) and \( \hat{\beta}_{{\textup{haj}, x}} \), respectively. An adjusted estimator can then be constructed as $\hat\beta_{\textup{adj}} = \hat\beta_{\textup{haj}} - \hat\beta_{{\textup{haj}, x}}^\top \gamma_x$, where the optimal adjustment coefficient is given by $\gamma_x = \operatorname{Cov}(\hat{\beta}_{{\textup{haj}, x}})^{-1} \operatorname{Cov}(\hat{\beta}_{{\textup{haj}, x}}, \hat\beta_{\textup{haj}})$,
and can be consistently estimated; see \cite{JiangandWei2019} and \cite{RothSantAnna2023}.
When using the estimated covariance matrix of $\hat{\beta}_{\textup{haj}}$ and $\hat\beta_{{\textup{haj}, x}}$ to estimate $\gamma_{x}$, the resulting estimated standard error is always smaller. 
We omit the details for these two alternative strategies because we focus on simpler regression-based estimators.

\section{Numerical examples}\label{sec:Numerical Illustration}
In this section, we first examine the finite-sample performance of our results with simulation and then apply our results to an empirical application.
Our analysis focuses on the exposure effect $\tau(t,t') = \mu(t) - \mu(t')$. 
We analyze another empirical example \cite{CaiJanvrySadoulet2015} in Appendix \ref{sec:Cai}.

\subsection{Simulation} \label{Simulation}
To achieve comparability with \cite{Leung2022}, we replicate the same scenario but with the inclusion of a covariate in the model. Regarding the results, we present the point and covariance estimators of the exposure effect from three specifications of WLS: unadjusted (Unadj), with additive covariates (Add), and with fully-interacted covariates (Sat).
% [PD: use Unadj? WLS is overloaded now....]
% Additionally, we report the adjusted covariance estimator for each regression-based HAC estimator to showcase its improvement in empirical coverage rates. 
We also report \cite{Leung2022}'s Horvitz--Thompson estimator and variance estimator. 

The study encompasses two outcome models: the linear-in-means model and the complex contagion model. 
% Define $\tilde{A}$ as the row-normalized version of $A$ (divide each row by its row sum). 
% [PD: the original presentation below is somewhat convoluted at least for me. check my modifications below.]
Define 
\begin{equation}
V_i( {D},  {A},  {x},  {\varepsilon})
= \alpha+\beta \sum_{j=1}^n \tilde{A}_{i j} Y_j +\delta \sum_{j=1}^n \tilde{A}_{i j} D_j + \xi D_i + \gamma x_i +\varepsilon_i. 
\label{eq:outcomeV}
\end{equation}
where $\tilde{A}_{i j} = A_{ij}/\sum_{j=1}^n A_{i j}$ is the $(i,j)$th entry of $\tilde{A}$, the row-normalized version of $A$. 
For the linear-in-means model, we set $Y_i = V_i( {D},  {A},  {x},  {\varepsilon})$ 
%as follows: 
%\[
%V_i( {D},  {A},  {x},  {\varepsilon})
%= \alpha+\beta \frac{\sum_{j=1}^n A_{i j} Y_j}{\sum_{j=1}^n A_{i j}}+\delta \frac{\sum_{j=1}^n A_{i j} D_j}{\sum_{j=1}^n A_{i j}}+\xi D_i  + \gamma  x_i+\varepsilon_i
%\]
with $(\alpha, \beta, \delta, \xi, \gamma)=(-1,0.8,1,1,3)$. 
% [PD: you need to solve for the outcomes? it is not clear enough here... see my next comment.] 
The model defines potential outcomes $Y_i(D)$ through its reduced form:
\[
Y = \alpha (I - \beta \tilde{A})^{-1} \iota + (I - \beta \tilde{A})^{-1} (\delta \tilde{A} + \xi I) D + (I - \beta \tilde{A})^{-1} \gamma x + (I - \beta \tilde{A})^{-1} \varepsilon.
\]
For the complex contagion model, we set $Y_i={1}(V_i( {D},  {A},  {x}, {\varepsilon})>0)$ with $(\alpha, \beta, \delta,\xi, \gamma)=(-1,1.5,1,1,3)$. 
% [PD: the complex contagion model is more complicated? you need to iterate?] 
The complex contagion model can be generated from the dynamic process:
\[
Y_i^t 
= 1\left( \alpha+\beta \sum_{j=1}^n \tilde{A}_{i j} Y_j^{t-1} + \delta \sum_{j=1}^n \tilde{A}_{i j} D_j + \xi D_i  + \gamma x_i + \varepsilon_i > 0 \right)
\]
with initialization at period 0 as 
% $Y_i^0 = 1( \alpha+\delta \sum_{j=1}^n \tilde{A}_{i j} D_j + \xi D_i  + \gamma  x_i+\varepsilon_i > 0 )$.
\begin{equation*}
Y_i^0 = 1\left( \alpha+\delta \sum_{j=1}^n \tilde{A}_{i j} D_j + \xi D_i  + \gamma x_i + \varepsilon_i > 0 \right).
\end{equation*}
We run the dynamic process to obtain new outcomes $Y^t=(Y_i^t)_{i=1}^n$ from last period's outcomes $Y^{t-1}$ until the first period $T$ such that $Y^T = Y^{T-1}$. We then take $Y^T$ as the vector of observed outcomes $Y$, which yields outcomes $(Y_i(D))^n_{i=1}$. As a result, this process implicitly defines potential outcomes (\citet[Section 3.1]{Leung2022}). 
Without covariates $x_i$, \cite{Leung2022} derived conditions on the model parameters of the linear-in-means model and complex contagion model so that ANI holds. 
We can extend his proof to the models with additive covariates as in \eqref{eq:outcomeV}, or covariates interacted with the network $A$, given that the covariates are fixed. 
We choose parameters to satisfy those conditions to ensure ANI.
% \footnote{For linear-in-means model, ANI holds with $|\beta|<1$. For complex contagion model, ANI holds with a bound on the magnitude of cumulative network interactions, thus ensuring the stability of network effects as $n$ grows large. Define a weighted directed network $G$ on $\mathcal{N}_n$ with entries $G_{ij} = A_{ij}\mathbb{E}[\varphi_j]$ for $\varphi_j = 1\{0 < \phi(D_j,\varepsilon_j) \leq \beta\}$.
% Then, ANI holds with $\sup_n \rho_n(\bar{s}) < 1$ for some $\bar{s}>0$ with
% \[
% \rho_n(\bar{s}) = \sup_{s\geq\bar{s}}\|G^s\|_{\infty}^{1/s} \equiv \sup_{s\geq\bar{s}}\left(\max_{i\in\mathcal{N}_n}\sum_{j=1}^{n}(G^s)_{ij}\right)^{1/s}.
% \]
% See more details in \citet[Section 3.1]{Leung2022}.
% }

% As a remark, to generate the outcomes, we first need to derive the reduced form of the outcome model to avoid the issue of simultaneity. [PD: the last sentence is dense and it is too late.]

Following \cite{Leung2022}, we generate the adjacency matrix $ {A}$ from a random geometric graph model. Specifically, for each node $i$, we randomly generate its position $\rho_i$ in a two-dimensional space from $\mathcal{U}([0,1]^2)$. 
An edge between nodes $i$ and $j$ is created if the Euclidean distance between their positions is less than or equal to a threshold value $r_n$: $A_{i j}={1}\{\left\|\rho_i-\rho_j\right\| \leq r_n\}$, where the threshold value is chosen as $r_n=(\kappa /(\pi n))^{1/2}$. 
We set $\kappa$ as the average degree $\delta({A})$, calculated based on the experimental data in Section \ref{app:paluck}, in order to better mimic real-world scenarios.
We also generate a sequence $\{\nu_i\}_{i=1}^n \stackrel{\text{IID}}{\sim} \mathcal{N}(0,1)$ independent of $ {A}$. 
% [PD: in V? maybe eqref the equation of V above?]
The error term in \eqref{eq:outcomeV} is generated as $\varepsilon_i=\nu_i + (\rho_{i 1}-0.5)$, where $\rho_{i 1}$ is the first component of $i$'s ``location'' $\rho_{i}$ generated above.
% [PD: what is $\rho_{i 1}$?] 
This inclusion accounts for unobserved homophily, as units with similar $\rho_{i1}$ values are more likely to form links.
Finally, we generate the covariate $\{x_i\}_{i=1}^n \stackrel{\text{IID}}{\sim} \mathcal{N}(0,1)$.
% the covariate $x_i$'s are IID $\mathcal{N}(0,1)$. 

We use the sample of the two largest treated schools from the network experiment in Section \ref{app:paluck} to calibrate the network models. The network size $n$ is $1456$. 
% The results pertaining to network sizes $n=805$ and $2725$ are included in the online appendix.
We also conduct simulation with network sizes $n=805$ and $2725$ to illustrate variations in population sizes. See results in the Appendix \ref{sec:additional_sim}.
We treat the schools as a single network by pooling the degree sequences across them. 
We randomly assign treatments to units classified as eligible in the experimental data with a probability $0.5$.
Since we work within a finite-population framework, we generate $ {A}$, $ {\varepsilon}$'s, and $ {x}$'s once and only redraw $ {D}$ for each simulation draw. This differs from the superpopulation design simulation in \cite{Leung2022}, where he regenerated $D$, $A$ and $\varepsilon$'s for each simulation draw. 

For the spillover effect of having at least one treated friend versus non-treated friends $\tau(1,0)$, we define the exposure mapping as $T_i={1}(\sum_{j=1}^n A_{i j} D_j>0)$ and analyze only the population of units with at least one friend who is eligible for treatment to satisfy Assumption \ref{asu2}.
% [PD: what is Bernoulli randomization? just the IID randomization of $D$?]
Under the IID randomization of $D$, we can compute the propensity score $\pi_i(1)$'s and $\pi_i(0)$'s for each student using Binomial probabilities.
% [PD: $\pi_i(1)$'s and $\pi_i(0)$'s]

Table \ref{table: Haj RGG 2} presents the results. 
The top panels display our regression-based results. 
We report the estimand under ``${\tau}(1,0)$,'' approximated by the unbiased Horvitz--Thompson estimator $\hat{\tau}_\text{ht}(1,0)$, computed over $10,000$ simulation draws.
We report ``Oracle SE,'' denoted by $\operatorname{Var}(\hat{\tau}(1,0))^{1/2}$, which are calculated as the standard deviation of the point estimators from corresponding WLS fits over $10,000$ simulation draws.
For the estimation results, we conduct another independent $10,000$ simulation draws.
We present the point estimate from each WLS fit under ``$\hat{\tau}(1,0)$.'' 
We present the HAC standard errors obtained from each WLS fit under ``WLS SE,'' and the corresponding adjusted HAC standard errors under ``WLS$^+$ SE'', where the suggested bandwidth based on \eqref{eq:bandwidth} is $b_n = 3$. We report the Eicker-Huber-White standard errors assuming no interference under ``EHW SE'' to illustrate the degree of dependence in the data.
We also report the empirical coverage rate of $95\%$ confidence intervals (CIs) in the “Coverage'' rows for the corresponding standard errors. 
The effective sample size of exposure mapping value $t$ is defined as $\hat{n}(t)=\sum_{i=1}^n 1(T_i=t)$.
% We report the effective sample size of $\hat{\mu}(t)$ as $\hat{n}(t)=\sum_{i=1}^n 1(T_i=t)$ under ``$\hat{n}(t)$.''
% The ``$b_n$" row reports the bandwidth in \eqref{eq:bandwidth} with $K=1$. 
% The ``APL'' row reports the average path length $\mathcal{L}({A})$ of each network. 
% % We adopt the bandwidth choice recommended by \cite{Leung2022}; however, it is worthwhile to investigate whether an optimal choice exists [PD: why do we mention the optimality issue here?].
% The ``Effective size" is $\hat{n}(1) + \hat{n}(0)$, as we only consider units with at least one friend who is eligible for treatment to satisfy Assumption \ref{asu2}.

The result table demonstrates that the standard errors obtained from the WLS fits can be anti-conservative, underestimating the true standard error. However, by utilizing the adjusted HAC standard errors, we can improve the empirical coverage and ensure a conservative estimation of the standard error. In this setting,  the estimator from the fully-interacted WLS fit is at least as efficient as the estimators from the unadjusted or additive WLS fits. 

In the middle panel of Table \ref{table: Haj RGG 2}, we report the results of standard errors and coverage rates of $95\%$ CIs using the kernel ${K}^{\textup{L}2019}_n$ in \cite{Leung2019e} and the kernel ${K}^{\textup{K}2021}_n$ in \cite{Kojevnikov2021}.
% where the $(i,j)$th element is
% \[
% {K}^{\textup{L}2019}_{n,ij}
% =\frac{\left|\mathcal{N}(i, b_n;A) \cap \mathcal{N}(j, b_n;A)\right|}{\left|\mathcal{N}(i, b_n;A)\right|^{1 / 2}\left|\mathcal{N}(j, b_n;A)\right|^{1 / 2}}
% \]
% and the kernel ${K}^{\textup{K}2021}_n$ in \cite{Kojevnikov2021} where the $(i,j)$th element is
% \[
% {K}^{\textup{K}2021}_{n,ij}=\frac{\left|\mathcal{N}(i, b_n;A) \cap \mathcal{N}(j, b_n,A)\right|}{n^{-1} \sum_{k=1}^n\left|\mathcal{N}(k, b_n;A)\right|},
% \]
% respectively.
Both ${K}^{\textup{L}2019}_n$ and ${K}^{\textup{K}2021}_n$ are positive semi-definite, ensuring the positive semi-definiteness of the covariance estimators. However, we can see that they substantially overreject even in moderately sized samples. 

The bottom panel of Table \ref{table: Haj RGG 2} present the results of the Horvitz--Thompson estimator and variance estimator from \cite{Leung2022}. By comparing the ``Oracle SE'' from the top and bottom panels, we can see the WLS estimators from all three specifications exhibit higher efficiency compared with the Horvitz--Thompson estimator.
% In Table \ref{table: Haj RGG 4}, our regression-based standard errors are approximately half of those reported using \cite{Leung2022}'s method, indicating a significant spillover effect at the 5\% significance level. 
% In contrast, \cite{Leung2022}'s method yields an insignificant effect. 
Moreover, \cite{Leung2022}’s standard errors are smaller than the oracle standard errors, resulting in under coverage.

\begin{table*}
\centering
\caption{Simulation results: network size $n=1456$}
\begin{tabular}{lcccccc}
\hline
\hline
% \multicolumn{7}{c}{SIMULATION RESULTS}  \\ 
% \hline
% & \multicolumn{6}{c}{\# Schools = 2} \\
% \cline{2-7}
Outcome model & \multicolumn{3}{c}{Linear-in-Means} & \multicolumn{3}{c}{Complex   Contagion} \\
\hline
% \cline{2-4} \cline{5-7}
WLS specification & \multicolumn{1}{c}{Unadj} & \multicolumn{1}{c}{Add} & \multicolumn{1}{c}{Sat} & \multicolumn{1}{c}{Unadj} & \multicolumn{1}{c}{Add} & \multicolumn{1}{c}{Sat} \\
\hline													
${\tau}(1,0)$	&		&	0.616	&		&		&	0.016	&		\\
$\hat{\tau}(1,0)$	&	0.620	&	0.617	&	0.617	&	0.017	&	0.017	&	0.017	\\
Oracle SE	&	0.842	&	0.639	&	0.639	&	0.041	&	0.027	&	0.027	\\
WLS SE	&	0.802	&	0.606	&	0.604	&	0.040	&	0.027	&	0.027	\\
WLS$^+$ SE	&	0.874	&	0.650	&	0.648	&	0.050	&	0.034	&	0.034	\\
EHW SE	&	0.407	&	0.272	&	0.272	&	0.040	&	0.027	&	0.027	\\
Oracle Coverage	&	0.952	&	0.950	&	0.950	&	0.951	&	0.953	&	0.953	\\
WLS Coverage	&	0.939	&	0.932	&	0.931	&	0.936	&	0.942	&	0.943	\\
WLS$^+$ Coverage	&	0.958	&	0.948	&	0.947	&	0.980	&	0.982	&	0.982	\\
EHW Coverage	&	0.659	&	0.596	&	0.595	&	0.944	&	0.947	&	0.947	\\
\hline													
\cite{Leung2019e} SE	&	0.748	&	0.562	&	0.560	&	0.041	&	0.027	&	0.027	\\
\cite{Kojevnikov2021} SE	&	0.734	&	0.553	&	0.551	&	0.041	&	0.028	&	0.028	\\
\cite{Leung2019e} Coverage	&	0.919	&	0.911	&	0.910	&	0.943	&	0.947	&	0.948	\\
\cite{Kojevnikov2021} Coverage	&	0.914	&	0.907	&	0.906	&	0.944	&	0.949	&	0.948	\\
\hline													
$\hat{\tau}_\text{ht}(1,0)$	&	0.709	&		&		&	0.020	&		&		\\
Oracle SE	&	1.380	&		&		&	0.112	&		&		\\
Leung SE	&	1.335	&		&		&	0.109	&		&		\\
Oracle Coverage	&	0.952	&		&		&	0.952	&		&		\\
Leung Coverage	&	0.934	&		&		&	0.937	&		&		\\
\hline
\hline
\end{tabular}
\label{table: Haj RGG 2}
% \captionsetup{font=small,labelfont=bf}
\caption*{ \footnotesize Note: The effective sample size for each exposure mapping value is $\hat{n}(1)=426$ and $\hat{n}(0)=296$, with a total of $\hat{n}(1) + \hat{n}(0) = 722$.
The suggested bandwidth in \eqref{eq:bandwidth} is $b_n = 3$. 
The average path length is $\mathcal{L}({A})=18.25$.}
\end{table*}

\subsection{Empirical Application I: \cite{PaluckShepherdAronow2016}}
\label{app:paluck}
In this subsection, we revisit \cite{PaluckShepherdAronow2016} and apply our regression-based analysis to their network experiment, which examines how an anti-conflict intervention influences teenagers’ social norms regarding hostile behaviors such as bullying, social exclusion, harassment, and rumor-spreading.
We now provide a detailed description of the empirical setting.
In the experimental design, half of $56$ schools were randomly assigned to the treatment group. Within these treated schools, a subset of students was selected as eligible for treatment based on certain characteristics. Half of the eligible students were then block-randomized into treatment by gender and grade. Those treated students were invited to participate in bi-weekly meetings that incorporated an anti-conflict curriculum.
Following \cite{Leung2022}, we choose self-reported data on wristband wearing as the outcome of interest, which serves as the reward for students who exhibit anti-conflict behavior. 
We incorporate both  gender and grade for covariate adjustment. 
The network is measured by asking students to name up to ten students at the school they spent time with in the last few weeks.
More details about this network experiment can be found in \cite{PaluckShepherdAronow2016}. 

To align with the results reported in \cite{Leung2022}, we restrict the data to the five largest treated schools.  
Our primary interest lies in assessing the direct effect of the anti-conflict intervention and the spillover effect of having at least one friend assigned to the treatment versus none such friends. 
We first calculate both effects by defining two one-dimensional exposure mappings and report the results in Table \ref{table: paluck}. 
To examine both effects simultaneously,
we define a two-dimensional exposure mapping and report the results in Table \ref{table: paluck1}. 
The network, obtained from surveys, is directed. 
When calculating the number of treated friends for the exposure mappings, we take into account the direction of links. However, when computing network neighborhoods for our covariance estimators, we disregard the directionality of links to conservatively define larger neighborhoods.
For each exposure mapping, our analysis involves three WLS specifications: unadjusted (Unadj), with additive covariates (Add), and with fully-interacted covariates (Sat). 
We also include the results from \citet[Table 1]{Leung2022} in the column ``Leung.''

% [PD:  needs better structure in this subsection. maybe use paragraph ``One-dimensional exposure mapping'' and ``Two-dimensional exposure mapping''? check my modification below.]

\paragraph*{One-dimensional exposure mapping}
For the direct effect, we define $T_i=D_i$ as in Example \ref{ex:direct} and limit the analysis to the students eligible for treatment, totaling $320$ students. The propensity score is $\pi_i(t)=0.5$ for each student. For the spillover effect, we employ $T_i={1}(\sum_{j=1}^n A_{i j} D_j>0)$ as the exposure mapping as in Example \ref{ex:one-dim}, indicating whether at least one friend has been assigned to the treatment. 
We restrict the effective sample to units with at least one eligible friend.
% [PD: block on what? be precise]
Under block randomization, we can compute the propensity score $\pi_i(0)$ and $\pi_i(1)$ for each student using Hypergeometric probabilities.
% [PD: propensity score $\pi_i(1)$ and $\pi_i(0)$? ]

The results are presented in Table \ref{table: paluck}.
% Table \ref{table: paluck} displays the results of both direct and spillover effects. 
The suggested bandwidths in \eqref{eq:bandwidth} are $b_n = 2$ for both exposure mappings. 
We present results for the range of bandwidths $\{0, \ldots, 3\}$, where $0$ yields the standard errors in the absence of interference.
The first row, labeled as ``Estimate,'' presents the point estimates obtained from corresponding WLS fits.
The rows labeled as ``$b_n=k$'' present the HAC standard errors with the specific bandwidth values stated. 
% Additionally, the rows labeled as ``$K_n$ PSD'' indicate whether the associated matrix $ {K}_n$ is positive semi-definite or not. Furthermore, the rows labeled as ``WLS$^{+}$ SE'' present the adjusted HAC standard errors with corresponding bandwidth values.
We find that the kernel matrix $K_n$ is not positive semi-definite for all bandwidths in $\{1, 2, 3\}$, so we report the adjusted HAC standard errors under ``WLS$^{+}$ SE''.
The direct effect is statistically significant at $5\%$ level across all specifications, bandwidths, and after adjustment to the covariance estimation. The spillover effect is significant at $5\%$ level except when $b_n = 3$, both before and after adjustment to the covariance estimation. 
While our results align with the conclusions of \cite{Leung2022}, our regression-based estimation approach provides higher precision. Also, the $K_n$ is not positive semi-definite indicating that \cite{Leung2022}'s variance estimators may be anti-conservative.

\paragraph*{Two-dimensional exposure mapping}
We define the exposure mapping and $G$ as in Example \ref{ex:two-dim}: $T_i = (D_i,  {1} (\sum_{j=1}^n A_{ij}D_j > 0) )$.  
% where $T_i$ takes values of $\{(0,0), (0,1), (1,0), (1,1)\}$. 
% With
% \[
% G = 2^{-1} 
% \begin{pmatrix}
% -1 & -1 & 1 & 1 \\  
% -1 & 1 & -1 & 1
% \end{pmatrix}, 
% \]
We focus on the first two components of $\tau=G\mu$, where the first component captures the direct effect and the second component captures the spillover effect.
% For the direct effect, we set $G=2^{-1}(-1,-1,1,1)$, and for the spillover effect, we set $G=2^{-1}(-1,1,-1,1)$. 
We restrict the effective sample to students who are eligible for treatment and have at least one eligible friend, resulting in a total of $150$ students. 
% Table \ref{table: paluck1} presents the results. 
% Our analysis focuses on students who are eligible for treatment and have at least one eligible friend, resulting in a total of 150 students. 

The results are presented in the top panel of Table \ref{table: paluck1}.
The average out-degree, $n^{-1} \sum_{ij} A_{ij}$, is 7.96. The APL is 3.37 across our five schools. Given $n = 3306$ students, we have $\log n / \log \delta(A) = 3.96$, which is close to 3.37. Thus, the suggested bandwidth in \eqref{eq:bandwidth} is $b_n = 2$ with $K = 1$, and we report results for the range of bandwidths $\{0, \ldots, 3\}$.
We observe that the magnitude and standard errors of the direct effect remain relatively stable. 
Regarding the spillover effect, its magnitude notably increases, and it remains statistically significant at the $5\%$ significance level across all specifications and bandwidths, even after adjustment to the covariance estimation.

To investigate whether these changes in results arise from shifts in the target population or potential misspecification of the exposure mappings, we provide results using two one-dimensional exposure mappings and focusing on treatment-eligible students with at least one eligible friend. These results are displayed in the bottom panel of Table \ref{table: paluck1}. 
Upon comparing the top and bottom panels, we can observe that there are minor differences in the point estimates and standard errors, but the overall message does not change. Specifically, the spillover effect is more pronounced and significant for the subset of students who are both eligible for treatment and have at least one eligible friend, in comparison to the subset with at least one eligible friend. Table \ref{table: paluck1} also demonstrates that our methods are robust to various specifications of exposure mappings.

\begin{table}
\centering
\caption{Estimates and SEs (one-dimensional exposure mapping).}
\label{table: paluck}
\begin{tabular}{lcccccccc}
% \toprule
\hline
\hline
& \multicolumn{4}{c}{Direct effect} & \multicolumn{4}{c}{Spillover effect} \\
\hline
Estimator & Unadj & Add & Sat & Leung & Unadj & Add & Sat & Leung \\
\midrule
Estimate	&	0.150	&	0.147	&	0.147	&	0.150	&	0.048	&	0.045	&	0.045	&	0.041	\\
$b_n=0$	&	0.040	&	0.040	&	0.040	&	0.044	&	0.016	&	0.016	&	0.016	&	0.017	\\
$b_n=1$	&	0.041	&	0.040	&	0.040	&	0.046	&	0.016	&	0.016	&	0.016	&	0.018	\\
WLS$^{+}$ SE	&	0.042	&	0.042	&	0.042	&		&	0.020	&	0.020	&	0.020	&		\\
$b_n=2$	&	0.035	&	0.035	&	0.033	&	0.039	&	0.017	&	0.017	&	0.016	&	0.021	\\
WLS$^{+}$ SE	&	0.050	&	0.049	&	0.048	&		&	0.027	&	0.028	&	0.027	&		\\
$b_n=3$	&	0.040	&	0.039	&	0.038	&	0.047	&	0.017	&	0.016	&	0.016	&	0.017	\\
WLS$^{+}$ SE	&	0.058	&	0.057	&	0.056	&		&	0.030	&	0.030	&	0.030	&		\\	
% \hline
% % \hline
% \multicolumn{7}{c}{Results copied from \cite{Leung2022}}\\
% \hline
% % \hline
% Estimate &			0.1500	&		&		&	0.0407	&	& \\
% $b_n=0$	&			0.0443 	&		&		&	0.0167	&	& \\
% $b_n=1$	&			0.0460 	&		&		&	0.0184	&	& \\
% $b_n=2$	&			0.0394 	&		&		&	0.0205	&	& \\
% $b_n=3$	&			0.0470 	&		&		&	0.0170	&	& \\
\hline
\hline
% \bottomrule
\end{tabular}
% \captionsetup{font=small,labelfont=bf}
\caption*{\footnotesize Note: Columns display results for the treatment ($n=320$) and spillover ($n=1685$) effects. }
\end{table}

\begin{table}[ht]
\centering
\caption{Estimates and SEs ($n=150$).}
\label{table: paluck1}
\begin{tabular}{lcccccc}
% \toprule
\hline
\hline
& \multicolumn{3}{c}{Direct effect} & \multicolumn{3}{c}{Spillover effect} \\
\hline
Estimator & Unadj & Add & Sat & Unadj & Add & Sat \\
\hline
% \hline
\multicolumn{7}{c}{Two-dimensional exposure mapping}\\
\hline
% \hline
Estimate	&	0.155	&	0.144	&	0.142	&	0.149	&	0.147	&	0.165	\\
$b_n=0$	&	0.051	&	0.050	&	0.051	&	0.051	&	0.050	&	0.051	\\
$b_n=1$	&	0.052	&	0.050	&	0.053	&	0.054	&	0.054	&	0.055	\\
WLS$^{+}$ SE	&	0.053	&	0.051	&	0.054	&	0.055	&	0.055	&	0.057	\\
$b_n=2$	&	0.046	&	0.044	&	0.049	&	0.055	&	0.056	&	0.062	\\
WLS$^{+}$ SE	&	0.054	&	0.052	&	0.058	&	0.061	&	0.061	&	0.067	\\
$b_n=3$	&	0.043	&	0.044	&	0.044	&	0.050	&	0.053	&	0.063	\\
WLS$^{+}$ SE	&	0.059	&	0.058	&	0.061	&	0.067	&	0.068	&	0.076	\\
\hline													
\multicolumn{7}{c}{One-dimensional	exposure	mapping}\\											
\hline													
Estimate	&	0.170	&	0.155	&	0.155	&	0.168	&	0.164	&	0.165	\\
$b_n=0$	&	0.057	&	0.057	&	0.057	&	0.053	&	0.053	&	0.053	\\
$b_n=1$	&	0.058	&	0.058	&	0.058	&	0.058	&	0.058	&	0.058	\\
WLS$^{+}$ SE	&	0.059	&	0.059	&	0.060	&	0.059	&	0.059	&	0.059	\\
$b_n=2$	&	0.049	&	0.051	&	0.051	&	0.061	&	0.061	&	0.061	\\
WLS$^{+}$ SE	&	0.060	&	0.061	&	0.061	&	0.067	&	0.066	&	0.066	\\
$b_n=3$	&	0.041	&	0.040	&	0.040	&	0.058	&	0.060	&	0.061	\\
WLS$^{+}$ SE	&	0.061	&	0.061	&	0.061	&	0.074	&	0.073	&	0.074	\\
\hline
\hline
% \bottomrule
\end{tabular}
\caption*{ \footnotesize Note: The top panel presents results from a two-dimensional exposure mapping, while the bottom panel shows results from two one-dimensional exposure mappings, using the same effective sample as the two-dimensional exposure mapping ($n=150$).}
% % \captionsetup{font=small,labelfont=bf}
% % \caption*{Note: Columns display results for the treatment and spillover ($n=150$) effects under three regression specifications. }
\end{table}

\section{Extensions to continuous exposure mapping} \label{sec:conclusions}

Our theory focuses on discrete exposure mappings with finite support. 
However, continuous or growing-dimensional exposure mappings, such as the number or share of treated friends, are also common in practice, e.g., \cite{MuralidharanNiehausSukhtankar2023}.
For growing-dimensional exposure mappings that vary with $n$, valid inference is possible when the network is sparse, meaning that the maximum or average degree is substantially smaller than the network size (e.g., \citealt{Leung2020}).
For continuous exposure mappings, estimating \( \mu(t) \) is conceptually straightforward by extending the propensity score to a treatment density function, defined as \( \pi_i(t) = f_{T_i}(t) \) \citep{HiranoImbens2004}. 

Without imposing any modeling assumption on $\mu(t)$, we can use the following nonparametric estimator:
\begin{align}
    \hat{\mu}_h(t) = \frac{\frac{1}{n h} \sum_{i=1}^n \frac{ 1(|T_i - t| \leq h) }{ \mathbb{P}(|T_i - t| \leq h) } Y_i}{\frac{1}{n h} \sum_{i=1}^n \frac{ 1(|T_i - t| \leq h) }{ \mathbb{P}(|T_i - t| \leq h) } }  
    \text{ where } 
    \mathbb{P}(|T_i - t| \leq h) = \int_{t-h}^{t+h} \pi_i(s) \textup{d}s,
\label{eq:nonparametric}
\end{align}
which locally averages the $Y_i$ values whose $T_i$ falls within the bandwidth $h$ around $t$.
Since the exposure mapping $T_i$ and the treatment assignments are known, one can compute $\mathbb{P}(|T_i - t| \leq h)$ either in closed form or via Monte Carlo simulation.
The main technical challenges are (i) ensuring sufficient smoothness of the estimand $\mu(t)$ and (ii) choosing an appropriate bandwidth $h$ to trade off the bias and variance for estimating $\mu(t)$. 
Here, we use the uniform kernel in \eqref{eq:nonparametric} as an illustrative example, although general kernel functions could be employed.
% Similarly, \cite{Leung2020} proposes nonparametric estimators for treatment and spillover effects, or OLS estimators under a linearity assumption, while assuming a sparse network.

As noted by \cite{FaridaniNiehaus2024}, regression-based analysis with continuous exposure mappings typically relies on either a linear outcome model or restrictions on the experimental design. Consider the potential outcome model \( Y_i(t) = Y_i(0) + \beta_i t \), where the individual effects $\beta_i$'s can vary across units. If we regress the outcome \( Y_i \) on the centered exposure mapping  \( T_i - \mathbb{E}(T_i) \) with weight $1/\sqrt{\operatorname{Var}(T_i)}$, the WLS coefficient 
\begin{align}
\hat{\beta}
= \frac{\frac{1}{n} \sum_{i=1}^n \frac{1}{\operatorname{Var}(T_i)} (T_i - \mathbb{E}(T_i)) Y_i}{\frac{1}{n} \sum_{i=1}^n \frac{1}{\operatorname{Var}(T_i)} (T_i - \mathbb{E}(T_i))^2}  
\label{eq:hat_beta}
\end{align}
identifies:
\[
\beta = \frac{\frac{1}{n} \sum_{i=1}^n \frac{1}{\operatorname{Var}(T_i)} \operatorname{Cov}(Y_i, T_i - \mathbb{E}(T_i))}{\frac{1}{n} \sum_{i=1}^n \frac{1}{\operatorname{Var}(T_i)} \operatorname{Var}(T_i)}
= \frac{1}{n} \sum_{i=1}^n \beta_i,
\]
which represents the average of the $\beta_i$'s.
With constant treatment effect $\beta_i = \beta$, the WLS coefficient $\hat{\beta}$ identifies $\beta$. 

We outline future directions for continuous exposure mapping above, leaving many technical issues for further research. For example, what is the optimal choice of bandwidth $h$ in estimator \eqref{eq:nonparametric}? More importantly, we aim to develop rigorous statistical inference procedures for both the nonparametric estimator in \eqref{eq:nonparametric} and the WLS estimator in \eqref{eq:hat_beta}.

\bibliographystyle{ecta}
\bibliography{bib}

\newpage

\appendix 

\pagenumbering{arabic} %reset page counter to 1
\renewcommand*{\thepage}{S\arabic{page}}

\setcounter{equation}{0} 
\global\long\def\theequation{S\arabic{equation}}%
%  \setcounter{section}{0} 
% \global\long\def\thesection{S\arabic{section}}%
 \setcounter{assumption}{0} 
\global\long\def\theassumption{S\arabic{assumption}}%
 \setcounter{theorem}{0} 
\global\long\def\thetheorem{S\arabic{theorem}}%
 \setcounter{proposition}{0} 
\global\long\def\theproposition{S\arabic{proposition}}%
 \setcounter{definition}{0} 
\global\long\def\thedefinition{S\arabic{definition}}%
 \setcounter{example}{0} 
\global\long\def\theexample{S\arabic{example}}%
 \setcounter{figure}{0} 
\global\long\def\thefigure{S\arabic{figure}}%
 \setcounter{table}{0} 
\global\long\def\thetable{S\arabic{table}}%

\begin{appendix} 

% \begin{small}

\begin{center}
  \LARGE {\bf Appendix for ``Causal inference in network experiments: 
regression-based analysis and design-based properties''}
\end{center}

Section \ref{sec:add_results} presents additional results that complement the main paper.
Section \ref{app:asu8} gives some numerical justification of Assumption \ref{asu8}.
Section \ref{Counterexample} gives the counterexample of three cases without efficiency gain from fully-interacted regression adjustment. 
Section \ref{sec:additional_sim} gives additional simulation results. 
Section \ref{sec:Cai} analyzes the network experiment of \cite{CaiJanvrySadoulet2015}.
Section \ref{app:HT} gives the regression-based analysis for recovering the Horvitz--Thompson estimator, and modifies \cite{Leung2022}'s variance estimator to guarantee conservativeness. 

Section \ref{sec:proof} contains the proofs of the results stated in the main paper.
Section \ref{app:Auxiliary} contains the auxiliary results that are used in the proofs.
Section \ref{proof:Hájek} contains the proofs of the results in Section \ref{sec:WLS}. 
Section \ref{app:covadj} contains the proofs of the results in Section \ref{sec:covadju}.

\paragraph*{Notation}
Let $\|\cdot\|_\textup{F}$ denote the Frobenius norm, i.e., $\|A\|_\textup{F} = \sqrt{\text{tr}(A^\top A)}$ for a real matrix $A$. 
% Let $\|\cdot\|$ denote the Euclidean norm, i.e., $\|w\| = \sqrt{w^\top w}$ for $w\in\mathbb{R}^v$. 
For any Lipschitz function $f: \mathbb{R}^{v\times a} \rightarrow \mathbb{R}$, let $\operatorname{Lip}(f)$ be its Lipschitz constant and $\|\cdot\|_{\infty}$ be the sup-norm of $f$, i.e., $\|f\|_{\infty}=\sup _{x \in \mathcal{X}^d}|f(x)|$, where $\mathcal{X} \subseteq \mathbb{R}$ is any compact set.
Let $\Phi$ denote the cumulative distribution function of $\mathcal{N}(0,1)$. Throughout the Appendix, we denote the $(i,j)$th entry of matrix $B$ as $B(i,j)$ or $B_{ij}$ and define $1_i(t) = 1(T_i = t)$ for simplicity of notation.

\section{Additional results} \label{sec:add_results}

\subsection{Numerical justification of Assumption \ref{asu8}} \label{app:asu8}
We provide some numerical justification of Assumption \ref{asu8}, in response to Remark \ref{justification}, with two classic network generation models: the random geometric graph model and the Erdős--Rényi model.
Since the network is observed, both $K_n$ and $K_n^-$ are known. 
Moreover, as noted by \citet{KojevnikovMarmerSong2019}, one can compute $M_n^-(b,1)$, $M_n^-(b,2)$ and $\mathcal{J}_n^-(s,b)$ for any $s$ using the data across a range of values of bandwidth $b$, e.g., $b\in[1,10]$. 
Suppose that the sequence $\{\tilde{\theta}_{n, s}\}$ is summable, i.e. $\sum_{s=0}^{n} \tilde{\theta}_{n, s}=O(1)$. It suffices to justify that $\max_s \mathcal{J}_n^-(s,b_n) = o(n^2)$ to satisfy Assumption \ref{asu8}(d). Suppose further that we have a sequence $\mathcal{H}_n(b_n) = O(b_n^\beta)$. The coefficient $\beta$ can be estimated by regression $\log(\mathcal{H}_n(b))$ against $\log(b)$ and a constant. We follow this idea to justify Assumption \ref{asu8}. For each model, we consider three different sample sizes for the number of nodes: $n=500, 1000$ and $5000$.

\paragraph*{Random geometric graph model.}  
The random geometric graph model exhibits the polynomial growth rates in the sense that for sufficiently large $s$
\[
\sup _n \max _{i \in \mathcal{N}_n}\left|\mathcal{N}_A(i, s)\right|=C s^d,  
\]
where $C>0$ and $d$ equals the underlying network dimension with $d\ge 1$ \citep{Leung2019e}.
\citet[Appendix]{Leung2022} gave a justification that the bandwidth in \eqref{eq:bandwidth} for network with polynomial growth rates is $b_n\approx n^{1/3d}$. 

To illustrate, we generate a network with $n$ nodes in $d=2$, where the network edges are determined as follows:
\[
A_{i j}={1}(\left\|\rho_i-\rho_j\right\| \leq (5 /(\pi n))^{1/2}),
\] 
where $\rho_i \stackrel{\text{IID}}{\sim} \mathcal{U}([0,1]^2)$. 
% The coefficient can be estimated by regressing $\log(\mathcal{J}_n^-(s,b))$ against $\log(b)$ and a constant.
Figures \ref{fig:RGG1}, \ref{fig:RGG2} and \ref{fig:RGG3} display the plots of $\log(M_n^-(b,1))$, $\log(M_n^-(b,2))$ and $\log(\max_s \mathcal{J}_n^-(s,b))$ against $\log(b)$, repectively. 
For $M_n^-(b,1)$, the coefficients vary within the range of $1.18-1.48$ across different sample sizes. 
For $M_n^-(b,2)$, the coefficients also vary within the range of $2.43-3.00$. For $\max_s \mathcal{J}_n^-(s,b)$, the coefficients exhibit variations within the range of $2.66-3.05$.
With $d=2$, we can see that the Assumption \ref{asu8}(b)--(d) aligns with the behavior of $K_n^-$ for the random geometric graph model. Additionally, we present the plot of $\log(M_n(b,1))$ against $\log(b)$ in Figure \ref{fig:RGG0}, with coefficients varying within the range of $1.27-1.42$. This serves as a validation of Assumption \ref{asu7}(b).

\begin{figure} 
\begin{subfigure}{0.45\textwidth}
\centering
\includegraphics[width=\linewidth]{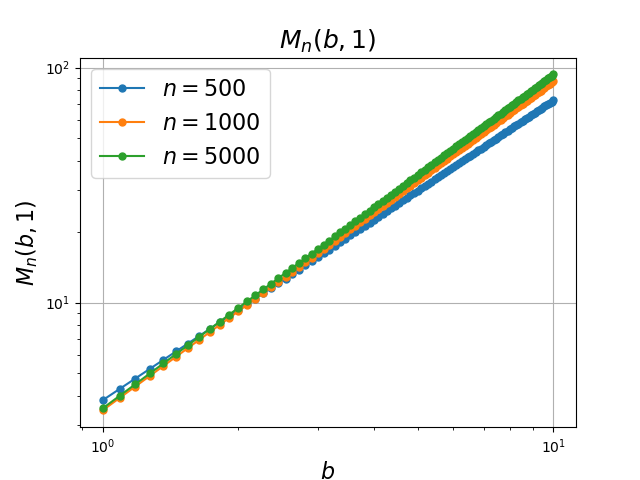}
\caption{$M_n(b,1)$}\label{fig:RGG0}
\end{subfigure}
\centering
\begin{subfigure}{0.45\textwidth}
\includegraphics[width=\linewidth]{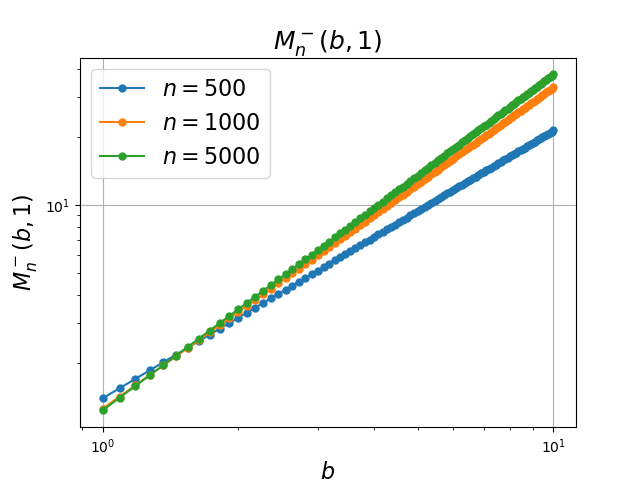}
\caption{$M_n^-(b,1)$}\label{fig:RGG1}
\end{subfigure}
\begin{subfigure}{0.45\textwidth}
\centering
\includegraphics[width=\linewidth]{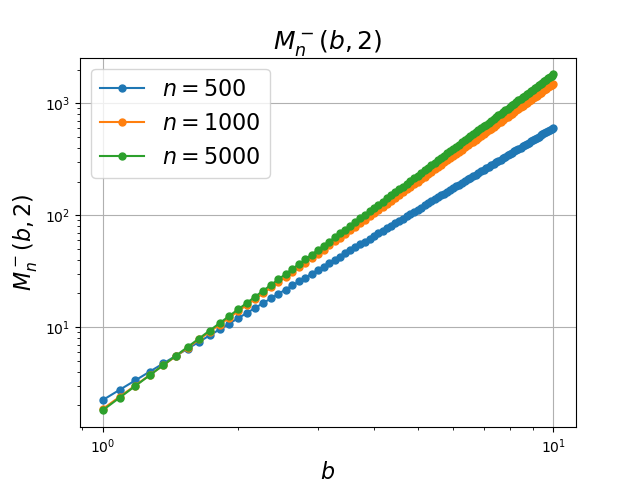}
\caption{$M_n^-(b,2)$}\label{fig:RGG2}
\end{subfigure}
\begin{subfigure}{0.45\textwidth}
\centering
\includegraphics[width=\linewidth]{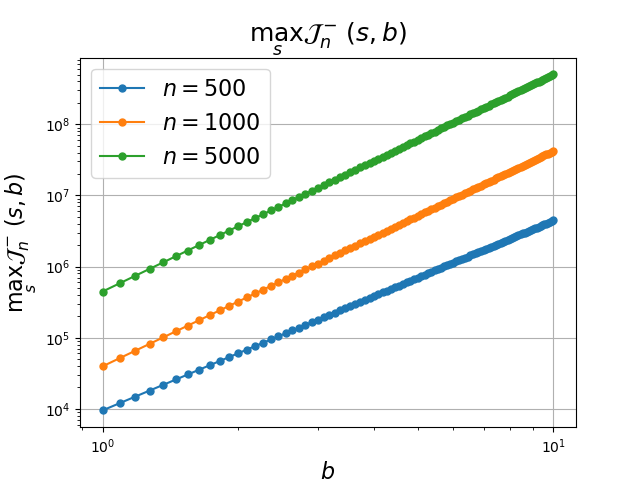}
\caption{$\max_s \mathcal{J}_n^-(s,b)$}\label{fig:RGG3}
\end{subfigure}
% Add more subfigures as needed
\caption{random geometric graph model. The log-log plots of $M_n(b,1)$ (on the top left panel), $M_n^-(b,1)$ (on the top right panel), $M_n^-(b,2)$ (on the bottom left panel) and $\max_s \mathcal{J}_n^-(s,b)$ (on the bottom right panel) against $b$.}
\end{figure}

\paragraph*{Erdős--Rényi model.} 
The Erdős--Rényi model exhibits the exponential growth rate in the sense that for sufficiently large $s$ 
\[
\sup _n \max _{i \in \mathcal{N}_n}\left|\mathcal{N}_A(i, s)\right|=C e^{\beta s},
\]
where $C, \beta>0$ and $\beta \approx \log \delta(A)$, with $\delta( {A})=$ $n^{-1} \sum_{i=1}^n \sum_{j=1}^n A_{i j}$ denoting the average degree \citep{BollobasJansonRiordan2007a,Barabasi2015}. 
\cite{Leung2022} justified that the bandwidth in \eqref{eq:bandwidth} is $b_n\approx 0.5\log n /\log \delta(A)$. 

To illustrate, we generate the Erdős--Rényi model with $A_{ij} \stackrel{\text{IID}}{\sim} \text{Bern}(5/n)$, so we have $\delta(A)=O(1)$. 
Figures \ref{fig:ER1}, \ref{fig:ER2} and \ref{fig:ER3} display the plots of $M_n^-(b,1)$, $M_n^-(b,2)$ and $\max_s \mathcal{J}_n^-(s,b)$ against $\log(b)$, respectively. 
For $M_n^-(b,1)$, the coefficients vary within the range of $0.30-0.77$ across different sample sizes. For $M_n^-(b,2)$, the coefficients also vary within the range of $1.25-2.34$. 
For $\max_s \mathcal{J}_n^-(s,b)$, the coefficients exhibit variations within the range of $1.55-3.07$.
We can see that the Assumption \ref{asu8}(b)--(d) aligns with the behavior of $K_n^-$ for the Erdős--Rényi model. Additionally, we present the plot of $\log(M_n(b,1))$ against $\log(b)$ in Figure \ref{fig:RGG0}, with coefficients varying within the range of $0.64 - 0.83$. 
This serves as a validation of Assumption \ref{asu7}(b).

\begin{figure} 
\begin{subfigure}{0.45\textwidth}
\centering
\includegraphics[width=\linewidth]{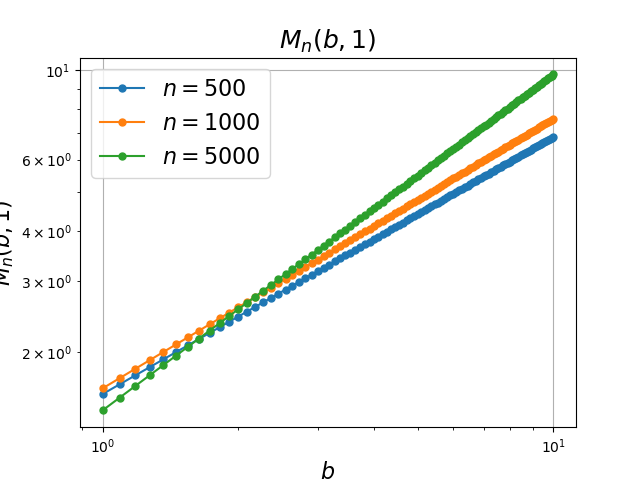}
\caption{$M_n(b,1)$}\label{fig:ER0}
\end{subfigure}
\centering
\begin{subfigure}{0.45\textwidth}
\includegraphics[width=\linewidth]{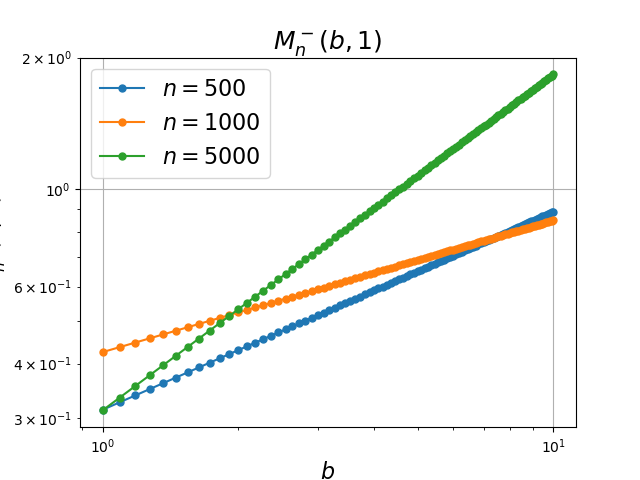}
\caption{$M_n^-(b,1)$}\label{fig:ER1}
\end{subfigure}
\begin{subfigure}{0.45\textwidth}
\centering
\includegraphics[width=\linewidth]{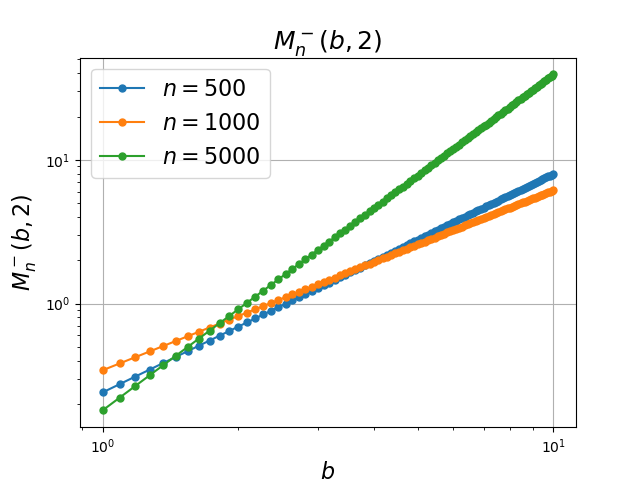}
\caption{$M_n^-(b,2)$}\label{fig:ER2}
\end{subfigure}
\begin{subfigure}{0.45\textwidth}
\centering
\includegraphics[width=\linewidth]{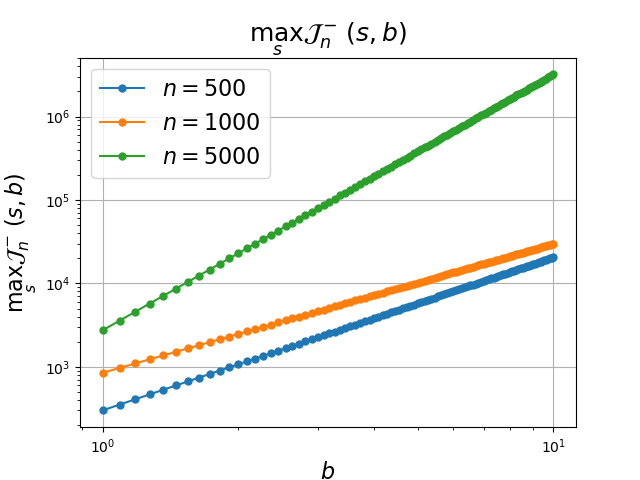}
\caption{$\max_s \mathcal{J}_n^-(s,b)$}\label{fig:ER3}
\end{subfigure}
% Add more subfigures as needed
% \caption{}
\caption{Erdős--Rényi model. The log-log plots of $M_n(b,1)$ (on the top left panel), $M_n^-(b,1)$ (on the top right panel), $M_n^-(b,2)$ (on the bottom left panel) and $\max_s \mathcal{J}_n^-(s,b)$ (on the bottom right panel) against $b$.}
\end{figure}

\subsection{Efficiency Gain: Examples and Counterexamples}
\label{Counterexample} 
In this section, we compare the performance of WLS fits with fully-interacted covariates with WLS fits without covariates and with additive covariates, as discussed in Section \ref{sec:efficiency}. 
% Specifically, we present three simulation designs showing that, when propensity scores are heterogeneous or interference is present, including fully-interacted covariates can yield larger asymptotic variance than the unadjusted Hájek estimator.
We introduce the following $\Delta$'s to simplify the presentation: 
\begin{eqnarray*}
\Delta_{i}(t;{\gamma}_\textsc{l}) 
&=& 
{1}_i(t) \pi_i(t)^{-1}
x_i^\top {\gamma}_\textsc{l}(t)
- x_i^\top {\gamma}_\textsc{l}(t),  \\
\Delta_{i}(t;{\gamma}_{\textsc{l}-\textsc{f}}) 
&=& {1}_i(t) \pi_i(t)^{-1} x_i^\top ({\gamma}_\textsc{l}(t)- {\gamma}_\textsc{f})  - x_i^\top ({\gamma}_\textsc{l}(t)- {\gamma}_\textsc{f}),  \\
\tilde{\Delta}_{i}(t;{\gamma}_\textsc{l}) 
&=& {1}_i(t) \pi_i(t)^{-1} (Y_i- x_i^\top {\gamma}_\textsc{l}(t) -\mu(t) )  - (\mu_i(t)  -  x_i^\top {\gamma}_\textsc{l}(t) - \mu(t)).   
\end{eqnarray*}
\begin{theorem} \label{thm:efficiency}
Under Assumptions \ref{asu1}--\ref{asu4}, \ref{asu7} and \ref{asu9}, 
we have
\begin{align*}
{{\Sigma}}_{*, \textup{haj}} 
=& {{\Sigma}}_{*, \textup{haj},\textsc{l}}
+ \left( \frac{1}{n} \sum_{i=1}^n \sum_{j=1}^n
\left( \Delta_{i}(t;{\gamma}_\textsc{l}) 
+ 2 \tilde{\Delta}_{i}(t;{\gamma}_\textsc{l})  \right)
\Delta_{j}(t';{\gamma}_\textsc{l})
K_{n}(i,j)
\right)_{t,t'\in\mathcal{T}},  \\
{{\Sigma}}_{*, \textup{haj},\textsc{f}}
=& {{\Sigma}}_{*, \textup{haj},\textsc{l}}
+ \left( \frac{1}{n} \sum_{i=1}^n \sum_{j=1}^n
\left( \Delta_{i}(t;{\gamma}_{\textsc{l}-\textsc{f}})  
+ 2 \tilde{\Delta}_{i}(t;{\gamma}_\textsc{l})  \right)
\Delta_{j}(t';{\gamma}_{\textsc{l}-\textsc{f}})  
K_{n}(i,j)
\right)_{t,t'\in\mathcal{T}}.
\end{align*}
\end{theorem}
Theorem \ref{thm:efficiency} highlights the lack of a clear efficiency gain from fully interacted regression adjustment. 
In most settings, we observe efficiency gain as in Section \ref{Simulation}.
However, counterexamples do exist, as demonstrated in Designs 1, 2 and 3.
% The first two examples serve as explicit illustrations, effectively showcasing that when fully interacted covariates are integrated, the asymptotic covariance might increase. Conversely, the final example demonstrates a scenario in which including fully interacted covariates leads to efficiency improvement.
In all three cases, we focus on $\tau(t,t')=\mu(t)-\mu(t')$, the specific contrast between two exposure mapping values, $t$ and $t'$.
For each simulation design, we present the results from $10,000$ simulation draws with a sample size of $n = 1000$. The results are shown in Table \ref{table: Counterexample}, which includes the estimand (Estimand) and oracle standard errors (Oracle SE) from three WLS fits.
% [PD: i use subsubsections below.]

\begin{table}[tbp]
\centering
\caption{Simulation results of counterexamples}
\label{table: Counterexample}
\begin{tabular}{lrrrrrrrrr}
\hline
\hline
% \multicolumn{4}{c}{SIMULATION RESULTS}  \\ 
% \hline
% & \multicolumn{6}{c}{\# Schools = 1} \\
% \cline{2-7}
% & \multicolumn{3}{c}{Linear-in-Means}  \\
% \cline{2-4} 
& \multicolumn{3}{c}{Design 1} & \multicolumn{3}{c}{Design 2} & \multicolumn{3}{c}{Design 3}\\
\hline
WLS  & \multicolumn{1}{r}{Unadj} & \multicolumn{1}{r}{Add} & \multicolumn{1}{r}{Sat} & \multicolumn{1}{r}{Unadj} & \multicolumn{1}{r}{Add} & \multicolumn{1}{r}{Sat} & \multicolumn{1}{r}{Unadj} & \multicolumn{1}{r}{Add} & \multicolumn{1}{r}{Sat} \\
\hline
% \multicolumn{4}{c}{Design 1} \\
% \hline
Estimand	&		&	11.673	&		&		&	6.225	&		&		&	6.290	&		\\
Oracle SE	&	5.052	&	5.539	&	5.356	&	0.472	&	0.491	&	0.484	&	0.322	&	0.338	&	0.337	\\
% WLS SE	&	1.4349	&	1.2858	&	1.2874 &	0.5037	&	0.4857	&	0.4883	&	0.8643	&	0.7943	&	0.7901 \\
% \hline
% \multicolumn{4}{c}{Design 2} \\
% \hline
% Estimand 	&	0.5702	&	0.5702	&	0.5702	\\
% Estimate	&	0.5681	&	0.5779	&	0.5786	\\
% Oracle SE	&	0.5061	&	0.4857	&	0.4890	\\
% WLS SE &	0.5037	&	0.4857	&	0.4883	\\
% \hline
% \multicolumn{4}{c}{Design 3} \\
% \hline
% Estimand 	&	1.6246	&	1.6246	&	1.6246	\\
% Estimate	&	1.6202	&	1.6667	&	1.7024	\\
% Oracle SE	&	0.8634	&	0.7979	&	0.8081	\\
% SE	&	0.8643	&	0.7943	&	0.7901	\\
\hline
\hline
\end{tabular}
\caption*{\footnotesize Note: Design 1 is no interference but with varying propensity scores, Design 2 is with interference and constant propensity score, and Design 3 is with interference and varying propensity scores.}
\end{table}

% \subsubsection{Design 1: No interference but with varying propensity scores}
\subsubsection{Design 1: No interference but with varying propensity scores}
\label{design1}
Under no interference, we set $b_n=0$.
We first simplify the formulations in Theorem \ref{thm:efficiency}. 
Recall $M(i,t) = \mu_i(t)-\mu(t)$. 
We introduce the following $Q$'s to simplify presentations:
\begin{align*}
Q_{xx} 
=& n^{-1} \sum_{i=1}^n x_ix_i^\top, 
\quad  
Q_{xx}(t;\pi) 
= n^{-1} \sum_{i=1}^n 
x_ix_i^\top \tfrac{1- \pi_i(t)}{\pi_i(t)}, \\
Q_{xx}(t;\pi^{-1}) 
=& n^{-1} \sum_{i=1}^n
x_ix_i^\top \tfrac{1}{ \pi_i(t) }, 
\quad 
Q_{xx}(t,t';\pi) 
= n^{-1}  \sum_{i=1}^n
x_ix_i^\top \tfrac{ \pi_i(t) + \pi_i(t') }{\pi_i(t)\pi_i(t')}. 
\end{align*}
By comparing variances, we have
% \gms{unadjusted VS interactive}
\begin{align*}
{{\Sigma}}_{*, \textup{haj}}(t,t') 
&= {{\Sigma}}_{*, \textup{haj}}(t,t', \gamma_\textsc{l}) 
+ {\gamma}^\top_\textsc{l}(t)
Q_{xx}(t;\pi)
{\gamma}_\textsc{l}(t)
+ {\gamma}^\top_\textsc{l}(t')
Q_{xx}(t';\pi)
{\gamma}_\textsc{l}(t')
+ 2 {\gamma}^\top_\textsc{l}(t)
Q_{xx}
{\gamma}_\textsc{l}(t') \\
&+ \underbrace{\frac{2}{n} \sum_{i=1}^n
% \left[
% \begin{array}{c}
\left(  \frac{ (M(i,t) -  {\gamma}_\textsc{l}(t)^\top x_i  ) {\gamma}_\textsc{l}(t)^\top }{\pi_i(t)}  
+ \frac{ (M(i,t') -  {\gamma}_\textsc{l}(t')^\top x_i  ) {\gamma}_\textsc{l}(t')^\top  }{\pi_i(t')}  \right) x_i 
% \left( \frac{{1}_i(t){\gamma}^\top_\textsc{l}(t)x_i }{\pi_i(t)} 
% - \frac{{1}_i(t'){\gamma}^\top_\textsc{l}(t')x_i }{\pi_i(t')} \right)
% \end{array}
}_{T_\textsc{l}} 
% -& ({\gamma}_\textsc{l}(t) - {\gamma}_\textsc{l}(t'))^\top
% \frac{2}{n} \sum_{i=1}^n
% % \left[
% \begin{array}{c}
% \left(  (\mu_i(t)-{\gamma}^\top_\textsc{l}(t) x_i-\mu(t) )  - (\mu_i(t')-{\gamma}^\top_\textsc{l}(t') x_i-\mu(t') )  \right) 
%  x_i
% \end{array}
+ o_\mathbb{P}(1),
\end{align*} 
% Under constant propensity score,
% \begin{align*}
% & \hat{{\Sigma}}_{*, \textup{haj}}^2(t,t') 
% = \hat{{\Sigma}}_{*, \textup{haj}}^2(t,t', \gamma_\textsc{l})
% + {\gamma}^\top_\textsc{l}(t)
% \frac{2}{n} \sum_{i=1}^n 
% x_ix_i^\top
% {\gamma}_\textsc{l}(t') \\
% +& {\gamma}_\textsc{l}(t)^\top
% \frac{1- \pi(t) }{\pi(t)}
% \frac{1}{n} \sum_{i=1}^n 
% x_ix_i^\top
% {\gamma}_\textsc{l}(t)
% + {\gamma}_\textsc{l}(t')^\top
% \frac{1- \pi(t') }{\pi(t')}
% \frac{1}{n} \sum_{i=1}^n 
% x_ix_i^\top
% {\gamma}_\textsc{l}(t') \\
% +& \frac{2}{n} \sum_{i=1}^n
% % \left[
% \begin{array}{c}
% \left(  \frac{ {\gamma}_\textsc{l}(t)^\top }{\pi(t)}  (\mu_i(t)x_i- x_ix_i^\top {\gamma}_\textsc{l}(t) - \mu(t)x_i )
% +  \frac{ {\gamma}_\textsc{l}(t')^\top }{\pi(t')} (\mu_i(t')x_i -  x_ix_i^\top {\gamma}_\textsc{l}(t') - \mu(t')x_i ) \right) 
% \end{array} 
% + o_\mathbb{P}(1) \\
% =& \hat{{\Sigma}}_{*, \textup{haj}}^2(t,t', \gamma_\textsc{l})
% + {\gamma}^\top_\textsc{l}(t)
% \frac{2}{n} \sum_{i=1}^n 
% x_ix_i^\top
% {\gamma}_\textsc{l}(t') \\
% +& {\gamma}_\textsc{l}(t)^\top
% \frac{1- \pi(t) }{\pi(t)}
% \frac{1}{n} \sum_{i=1}^n 
% x_ix_i^\top
% {\gamma}_\textsc{l}(t)
% + {\gamma}_\textsc{l}(t')^\top
% \frac{1- \pi(t') }{\pi(t')}
% \frac{1}{n} \sum_{i=1}^n 
% x_ix_i^\top
% {\gamma}_\textsc{l}(t') 
% + o_\mathbb{P}(1)
% \end{align*}
% \gms{additive VS interactive}
and
\begin{align*}
{{\Sigma}}_{*, \textup{haj}}& (t,t', {\gamma}_\textsc{f}) 
= {{\Sigma}}_{*, \textup{haj}}(t,t', \gamma_\textsc{l})  
+ ({\gamma}_\textsc{l}(t) - {\gamma}_\textsc{f})^\top
Q_{xx}(t;\pi)
({\gamma}_\textsc{l}(t) - {\gamma}_\textsc{f}) \\
+& ({\gamma}_\textsc{l}(t') - {\gamma}_\textsc{f})^\top
Q_{xx}(t';\pi)
({\gamma}_\textsc{l}(t') - {\gamma}_\textsc{f})
+ 2 ({\gamma}_\textsc{l}(t) - {\gamma}_\textsc{f})^\top
Q_{xx}
({\gamma}_\textsc{l}(t') - {\gamma}_\textsc{f}) \\
+& \underbrace{\frac{2}{n} \sum_{i=1}^n
% \begin{array}{c}
\left( \frac{(M(i,t)-{\gamma}^\top_\textsc{l}(t) x_i )({\gamma}_\textsc{l}(t)- {\gamma}_\textsc{f})^\top }{\pi_i(t)}  
+ \frac{(M(i,t')-{\gamma}^\top_\textsc{l}(t') x_i )({\gamma}_\textsc{l}(t')- {\gamma}_\textsc{f})^\top }{\pi_i(t')}  \right) x_i
% \end{array} 
}_{T_{\textsc{l}-\textsc{f}}}
+ o_\mathbb{P}(1).
\end{align*}
Recall ${\gamma}_\textsc{l}(t) = (\sum_{i=1}^n x_ix_i^\top)^{-1} \sum_{i=1}^n x_i \mu_i(t)$. When the propensity score is constant, both $T_{\textsc{l}}$ and $T_{\textsc{l} - \textsc{f}}$ are equal to zero. The variance difference between the unadjusted regression and fully-interacted regression,
${{\Sigma}}_{*, \textup{haj}}(t,t')- {{\Sigma}}_{*, \textup{haj}}(t,t', \gamma_\textsc{l})$, recovers Corollary 1.1 in \cite{Lin2013}.
If we further assume there are only two exposure mapping values, i.e., $\pi(t)+\pi(t')=1$, then ${{\Sigma}}_{*, \textup{haj}}(t,t',\gamma_\textsc{f})- {{\Sigma}}_{*, \textup{haj}}(t,t', \gamma_\textsc{l}) $ recovers Corollary 1.2 in \cite{Lin2013}.
Therefore, including the fully interacted covariates leads to an efficiency gain. 
However, this efficiency gain may be compromised when the propensity scores vary; see also \cite{SuDing2021}, \cite{ZhaoDing2022} and \cite{ZhaoDingLi2024}. 

Now we present a data generation process where the inclusion of covariates can lead to a less efficient result. 
% We set the sample size $n=1000$ and monte carlo simulation $B=20,000$.
We set the potential outcome model as
% $Y_i = V_i( {D},  {x},  {\varepsilon})$ with 
\begin{equation*}
Y_i(D)
= \beta_0 + \beta_1 D_i  + \beta_2  x_i + \beta_3 D_i \exp(x_i^2) + \varepsilon_i
\end{equation*}
with $(\beta_0, \beta_1, \beta_2, \beta_3)=(1,4,2,0.1)$.
% For the complex contagion model, we set $Y_i={1}\left\{V_i( {D},  {x},  {\varepsilon})>0\right\}$. 
The propensity score $\pi_i$'s are drawn from $\text{Uniform}[0.1, 0.9]$, and the treatment assignment $D_i$'s are drawn from $\text{Bern}(\pi_i)$ for each $i$.
The covariate $x_i$ is drawn from $\mathcal{N}(0,1)$ and $\varepsilon_i$ is drawn from $\mathcal{N}(0,1)$. 
We study the treatment effect with exposure mapping $T_i = D_i$. 
% Table \ref{table: Counterexample1} displays the estimand and corresponding oracle standard error of $\tau(t,t')$. 
The results of Table \ref{table: Counterexample} under Design 1 indicate that incorporating the fully interacted covariates can result in a less efficient estimator compared with those from the unadjusted and additive WLS fits.

% \begin{table}[tbp]
% \centering
% \caption{\label{table: Counterexample1}No interference but with varying propensity scores }
% \begin{tabular}{lrrr}
% \hline
% \hline
% \multicolumn{4}{c}{SIMULATION RESULTS}  \\ 
% \hline
% % & \multicolumn{6}{c}{\# Schools = 1} \\
% % \cline{2-7}
% % & \multicolumn{3}{c}{Linear-in-Means} \\
% % \cline{2-4} 
% & \multicolumn{1}{r}{Unadj} & \multicolumn{1}{r}{Add} & \multicolumn{1}{r}{Sat} \\
% \hline
% Estimand 	&		&	0.5702	&		\\
% Estimate	&	0.5681	&	0.5779	&	0.5786	\\
% Oracle SE	&	0.5061	&	0.4857	&	0.4890	\\
% SE	&	0.5037	&	0.4857	&	0.4883	\\
% \hline
% \hline
% \end{tabular}
% \captionsetup{font=small,labelfont=bf}
% % \caption*{Note: over 20,000 simulations.}
% \end{table}

\subsubsection{Design 2: With interference and constant propensity score}
\label{design2}
Under constant propensity score, $\pi_i(t') = \pi(t)$, for all $t\in\mathcal{T}$. The formulations in Theorem \ref{thm:efficiency} can be simplified as follows: 
\begin{align*}
{{\Sigma}}_{*, \textup{haj}}&(t,t')
= {{\Sigma}}_{*, \textup{haj}}(t,t', \gamma_\textsc{l})
+ \tfrac{1- \pi(t) }{\pi(t)}
{\gamma}_\textsc{l}(t)^\top
Q_{xx}
{\gamma}_\textsc{l}(t)
+ \tfrac{1- \pi(t') }{\pi(t')}
{\gamma}_\textsc{l}(t')^\top
Q_{xx}
{\gamma}_\textsc{l}(t')
+ 2 {\gamma}_\textsc{l}(t)^\top
Q_{xx}
{\gamma}_\textsc{l}(t') \\
&+  \underbrace{\frac{1}{n} 
\sum_{\substack{i \ne j: \\
\ell_{A}(i, j) \leq b_n }}
\left( 
\Delta_{i}(t;{\gamma}_\textsc{l})
- \Delta_{i}(t';{\gamma}_\textsc{l})
+ 2(\tilde{\Delta}_{i}(t;{\gamma}_\textsc{l}) - \tilde{\Delta}_{i}(t';{\gamma}_\textsc{l}) )
\right)
\left( 
\Delta_{j}(t;{\gamma}_\textsc{l})
- \Delta_{j}(t';{\gamma}_\textsc{l})
\right) }_{\tilde{T}_\textsc{l}},
% &+ \frac{2}{n} 
% \sum_{i \ne j: \ell_{ {A}}(i, j) \leq b_n}
% \left( 
% \tilde{\Delta}_{i}(t;{\gamma}_\textsc{l}) - \tilde{\Delta}_{i}(t';{\gamma}_\textsc{l}) 
% \right) 
% \left( 
% \Delta_{j}(t;{\gamma}_\textsc{l})
% - \Delta_{j}(t';{\gamma}_\textsc{l})
%  \right)
\end{align*} 
% \begin{align*}
%   Q_{xx} =& n^{-1} \sum_{i=1}^n x_ix_i^\top, \quad  
%   Q_{xx}(t;\pi) = n^{-1} \sum_{i=1}^n 
%   x_ix_i^\top \frac{1- \pi_i(t) }{\pi_i(t)}, \\
%   Q_{xx}(t;\pi^{-1}) =& n^{-1} \sum_{i=1}^n
%   \frac{ x_ix_i^\top }{ \pi_i(t) }, \quad 
%   Q_{xx}(t,t';\pi) =  n^{-1}  \sum_{i=1}^n
%   x_ix_i^\top \frac{ \pi_i(t) + \pi_i(t') }{\pi_i(t)\pi_i(t')}. 
%   \end{align*}
and 
\begin{align*}
& {{\Sigma}}_{*, \textup{haj}}(t,t',{\gamma}_\textsc{f})
= {{\Sigma}}_{*, \textup{haj}}(t,t', \gamma_\textsc{l}) 
+ \tfrac{1-\pi(t)}{\pi(t)}
({\gamma}_\textsc{l}(t) - {\gamma}_\textsc{f})^\top
Q_{xx}
({\gamma}_\textsc{l}(t) - {\gamma}_\textsc{f}) \\
&+ \tfrac{1-\pi(t')}{\pi(t')}
({\gamma}_\textsc{l}(t') - {\gamma}_\textsc{f})^\top
Q_{xx}
({\gamma}_\textsc{l}(t') - {\gamma}_\textsc{f}) 
+ 2 ({\gamma}_\textsc{l}(t) - {\gamma}_\textsc{f})^\top
Q_{xx}
({\gamma}_\textsc{l}(t') - {\gamma}_\textsc{f}) \\
&+ \underbrace{\frac{1}{n} 
\sum_{\substack{i \ne j: \\
\ell_{A}(i, j) \leq b_n }}
\left( \Delta_{i}(t;{\gamma}_{\textsc{l}-\textsc{f}}) - \Delta_{i}(t';{\gamma}_{\textsc{l}-\textsc{f}})
+ 2(\tilde{\Delta}_{i}(t;{\gamma}_\textsc{l}) - \tilde{\Delta}_{i}(t';{\gamma}_\textsc{l})) \right) 
( \Delta_{j}(t;{\gamma}_{\textsc{l}-\textsc{f}}) - \Delta_{j}(t';{\gamma}_{\textsc{l}-\textsc{f}}) )}_{\tilde{T}_{\textsc{l}-\textsc{f}}}. 
% +& \frac{2}{n} 
% \sum_{i\ne j: \ell_{ {A}}(i, j) \leq b_n}
% \left( \tilde{\Delta}_{i}(t;{\gamma}_\textsc{l}) - \tilde{\Delta}_{i}(t';{\gamma}_\textsc{l}) \right) 
% \left( \Delta_{j}(t;{\gamma}_{\textsc{l}-\textsc{f}}) - \Delta_{j}(t';{\gamma}_{\textsc{l}-\textsc{f}}) \right).
\end{align*}
If there is no interference, $\tilde{T}_\textsc{l}=0$ and $\tilde{T}_{\textsc{l}-\textsc{f}}=0$.
Once interference is present, the efficiency gains from incorporating fully interacted covariates can be compromised, even under a constant propensity score.
% In our simulation study described below, we established a sample size of $n=1000$ and conducted a Monte Carlo simulation involving $B=20,000$ repetitions.
As an example, we consider a potential outcome model of the form
\[
Y_i(D) 
= \beta_0 + \beta_1 \sum_{j=1}^n \tilde{A}_{i j} D_j + \beta_2 D_i + \beta_3 x_i + \beta_4 D_i \exp(x_i) + \beta_5 \sum_{j=1}^n \tilde{A}_{i j} x_j + \varepsilon_i,
\]
where $\tilde{A}_{i j} = A_{ij}/\sum_{j=1}^n A_{i j}$, and so that $b_n=2$.
As a special case, $\tilde{T}_\textsc{l}$ and $\tilde{T}_{\textsc{l}-\textsc{f}}$ become negative when the exposure mapping indicators are negatively correlated for pairs $(i,j)$ satisfying $\ell_{A}(i, j) \leq b_n$ with negative $\beta_1$, $\beta_3$ and $\beta_5$.
In our simulation, we set $(\beta_0, \beta_1, \beta_2, \beta_3, \beta_4, \beta_5) = (1,-0.9,6,-1, 0.2,-3)$. 
We simulate experiments in which 1/10 of the units are randomly assigned to treatment. 
To investigate the direct effect, we use the exposure mapping $T_i = D_i$ with a constant propensity score $\pi(1) = 1/10$.
Under this design, the exposure indicators $1_i(t)$ and $1_j(t)$ are negatively correlated because treating unit $i$ reduces the probability of treating unit $j$.
We simulate ${A}$ from random geometric graph models where $A_{i j}={1}\left\{\left\|\rho_i-\rho_j\right\| \leq r_n\right\}$ for $\rho_i \stackrel{\text{ IID }}{\sim} \text{Uniform}\left([0,1]^2\right)$ and $r_n=(\kappa /(\pi n))^2$ with $\kappa = 8$. 
The covariate $x_i$ is drawn from $\mathcal{N}(0,1)$ and $\varepsilon_i$ is drawn from $\mathcal{N}(0,16)$. 
% In this design, we follow the DGP in the subsection \ref{Simulation}, but set $(\alpha, \beta, \delta, \xi, \gamma)=(-1,0.8, 1, 3,-0.01)$ and $(\alpha, \beta, \delta, \xi, \gamma)=(-1, 1.5, 1, 3,-0.01)$. 
The results of Table \ref{table: Counterexample} under Design 2 indicate that incorporating the fully interacted covariates can result in a less efficient estimator compared with those from the unadjusted and/or additive WLS fits. 
In other words, when the exposure indicators among neighbors are positively correlated, efficiency gains from fully interacted covariate adjustment are more likely to be observed.

\subsubsection{Design 3: With interference and varying propensity scores} 
\label{design3}
% We present an illustrative example where the inclusion of fully interacted covariates can indeed lead to an efficiency gain. 
% In this instance, we consider a scenario with a sample size of $n=1000$ and $B=20,000$ Monte Carlo iterations.
Following the discussions in Sections \ref{design1} and \ref{design2}, it is not surprising that the efficiency gains from incorporating fully interacted covariates can be compromised in settings with interference and varying propensity scores. 
We consider the following outcome model:
\[
Y_i=\beta_0 + \beta_1 \sum_{j=1}^n \tilde{A}_{i j} D_j + \beta_2 D_i + \beta_3 x_i + \beta_4 D_i \exp(x_i) + \beta_5 \sum_{j=1}^n \tilde{A}_{i j} x_j + \varepsilon_i,
\]
where $\tilde{A}_{i j} = A_{ij}/\sum_{j=1}^n A_{i j}$, and
so that $b_n=2$.
Similar to the special case in Design 2, an efficiency loss occurs when the exposure mapping indicators are negatively correlated for pairs $(i,j)$ satisfying $\ell_{A}(i, j) \leq b_n$ and when $\beta_1$, $\beta_3$ and $\beta_5$ are negative.
In this setting, we set $(\beta_0, \beta_1, \beta_2, \beta_3, \beta_4, \beta_5) = (1,-0.9,6,-1,0.2,-3)$. 
We simulate ${A}$ from random geometric graph models with $A_{i j}={1}\left\{\left\|\rho_i-\rho_j\right\| \leq r_n\right\}$ for $\rho_i \stackrel{\text { IID }}{\sim} \text{Uniform}\left([0,1]^2\right)$ and $r_n=(\kappa /(\pi n))^2$ with $\kappa = 5$. The covariate $x_i$ is drawn from $\mathcal{N}(0,1)$ and $\varepsilon_i$ is drawn from $\mathcal{N}(0,16)$.  
In the experimental design, units are randomly ordered and treatment is assigned sequentially. Each unit $i$ is initially assigned a baseline treatment probability $p_i$ drawn from $\text{Uniform}[0.4, 0.8]$.
To investigate the direct effect, we use the exposure mapping $T_i = D_i$.
Then for each unit $i$, the effective propensity score for unit $i$ is given by
\[
\pi_i(1) =
\begin{cases}
p_i / 4, & \text{if at least one neighbor (processed before $i$) is treated}, \\
p_i,   & \text{otherwise}.
\end{cases}
\] 
This describes an experimental design in which the treatment probability for a unit is reduced to one-quarter whenever at least one of its neighbors (already processed) has been treated.
The results of Table \ref{table: Counterexample} under Design 3 show that including the fully interacted covariates can lead to an less efficient estimator compared with those from the unadjusted and/or additive WLS fits.

\subsection{Additional simulation result}
\label{sec:additional_sim}
To illustrate variations in population sizes, we also conduct simulation with $n=805$ and $n=2725$ using the sample of the largest and four largest treated schools from the network experiment in Section \ref{app:paluck} to calibrate the network models. 
% The network sample size $n=805$ and $2725$, respectively.
The data generation process is descibed in Section \ref{Simulation}.
Tables \ref{table: Haj RGG 1} and \ref{table: Haj RGG 4} present the results. 
For each table, we provide results from two outcome models: linear-in-means and complex contagion models. 
The top panels of Tables \ref{table: Haj RGG 1} and \ref{table: Haj RGG 4} display our regression-based results. 
The middle panels report the results of standard errors and coverage rates of $95\%$ CIs using the kernel ${K}^{\textup{L}2019}_n$ in \cite{Leung2019e} and the kernel ${K}^{\textup{K}2021}_n$ in \cite{Kojevnikov2021}. The bottom panels present the results of the Horvitz--Thompson estimator and variance estimator from \cite{Leung2022}.

The result tables demonstrate that the Hájek estimator can be biased when the sample size is small, but the bias diminishes as the sample size increases.
The coverage rate of the adjusted HAC standard errors improves as the (effective) sample size increases.
In Table \ref{table: Haj RGG 4}, our regression-based standard errors are approximately half of those reported using \cite{Leung2022}'s method, indicating a significant spillover effect at the $5\%$ significance level. 
In contrast, \cite{Leung2022}'s method yields an insignificant effect. 
In this setting, the estimator from the fully-interacted WLS fit is at least as efficient as the estimators from the unadjusted or additive WLS fits.
% By comparing the ``Oracle SE'' from the top and bottom panels, we can see the WLS estimators from all three specifications exhibit higher efficiency compared with the Horvitz--Thompson estimator.
% In Table \ref{table: Haj RGG 4}, our regression-based standard errors are approximately half of those reported using \cite{Leung2022}'s method, indicating a significant spillover effect at the $5\%$ significance level. 
% In contrast, \cite{Leung2022}'s method yields an insignificant effect. 
% Moreover, \cite{Leung2022}’s standard errors are smaller than the oracle standard errors, resulting in under coverage.

\begin{table*}
\centering
\caption{Simulation results: network size $n=805$ }
\begin{tabular}{lcccccc}
\hline
\hline
% \multicolumn{7}{c}{SIMULATION RESULTS}  \\ 
% \hline
Outcome model & \multicolumn{3}{c}{Linear-in-Means} & \multicolumn{3}{c}{Complex   Contagion} \\
% \cline{2-4} \cline{5-7}
\hline
WLS specification & \multicolumn{1}{r}{Unadj} & \multicolumn{1}{r}{Add} & \multicolumn{1}{r}{Sat} & \multicolumn{1}{r}{Unadj} & \multicolumn{1}{r}{Add} & \multicolumn{1}{r}{Sat} \\
\hline													
${\tau}(1,0)$	&		&	0.470	&		&		&	0.023	&		\\
$\hat{\tau}(1,0)$	&	0.449	&	0.482	&	0.478	&	0.022	&	0.026	&	0.026	\\
Oracle SE	&	1.104	&	0.851	&	0.850	&	0.063	&	0.042	&	0.042	\\
WLS SE	&	1.031	&	0.779	&	0.772	&	0.060	&	0.041	&	0.041	\\
WLS$^+$ SE	&	1.077	&	0.800	&	0.793	&	0.069	&	0.047	&	0.047	\\
EHW SE	&	0.529	&	0.356	&	0.354	&	0.054	&	0.035	&	0.035	\\
Oracle Coverage	&	0.955	&	0.952	&	0.952	&	0.954	&	0.951	&	0.951	\\
WLS Coverage	&	0.945	&	0.938	&	0.936	&	0.935	&	0.934	&	0.933	\\
WLS$^+$ Coverage	&	0.954	&	0.946	&	0.943	&	0.970	&	0.964	&	0.964	\\
EHW Coverage	&	0.655	&	0.595	&	0.593	&	0.909	&	0.898	&	0.897	\\
\hline													
\cite{Leung2019e} SE	&	0.960	&	0.710	&	0.703	&	0.060	&	0.040	&	0.040	\\
\cite{Kojevnikov2021} SE	&	0.986	&	0.722	&	0.714	&	0.062	&	0.041	&	0.041	\\
\cite{Leung2019e} Coverage	&	0.927	&	0.911	&	0.908	&	0.938	&	0.929	&	0.929	\\
\cite{Kojevnikov2021} Coverage	&	0.934	&	0.917	&	0.914	&	0.947	&	0.938	&	0.937	\\
\hline													
$\hat{\tau}_\text{ht}(1,0)$	&	0.649	&		&		&	0.027	&		&		\\
Oracle SE	&	1.920	&		&		&	0.141	&		&		\\
\cite{Leung2022} SE	&	1.814	&		&		&	0.137	&		&		\\
Oracle Coverage	&	0.950	&		&		&	0.952	&		&		\\
\cite{Leung2022} Coverage	&	0.932	&		&		&	0.937	&		&		\\
\hline
\hline
\end{tabular}
\label{table: Haj RGG 1}
% \captionsetup{font=small,labelfont=bf}
\caption*{\footnotesize Note: The effective sample size for each exposure mapping value is $\hat{n}(1)=226$ and $\hat{n}(0)=170$, with a total of $\hat{n}(1) + \hat{n}(0) = 396$.
The suggested bandwidth in \eqref{eq:bandwidth} is $b_n = 2$. 
The average path length is $\mathcal{L}({A})=14.29$. 
}
\end{table*}

\begin{table*}[tbp]
\centering
\caption{Simulation results: network size $n=2725$}
\label{table: Haj RGG 4}
\begin{tabular}{lcccccc}
\hline
\hline
% \multicolumn{7}{c}{SIMULATION RESULTS}  \\ 
% \hline
% & \multicolumn{6}{c}{\# Schools = 4} \\
% \cline{2-7}
Outcome model & \multicolumn{3}{c}{Linear-in-Means} & \multicolumn{3}{c}{Complex   Contagion} \\
\hline
% \cline{2-4} \cline{5-7}
WLS specification & \multicolumn{1}{c}{Unadj} & \multicolumn{1}{c}{Add} & \multicolumn{1}{c}{Sat} & \multicolumn{1}{r}{Unadj} & \multicolumn{1}{c}{Add} & \multicolumn{1}{c}{Sat} \\
\hline
${\tau}(1,0)$	&		&	0.685	&		&		&	0.028	&		\\
$\hat{\tau}(1,0)$	&	0.678	&	0.027	&	0.666	&	0.026	&	0.666	&	0.026	\\
Oracle SE	&	0.573	&	0.031	&	0.445	&	0.019	&	0.445	&	0.019	\\
WLS SE	&	0.560	&	0.031	&	0.434	&	0.020	&	0.434	&	0.020	\\
WLS$^+$ SE	&	0.608	&	0.038	&	0.462	&	0.024	&	0.462	&	0.024	\\
EHW SE	&	0.279	&	0.028	&	0.185	&	0.019	&	0.185	&	0.019	\\
Oracle Coverage	&	0.954	&	0.952	&	0.954	&	0.948	&	0.954	&	0.947	\\
WLS Coverage	&	0.947	&	0.947	&	0.946	&	0.944	&	0.946	&	0.943	\\
WLS$^+$ Coverage	&	0.965	&	0.983	&	0.960	&	0.984	&	0.960	&	0.984	\\
EHW Coverage	&	0.670	&	0.928	&	0.591	&	0.940	&	0.590	&	0.940	\\
\hline													
\cite{Leung2019e} SE	&	0.524	&	0.031	&	0.402	&	0.019	&	0.402	&	0.019	\\
\cite{Kojevnikov2021} SE	&	0.532	&	0.032	&	0.408	&	0.020	&	0.408	&	0.020	\\
\cite{Leung2019e} Coverage	&	0.930	&	0.946	&	0.926	&	0.944	&	0.926	&	0.944	\\
\cite{Kojevnikov2021} Coverage	&	0.934	&	0.952	&	0.930	&	0.950	&	0.929	&	0.950	\\
\hline													
$\hat{\tau}_\text{ht}(1,0)$	&	0.677	&		&		&	0.025	&		&		\\
Oracle SE	&	0.927	&		&		&	0.082	&		&		\\
\cite{Leung2022} SE	&	0.916	&		&		&	0.081	&		&		\\
Oracle Coverage	&	0.949	&		&		&	0.951	&		&		\\
\cite{Leung2022} Coverage	&	0.942	&		&		&	0.944	&		&		\\
\hline
\hline
\end{tabular}
% \captionsetup{font=small,labelfont=bf}
\caption*{\footnotesize Note: The effective sample size for each exposure mapping value is $\hat{n}(1)=849$ and $\hat{n}(0)=595$, with a total of $\hat{n}(1) + \hat{n}(0) = 1444$.
The suggested bandwidth in \eqref{eq:bandwidth} is $b_n = 3$. 
The average path length is $\mathcal{L}({A})=24.81$. }
\end{table*}

\subsection{Empirical Application II: \cite{CaiJanvrySadoulet2015}}
\label{sec:Cai}

\cite{CaiJanvrySadoulet2015} conducted an experiment in rural China to investigate how farmers' understanding of a weather insurance policy affects their purchasing decisions. 
% The study focused on determining the impact of two types of information sessions on insurance adoption among 4902 households in 173 small rice-producing villages (47 administrative villages) located in three regions within the Jiangxi province, China.
The main outcome of interest was whether a household decided to purchase the insurance policy or not. In each village, the experiment included two rounds of information sessions to introduce the insurance product. Each round consisted of two simultaneous sessions: a simple session with less information and an intensive session. The second round of information sessions was scheduled three days after the first round, allowing farmers to communicate with friends. However, this time gap was designed to prevent all the information from the first round from spreading widely throughout the entire population via the network.

While the original experiment included a village-level randomization with price variation and a second round of sessions, we focus only on the household-level randomization.
For household-level randomization, \cite{CaiJanvrySadoulet2015} initially computed the median values of household size and area of rice production per capita within each village. They then created dummy variables for each household, indicating whether their respective variables were above or below the median. Using this information, households were divided into four strata groups. 
All households in the sample are randomly assigned to one of four sessions: first round simple, first round intensive, second round simple, or second round intensive.
We use the variables $\text{Delay}_i$ and $\text{Int}_i$ to indicate whether household $i$ attended the first round ($\text{Delay}_i = 0$) or the second round ($\text{Delay}_i = 1$) of sessions, and whether they attended a simple session ($\text{Int}_i = 0$) or an intensive session ($\text{Int}_i = 1$), respectively.
In Section \ref{sec:setup}, we consider a binary treatment for simplicity, although this assumption is not crucial to our theory. We maintain the flexibility to extend it to discrete treatments with finite and fixed dimensions, like $D_i = (\text{Delay}_i, \text{Int}_i) \in \{0, 1\}^2$ in this experiment. 
% We will analyze an empirical application where $D_i \in \{0, 1\}^2$.

The network information is measured by asking household heads to list five close friends, either within or outside the village, with whom they most frequently discussed rice production or financial issues.  Consequently, ${A}$ is directed. Moreover, respondents were also asked to rank these friends based on which one would be consulted first, second, etc. But in our paper, we do not consider this ranking and instead assign equal weight to each link. Again, 
we incorporate link directionality when calculating the number of treated friends for exposure mappings but omit it when defining network neighborhoods in a conservative manner for covariance estimators.
We include age and education as covariates.

Define $A^2$ as the square of the adjacency matrix $A$, with $(i,j)$th entry $(A^2)_{ij} \;=\;\sum_{k=1}^n A_{ik}\,A_{kj}$, which counts the number of common friends between units $i$ and $j$. 
Define $B_{ij} = 1((A^2)_{ij} \ge 1 \text{ and } A_{ij}=0, i\neq j)$ to denote that units $i$ and $j$ share at least one friend but are not directly connected.
We define the following exposure mappings: 
\begin{align*}
T_{1i} =& \text{Int}_i \cdot \text{Delay}_i, \\
T_{2i} =& 1\left( \sum_{j=1}^n A_{ij}(1-\text{Delay}_j)\text{Int}_j > 0 \right), \\
T_{3i} =& 1\left( \sum_{j=1}^n B_{ij}(1-\text{Delay}_j)\text{Int}_j > 0 \right),
\end{align*}
where $T_{1i}$ captures the direct effect of participating in the second-round intensive sessions, $T_{2i}$ captures the spillover effect of having at least one friend attend the first-round intensive sessions, 
and $T_{3i}$ captures the spillover effect of having at least one friend-of-friend attend the first-round intensive sessions. 

% In Tables \ref{table: cai_one}, \ref{table: cai_two} and \ref{table: cai_three}, we present the results of all effects by defining three one-dimensional exposure mappings, one two-dimensional exposure mapping, and one three-dimensional exposure mapping, respectively.

We consider regression specifications with exposure mappings of different dimensions to assess how robust our results are to the choice of mapping.
For all specifications, we restrict our effective sample to households that attended the second-round session, had at least one friend attend the first-round sessions and had at least one friend-of-friend attend the first-round sessions to satisfy Assumption \ref{asu2}, resulting in a total of $1527$ households.
\paragraph*{One-dimensional exposure mappings}
% We present the results of the direct and spillover effects using two one-dimensional exposure mappings in the bottom panel of Table \ref{table: cai}.
We consider three one-dimensional exposure mappings: $T_{1i}$, $T_{2i}$ and $T_{3i}$. We set $G = (-1, 1)$ for each of them.
% First, we define \( T_{1i} = \text{Int}_i \cdot \text{Delay}_i \), as in Example~\ref{ex:direct}, to calculate the direct effect of participating in the second-round intensive session. 
% Second, we define \( T_{2i} = 1( \sum_{j=1}^n A_{ij}(1-\text{Delay}_j)\text{Int}_j > 0 ) \), as in Example~\ref{ex:one-dim}, to calculate the spillover effect of having at least one friend attend the first-round intensive sessions. 
% Let \( (A^2)_{ij} \) denote whether units \( i \) and \( j \) share any common friends. 
% Third, we define \( T_{3i} = 1( \sum_{j=1}^n (A^2)_{ij}(1-\text{Delay}_j)\text{Int}_j > 0 ) \) to calculate the spillover effect of having at least one friend-of-friend attend the first-round intensive sessions. 
The results are presented in the top panel of Table~\ref{table: cai_one}.
For comparison, we also report results using the standard errors proposed by \citet{Leung2022} in the bottom panel of Table~\ref{table: cai_one}.

\paragraph*{Two-dimensional exposure mapping}
We define the two-dimensional exposure mapping as $T_i = (T_{1i}, T_{2i}) \in \{(0,0), (0,1), (1,0), (1,1)\}$
% $T_i = ( \text{Int}_i \cdot \text{Delay}_i, {1}(\sum_{j=1}^n A_{ij}(1-\text{Delay}_j)\text{Int}_j > 0),{1}(\sum_{j=1}^n (A^2)_{ij}(1-\text{Delay}_j)\text{Int}_j > 0) )$,  
and set $G=\left(g_{1}, g_{2}\right)^\top$ with $g_{1}=2^{-1}(-1, -1, 1,  1)^\top $ and $ g_{2}=2^{-1}(-1, 1,  -1,  1)^\top$, as in Example \ref{ex:two-dim}. Then the estimand $\tau$ recovers the direct effect and the spillover effect of having at least one treated friends. 
The results are presented in Table \ref{table: cai_two}. 

\paragraph*{Three-dimensional exposure mapping}
We define the three-dimensional exposure mapping as $T_i = (T_{1i}, T_{2i}, T_{3i}) \in \{(0,0,0), (0,0,1), (0,1,0), (0,1,1),(1,0,0), (1,0,1), (1,1,0), (1,1,1)\}$.
% $T_i = ( \text{Int}_i \cdot \text{Delay}_i, {1}(\sum_{j=1}^n A_{ij}(1-\text{Delay}_j)\text{Int}_j > 0),{1}(\sum_{j=1}^n (A^2)_{ij}(1-\text{Delay}_j)\text{Int}_j > 0) )$,  
% where $T_i$ takes values in
% \[
% \{(0,0,0), (0,0,1), (0,1,0), (0,1,1),(1,0,0), (1,0,1), (1,1,0), (1,1,1)\}.
% \] 
Setting $G=\left(g_{1}, g_{2}, g_{3}\right)^\top$ with $g_{1}=4^{-1}(-1, -1, -1, -1,  1,  1,  1,  1)^\top, g_{2}=4^{-1}(-1, -1,  1,  1, -1, -1,  1,  1)^\top$, and $g_{3}=4^{-1}(-1,  1, -1,  1, -1,  1, -1,  1)^\top$, then the estimand $\tau$ recovers the direct effect, the spillover effect of having at least one treated friends, and the spillover effect of having at least one treated friend-of-friends. 
The results are presented in Table \ref{table: cai_three}. 

% We restrict the effective sample to the households attending the second round session with at least one friend attending the first round session, resulting in a total of 1056 households. 
% We first show the results of the treatment and spillover effects by running two separate regressions, and then show the results of calculating treatment and spillover effect simultaneously. 

% We can also simultaneously examine the treatment and spillover effects by 

For $T_{1i}$ and $T_{2i}$, the suggested bandwidth in \eqref{eq:bandwidth} is $b_n = 3$ with $K = 0$ and $K = 1$, while for $T_{3i}$, the suggested bandwidth is $b_n = 4$ with $K = 2$.
We present results for the bandwidths in $\{0, 2, 3, 4, 5\}$.
In Table~\ref{table: cai_one}, the direct effect is not statistically significant across regression specifications and bandwidths.  
In contrast, the spillover effect of having at least one treated friend, estimated using the Hájek estimator, is statistically significant at the 5\% level using the regression-based standard error and at the 10\% level using the modified standard error, although this significance does not hold consistently across all bandwidths.
Compared with the Horvitz--Thompson estimator results reported in \citet{Leung2022}, our method yields smaller standard errors, resulting in statistically significant findings.  
We also examined the spillover effect of having at least one friend-of-friend attend the first-round intensive session, but this effect was not statistically significant. This suggests that the spillover effect diminishes with social distance.
Overall, our findings closely align with the estimates reported in Table~2 of \citet{CaiJanvrySadoulet2015}.  
More specifically, providing intensive sessions on insurance and emphasizing the product's expected benefits to a targeted group of farmers generates a significant and positive spillover effect on others.  
The difference in the magnitude of point estimates stems from \citet{CaiJanvrySadoulet2015} using the count of friends attending the first-round intensive session, whereas we focus solely on whether at least one friend attended.
The findings presented in Tables~\ref{table: cai_one} and~\ref{table: cai_two} align with each other, demonstrating the robustness of our methods to variations in regression specifications and exposure mappings, provided that the exposure mappings are independent or only weakly dependent.
When including all exposure mappings in Table \ref{table: cai_three}, the direct effect and the spillover effect of having at least one treated friend-of-friend remain statistically insignificant. 
However, the spillover effect of having at least one treated friend becomes less significant compared with the results in Table~\ref{table: cai_one}, as the inclusion of many additional regressors substantially increases the standard error.
We also observe an efficiency loss associated with the inclusion of fully interacted covariates in this application, in line with the theoretical results presented in Section~\ref{Counterexample}. This arises because, given the number of treated units in each stratum, the number of treated friends and the number of treated friends-of-friends are negatively dependent.
The insignificant results in Table \ref{table: cai_three} likely reflects the small effective sample sizes produced by the more flexible exposure mapping. 
Comparing the results across Tables \ref{table: cai_one}--\ref{table: cai_three} reveals the traditional bias-variance trade-off: more flexible exposure mappings reduce bias but result in noisier estimates.

\begin{table}[ht]
\centering
\caption{\label{table: cai_one}Estimates and SEs (one-dimensional exposure mappings). }
\begin{tabular}{lccccccccc}
\hline \hline
& \multicolumn{3}{c}{Direct effect} 
& \multicolumn{3}{c}{Effect of direct friends} 
& \multicolumn{3}{c}{Effect of indirect friends} \\
\hline
WLS specification & Unadj & Add & Sat & Unadj & Add & Sat & Unadj & Add & Sat \\
\hline
% \multicolumn{10}{c}{One-dimensional exposure mapping}\\
% \hline
Estimate	&	0.009	&	0.010	&	0.011	&	0.054	&	0.056	&	0.056	&	0.014	&	0.014	&	0.014	\\
$b_n=0$	&	0.025	&	0.025	&	0.025	&	0.027	&	0.027	&	0.027	&	0.047	&	0.047	&	0.047	\\
$b_n=2$	&	0.027	&	0.027	&	0.027	&	0.029	&	0.029	&	0.029	&	0.047	&	0.047	&	0.047	\\
WLS$^{+}$ SE	&	0.031	&	0.030	&	0.030	&	0.033	&	0.033	&	0.033	&	0.055	&	0.055	&	0.055	\\
$b_n=3$	&	0.026	&	0.026	&	0.026	&	0.026	&	0.026	&	0.026	&	0.048	&	0.048	&	0.048	\\
WLS$^{+}$ SE	&	0.029	&	0.029	&	0.029	&	0.030	&	0.030	&	0.030	&	0.054	&	0.054	&	0.054	\\
$b_n=4$	&	0.027	&	0.027	&	0.026	&	0.026	&	0.026	&	0.026	&	0.047	&	0.047	&	0.047	\\
WLS$^{+}$ SE	&	0.031	&	0.031	&	0.030	&	0.031	&	0.031	&	0.031	&	0.054	&	0.053	&	0.053	\\
$b_n=5$	&	0.027	&	0.027	&	0.026	&	0.025	&	0.025	&	0.025	&	0.049	&	0.049	&	0.049	\\
WLS$^{+}$ SE	&	0.031	&	0.031	&	0.030	&	0.031	&	0.031	&	0.031	&	0.056	&	0.056	&	0.056	\\
\hline		
\multicolumn{10}{c}{Results using estimates and SE from \cite{Leung2022}} \\
\hline
Estimate	&	0.012	&	&		&	0.057	&	&	&	0.062	&		&		\\
$b_n=0$	&	0.034	&		&		&	0.036	&	&	&	0.055	&		&		\\
$b_n=2$	&	0.033	&		&		&	0.049	&	&	&	0.073	&		&		\\
$b_n=3$	&	0.032	&		&		&	0.048	&	&	&	0.081	&		&		\\
$b_n=4$	&	0.033	&		&		&	0.046	&	&	&	0.082	&		&		\\
$b_n=5$	&	0.030	&		&		&	0.046	&	&	&	0.080	&		&		\\
\hline \hline
\end{tabular}
\end{table}

\begin{table}[ht]
\centering
\caption{\label{table: cai_two}Estimates and SEs (two-dimensional exposure mapping). }
\begin{tabular}{lcccccc}
\hline \hline
& \multicolumn{3}{c}{Direct effect} 
& \multicolumn{3}{c}{Effect of direct friends} \\
% & \multicolumn{3}{c}{Effect of indirect friends} \\
\hline
WLS specification & Unadj & Add & Sat & Unadj & Add & Sat \\
% \midrule
\hline
% \hline
% \multicolumn{7}{c}{Two-dimensional exposure mapping}\\
% \hline
Estimate	&	0.013	&	0.014	&	0.017	&	0.056	&	0.058	&	0.055	\\
$b_n=0$	&	0.027	&	0.027	&	0.027	&	0.027	&	0.027	&	0.027	\\
$b_n=2$	&	0.028	&	0.028	&	0.027	&	0.027	&	0.027	&	0.027	\\
WLS$^{+}$ SE	&	0.033	&	0.032	&	0.032	&	0.033	&	0.033	&	0.033	\\
$b_n=3$	&	0.028	&	0.028	&	0.028	&	0.026	&	0.026	&	0.026	\\
WLS$^{+}$ SE	&	0.032	&	0.032	&	0.032	&	0.031	&	0.031	&	0.031	\\
$b_n=4$	&	0.028	&	0.028	&	0.028	&	0.025	&	0.026	&	0.025	\\
WLS$^{+}$ SE	&	0.034	&	0.033	&	0.033	&	0.031	&	0.031	&	0.031	\\
$b_n=5$	&	0.027	&	0.027	&	0.027	&	0.025	&	0.025	&	0.024	\\
WLS$^{+}$ SE	&	0.034	&	0.034	&	0.033	&	0.032	&	0.032	&	0.032	\\
\hline \hline
% \bottomrule
\end{tabular}
% \captionsetup{font=small,labelfont=bf}
% \caption*{Note: Columns display results for the treatment and spillover effects ($n=1056$) under three regression specifications. }
\end{table}

\begin{table}[ht]
\centering
\caption{\label{table: cai_three}Estimates and SEs (three-dimensional exposure mapping). }
\begin{tabular}{lccccccccc}
\hline \hline
& \multicolumn{3}{c}{Direct effect} 
& \multicolumn{3}{c}{Effect of direct friends} 
& \multicolumn{3}{c}{Effect of indirect friends} \\
\hline
WLS specification & Unadj & Add & Sat & Unadj & Add & Sat & Unadj & Add & Sat \\
% \midrule
\hline
% \hline
% \multicolumn{10}{c}{Two-dimensional exposure mapping}\\
% \hline
Estimate	&	0.046	&	0.048	&	0.063	&	0.077	&	0.079	&	0.075	&	$-$0.001	&	$-$0.001	&	$-$0.004	\\
$b_n=0$	&	0.049	&	0.048	&	0.049	&	0.050	&	0.050	&	0.051	&	0.051	&	0.051	&	0.051	\\
$b_n=2$	&	0.044	&	0.043	&	0.043	&	0.049	&	0.049	&	0.050	&	0.052	&	0.052	&	0.053	\\
WLS$^{+}$ SE	&	0.053	&	0.053	&	0.053	&	0.060	&	0.060	&	0.061	&	0.062	&	0.062	&	0.063	\\
$b_n=3$	&	0.042	&	0.042	&	0.041	&	0.048	&	0.048	&	0.050	&	0.054	&	0.054	&	0.056	\\
WLS$^{+}$ SE	&	0.049	&	0.049	&	0.048	&	0.056	&	0.056	&	0.058	&	0.062	&	0.062	&	0.063	\\
$b_n=4$	&	0.043	&	0.042	&	0.042	&	0.047	&	0.048	&	0.050	&	0.055	&	0.055	&	0.057	\\
WLS$^{+}$ SE	&	0.051	&	0.050	&	0.049	&	0.056	&	0.056	&	0.058	&	0.061	&	0.061	&	0.063	\\
$b_n=5$	&	0.042	&	0.042	&	0.041	&	0.047	&	0.047	&	0.049	&	0.055	&	0.055	&	0.057	\\
WLS$^{+}$ SE	&	0.051	&	0.050	&	0.049	&	0.055	&	0.055	&	0.056	&	0.063	&	0.063	&	0.065	\\
\hline \hline
% \bottomrule
\end{tabular}
% \captionsetup{font=small,labelfont=bf}
% \caption*{Note: Columns display results for the treatment and spillover effects ($n=1056$) under three regression specifications. }
\end{table}

% \begin{table}[ht]
% \centering
% \caption{\label{table: cai both}Estimates and SEs\citep{CaiJanvrySadoulet2015}. }
% \begin{tabular}{lrrrrrr}
% \toprule
% & \multicolumn{3}{c}{Treatment Effect} & \multicolumn{3}{c}{Spillover Effect} \\
% \cline{2-7}
%  & WLS & Add & Sat & WLS & Add & Sat \\
% \midrule
% \bottomrule
% \end{tabular}
% \caption*{Note: Columns display results for the treatment and spillover effects ($n=1056$) under three regression specifications. }
% \end{table}

\subsection{Regression to recover the Horvitz--Thompson estimator}\label{app:HT}
% The Horvitz–-Thompson estimator of $\mu(t)$ is 
% \[
% \hat{Y}_{\text{ht}}(t) 
% = \frac{1}{n}\sum_{i=1}^n \frac{{1}_i(t)Y_i}{\pi_i(t)}
% \]
In this section, we demonstrate how to recover the Horvitz--Thompson estimator from a regression-based approach, as noted in Remark \ref{re:HT}.
We begin with showing that a particular WLS fit reproduces \cite{Leung2022}’s Horvitz--Thompson estimator and analyze the asymptotic performance of the resulting regression-based HAC variance estimator.  
We then modify \cite{Leung2022}'s variance estimator to resolve the issue of anti-conservative variance estimation and compare the efficiency of the Horvitz--Thompson and Hájek estimators.

\subsubsection{WLS-based analysis of the Horvitz--Thompson estimator}
Define the adjusted outcome $\tilde{Y}_i = \left(\sum_{t\in \mathcal{T}} {1}_i(t) \hat{1}_{\text{ht}}(t) \right)Y_i$.
% Motivated by the use of inverse probability weighting in constructing the Horvitz–-Thompson estimator, we consider the WLS fitting of transformed outcome $(\sum_{t\in \mathcal{T}} {1}_i(t) \alpha_{w,t})Y_i$ on $\sum_{t\in \mathcal{T}}{1}_i(t)$ with weight $\tilde{w}_i = \sum_{t\in \mathcal{T}} \frac{1}{\alpha_{w,t}}\frac{{1}_i(t)}{\pi_i(t)}$, and estimates $\{\mu(t)\}_{t\in\mathcal{T}}$ by the resulting WLS coefficients. 
We consider the WLS fit: 
\begin{equation}
\text{regress } \tilde{Y}_i \text{ on } z_i \text{ with weights } \tilde{w}_i = 1/(\hat{1}_{\text{ht}}(t)\pi_i(T_i)).  
% \text{lm}(\tilde{Y}_i \sim z_i) \text{ with weights } \tilde{w}_i = 1/(\hat{1}_{\text{ht}}(t)\pi_i(T_i)).  
\label{eq:HT_WLS}
\end{equation}
Let $\hat{\beta}_{\text{ht}}$ denote the estimtors of coefficients for $z_i$. 
Let $\tilde{ {W}}=\text{diag}\{ \tilde{w}_i: i=1,\ldots, n\}$ and vectorize $\tilde{Y}_i$'s as $\tilde{ {Y}}$.
% Proposition \ref{prop: ht} below states the numerical equivalence between $\hat{\beta}_{\text{ht}}$ and $\hat{Y}_{\text{ht}}$. 

% Define $d_i$ as the vectorization of $\{{1}_i(t): t\in\mathcal{T}\}$ and $D$ as the matrix with rows $d_i$. 
\begin{proposition} \label{prop: ht}
$\hat{\beta}_{\textup{ht}} = \hat{Y}_{\textup{ht}}$. 
% and $\hat{ {V}}_{\text{ht}} = (D^{\top}\tilde{W} D)^{-1}
% (D^{\top}\tilde{W} e  {K}_n e^{\top} \tilde{W} D)
% (D^{\top}\tilde{W} D)^{-1}$. 
\end{proposition}
\begin{proof}[Proof of Proposition \ref{prop: ht}]
% Let $\tilde{W}=diag\{ \tilde{w}_i\}_{i=1}^n$ and $\tilde{Y}$ be the vectorization of $\{ (\sum_{t\in \mathcal{T}} {1}_i(t) \alpha_{w,t})Y_i: i=1,\ldots, n \}$.
The result follows from $\hat{\beta}_{\text{ht}} = ( {Z}^{\top}\tilde{ {W}}  {Z})^{-1} 
{Z}^{\top}\tilde{ {W}} \tilde{ {Y}}$. 
\end{proof}
Proposition \ref{prop: ht} is numerical and shows the utility of WLS fit in reproducing the Horvitz--Thompson estimator.
We exclude it from the main paper due to the unnaturalness in both its weighting and outcome transformation schemes.

% Proposition \ref{ref: ht} is numeric and shows that we can recover the Horvitz–-Thompson estimator for $\tau(t,t')$ with $\hat{\tau}_{\text{ht}}(t,t') = \hat{\beta}_{\text{ht}}(t) - \hat{\beta}_{\text{ht}}(t') = ({1}(t) - {1}(t'))^{\top}\hat{\beta}_{\text{ht}}$.

The residual from the WLS fit in \eqref{eq:HT_WLS} is
\begin{align*}
e_{\text{ht},i}  
=& \sum_{t\in \mathcal{T}} {1}_i(t) \hat{1}_{\text{ht}}(t) \left( Y_i - \frac{1}{\hat{1}_{\text{ht}}(t)} \frac{1}{n}\sum_{i=1}^n \frac{Y_i{1}_i(t)}{\pi_i(t)} \right) 
= \tilde{Y}_i - \sum_{t\in \mathcal{T}} {1}_i(t) \hat{\beta}_{\text{ht}}(t).
\end{align*} 
% Let $e$ be the vectorization of $\{e_i: i=1,\ldots, n \}$ and $W$ be a symmetrys matrix with $ij$th entry being ${1}(\ell_{ {A}}(i,j)\le b_n)$. 
Let $ {e}_\text{ht} = \text{diag}\{e_{\text{ht},i}: i=1,\ldots, n\}$.
The HAC variance estimator for $\hat{\beta}_{\text{ht}}$ based on WLS fit in \eqref{eq:HT_WLS} equals 
\begin{align*}
\hat{ {V}}_\text{ht}
=& ( {Z}^{\top}\tilde{ {W}}  {Z})^{-1}
( {Z}^{\top}\tilde{ {W}}  {e}_\text{ht}  {K}_n  {e}_\text{ht}^{\top} \tilde{ {W}}  {Z})
( {Z}^{\top}\tilde{ {W}}  {Z})^{-1}. 
% =&  \left(\sum_{i=1}^n \tilde{w}_ix_ix_i'\right)^{-1}
% \sum_{i,j}\tilde{w}_i\tilde{w}_jx_ie_ie_jx_j 
% {1}(\ell_{ {A}}(i,j)\le b_n)
% \left(\sum_{i=1}^n \tilde{w}_ix_ix_i'\right)^{-1} \\
% =& \begin{pmatrix}
% \frac{\sum_{i,j} \tilde{w}_ie_i{1}_i(t)\tilde{w}_je_j{1}_j(t) {1}(\ell_{ {A}}(i,j)\le b_n)}{(\sum_i \tilde{w}_i {1}_i(t) )^2}  & \ldots & \frac{\sum_{i,j} \tilde{w}_ie_i{1}_i(t)\tilde{w}_je_j{1}_j(t') {1}(\ell_{ {A}}(i,j)\le b_n)}{\sum_i \tilde{w}_i {1}_i(t) \sum_i \tilde{w}_i {1}_i(t') } \\
% \vdots & & \vdots \\
% \frac{\sum_{i,j} \tilde{w}_ie_i{1}_i(t')\tilde{w}_je_j{1}_j(t) {1}(\ell_{ {A}}(i,j)\le b_n)}{\sum_i \tilde{w}_i {1}_i(t') \sum_i \tilde{w}_i {1}_i(t)} & \ldots & \frac{\sum_{i,j} \tilde{w}_ie_i{1}_i(t')\tilde{w}_je_j{1}_j(t') {1}(\ell_{ {A}}(i,j)\le b_n)}{(\sum_i \tilde{w}_i {1}_i(t') )^2}
% \end{pmatrix}
\end{align*}
% Define $d(t)$ as the realized value of $d_i$ with ${1}_i(t)=1$ and other components being 0.
% We can express the Horvitz–-Thompson estimator as $\hat{\tau}_{\text{ht}}(t,t') = ({1}(t) - {1}(t'))^{\top}\hat{\beta}_{\text{ht}}$. 
% Define $\hat{1}_{\text{ht}}(t) 
% = \frac{1}{n}\sum_{i=1}^n \frac{{1}_i(t)}{\pi_i(t)}$.
% We have 
% \begin{align*}
% \hat{\Sigma}^2_{w,ht}(t, t^{\prime})) 
% =&  n ({1}(t) - {1}(t'))^{\top} 
% \hat{ {V}}_{w,ht}
% ({1}(t) - {1}(t'))
% \end{align*}
To facilitate a more direct comparison with the results of \cite{Leung2022}, we focus on the estimation and inference of the estimand, $\tau(t,t') = \mu(t) - \mu(t')$, the contrast between exposure mapping values $t$ and $t'$. Proposition \ref{prop: ht} shows we can use the WLS estimator, $\hat{\tau}_{\text{ht}}(t,t') = G\hat{\beta}_{\text{ht}} = \hat{\beta}_{\text{ht}}(t) - \hat{\beta}_{\text{ht}}(t')$, where $G$ is a $1\times |\mathcal{T}|$ vector containing a value of $1$ for the element corresponding to exposure mapping value $t$, a value of $-1$ for the element corresponding to $t'$, and 0 for all other elements.
We define $\hat{\sigma}^2_{\text{ht}}(t,t')$ as the regression-based HAC variance estimator for $\hat{\tau}_{\text{ht}}(t,t')$ based on $\hat{ {V}}_{\text{ht}}$, i.e., $\hat{\sigma}^2_{\text{ht}}(t, t^{\prime})
=  G 
\hat{ {V}}_{\text{ht}}
G^\top$.
Define $\Delta_i(t,t') = ( {1}_i(t)\pi_i(t)^{-1} - {1}_i(t')\pi_i(t')^{-1} )Y_i $ and ${\sigma}_{\text{ht}}^2(t,t') = \operatorname{Var}\left( \sqrt{n} \hat{\tau}_{\text{ht}}(t,t') \right)$.
Theorem \ref{thm:WLS HT asm} below states the asymptotic normality of $\hat{\tau}_{\text{ht}}(t,t')$. 
\begin{theorem} \label{thm:WLS HT asm}
Define 
\begin{equation}
{\sigma}^2_*(t, t^{\prime})
= \frac{1}{n} \sum_{i=1}^n \sum_{j=1}^n
\left(\Delta_i(t,t')-\tau_i(t, t^{\prime})\right)
\left(\Delta_j(t,t')-\tau_j(t, t^{\prime})\right) 
{1}(\ell_{ {A}}(i, j) \leq b_n).
\label{eq:sigma_ht}    
\end{equation}
Under Assumptions \ref{asu1}--\ref{asu6}, we have
\begin{align*}
{\sigma}_{\textup{ht}}^{-1}(t,t') 
\sqrt{n}
\left( \hat{\tau}_{\textup{ht}}(t,t') - \tau(t,t') \right) 
\stackrel{\textup{d}}{\rightarrow}  \mathcal{N}(0,1)
\text{ and } 
{\sigma}^2_*(t, t^{\prime})
= {\sigma}_{\textup{ht}}^2(t,t') + o_\mathbb{P}(1).
\end{align*}
\end{theorem}
\begin{proof}[Proof of Theorem \ref{thm:WLS HT asm}]
The asymptotic normality of $\hat{\tau}_{\text{ht}}(t,t')$ follows from Theorem 3 in \cite{Leung2022} and Proposition \ref{prop: ht}.
The consistency result
${\sigma}^2_*(t, t^{\prime}) = {\sigma}_{\text{ht}}^2(t,t') + o_\mathbb{P}(1)$ is proved in Theorem 4 of \cite{Leung2022}.
\end{proof}

Define the individual-level exposure effect $\tau_i(t, t^{\prime}) = \mu_i(t) - \mu_i(t')$.
Define $\tilde{\tau}_i(t,t') 
= {1}_i(t)\pi_i(t)^{-1} \mu(t)
- {1}_i(t')\pi_i(t')^{-1} \mu(t')$. 
Theorem \ref{thm:WLS HT} below establishes the asymptotic bias of HAC variance estimator $\hat{\sigma}^2_{\text{ht}}(t,t')$ for the asymptotic variance of $\hat{\tau}_{\text{ht}}(t,t')$. It shows that 
the regression-based variance estimation of $\hat{\tau}_{\text{ht}}(t,t')$ is not asymptotically exact even when the individual-level exposure effects are constant.
\begin{theorem} \label{thm:WLS HT}
Define \begin{align}
% \hat{\sigma}^2_{\text{ht}}(t, t^{\prime})
% &=  G 
% \hat{ {V}}_{\text{ht}}
% G^\top \\
% \hat{\sigma}^2_*(t, t^{\prime})
% &= \frac{1}{n} \sum_{i=1}^n \sum_{j=1}^n
% \left(\Delta_i(t,t')-\tau_i(t, t^{\prime})\right)
% \left(\Delta_j(t,t')-\tau_j(t, t^{\prime})\right) 
% {1}(\ell_{ {A}}(i, j) \leq b_n), \label{eq:sigma_ht} \\
R_{\textup{ht}}(t,t')
% =& \frac{1}{n} \sum_{i=1}^n \sum_{j=1}^n 
% \left(\tau_i(t,t')- \textcolor{red}{\check{\tau}_i(t,t')}  \right)
% \left({\tau}_j(t,t') - \textcolor{red}{\check{\tau}_j(t,t')}  \right) 
% {1}(\ell_{ {A}}(i, j) \leq b_n)\\
% &+ \frac{2}{n} \sum_{i=1}^n \sum_{j=1}^n 
% \left( \Delta_i(t,t') - {\tau}_i(t,t') \right)
% \left({\tau}_j(t,t') - \textcolor{red}{\check{\tau}_j(t,t')} \right) 
% {1}(\ell_{ {A}}(i, j) \leq b_n)
&= \frac{1}{n} \sum_{i=1}^n \sum_{j=1}^n 
\left(\tau_i(t,t')- \tilde{\tau}_i(t,t')   \right) 
\left({\tau}_j(t,t') - \tilde{\tau}_j(t,t') \right) 
{1}(\ell_{ {A}}(i, j) \leq b_n)\notag \\
&+ \frac{2}{n} \sum_{i=1}^n \sum_{j=1}^n 
\left( \Delta_i(t,t') - \tilde{\tau}_i(t,t')   \right) 
\left({\tau}_j(t,t') - \tilde{\tau}_j(t,t')  \right) 
{1}(\ell_{ {A}}(i, j) \leq b_n). 
\label{eq:bias_HT}
\end{align}
Under Assumptions \ref{asu1}--\ref{asu4} and \ref{asu7},
we have $\hat{\sigma}^2_{\textup{ht}}(t,t')
= {\sigma}^2_*(t, t^{\prime}) + R_{\textup{ht}}(t,t')+o_\mathbb{P}(1)$.
% \begin{align*}
% % \hat{\sigma}^2_*(t, t^{\prime})
% % =&  \operatorname{Var}\left(\sqrt{n} \hat{\tau}_{\text{ht}}(t, t^{\prime})\right)
% % + o_\mathbb{P}(1), \\
% \hat{\sigma}^2_{\text{ht}}(t,t')
% =&~ \hat{\sigma}^2_*(t, t^{\prime}) + R_{\text{ht}}(t,t')+o_\mathbb{P}(1).
% \end{align*}
\end{theorem}

\begin{proof}[Proof of Theorem \ref{thm:WLS HT}]
\textbf{Formula of $\hat{\sigma}^2_\text{ht}(t,t')$.}
% Define ${1}(t)$ as the realized value of $\{{1}_i(t): t\in\mathcal{T}\}$ where $T_i=t$.
% We can express the Horvitz–-Thompson estimator as $\hat{\tau}_{\text{ht}}(t,t') = \hat{\beta}_{\text{ht}}(t) - \hat{\beta}_{\text{ht}}(t') = ({1}(t) - {1}(t'))^{\top}\hat{\beta}_{\text{ht}}$. 
By direct algebra, 
the HAC variance for $\hat{\tau}_{\text{ht}}(t,t')$ equals
\begin{align*}
\hat{\sigma}^2_\text{ht}(t, t^{\prime}) 
% =  ({1}(t) - {1}(t'))^{\top} 
% \hat{ {V}}_\text{ht}
% ({1}(t) - {1}(t'))\\
% =&  \frac{\sum_{\ell_{ {A}}(i,j)\le b_n} \tilde{w}_ie_i{1}_i(t)\tilde{w}_je_j{1}_j(t) }{(\sum_i \tilde{w}_i {1}_i(t) )^2}
% + \frac{\sum_{\ell_{ {A}}(i,j)\le b_n} \tilde{w}_ie_i{1}_i(t')\tilde{w}_je_j{1}_j(t') }{(\sum_i \tilde{w}_i {1}_i(t') )^2} \\
% & - \frac{\sum_{\ell_{ {A}}(i,j)\le b_n} \tilde{w}_ie_i{1}_i(t)\tilde{w}_je_j{1}_j(t') }{\sum_i \tilde{w}_i {1}_i(t) \sum_i \tilde{w}_i {1}_i(t')} 
% - \frac{\sum_{\ell_{ {A}}(i,j)\le b_n} \tilde{w}_ie_i{1}_i(t')\tilde{w}_je_j{1}_j(t) }{\sum_i \tilde{w}_i {1}_i(t') \sum_i \tilde{w}_i {1}_i(t)} \\
% =& \sum_{i=1}^n \sum_{j=1}^n \left( \frac{\tilde{w}_i {1}_i(t)}{\sum_i \tilde{w}_i {1}_i(t)}   - \frac{\tilde{w}_i {1}_i(t')}{\sum_i \tilde{w}_i {1}_i(t')} \right) e_i \left( \frac{\tilde{w}_j {1}_j(t)}{\sum_j \tilde{w}_j {1}_j(t)}   - \frac{\tilde{w}_j {1}_j(t')}{\sum_j \tilde{w}_j {1}_j(t')} \right) e_j  {1}(\ell_{ {A}}(i,j)\le b_n) \\
= \frac{1}{n} \sum_{i=1}^n \sum_{j=1}^n 
\left( \frac{ {1}_i(t)}{\pi_i(t)\hat{1}_{\text{ht}}(t)} - \frac{ {1}_i(t')}{\pi_i(t')\hat{1}_{\text{ht}}(t')}  \right) e_{\text{ht},i}  
\left( \frac{ {1}_j(t)}{\pi_j(t)\hat{1}_{\text{ht}}(t)} - \frac{ {1}_j(t')}{\pi_j(t')\hat{1}_{\text{ht}}(t')}  \right) e_{\text{ht},j} 
K_n(i,j), 
\end{align*}
where 
\begin{align*}
\left( \frac{ {1}_i(t)}{\pi_i(t)\hat{1}_{\text{ht}}(t)} - \frac{ {1}_i(t')}{\pi_i(t')\hat{1}_{\text{ht}}(t')}  \right) e_{\text{ht},i} 
% =& \left( \frac{ {1}_i(t)}{\pi_i(t)} - \frac{{1}_i(t')}{\pi_i(t')} \right) Y_i 
% - \left( \frac{\frac{ {1}_i(t)}{\pi_i(t)}\frac{1}{n} \sum_{i=1}^n \frac{ Y_i {1}_i(t)}{\pi_i(t)} }{\hat{1}_{\text{ht}}(t)}  - \frac{\frac{ {1}_i(t')}{\pi_i(t')}\frac{1}{n} \sum_{i=1}^n \frac{ Y_i {1}_i(t')}{\pi_i(t')} }{\hat{1}_{\text{ht}}(t')}   \right) \\
% =&  \frac{ {1}_i(t)}{\pi_i(t)}(Y_i - \hat{\beta}_{\textup{haj}}(t)) - \frac{ {1}_i(t')}{\pi_i(t')} (Y_i - \hat{\beta}_{\textup{haj}}(t')) \\
=& \Delta_i(t,t') 
- \left(\frac{ {1}_i(t)}{\pi_i(t)} \hat{\beta}_{\textup{haj}}(t)
- \frac{ {1}_i(t')}{\pi_i(t')} \hat{\beta}_{\textup{haj}}(t') \right).
% =& \Delta_i(t,t') - 
% \left( \frac{\frac{ {1}_i(t)}{\pi_i(t)}\frac{1}{n} \sum_i \frac{ Y_i {1}_i(t)}{\pi_i(t)} }{\hat{1}_{\text{ht}}(t)} 
% - \frac{\frac{ {1}_i(t')}{\pi_i(t')}\frac{1}{n} \sum_i \frac{ Y_i {1}_i(t')}{\pi_i(t')} }{\hat{1}_{\text{ht}}(t')} \right) 
% % \equiv & \Delta_i(t,t') - \tilde{\tau}_i(t,t')
\end{align*} 

\noindent\textbf{Bias of $\hat{\sigma}^2_\text{ht}(t,t')$.}
By direct algebra, we have
\begin{align*}
\hat{\sigma}^2_\text{ht}(t,t')
&=  {\sigma}^2_*(t,t') 
+ \frac{1}{n} \sum_{i=1}^n \sum_{j=1}^n
\left( \tau_i(t,t') - \tilde{\tau}_i(t,t') \right) 
\left( \tau_j(t,t') - \tilde{\tau}_j(t,t')  \right) 
K_n(i,j) \\
+& \frac{2}{n} \sum_{i=1}^n \sum_{j=1}^n
\left( \Delta_i(t,t') - \tilde{\tau}_i(t,t') \right) 
\left({\tau}_j(t,t') - \tilde{\tau}_j(t,t')  \right) 
K_n(i,j)\\
+& \frac{1}{n}  \sum_{i=1}^n \sum_{j=1}^n 
\left[
\begin{array}{c}
\left( 
\frac{{1}_i(t) (\mu(t)- \hat{\beta}_{\textup{haj}}(t))}{\pi_i(t)}  -\frac{{1}_i(t')(\mu(t')- \hat{\beta}_{\textup{haj}}(t'))}{\pi_i(t')} 
+ 2 \left( \Delta_i(t,t') - \tilde{\tau}_i(t,t')
\right)  \right)  \\
\left( 
\frac{{1}_j(t)(\mu(t)- \hat{\beta}_{\textup{haj}}(t))}{\pi_j(t)}  -\frac{{1}_j(t')(\mu(t')- \hat{\beta}_{\textup{haj}}(t'))}{\pi_j(t')} 
\right)
\end{array}
\right]
K_n(i,j). 
\end{align*}
Since $\hat{\beta}_{\textup{haj}}(t) - \mu(t) = O_{\mathbb{P}}(n^{-1/2})$ for all $t\in\mathcal{T}$ by Lemma \ref{lemma: beta_haj}, and $Y_i$ and ${1}_i(t)\pi_i(t)^{-1}$ are uniformly bounded by Assumptions \ref{asu2} and \ref{asu3}, then for some $C >0$ and any $n$, we have
\begin{align*}
& \left| 
\frac{1}{n} \sum_{i=1}^n \sum_{j=1}^n  
\left[
\begin{array}{c}
\left( 
\frac{{1}_i(t) (\mu(t)- \hat{\beta}_{\textup{haj}}(t))}{\pi_i(t)}  -\frac{{1}_i(t')(\mu(t')- \hat{\beta}_{\textup{haj}}(t'))}{\pi_i(t')} 
+ 2 \left( \Delta_i(t,t') - \tilde{\tau}_i(t,t') \right) \right)  \\
\left( 
\frac{{1}_j(t)(\mu(t)- \hat{\beta}_{\textup{haj}}(t))}{\pi_j(t)}  -\frac{{1}_j(t')(\mu(t')- \hat{\beta}_{\textup{haj}}(t'))}{\pi_j(t')} 
\right)
\end{array}
\right] K_n(i,j)
\right| \\
& \le 
C \left[
| \mu(t) - \hat{\beta}_{\textup{haj}}(t) | 
+ | \mu(t') - \hat{\beta}_{\textup{haj}}(t') |   
\right]
\frac{1}{n} \sum_{i=1}^n \sum_{j=1}^n
K_n(i,j) \\
&= o_\mathbb{P}(1).
% +& \frac{2}{n} \sum_{i=1}^n \sum_{j=1}^n 
% \left(
% \Delta_i(t,t')-\check{\tau}_i(t, t^{\prime}))
% \right) 
% \left(
% \check{\tau}_j(t,t')-\tilde{\tau}_j(t, t^{\prime}))
% \right)
% {1}(\ell_{ {A}}(i, j) \leq b_n)\\
% =& \frac{1}{n} \sum_{\ell_{ {A}}(i, j) \leq b_n}
% \left(
% \frac{ {1}_i(t)}{\pi_i(t)} (\mu(t) - \hat{\beta}_{\textup{haj}}(t) )
% - \frac{ {1}_i(t')}{\pi_i(t')} (\mu(t') - \hat{\beta}_{\textup{haj}}(t')  )
% \right)
% \left(
% \frac{ {1}_j(t)}{\pi_j(t)} (\mu(t) - \hat{\beta}_{\textup{haj}}(t) )
% - \frac{ {1}_j(t')}{\pi_j(t')} (\mu(t') - \hat{\beta}_{\textup{haj}}(t')  )
% \right) \\
% +& \frac{2}{n} \sum_{i=1}^n \sum_{j=1}^n 
% \left(
% \Delta_i(t,t')-\check{\tau}_i(t, t^{\prime}))
% \right) 
% \left(
% \frac{{1}_j(t)}{\pi_j(t)} (\mu(t) - \hat{\beta}_{\textup{haj}}(t) )
% - \frac{ {1}_j(t')}{\pi_j(t')} (\mu(t') - \hat{\beta}_{\textup{haj}}(t')  )
% \right)
% {1}(\ell_{ {A}}(i, j) \leq b_n)\\
% =& o_\mathbb{P}(1)
\end{align*}
The last line follows by Assumption \ref{asu7}(b).
Thus, we complete the proof.  
\end{proof}

We do not provide modifications to the HAC covariance estimator based on the WLS fit in \eqref{eq:HT_WLS} for two reasons. First, \eqref{eq:HT_WLS} requires transformations in both its weighting and outcome. Second, the asymptotic bias term from the regression-based variance estimator is not quadratic, as demonstrated in \eqref{eq:bias_HT}. This characteristic indicates that issues of anti-conservativeness cannot be resolved even after applying the modified kernel.
More importantly, the HAC variance estimator is
not guaranteed to be exact for inference even if the individual effects are constant.

\subsubsection{Modification of variance estimation in \cite{Leung2022}}
\cite{Leung2022} proposed the following variance estimator,
\[
\hat{\sigma}^2(t, t^{\prime})
=\frac{1}{n} \sum_{i=1}^n \sum_{j=1}^n
\left( \Delta_i(t,t')-\hat{\tau}(t, t^{\prime}) \right)
\left( \Delta_j(t,t')-\hat{\tau}(t, t^{\prime}) \right) 
{1}(\ell_{ {A}}(i, j) \leq b_n).
\]
Theorem 4 in \cite{Leung2022} establishes that 
\begin{align*}
% & \hat{\sigma}^2_*(t, t^{\prime})
% = \operatorname{Var}\left(\sqrt{n} \hat{\tau}_\text{ht}(t, t^{\prime})\right) + o_\mathbb{P}(1), \\
& \hat{\sigma}^2(t, t^{\prime})
= {\sigma}^2_*(t, t^{\prime}) + R_n(t, t^{\prime}) + o_\mathbb{P}(1), 
\end{align*}
where ${\sigma}^2_*(t, t^{\prime})$ is defined in \eqref{eq:sigma_ht} and 
\begin{align*}
% \hat{\sigma}^2_*(t, t^{\prime})
% =& \frac{1}{n} \sum_{i=1}^n \sum_{j=1}^n
% \left(\Delta_i(t,t')-\tau_i(t, t^{\prime})\right)
% \left(\Delta_j(t,t')-\tau_j(t, t^{\prime})\right) 
% {1}(\ell_{ {A}}(i, j) \leq b_n),  \\
R_n(t, t^{\prime})
=& \frac{1}{n} \sum_{i=1}^n \sum_{j=1}^n
\left(\tau_i(t, t^{\prime})-\tau(t, t^{\prime})\right)
\left(\tau_j(t, t^{\prime})-\tau(t, t^{\prime})\right) 
{1}(\ell_{ {A}}(i, j) \leq b_n).   
\end{align*}
% The truncated kernel is not guaranteed to be positive semi-definite, so this variance is not guaranteed to be positive  and conservative. In this subsection, we propose an adjustment to \cite{Leung2022}'s variance estimator. 
The bias $R_n(t, t^{\prime})$ is equivalent to the $(t,t')$th element of $R_{\textup{haj}}$, as defined in Theorem \ref{thm:Hájek_bias}. The variance estimation is asymptotically exact with constant individual-level exposure effects.
The variance estimator proposed by \cite{Leung2022} is not guaranteed to be conservative.
% % We define the weight matrix as $ {K}_n$ with the $(i,j)$th entry being ${1}(\ell_{ {A}}(i, j) \leq b_n)$.
% Let $Q_n \Lambda_n Q_n^{\top}$ be the eigen-decomposition of $ {K}_n$ (since $ {K}_n$ is symmetry all its eigenvalues are real). 
% % Also let $\underline{\lambda}(A)$ denote the smallest eigenvalue of $A$, e.g., $\underline{\lambda}\left(\hat{\Sigma}_n\right)=\min _{1 \leq k \leq v} \Lambda_n$. 
% % Consider a sequence of small positive real numbers $c_n \searrow 0$. 
% We define
% \[
%  {K}_n^{+}:= Q_n\left(\Lambda_n \vee 0 \right) Q_n^{\top}
% =  {K}_n
% +  {K}_n^{-},
% \]
% where the maximum is taken element-wise. By construction, the matrix $ {K}_n^{+}$ is positive definite, and we denote the $(i,j)$th entry of $Q_n \min\{ \Lambda_n, 0\} Q_n^{\top}$ as ${1}^{-}\left\{\ell_{ {A}}(i, j) \leq b_n\right\}$. 
% % Moreover, in the case when the smallest eigenvalue of $\Sigma_n$ is bounded from below
% Define $\mathbf{\Delta}(t,t')$ as the stacked vector of $\Delta_i(t,t')$ and $ \tau(t,t') $ as the stacked vector of $\tau_i(t, t^{\prime}))$. 
Define ${\Delta}(t,t')$ as the stacked vector of $\Delta_i(t,t')$ and $\bar{\tau}(t,t')$ as the stacked vector of $\tau_i(t,t')-\tau(t,t')$.
The variance estimator in \cite{Leung2022} can be represented as
\[
\hat{\sigma}^2(t,t') 
= n^{-1} 
( {\Delta}(t,t')-\hat{\tau}(t, t^{\prime}))^{\top}  
{K}_n 
( {\Delta}(t,t')-\hat{\tau}(t, t^{\prime})).
\]
We propose the adjusted variance estimator as
% \[
% \widehat{ {V}}^+_\text{ht}
% = ( {Z}^{\top}\tilde{ {W}}  {Z})^{-1}
% ( {Z}^{\top}\tilde{ {W}}  {e}_\text{ht}  {K}_n  {e}_\text{ht}^{\top} \tilde{ {W}}  {Z})
% ( {Z}^{\top}\tilde{ {W}}  {Z})^{-1} 
% \]
\[
\hat{\sigma}^{2,+}(t,t') 
= n^{-1} 
( {\Delta}(t,t')-\hat{\tau}(t, t^{\prime}))^{\top} 
{K}_n^{+} 
( {\Delta}(t,t')-\hat{\tau}(t, t^{\prime})).
\]
% \begin{assumption}[xxx]\label{asu8}
% $ {K}_n^{-}$ has the same growing rate as $ {K}_n$.
% \end{assumption}
\begin{theorem} \label{thm:mofidy_ht}
Define
\begin{align*}
R_{n}^+(t,t')
=& n^{-1} 
\overline{ \tau}(t,t')^{\top} 
{K}_n^{+} 
\overline{ \tau}(t,t') 
+ n^{-1}
( {\Delta}(t,t')- \tau(t,t') )^{\top} 
{K}_n^{-} 
( {\Delta}(t,t')-  \tau(t,t') ).    
\end{align*}
Under Assumptions \ref{asu1}-\ref{asu4} and \ref{asu8},
we have
\[
\hat{\sigma}^{2,+}(t,t')
= {\sigma}^2_{*}(t,t') + R_{n}^+(t,t')+o_\mathbb{P}(1),
\]
where ${\sigma}^2_{*}(t,t')$ is defined in \eqref{eq:sigma_ht}.
% which gaurantees $\hat{\Sigma}_{+}^{2}(t,t')$ is asymptotically conservative.
\end{theorem}

% \section{Horvitz–-Thompson estimator}

\begin{proof}[Proof of Theorem \ref{thm:mofidy_ht}]
% [Proof of Theorem \ref{thm:mofidy_ht}] 
% We define the weight matrix as $ {K}_n$ with the $(i,j)$th entry being ${1}(\ell_{ {A}}(i, j) \leq b_n)$.
% Let $Q_n \Lambda_n Q_n^{\top}$ be the eigen-decomposition of $ {K}_n$ (since $ {K}_n$ is symmetry all its eigenvalues are real). 
% % Also let $\underline{\lambda}(A)$ denote the smallest eigenvalue of $A$, e.g., $\underline{\lambda}\left(\hat{\Sigma}_n\right)=\min _{1 \leq k \leq v} \Lambda_n$. 
% % Consider a sequence of small positive real numbers $c_n \searrow 0$. 
% We define
% \[
%  {K}_n^{+}:= Q_n\left(\Lambda_n \vee 0 \right) Q_n^{\top}
% =  {K}_n
% +  {K}_n^{-},
% \]
% where the maximum is taken element-wise. By construction, the matrix $ {K}_n^{+}$ is positive definite, and we denote the $(i,j)$th entry of $Q_n \min\{ \Lambda_n, 0\} Q_n^{\top}$ as ${1}^{-}\left\{\ell_{ {A}}(i, j) \leq b_n\right\}$. 
% % Moreover, in the case when the smallest eigenvalue of $\Sigma_n$ is bounded from below
% Define $\mathbf{Z}$ as the stacked vector of $Z_i$ and $ \tau $ as the stacked vector of $\tau_i(t, t^{\prime}))$. 

By direct algebra, 
\begin{align*}
\hat{\sigma}^{2,+}(t,t')
% =& n^{-1}
% \left(\mathbf{\Delta}(t,t')-\hat{\tau}(t, t^{\prime})\right)^{\top} 
%  {K}_n^{+} 
% \left(\mathbf{\Delta}(t,t')-\hat{\tau}(t, t^{\prime})\right) \\
=& \hat{\sigma}^{2}(t,t')
+ n^{-1}
\left( {\Delta}(t,t')-\hat{\tau}(t, t^{\prime})\right)^{\top} 
{K}_n^{-} 
\left( {\Delta}(t,t')-\hat{\tau}(t, t^{\prime})\right) \\
=& {\sigma}^{2}_{*}(t,t') + R_n(t,t')
+ n^{-1}
\left( {\Delta}(t,t')-\hat{\tau}(t, t^{\prime})\right)^{\top} 
{K}_n^{-} 
\left( {\Delta}(t,t')-\hat{\tau}(t, t^{\prime})\right)
+ o_\mathbb{P}(1) \\
=& {\sigma}^{2}_{*}(t,t') 
+ n^{-1}
\overline{ \tau}(t,t')^{\top} 
{K}_n^{+} 
\overline{ \tau}(t,t') 
- n^{-1} 
\overline{ \tau}(t,t')^{\top} 
{K}_n^{-} 
\overline{ \tau}(t,t') \\
&+ n^{-1} 
\left( {\Delta}(t,t')-\hat{\tau}(t, t^{\prime})\right)^{\top} 
{K}_n^{-} 
\left( {\Delta}(t,t')-\hat{\tau}(t, t^{\prime})\right)
+ o_\mathbb{P}(1),
\end{align*}
where the first and third lines hold by the definition of $K_n^+$, and the second line holds by Theorem 4 in \cite{Leung2022}.
Applying the proof of Theorem 4 in \cite{Leung2022} but replacing Assumption \ref{asu7} with Assumption \ref{asu8}, we have
\begin{align*}
& n^{-1} 
\left( {\Delta}(t,t')-\hat{\tau}(t, t^{\prime})\right)^{\top} 
{K}_n^{-} 
\left( {\Delta}(t,t')-\hat{\tau}(t, t^{\prime})\right) \\
=& n^{-1}
\left( {\Delta}(t,t')- \tau(t,t') \right)^{\top} 
{K}_n^{-} 
\left( {\Delta}(t,t')-  \tau(t,t') \right)
+ n^{-1}
\overline{ \tau}(t,t')^{\top} 
{K}_n^{-} 
\overline{ \tau}(t,t')
+ o_\mathbb{P}(1).
\end{align*}
% The details of the proof are omitted. 
Therefore, we complete the proof. 
% \begin{align*}
% \hat{\sigma}^{2,+}(t,t')
% = \hat{\sigma}^2_{*}(t,t') 
% + n^{-1} 
% \overline{ \tau}(t,t')^{\top} 
%  {K}_n^{+} 
% \overline{ \tau}(t,t')
% + n^{-1} 
% \left( {\Delta}(t,t')- \tau(t,t') \right)^{\top} 
%  {K}_n^{-} 
% \left( {\Delta}(t,t')-  \tau(t,t') \right)
% + o_\mathbb{P}(1).
% \end{align*}
\end{proof}

% \subsection{Compare Variance with \cite{Leung2022}}
% In this subsection, we compare the bias term of our variance estimator with that of \cite{Leung2022}. 
% \begin{align*}
% R_{\text{ht}}(t,t')
% &= \frac{1}{n} \sum_{\ell_{ {A}}(i, j) \leq b_n}
% \left[
% \begin{array}{c}
% \left(\tau_i(t,t')- (\frac{{1}_i(t)}{\pi_i(t)} \mu(t) -\frac{{1}_i(t')}{\pi_i(t')} \mu(t')) + 2 ( \Delta_i(t,t') - (\frac{{1}_i(t)}{\pi_i(t)} \mu(t) -\frac{{1}_i(t')}{\pi_i(t')} \mu(t'))  )  \right) \\
% \left({\tau}_j(t,t') - \left(\frac{{1}_j(t)}{\pi_j(t)} \mu(t) -\frac{{1}_j(t')}{\pi_j(t')} \mu(t') \right)  \right) 
% \end{array}
% \right] \\
% R_n(t, t^{\prime})
% &= \frac{1}{n} \sum_{i=1}^n \sum_{j=1}^n
% \left(\tau_i(t, t^{\prime})-\tau(t, t^{\prime})\right)
% \left(\tau_j(t, t^{\prime})-\tau(t, t^{\prime})\right) 
% {1}(\ell_{ {A}}(i, j) \leq b_n)
% \end{align*}

% \subsection{Simulations}
% Please add the following required packages to your document preamble:
% \usepackage[table,xcdraw]{xcolor}
% If you use beamer only pass "xcolor=table" option, i.e. \documentclass[xcolor=table]{beamer}

In Table \ref{table: HT RGG}, we present the results under the simulation design of \cite{Leung2022}.
We report the ``Estimand'' as the average of the Horvitz--Thompson estimator $\hat{\tau}_\text{ht}(1,0)$ over $10,000$ simulation draws. We report the ``Oracle SE,'' denoted by $\operatorname{Var}(\hat{\tau}_\text{ht}(1,0))^{1/2}$, as the standard error of $\hat{\tau}_\text{ht}(1,0)$ over the same $10,000$ simulation draws.
We report the ``Estimate'' as the average of $\hat{\tau}_\text{ht}(1,0)$ over another $10,000$ simulation draws.
We present the coverage of the standard error under ``Oracle Coverage.''
We present \cite{Leung2022}'s standard error and the corresponding coverage in the ``Leung SE'' and ``Leung Coverage'' rows.
We present our adjusted standard error and the corresponding coverage in the ``Leung$^+$ SE'' and ``Leung$^+$ Coverage'' rows. 
% The remaining rows correspond to the results documented in Table 3 of \cite{Leung2022}. 
We can see that \cite{Leung2022}'s standard error can be anti-conservative and our adjusted standard error improves the empirical coverage rate.   

% We report the estimands $\tau(1,0)$, approximated by the Horvitz–-Thompson estimator $\hat{\tau}_\text{ht}(1,0)
% = n^{-1} \sum_{i=1}^n ( {1}_i(1)\pi_i(1)^{-1} - {1}_i(0)\pi_i(0)^{-1} )Y_i$ over 10,000 simulation draws.
% We report ``Oracle SE", denoted by $\operatorname{Var}(\hat{\tau}(1,0))^{1/2}$, which are calculated as the standard deviation of the point estimators from corresponding WLS fits over 10,000 simulation draws.
% We proesent the HAC standard errors obtained from each WLS fit under ``WLS-based SE", and the corresponding adjusted HAC standard errors under ``WLS Adjusted SE". We report the standard errors assuming no interference under ``Naïve SE"  to illustrate the degree of dependence in the data.
% We also report the empirical coverage of 95\% CIs in the “Coverage” rows for the corresponding standard errors. 
% We report the effective sample size of $\hat{\mu}(t)$ as $\hat{n}(t)=\sum_i {1}_i(t)$.
% The ``$b_n$" reports the bandwidth in \eqref{eq:bandwidth} with $K=1$. 

\begin{table}[tbp]
\centering
\caption{\label{table: HT RGG}Horvitz--Thompson Estimation}
\begin{tabular}{lcccccc}
\hline 
\hline 
% \multicolumn{7}{c}{SIMULATION RESULTS: RGG MODEL}  \\ 
% \hline
% & \multicolumn{6}{c}{\# Schools} \\
% \cline{2-7}
Outcome models & \multicolumn{3}{c}{Linear-in-Means}   
& \multicolumn{3}{c}{Complex Contagion}   \\
% \cline{2-4} \cline{5-7}
\hline
\# Schools & \multicolumn{1}{c}{1}  & \multicolumn{1}{c}{2} & \multicolumn{1}{c}{4} & \multicolumn{1}{c}{1} & \multicolumn{1}{c}{2} & \multicolumn{1}{c}{4} \\ 
\hline
% $\hat{n}(1)$ & 226.412  & 426.04 & 845.475 & 226.412 & 426.04 & 845.475 \\
% $\hat{n}(0)$ & 169.588  & 295.96 & 598.525 & 169.588 & 295.96 & 598.525 \\
Estimand	&	0.678	&	0.710	&	0.698	&	0.027	&	0.019	&	0.024	\\	
Estimate	&	0.649	&	0.709	&	0.677	&	0.027	&	0.020	&	0.025	\\	
Oracle	SE	&	1.920	&	1.380	&	0.927	&	0.141	&	0.112	&	0.082	\\
Leung	SE	&	1.814	&	1.335	&	0.916	&	0.137	&	0.109	&	0.081	\\
Leung$^+$	SE	&	1.901	&	1.491	&	1.016	&	0.147	&	0.124	&	0.092	\\
Oracle	Coverage	&	0.950	&	0.952	&	0.949	&	0.952	&	0.952	&	0.951	\\
Leung	Coverage	&	0.932	&	0.934	&	0.942	&	0.937	&	0.937	&	0.944	\\
Leung$^+$	Coverage	&	0.946	&	0.963	&	0.965	&	0.954	&	0.964	&	0.971	\\													
% Naïve Coverage & 0.4886 & 0.5286 & 0.5322 & 0.6612 & 0.6646 & 0.6594 \\
% $b_n$ & 2 & 3 & 3 & 2 & 3 & 3 \\
% APL & 14.2916 & 18.2498 & 24.8087 & 14.2916  & 18.2498 & 24.8087 \\
% network size & {805} & 1456 & 2725 & 805 & 1456 & 2725 \\ 
\hline
\hline 
\end{tabular}
\end{table}

\subsubsection{Compare the Horvitz--Thompson and Hájek Estimators}\label{sec:compare_HT_Hajek}
In this subsection, we provide a brief discussion on the efficiency comparison between the Horvitz--Thompson and Hájek estimators. By considering a special case with $G=(1,-1)$, we compare the oracle variances of the two estimators.
Define
\begin{align*}
{\Delta}_{\textup{haj},i}(t, t^{\prime}) 
=& 1_i(t)\pi_i(t)^{-1}(Y_i - \mu(t)) - (\mu_i(t) - \mu(t) )
- 1_i(t')\pi_i(t')^{-1}(Y_i - \mu(t')) - (\mu_i(t') - \mu(t') ) \\
=& \Delta_i(t,t') - \tau_i(t, t^{\prime}) 
- \left( {1}_i(t)\pi_i(t)^{-1} \mu(t) - {1}_i(t')\pi_i(t')^{-1} \mu(t) 
- (\mu(t) - \mu(t') ) \right).
\end{align*}
Then, by Theorems \ref{thm:Hájek_bias} and \ref{thm:WLS HT asm}, the oracle variances of the Horvitz--Thompson and Hájek estimators can be related as follows:
\begin{align*}
& {\sigma}^2_{*}(t, t^{\prime})
= \frac{1}{n} \sum_{i=1}^n \sum_{j=1}^n
\left( \Delta_i(t,t') - \tau_i(t, t^{\prime}) \right) 
\left( \Delta_j(t,t') - \tau_j(t, t^{\prime}) \right)
{1}(\ell_{ {A}}(i, j) \leq b_n) \\
&= \frac{1}{n} \sum_{i=1}^n \sum_{j=1}^n
\left( {\Delta}_{\textup{haj},i}(t, t^{\prime}) + \left( \tfrac{{1}_i(t)}{\pi_i(t)} \mu(t) - \tfrac{{1}_i(t')}{\pi_i(t')} \mu(t') - (\mu(t) - \mu(t') ) \right) \right) \\
& \quad \cdot \left( {\Delta}_{\textup{haj},j}(t, t^{\prime}) + \left( \tfrac{{1}_j(t)}{\pi_j(t)} \mu(t) - \tfrac{{1}_j(t')}{\pi_j(t')} \mu(t') - (\mu(t) - \mu(t') ) \right) \right)
{1}(\ell_{ {A}}(i, j) \leq b_n) \\
% &= {\sigma}^2_{*,\textup{haj}}(t, t^{\prime}) 
% + \frac{1}{n} \sum_{i=1}^n \sum_{j=1}^n
% \left[
% \begin{array}{c}
% \left( \tfrac{{1}_i(t)}{\pi_i(t)} \mu(t)  - \tfrac{{1}_i(t')}{\pi_i(t')}\mu(t')  - (\mu(t) - \mu(t') ) \right) \\
% \left( \tfrac{{1}_j(t)}{\pi_j(t)} \mu(t)  - \tfrac{{1}_j(t')}{\pi_j(t')}\mu(t')  - (\mu(t) - \mu(t') ) \right)   
% \end{array}
% \right]
% {1}(\ell_{ {A}}(i, j) \leq b_n) \\
% &\quad + 2 \frac{1}{n} \sum_{i=1}^n \sum_{j=1}^n
% {\Delta}_{\textup{haj},i}(t, t^{\prime})
% \left( \frac{{1}_j(t)}{\pi_j(t)} \mu(t) - \frac{{1}_j(t')}{\pi_j(t')} \mu(t') - (\mu(t) - \mu(t') ) \right)
% {1}(\ell_{ {A}}(i, j) \leq b_n) \\
&= {\sigma}^2_{*,\textup{haj}}(t, t^{\prime}) 
+ 2 \frac{1}{n} \sum_{i=1}^n \sum_{j=1}^n
\left( \Delta_i(t,t') - \tau_i(t, t^{\prime}) \right)
\left( \tfrac{{1}_j(t)}{\pi_j(t)} \mu(t) - \tfrac{{1}_j(t')}{\pi_j(t')} \mu(t') - (\mu(t) - \mu(t') ) \right)
{1}(\ell_{ {A}}(i, j) \leq b_n) \\
&\quad - \frac{1}{n} \sum_{i=1}^n \sum_{j=1}^n
\left[
\begin{array}{c}
\left( \tfrac{{1}_i(t)}{\pi_i(t)} \mu(t)  - \tfrac{{1}_i(t')}{\pi_i(t')}\mu(t')  - (\mu(t) - \mu(t') ) \right) \\
\left( \tfrac{{1}_j(t)}{\pi_j(t)} \mu(t)  - \tfrac{{1}_j(t')}{\pi_j(t')}\mu(t')  - (\mu(t) - \mu(t') ) \right)   
\end{array}
\right]
{1}(\ell_{{A}}(i, j) \leq b_n). 
\end{align*}
For the special case with no interference, this expression simplifies to:
\begin{align*}
{\sigma}^2_{*}(t, t^{\prime})
=& {\sigma}^2_{*,\textup{haj}}(t, t^{\prime})
+ 2 \frac{1}{n} \sum_{i=1}^n 
\left( \Delta_i(t,t') - \tau_i(t, t^{\prime}) \right)
\left( \tfrac{{1}_i(t)}{\pi_i(t)} \mu(t) - \tfrac{{1}_i(t')}{\pi_i(t')} \mu(t') - (\mu(t) - \mu(t') ) \right) \\ 
&- \frac{1}{n} \sum_{i=1}^n 
\left( \tfrac{{1}_i(t) }{\pi_i(t)} \mu(t) - \tfrac{{1}_i(t') }{\pi_i(t')} \mu(t') - (\mu(t) - \mu(t') ) \right)^2. 
\end{align*}
Although the Hájek estimator is often preferred in practice, it is challenging to precisely characterize the conditions under which it outperforms the Horvitz--Thompson estimator under the design-based framework. 
We discuss these conditions in some special cases.
The Horvitz--Thompson estimator tends to exhibit a smaller variance when the outcome variable $Y_i$ is proportional to the propensity score \citep{Fuller2011}.
\cite{SarndalSwenssonWretman2003} outline several situations in which the Hájek estimator is generally considered more efficient than the Horvitz--Thompson estimator: 
(1) the potential outcome is nearly constant; 
(2) the treatment groups are not balanced; 
(3) the propensity scores are weakly or negatively correlated with the potential outcomes.

The presence of interference further complicates the situation due to the additional dependence terms. 
Recall that $\tilde{A}_{i j} = A_{ij}/\sum_{j=1}^n A_{i j}$.
We consider the following outcome model without covariates:
\[
Y_i = \beta_0 + \beta_1 \sum_{j=1}^n \tilde{A}_{i j} D_j + \beta_2 D_i + \varepsilon_i,
\]
and so that $b_n=2$.
One scenario in which the Horvitz--Thompson estimator may be more efficient than the Hájek estimator is when the propensity scores are positively correlated for pairs $(i,j)$ satisfying ${1}(\ell_{ {A}}(i, j) \leq b_n)$, and when both $\beta_1$ and $\beta_2$ are negative.
We set $(\beta_0, \beta_1, \beta_2) = (1,-1,-1)$. 
We simulate ${A}$ from random geometric graph models with $A_{i j}={1}\left\{\left\|\rho_i-\rho_j\right\| \leq r_n\right\}$ for $\rho_i \stackrel{\text { IID }}{\sim} \text{Uniform}\left([0,1]^2\right)$ and $r_n=(\kappa /(\pi n))^2$ with $\kappa = 5$. 
The error term $\varepsilon_i$ is drawn from $\mathcal{N}(0,16)$.  
In the experimental design, units are randomly ordered and treatment is assigned sequentially. Each unit $i$ is initially assigned a baseline treatment probability $p_i$ drawn from $\text{Uniform}[0.2, 0.4]$.
To investigate the direct effect, we use the exposure mapping $T_i = D_i$.
Then for each unit $i$, the effective propensity score for unit $i$ is given by
\[
\pi_i(1) =
\begin{cases}
2 p_i, & \text{if at least one neighbor (processed before $i$) is treated}, \\
p_i,   & \text{otherwise}.
\end{cases}
\] 
This describes an experimental design in which the treatment probability for a unit is doubled if at least one of its neighbors (already processed) has been treated.
The result of Table \ref{table: HT} shows that the Horvitz--Thompson estimator can be more efficient than the Hájek estimator.
Nevertheless, the Hájek estimator outperforms the Horvitz--Thompson estimator in both our simulation and empirical studies.

\begin{table}[tbp]
\centering
\caption{Comparison of the Horvitz--Thompson and Hájek Estimators}
\label{table: HT}
\begin{tabular}{lcc}
\hline
\hline
% & \multicolumn{3}{c}{Design 1} & \multicolumn{3}{c}{Design 2} & \multicolumn{3}{c}{Design 3}\\
% \hline
Estimator  & \multicolumn{1}{c}{Hájek} & \multicolumn{1}{c}{Horvitz--Thompson} \\
\hline
% \multicolumn{4}{c}{Design 1} \\
% \hline
Estimand    & $-$1.067 & $-$1.065	\\
Oracle SE	& 0.307  & 0.301	\\
\hline
\hline
\end{tabular}
% \caption*{\footnotesize Note: Design 1 is no interference but with varying propensity scores, Design 2 is with interference and constant propensity score, and Design 3 is with interference and varying propensity scores.}
\end{table}

\section{Proofs}\label{sec:proof}
\subsection{Auxiliary results}\label{app:Auxiliary}
We first review the definition of weak network dependence in \cite{KojevnikovMarmerSong2021}. 
% \cite{Leung2022} simplied the notation with dimension $v=1$. [PD: dimension of what? $v$ appears later.]
% [PD: better to recall the definition of $ \mathcal{N}_n$]
% Recall that $\mathcal{N}_n = \{1,\ldots,n\}$ denotes the set of units.
For any $H, H^{\prime} \subseteq \mathcal{N}_n$, define $\ell_{ {A}}\left(H, H^{\prime}\right)=\min \left\{\ell_{ {A}}(i, j): i \in H, j \in H^{\prime}\right\}$ as the path distance between two subsets of units within network $A$. 
% [PD: Z for the treatment. does this notation bother you?]
For any random vector $U_i \in \mathbb{R}^v$, we denote its concatenation over $i \in H$  by $U_H=\left(U_i: i \in H\right)$. Let $\mathcal{L}_{v,a}$ be the set of bounded, $\mathbb{R}$-valued, Lipschitz functions on $\mathbb{R}^{v\times a}$. Denote by $\mathcal{P}_n (h, h^{\prime} ; s)$ the set of pairs $H, H' \subseteq \mathcal{N}_n$ with size $h$ and $h'$, respectively, such that the pairs are at least path distance $s$ apart:
$$
\mathcal{P}_n\left(h, h^{\prime} ; s\right)=\left\{\left(H, H^{\prime}\right): H, H^{\prime} \subseteq \mathcal{N}_n,|H|=h,\left|H^{\prime}\right|=h^{\prime}, \ell_{ {A}}\left(H, H^{\prime}\right) \geq s\right\}.
$$

% [PD: move $\|A\|_F$ to the notation section. also it should be $\|A\|_\textup{F}$]

% [PD: move it up to the notation section of the appendix]
% We introduce the definition of $\psi$-dependent from \cite{KojevnikovMarmerSong2021}.
\begin{definition}[$\psi$-dependent]
A triangular array $\{U_i \in \mathbb{R}^v \}_{i=1}^n$ is $\psi$-dependent if there exist (a) dependence coefficients $\{\tilde{\theta}_{n, s}\}_{s, n \in \mathbb{N}}$ that are uniformly bounded constants with $\tilde{\theta}_{n, 0}=1 $ for all $n$ such that $\sup _n \tilde{\theta}_{n, s} \rightarrow 0$ as $s \rightarrow \infty$, and (b) functionals $\left\{\psi_{h, h^{\prime}}(\cdot, \cdot)\right\}_{h, h^{\prime} \in \mathbb{N}}$ with $\psi_{h, h^{\prime}}: \mathcal{L}_{v,h} \times \mathcal{L}_{v,h^{\prime}} \rightarrow[0, \infty)$ such that
\begin{equation} \label{eq:dependent}
\left|\operatorname{Cov}\left(f\left( {U}_H\right), f^{\prime}\left( {U}_{H^{\prime}}\right)\right)\right| \leq \psi_{h, h^{\prime}}\left(f, f^{\prime}\right) \tilde{\theta}_{n, s}    
\end{equation}
for all $n, h, h^{\prime} \in \mathbb{N} ; s>0 ; f \in \mathcal{L}_{v,h} ; f^{\prime} \in \mathcal{L}_{v,h^{\prime}}$; and $\left(H, H^{\prime}\right) \in \mathcal{P}_n\left(h, h^{\prime} ; s\right)$.
\end{definition}
% Define $\hat{{1}}_\text{ht} = \text{diag}\{\hat{{1}}_\text{ht}(t)\}_{t\in\mathcal{T}}$. 

\begin{lemma}\label{lemma:KMS2021lemma}
\citep[][Lemma 2.1]{KojevnikovMarmerSong2021}
Consider an array $\{U_i  \in \mathbb{R}^v \}_{i=1}^n $. The array is $\psi$-dependent in that \eqref{eq:dependent} holds with the dependence coefficients $\{\tilde{\theta}_{n, s}\}_{s, n \in \mathbb{N}}$ that are uniformly bounded. For each $n \geq 1$, let $\{c_{i}\}_{i=1}^n$ be a sequence of vectors in $\mathbb{R}^v$ such that $\max _{i \in \mathcal{N}_n}\left\|c_{i}\right\| \leq 1$. Then the array $\{\tilde{U}_i\}_{i=1}^n $ defined by $\tilde{U}_i=c_{i}^{\top} U_i$ is $\psi$-dependent with the dependence coefficients $\{\tilde{\theta}_{n, s}\}_{s, n \in \mathbb{N}}$. 
\end{lemma}
% \begin{proof}
% \end{proof}
Lemma \ref{lemma:KMS2021lemma} shows that $\psi$-dependence of random vectors carries over to linear combinations of their elements. 

\begin{lemma}\label{lemma:KMS2019thm}
\citep[][Theorem 3.2]{KojevnikovMarmerSong2019}
Consider an array $\{U_i \in \mathbb{R} \}_{i=1}^n$. Define $\sigma_n^2 = \operatorname{Var}(n^{-1 / 2} \sum_{i=1}^n U_i)$.
Assume 
\begin{itemize}
\item[(a)] $\{U_i\}_{i=1}^n$ is $\psi$-dependent in that \eqref{eq:dependent} holds with the dependence coefficients $\{\tilde{\theta}_{n, s}\}_{s, n \in \mathbb{N}}$ that are uniformly bounded 
and $\mathbb{E}[U_i] = 0$ for all $i\in\mathcal{N}_n$,
% [PD: why a.s.? $\mathbb{E}[Z_i] = 0$ is not a random statement?]
\item[(b)] For some $p>4$, $\sup _{n \geq 1} \max _{i \in \mathcal{N}_n} (\mathbb{E}[|U_i|^p ])^{1 / p} < \infty$, 
% [PD: why a.s. again? also should $\max _{i \in N_n}$ be $\max _{i \in \mathcal{N}_n}$?]
\item[(c)] Recall the definitions of $M_n(s,k)$ and $\mathcal{H}_n(s, m)$ in \eqref{eq:Mn} and \eqref{eq:Hn}.
There exist $\epsilon>0$ and a positive sequence $\{m_n\}_{n \in \mathbb{N}}$ such that as $n\rightarrow \infty$ we have $m_n \rightarrow \infty$ and 
\[
\sigma_n^{-4} n^{-2} \sum_{s=0}^n\left|\mathcal{H}_n\left(s, m_n\right)\right| \tilde{\theta}_{n, s}^{1-\epsilon} \rightarrow 0,
\text{ }
\sigma_n^{-3} n^{-1 / 2} M_n\left(m_n, 2\right) \rightarrow 0, 
\text{ }
\sigma_n^{-1} n^{3 / 2} \tilde{\theta}_{n, m_n}^{1-\epsilon}  \rightarrow 0. 
\]
% \begin{align*}
% \sigma_n^{-4} n^{-2} \sum_{s=0}^n\left|\mathcal{H}_n\left(s, m_n\right)\right| \tilde{\theta}_{n, s}^{1-\epsilon} \rightarrow 0,
% \text{ }
% \sigma_n^{-3} n^{-1 / 2} M_n\left(m_n, 2\right) \rightarrow 0, 
% \text{ }
% \sigma_n^{-1} n^{3 / 2} \tilde{\theta}_{n, m_n}^{1-\epsilon}  \rightarrow 0,  
% \end{align*}
\end{itemize}
Then $\sup _{t \in \mathbb{R}}|\mathbb{P}
( \sigma_n^{-1}  n^{-1 / 2} \sum_{i=1}^n U_i \leq t 
) -\Phi(t)| 
\stackrel{}{\rightarrow} 0 \text {, as } n \rightarrow \infty$.
% $$
% \sup _{t \in \mathbb{R}}\left|\mathbb{P}
% \left(
%     \sigma_n^{-1}  n^{-1 / 2} \sum_{i=1}^n U_i \leq t 
% \right) -\Phi(t)\right| 
% \stackrel{}{\rightarrow} 0 \text {, as } n \rightarrow \infty.
% $$
% [PD: why a.s. again?]
% [PD: move $\Phi$ up to the notation section]
\end{lemma}
Lemma \ref{lemma:KMS2019thm} establishes the CLT for the normalized sum with weak dependence $\{U_i\}_{i=1}^n$.
% The published version \cite{KojevnikovMarmerSong2021} formulates (c) slightly differently.  
% [PD: they why do we prove it? just cite the papers?]
% \begin{proof}[Proof of Lemma \ref{lemma:KMS2019thm}]

% \end{proof}

\begin{lemma}\label{lemma:KMS2019prop}
% \citep[][Proposition 4.1]{KojevnikovMarmerSong2021}
Recall that $K_{n}(i,j) = {1}(\ell_{ {A}}(i, j)\le b_n)$.
Consider an array $\{ U_i \in \mathbb{R}^v \}_{i=1}^n$. The array is $\psi$-dependent in that \eqref{eq:dependent} holds with the dependence coefficients $\{\tilde{\theta}_{n, s}\}_{s, n \in \mathbb{N}}$ that are uniformly bounded  and $\mathbb{E}(U_i) = 0$ 
% [PD: as]
for all $i\in\mathcal{N}_n$. Define $V_n = \operatorname{Var}(n^{-1 / 2} \sum_{i=1}^n U_i)$ and $\tilde{V}_n = n^{-1} \sum_{i=1}^n \sum_{j=1}^n U_i U_j^\top 
K_n(i,j)$. 
% \[
% \tilde{V}_n = n^{-1} \sum_{i=1}^n \sum_{j=1}^n U_i U_j^\top 
% K_n(i,j).  
% \] 
Under Assumption \ref{asu7}, we have $\mathbb{E}[ \| \tilde{V}_n - V_n \|_\textup{F} ] \rightarrow 0$.
% $$
% \mathbb{E}\left[\left\| \tilde{V}_n - V_n \right\|_\textup{F} \right] \rightarrow 0.
% $$
% where $\|\cdot\|_F$ denotes the Frobenius norm. For a real matrix $A$, $\|A\|_F = \sqrt{\text{tr}(A\top A)}$.
% If, in addition, $D_n(b_n) / n \rightarrow 0$, then
% $$
% \mathbb{E}\left[\left\| \hat{ {\Sigma}}_{*,\textup{haj}} - {\Sigma}_\textup{haj}\right\|_F \right] \rightarrow 0 \text { a.s. },
% $$
% where $\hat{ {\Sigma}}_{*,\textup{haj}}$ is defined in \eqref{thm:Hájek_asym_n}. 
\end{lemma}
Lemma \ref{lemma:KMS2019prop} is a special case of Proposition 4.1 of \cite{KojevnikovMarmerSong2021} with uniform kernel and our Assumption \ref{asu7} implies their Assumption 4.1. 
Lemma \ref{lemma:Z_i weak dep} below serves as an analogue to Theorem 1 in \cite{Leung2022}, which establishes the $\psi$-dependence of the array
\[
\left\{ \left( 1_i(t)/\pi_i(t) - 1_i(t')/\pi_i(t') \right) Y_i \right\}_{i=1}^n,
\]
the average of which yields the Horvitz--Thompson estimator of $\tau(t,t')$. 
We rely on Lemma \ref{lemma:Z_i weak dep} to analyze the asymptotic properties of the Hájek estimator.
% [PD: it does not look good when compiled.] 

\begin{lemma}\label{lemma:Z_i weak dep}
% Define $Z_i = \left\{ \frac{{1}_i(t)(Y_i-\mu(t)) }{\pi_i(t)} \right\}_{t\in\mathcal{T}}$.
% Define 
% \begin{equation}
% \psi_{h, h^{\prime}}\left(f, f^{\prime}\right)
% = 2 \sqrt{|\mathcal{T}|} \underline{\pi}^{-1} \left(\|f\|_{\infty}\left\|f^{\prime}\right\|_{\infty}+h\left\|f^{\prime}\right\|_{\infty} \operatorname{Lip}(f)+h^{\prime}\|f\|_{\infty} \operatorname{Lip}\left(f^{\prime}\right)\right). 
% \label{eq:psi}
% \end{equation}
\begin{itemize}
\item[(a)] Under Assumptions \ref{asu1} and \ref{asu2}, $\{ (1_i(t)/\pi_i(t):{t\in\mathcal{T}}) \}_{i=1}^n$ is $\psi$-dependent in that \eqref{eq:dependent} holds with the dependence coefficients $\check{\theta}_{n, s} = {1}\{s\le 2K\}$ for all $n \in \mathbb{N}$ and $s>0$ and $\psi_{h, h^{\prime}} (f, f^{\prime} )
= 2 \sqrt{|\mathcal{T}|} \underline{\pi}^{-1} \|f\|_{\infty} \|f^{\prime} \|_{\infty}$
% \begin{equation*}
% \psi_{h, h^{\prime}}\left(f, f^{\prime}\right)
% = 2 \sqrt{|\mathcal{T}|} \underline{\pi}^{-1} \|f\|_{\infty}\left\|f^{\prime}\right\|_{\infty}
% % \label{eq:psi}
% \end{equation*}
for all $h, h^{\prime} \in \mathbb{N}, f \in \mathcal{L}_{v,h}$, and $f^{\prime} \in \mathcal{L}_{v,h^{\prime}}$;
% , or $h=h^{\prime}=1$ and $f=f^{\prime}=i^*$; 
\item[(b)] Under Assumptions \ref{asu1}-\ref{asu4}, $\{ (1_i(t)(Y_i-\mu(t))/\pi_i(t):{t\in\mathcal{T}})\}_{i=1}^n$ is $\psi$-dependent in that \eqref{eq:dependent} holds with the dependence coefficients $\tilde{\theta}_{n, s}$ defined in \eqref{eq:tildetheta} for all $n \in \mathbb{N}$ and $s>0$ and 
\begin{equation}
\psi_{h, h^{\prime}}\left(f, f^{\prime}\right)
= 2 \sqrt{|\mathcal{T}|} \underline{\pi}^{-1} \left(\|f\|_{\infty}\left\|f^{\prime}\right\|_{\infty}+h\left\|f^{\prime}\right\|_{\infty} \operatorname{Lip}(f)+h^{\prime}\|f\|_{\infty} \operatorname{Lip}\left(f^{\prime}\right)\right)
\label{eq:psi}
\end{equation}
for all $h, h^{\prime} \in \mathbb{N}, f \in \mathcal{L}_{v,h}$, and $f^{\prime} \in \mathcal{L}_{v,h^{\prime}}$.
% , or $h=h^{\prime}=1$ and $f=f^{\prime}=i^*$. 
\end{itemize}
\end{lemma}
\begin{proof}[Proof of Lemma \ref{lemma:Z_i weak dep}]
(b) follows from applying the proof of Theorem 1 in \cite{Leung2022} to the array of random vectors $\{ (1_i(t)(Y_i-\mu(t))/\pi_i(t):{t\in\mathcal{T}})\}_{i=1}^n$. (a) follows from an analogous argument to (b) and the fact that $\check{\theta}_{n, s} \le \tilde{\theta}_{n, s}$.  
% Proposition 2.3 of \cite{KojevnikovMarmerSong2021}.
\end{proof}

\subsection{Proofs of the results in Section \ref{sec:WLS}} \label{proof:Hájek}
% \subsection{WLS fit without covariates}
We start with some useful lemmas. To facilitate the discussion, 
we define
\[
\hat{Y}'_{\text{ht}}(t) = n^{-1} \sum_{i=1}^n \frac{1_i(t)}{\pi_i(t)} (Y_i- \mu(t))
\] 
% \[
% \hat{Y}'_{\text{ht}}(t) = \frac{1}{n} \sum_{i=1}^n \frac{1_i(t)}{\pi_i(t)}  (Y_i- \mu(t)) 
% \]
as the Horvitz--Thompson estimator for the centered outcome $Y_i-\sum_{t\in\mathcal{T}} 1_i(t) \mu(t)$. Let $\hat{Y}'_{\text{ht}}$ be the $|\mathcal{T}|\times 1$ vectorization of $\hat{Y}'_{\text{ht}}(t)$ and $\mu$ be the $|\mathcal{T}|\times 1$ vectorization of $\mu(t)$.
The difference between the Hájek estimator and the true finite-population average equals $\hat{Y}_{\textup{haj}}(t) - \mu(t) = \hat{Y}'_{\text{ht}}(t)/\hat{1}_{\text{ht}}(t)$.
% \[
% \hat{Y}_{\textup{haj}}(t) - \mu(t)
% = \frac{\hat{Y}_{\text{ht}}(t) - \hat{1}_{\text{ht}}(t)\mu(t)}{\hat{1}_{\text{ht}}(t)}  
% = \frac{\hat{Y}'_{\text{ht}}(t)}{\hat{1}_{\text{ht}}(t)}.
% \]
% where $\hat{1}_{\text{ht}}(t) = \frac{1}{n} \sum_{i=1}^n \frac{{1}_i(t)}{\pi_i(t)}$
Define $\hat{1}_{\text{ht}} = \text{diag}\{ \hat{1}_{\text{ht}}(t): t\in\mathcal{T} \}$, and then $\hat{Y}_{\textup{haj}} - \mu = \hat{1}_{\text{ht}}^{-1}\hat{Y}'_{\text{ht}}$. 
\begin{lemma} \label{lemma:asym}
Under Assumptions \ref{asu1}--\ref{asu6}, we have ${\Sigma}^{-1/2}_\textup{haj} \sqrt{n} \hat{Y}'_{\textup{ht}} 
\stackrel{\textup{d}}{\rightarrow}  \mathcal{N}(0,{I})$,
% \[
% {\Sigma}^{-1/2}_\textup{haj} \sqrt{n} \hat{Y}'_{\text{ht}} 
% \stackrel{\textup{d}}{\rightarrow}  \mathcal{N}(0,{I}),    
% \]
where ${\Sigma}_\textup{haj}$ is defined in \eqref{eq:Sigma_haj}.
% \[
% {\Sigma}_\textup{haj} 
% = \operatorname{Var} \left( \left\{ n^{-1/2} \sum_{i=1}^n \frac{{1}_i(t)}{\pi_i(t)} (Y_i-\mu(t)) \right\}_{t\in\mathcal{T}} \right). 
% \]
% with 
% \[
% {\Sigma}^2_\textup{haj} 
% = \operatorname{Var}\left( n^{-1/2} \left\{ \sum_{i=1}^n \frac{{1}_i(t)}{\pi_i(t)} (Y_i-\mu(t)) \right\}_{t\in\mathcal{T}} \right)    
% \]
\end{lemma}
\begin{proof}[Proof of Lemma \ref{lemma:asym}]
% The result follows by applying our Lemma \ref{lemma:Z_i weak dep} and Lemma \ref{lemma:KMS2019thm}.

Define $U_i = (U_i(t): t\in\mathcal{T})$ 
% [PD: confusing with Z for the exposure mapping. ] 
with 
\[
U_i(t) = 1_i(t)\pi_i(t)^{-1}  (Y_i- \mu(t)) - (\mu_i(t) - \mu(t)).
\]  
% Define $\sigma_n^2(t) = \operatorname{Var}(n^{-1 / 2} \sum_{i=1}^n Z_i(t))$.
% By Lemma \ref{lemma:KMS2019thm}, we have $\sigma_n(t)^{-1}n^{-1 / 2} \sum_{i=1}^n Z_i(t)
% \stackrel{\textup{d}}{\rightarrow}  \mathcal{N}(0,1)$. 
By construction,
${\Sigma}^{-1/2}_\textup{haj} \sqrt{n} \hat{Y}'_{\text{ht}} 
= {\Sigma}^{-1/2}_\textup{haj} n^{-1/2} \sum_{i=1}^n U_i$.
By the Cramér--Wold theorem, we have
$n^{-1/2} \sum_{i=1}^n U_i  
\stackrel{\textup{d}}{\rightarrow}  \mathcal{N}(0,{\Sigma}_\textup{haj})$
% \[
% n^{-1/2} \sum_{i=1}^n Z_i  
% \stackrel{\textup{d}}{\rightarrow}  \mathcal{N}(0,{\Sigma}_\textup{haj})
% \]
% \[
% {\Sigma}^{-1/2}_\textup{haj} \sqrt{n} \hat{Y}'_{\text{ht}} 
% \stackrel{\textup{d}}{\rightarrow}  \mathcal{N}(0,{I})    
% \]
if and only if $n^{-1/2} \sum_{i=1}^n w^\top U_i
\stackrel{\textup{d}}{\rightarrow}  \mathcal{N}(0, w^\top{\Sigma}_\textup{haj}w)$ for all $w = ( w_t: t\in\mathcal{T} ) \in \mathbb{R}^{|\mathcal{T}|}$.
Therefore, it suffices to show that as $n \rightarrow \infty$,
\begin{equation}
\sup_{t \in \mathbb{R}}\left|\mathbb{P}
\left( \frac{1}{\sqrt{w^\top{\Sigma}_\textup{haj}w}/\|w\|} n^{-1/2} \sum_{i=1}^n w^\top U_i/\|w\| \leq t 
\right) -\Phi(t)\right| 
\stackrel{\textup{a.s.}}{\rightarrow} 0. \label{eq:convergence}
\end{equation}
Define $\tilde{U}_i = w^\top U_i/\|w\|$ where $\mathbb{E}(\tilde{U}_i) = 0$ for all $i\in\mathcal{N}_n$.
The result \eqref{eq:convergence} follows from verifying the assumptions in Lemma \ref{lemma:KMS2019thm} for the array $\{\tilde{U}_i\}_{i=1}^n$ with $\sigma_n^2 = w^\top{\Sigma}_\textup{haj}w/\|w\|^2$.
By Lemma \ref{lemma:Z_i weak dep}, $\{U_i\}_{i=1}^n$ is $\psi$-dependent with the dependence coefficients $\tilde{\theta}_{n, s}$ defined in \eqref{eq:tildetheta}. 
By Lemma \ref{lemma:KMS2021lemma}, $\{\tilde{U}_i\}_{i=1}^n$ is also $\psi$-dependent with the dependence coefficients $\tilde{\theta}_{n, s}$ defined in \eqref{eq:tildetheta}. Thus Assumption (a) holds.
Assumption (b) holds by uniform boundedness of ${1}_i(t)\pi_i(t)^{-1}$ and $Y_i$ by Assumptions \ref{asu2} and \ref{asu3}.
Assumption (c) holds by Assumption \ref{asu6} and $\sigma_n^2 = w^\top{\Sigma}_\textup{haj}w/\|w\|^2 \ge \lambda_{\min}(\Sigma_{\textup{haj}})$.
Thus, we complete the proof.
% The desired result follows from Lemma \ref{lemma:KMS2019thm} in combination with the conditions given in the theorem.

\end{proof}

% Define $d_i$ as the vectorization of $\{{1}_i(t): t\in\mathcal{T}\}$ and $D$ as the matrix with rows $d_i$.
% Let $\tilde{W}=diag\{ \tilde{w}_i\}_{i=1}^n$, $\tilde{Y}$ and $e$ be the vectorization of $\{\tilde{Y}_i: i=1,\ldots, n \}$ and $\{e_i: i=1,\ldots, n \}$ where $e_i$ is the residual from the above weighted least square regression, and $ {K}_n$ be a symmetry matrix with $ij$th entry ${1}(\ell_{ {A}}(i,j)\le b_n)$. 

% Define $d_i$ as the vectorization of $\{{1}_i(t): t\in\mathcal{T}\}$ and $D$ as the matrix with rows $d_i$.
% Let ${W}=\text{diag}\{ {w}_i: i = 1,\ldots, n\}$ and ${Y}$ be the vectorization of $\{Y_i: i=1,\ldots, n \}$.

\begin{lemma} \label{lemma: beta_haj}
Under Assumptions \ref{asu1}--\ref{asu4} and \ref{asu7}(a), we have $\hat{1}_{\textup{ht}} = I + O_{\mathbb{P}}(n^{-1/2})$ and $\hat{Y}_{\textup{haj}} - \mu = O_{\mathbb{P}}(n^{-1/2})$.
\end{lemma}
\begin{proof}[Proof of Lemma \ref{lemma: beta_haj}]
% The result follows by applying the proof of Theorem 2 in \cite{Leung2022} with Assumption \ref{asu7}(a).
% Note that $\check{\theta}_{n, s} \le \tilde{\theta}_{n,s}$ with $\tilde{\theta}_{n,s}$ defined in (\ref{eq:tildetheta}). 
We prove the results element-by-element.
By applying the proof of Theorem 2 in \cite{Leung2022} with Assumption \ref{asu7}(a) to constant potential outcome $1$ and $\hat{Y}'_{\text{ht}}(t)$, we can show that $\hat{1}_{\text{ht}}(t) = 1 + O_{\mathbb{P}}(n^{-1/2})$ and $\hat{Y}'_{\text{ht}}(t) = O_{\mathbb{P}}(n^{-1/2})$, respectively.
Then the result $\hat{Y}_{\textup{haj}}(t) - \mu(t) = O_{\mathbb{P}}(n^{-1/2})$ follows from $\hat{Y}_{\textup{haj}}(t) - \mu(t) = \hat{1}_{\text{ht}}(t)^{-1}\hat{Y}'_{\text{ht}}(t)$. 

\end{proof}

% \subsection{WLS without covariate adjustment}
Below, we prove the main results in Section \ref{sec:WLS}.

% \begin{proof}[Proof of Proposition \ref{prop:haj}]
% The numerical equivalence between $\hat{\beta}_{\textup{haj}} = \hat{Y}_{\textup{haj}}$ follows from $\hat{\beta}_{\textup{haj}} = ( {Z}^{\top}  {W}  {Z})^{-1}  {Z}^{\top}  {W}  {Y}$. 
% % It is a well-known result on WLS fit \citep{AronowSamii2017}. 
% \end{proof}

% Let ${W}=\text{diag}\{ {w}_i: i = 1,\ldots, n\}$, ${Y}$ and $e_{\textup{haj}}$ be the vectorization of $\{{Y}_i: i=1,\ldots, n \}$ and $\{e_{i}: i=1,\ldots, n \}$ where $e_i$ is the residual from the above weighted least square regression.

% In this case, the regression residual is
% \begin{align*}
% e_i  
% = {Y}_i - \sum_{t\in \mathcal{T}} {1}_i(t) \hat{\beta}_{\textup{haj}}(t)
% \end{align*} 
% Let $e_{\textup{haj}}$ be the vectorization of $\{e_i: i=1,\ldots, n \}$ and $W$ be a symmetrys matrix with $ij$th entry being ${1}(\ell_{ {A}}(i,j)\le b_n)$. 
% The network robust covariances for $\hat{\beta}_{\textup{haj}}$ equal 
% \begin{align*}
% \widehat{V}_{\text{w,haj}}
% =& ( {D}^{\top}  {W}  {D})^{-1}
% (D^{\top} {W} e_{\textup{haj}} W e_{\textup{haj}}^{\top}  {W} D)
% ( {D}^{\top}  {W}  {D})^{-1} 
% \end{align*}

\begin{proof}[Proof of Theorem \ref{thm:Hájek_asym_n}] 
% Define $\Lambda_{\text{ht}} = \text{diag}\{ \hat{1}_{\text{ht}}(t): t\in\mathcal{T} \}$.
% \paragraph*{Asymptotic normality.}
% \textbf{Asymptotic normality.}
% We first show the asymptotic normality.
By Lemma \ref{lemma: beta_haj}, $\hat{1}_{\text{ht}} = {I} + o_\mathbb{P}(1)$.
Since $\hat{\beta}_{\textup{haj}} - \mu 
= \hat{1}_{\text{ht}}^{-1} \hat{Y}'_{\text{ht}}$, 
the asymptotic normality of $\hat{\beta}_{\textup{haj}}$ follows from Slutsky's theorem and Lemma \ref{lemma:asym}.

% By our Lemma \ref{lemma:asym} and Lemma \ref{lemma:KMS2019prop}, we have
% \[
% \hat{ {\Sigma}}_{*,\textup{haj}} 
% = \operatorname{Var}\left(\sqrt{n} \hat{Y}'_{\text{ht}}\right) 
% + o_\mathbb{P}(1)
% = {\Sigma}_\textup{haj}
% + o_\mathbb{P}(1).
% \]
\end{proof}

\begin{proof}[Proof of Theorem \ref{thm:Hájek_bias}] 

We first show the result of ${\Sigma}_{*, \textup{haj}}
= {\Sigma}_{\textup{haj}}
+ o_\mathbb{P}(1)$.
% \textbf{Oracle covariance.}
Define
\[
U_i = \left( \frac{1_i(t)}{\pi_i(t)}  (Y_i- \mu(t)) - (\mu_i(t) - \mu(t)): t \in \mathcal{T} \right). 
\] 
By Lemma \ref{lemma:Z_i weak dep}, $\left\{U_i\right\}_{i=1}^n$ is $\psi$-dependent with dependence coefficients $\tilde{\theta}_{n, s}$ defined in \eqref{eq:tildetheta} for all $n \in \mathbb{N}$ and $s>0$. Then by Lemma \ref{lemma:KMS2019prop} and Lemma \ref{lemma:asym}, we have 
\begin{align*}
\hat{ {\Sigma}}_{*,\textup{haj}} 
= \operatorname{Var}\left(n^{-1 / 2} \sum_{i=1}^n 
U_i\right) + o_\mathbb{P}(1)
= \operatorname{Var}\left(\sqrt{n} \hat{Y}'_{\text{ht}}\right) 
+ o_\mathbb{P}(1)  
= {\Sigma}_\textup{haj}
+ o_\mathbb{P}(1). 
\end{align*}
% \textbf{Formula of $\hat{ {V}}_{\textup{haj}}$.}
Then we show the result of $n \hat{ {V}}_{\textup{haj}}
= {\Sigma}_{*, \textup{haj}} + R_{\textup{haj}} + o_\mathbb{P}(1)$.
% With direct algebra, $\frac{1}{n}  {Z}^{\top}  {W}  {Z} 
% = \text{diag}\{\hat{1}_{\text{ht}}(t): t\in\mathcal{T}\}$
% and 
% \begin{align*}
%  {Z}^{\top}  {W}  {e}_{\textup{haj}}  {K}_n 
%  {e}_{\textup{haj}}  {W}  {Z}    
% =& \sum_{i=1}^n \sum_{j=1}^n
% \frac{{1}_i(t)}{\pi_i(t)} e_i
% \frac{{1}_j(t')}{\pi_j(t')} e_j
% {1}(\ell_{ {A}}(i,j)\le b_n)
% \end{align*}
By algebra,
\begin{align*}
\hat{ {V}}_{\textup{haj}}
% =& ( {Z}^{\top}  {W}  {Z})^{-1}
% ( {Z}^{\top}  {W} 
%     {e}_{\textup{haj}}
%     {K}_n 
%     {e}_{\textup{haj}}
%     {W}  {Z})
% ( {Z}^{\top}  {W}  {Z})^{-1} \\
=& \left( n^{-2} \sum_{i=1}^n \sum_{j=1}^n 
\frac{1_i(t) (Y_i- \hat{\beta}_{\textup{haj}}(t)) }{\pi_i(t) \hat{1}_{\text{ht}}(t)}  
\frac{1_j(t') (Y_j-\hat{\beta}_{\textup{haj}}(t')) }{\pi_j(t') \hat{1}_{\text{ht}}(t')} 
K_n(i,j)
\right)_{t,t'\in\mathcal{T}}. 
% =& \Lambda_{\text{ht}}^{-1}
%  {K}_n 
% \Lambda_{\text{ht}}^{-1}
\end{align*}
% \paragraph*{Bias of $\hat{V}_{\textup{haj}}$.}
Let $\hat{ {V}}_{\textup{haj}}(t,t')$ and $\hat{ {\Sigma}}_{*,\textup{haj}}(t,t')$ be the $(t,t')$th entry of $\hat{ {V}}_{\textup{haj}}$ and $\hat{ {\Sigma}}_{*,\textup{haj}}$, respectively. We have
\begin{align*}
& n \hat{ {V}}_{\textup{haj}}(t,t')
% = \frac{1}{n} \sum_{i=1}^n \sum_{j=1}^n 
% \frac{1_i(t) (Y_i- \hat{\beta}_{\textup{haj}}(t))}{\pi_i(t) \hat{1}_{\text{ht}}(t)}   
% \frac{1_j(t') (Y_j-\hat{\beta}_{\textup{haj}}(t')) }{\pi_j(t') \hat{1}_{\text{ht}}(t')} 
% K_n(i,j) \\
=
\frac{1}{n} \sum_{i=1}^n \sum_{j=1}^n 
\frac{1_i(t)(Y_i- \hat{\beta}_{\textup{haj}}(t))}{\pi_i(t)}   
\frac{1_j(t')(Y_j-\hat{\beta}_{\textup{haj}}(t'))}{\pi_j(t')} 
K_n(i,j) + o_\mathbb{P}(1) \\
% =& \hat{ {\Sigma}}_{*,\textup{haj}}(t,t') 
% + \frac{1}{n}\sum_{i=1}^n \sum_{j=1}^n 
% \left( \mu_i(t) - \mu(t)  \right)
% \left( \mu_j(t) - \mu(t)  \right)
% {1}(\ell_{ {A}}(i,j)\le b_n) \\
% \overset{(1)}{=}
% & \frac{1}{n} \sum_{i=1}^n \sum_{j=1}^n
% \begin{array}{c}
% \left( \frac{{1}_i(t)(Y_i - \hat{\beta}_{\textup{haj}}(t))}{\pi_i(t)} \right) 
% \left( \frac{{1}_j(t')(Y_j- \hat{\beta}_{\textup{haj}}(t'))}{\pi_j(t')}  \right)
% \end{array}
% {1}(\ell_{ {A}}(i,j)\le b_n)
% + o_\mathbb{P}(1) \\
&= \frac{1}{n} \sum_{i=1}^n \sum_{j=1}^n
% \begin{array}{c}
\frac{1_i(t)(Y_i - \mu(t))}{\pi_i(t)}
\frac{1_j(t')(Y_j- \mu(t'))}{\pi_j(t')}
% \end{array} 
K_n(i,j) + o_\mathbb{P}(1) \\
&+ \frac{1}{n} \sum_{i=1}^n \sum_{j=1}^n
% \begin{array}{c}
\left(\frac{1_i(t)(\mu(t) - \hat{\beta}_{\textup{haj}}(t))}{\pi_i(t)}   
+ 2 \frac{1_i(t)(Y_i - \mu(t)) }{\pi_i(t)}   
\right) 
\frac{1_j(t')(\mu(t')- \hat{\beta}_{\textup{haj}}(t'))}{\pi_j(t')}  
K_n(i,j),
% =& \frac{1}{n} \sum_{i=1}^n \sum_{j=1}^n
% \begin{array}{c}
% \left( \frac{{1}_i(t)(Y_i - \mu(t))}{\pi_i(t)}  \right) 
% \left( \frac{{1}_j(t')(Y_j- \mu(t'))}{\pi_j(t')} \right)
% \end{array} 
% {1}(\ell_{ {A}}(i,j)\le b_n) 
% % +& \frac{1}{n} \sum_{\ell_{ {A}}(i,j)\le b_n}
% % \begin{array}{c}
% % \left(\frac{{1}_i(t)(\mu(t) - \hat{\beta}_{\textup{haj}}(t))}{\pi_i(t)}   
% %  + 2 \frac{{1}_i(t)(Y_i - \mu(t)) }{\pi_i(t)}   
% %   \right) 
% % \left( \frac{{1}_j(t')(\mu(t')- \hat{\beta}_{\textup{haj}}(t'))}{\pi_j(t')}  \right)
% % \end{array}
% + o_\mathbb{P}(1),
\end{align*}
where the first equality holds by Slutsky's Theorem, Assumptions \ref{asu2}--\ref{asu3}, and Lemma \ref{lemma: beta_haj} that $\hat{1}_{\text{ht}}(t) = 1 + o_\mathbb{P}(1)$ for all $t\in\mathcal{T}$. 
By Lemma \ref{lemma: beta_haj}, $\hat{\beta}_{\textup{haj}}(t) - \mu(t) = O_{\mathbb{P}}(n^{-1/2})$ for all $t\in\mathcal{T}$, and under Assumptions \ref{asu2} and \ref{asu3}, ${1}_i(t)\pi_i(t)^{-1} (Y_i- \mu(t))$ is uniformly bounded. Then for some $C >0$ and any $n$, we have
\begin{align*}
& \left| \frac{1}{n} \sum_{i=1}^n \sum_{j=1}^n
% \begin{array}{c}
\left(\frac{{1}_i(t)(\mu(t) - \hat{\beta}_{\textup{haj}}(t))}{\pi_i(t)}    
+ 2 \frac{{1}_i(t)(Y_i - \mu(t))}{\pi_i(t)}   \right) 
\frac{{1}_j(t')(\mu(t') - \hat{\beta}_{\textup{haj}}(t')) }{\pi_j(t')} K_n(i,j)
% \end{array}
\right| \\
& \le 
C \left| \mu(t') - \hat{\beta}_{\textup{haj}}(t') \right| 
\frac{1}{n} \sum_{i=1}^n \sum_{j=1}^n
K_n(i,j) \\
& = o_\mathbb{P}(1),
\end{align*}
where the last line holds by Assumption \ref{asu7}(b).
Furthermore,
\begin{align*}
& n \hat{V}_{\textup{haj}}(t,t') 
= \frac{1}{n} \sum_{i=1}^n \sum_{j=1}^n
\frac{{1}_i(t)(Y_i - \mu(t))}{\pi_i(t)} 
\frac{{1}_j(t')(Y_j - \mu(t'))}{\pi_j(t')}
K_n(i,j)
+ o_\mathbb{P}(1) \\
=& \hat{ {\Sigma}}_{*,\textup{haj}}(t,t')
+ \frac{1}{n} \sum_{i=1}^n \sum_{j=1}^n
(\mu_i(t) - \mu(t))
(\mu_j(t') - \mu(t'))
K_n(i,j) 
+ r_n(t,t')
+ r_n(t',t)
+ o_\mathbb{P}(1),
\end{align*}
where 
\[
r_n(t,t')
= \frac{1}{n} \sum_{i=1}^n \sum_{j=1}^n
% \begin{array}{c}
\left( \frac{{1}_i(t)(Y_i - \mu(t))}{\pi_i(t)}     
- (\mu_i(t) - \mu(t)) \right)   
(\mu_j(t') - \mu(t'))
% \end{array}
K_n(i,j).
\]
Now we show $r_n(t,t') = o_\mathbb{P}(1)$.
Define $W_i=\sum_{j=1}^n\left(\mu_j(t^{\prime})-\mu(t^{\prime})\right) K_n(i,j)$. Then we have
$$
\begin{aligned}
& \mathbb{E}\left[
\left| 
\frac{1}{n} \sum_{i=1}^n \sum_{j=1}^n
% \begin{array}{c}
\left( \frac{{1}_i(t)(Y_i - \mu(t))}{\pi_i(t)}      
- (\mu_i(t) - \mu(t)) \right)
\left(\mu_j(t^{\prime})-\mu(t^{\prime})\right) 
% \end{array}
K_n(i,j) \right| \right] \\
& \leq  \mathbb{E}\left[\left(
\frac{1}{n} \sum_{i=1}^n
\left( \frac{{1}_i(t)(Y_i - \mu(t))}{\pi_i(t)}     
- (\mu_i(t) - \mu(t)) \right) W_i\right)^2\right]^{1 / 2} \\
& \leq \left(\frac{1}{n^2} \sum_{i=1}^n \operatorname{Var}\left( \frac{{1}_i(t)(Y_i - \mu(t))}{\pi_i(t)}    \right) W_i^2
+ C \frac{1}{n^2} \sum_{s=0}^n \tilde{\theta}_{n, s} \sum_{j \neq i} {1}(\ell_{ {A}}(i, j)=s)\left|W_i W_j\right|\right)^{1 / 2}
\end{aligned}
$$
for some $C>0$, where the first inequality holds by Jensen's inequality and the second inequality holds by Lemma \ref{lemma:Z_i weak dep}.
Since ${1}_i(t)(Y_i - \mu(t))/\pi_i(t)$ is uniformly bounded by Assumptions \ref{asu2} and \ref{asu3}, for some $C^{\prime}>0$,
\[
n^{-2} \sum_{i=1}^n \operatorname{Var}\left(\frac{{1}_i(t)(Y_i - \mu(t))}{\pi_i(t)} \right) W_i^2 \leq C^{\prime} n^{-1} M_n\left(b_n, 2\right),
\]
which is $o(1)$ by Assumption \ref{asu7}(c). Likewise,
$$
n^{-2} \sum_{s=0}^n \tilde{\theta}_{n, s} \sum_{i=1}^n \sum_{j \neq i} {1}(\ell_{ {A}}(i, j)=s)\left|W_i W_j\right| \leq \frac{C^{\prime \prime}}{n^2} \sum_{s=0}^n \tilde{\theta}_{n, s} \mathcal{J}_n\left(s, b_n\right)
$$
for some $C^{\prime \prime}>0$, and the right-hand side term is $o(1)$ by Assumption \ref{asu7}(d). The result $r_n(t',t) = o_\mathbb{P}(1)$ follows from symmetry. 
Thus, we complete the proof. 
\end{proof}

\begin{proof}[Proof of Theorem \ref{thm:Hájek_adj}] 
% To correct the bias term, we replace the truncated kernel with an adjusted weight matrix.
% Define $\hat{\mathbf{\Delta}}_{\textup{haj}}(t,t')$ as the stacked vector of
% \[
% \frac{\frac{{1}_i(t)}{\pi_i(t)}}{\hat{1}_{\text{ht}}(t)} (Y_i - \hat{\beta}_{\textup{haj}}(t))   - \frac{\frac{{1}_i(t')}{\pi_i(t')}}{\hat{1}_{\text{ht}}(t')}   (Y_i- \hat{\beta}_{\textup{haj}}(t')). 
% \]
% We define the adjusted covariance estimator as
% \[
% \hat{\Sigma}^2_{w,haj,+}(t,t') 
% = \frac{1}{n} \hat{\mathbf{\Delta}}_{\textup{haj}}(t,t')  {K}_n^{+} \hat{\mathbf{\Delta}}_{\textup{haj}}(t,t').
% \]
Let $\hat{ {V}}_{\textup{haj}}^+(t,t')$ be the $(t,t')$th element of $\hat{ {V}}_{\textup{haj}}^+$.  
We have
\begin{align*}
 n \hat{ {V}}_{\textup{haj}}^+(t,t') 
=& n \hat{ {V}}_{\textup{haj}}(t,t')
+ \frac{1}{n} \sum_{i=1}^n \sum_{j=1}^n 
\frac{{1}_i(t) (Y_i- \hat{\beta}_{\textup{haj}}(t))}{\pi_i(t) \hat{1}_{\text{ht}}(t)}   
\frac{{1}_j(t') (Y_j-\hat{\beta}_{\textup{haj}}(t'))}{\pi_j(t') \hat{1}_{\text{ht}}(t')}  
K_{n}^-(i,j) \\
% =& \hat{\Sigma}_{*,\textup{haj}}(t,t') + R_{\textup{haj}}(t,t') 
% +  \frac{1}{n} \sum_{i=1}^n \sum_{j=1}^n 
% % \begin{array}{c}
% \frac{{1}_i(t) (Y_i- \hat{\beta}_{\textup{haj}}(t)) {1}_j(t') (Y_j-\hat{\beta}_{\textup{haj}}(t'))}{\pi_i(t) \pi_j(t')}  
% % \frac{ }{}  
% % \end{array}
% {1}^-\{\ell_{ {A}}(i,j)\le b_n\} + o_\mathbb{P}(1) \\
=& \hat{\Sigma}_{*,\textup{haj}}(t,t')
+ \frac{1}{n} \sum_{i=1}^n \sum_{j=1}^n
(\mu_i(t) - \mu(t))
(\mu_j(t') - \mu(t')) 
\left(K_{n}^+(i,j) - K_{n}^-(i,j)\right)  \\
% &- \frac{1}{n} \sum_{i=1}^n \sum_{j=1}^n
% (\mu_i(t) - \mu(t))
% (\mu_j(t') - \mu(t')) 
% K_{n}^-(i,j) \\
+&  \frac{1}{n} \sum_{i=1}^n \sum_{j=1}^n 
\frac{{1}_i(t) (Y_i- \hat{\beta}_{\textup{haj}}(t)) }{\pi_i(t)}  
\frac{{1}_j(t') (Y_j-\hat{\beta}_{\textup{haj}}(t')) }{\pi_j(t')}  
K_{n}^-(i,j) + o_\mathbb{P}(1),
% =& \frac{1}{n} \hat{\mathbf{\Delta}}_{\textup{haj}}(t,t')^{\top} 
%  {K}_n^{+} 
% \hat{\mathbf{\Delta}}_{\textup{haj}}(t,t') \\
% =& \hat{\Sigma}^{2}_\textup{haj}(t,t')
% + \frac{1}{n} 
% \hat{\mathbf{\Delta}}_{\textup{haj}}(t,t')^{\top} 
%  {K}_n^{-}
% \hat{\mathbf{\Delta}}_{\textup{haj}}(t,t') \\
% =& \hat{\Sigma}_{*,\textup{haj}}^{2}(t,t') + R_\textup{haj}(t,t')
% + \frac{1}{n} 
% \hat{\mathbf{\Delta}}_{\textup{haj}}(t,t')^{\top} 
%  {K}_n^{-}
% \hat{\mathbf{\Delta}}_{\textup{haj}}(t,t')
% + o_\mathbb{P}(1) \\
% =& \hat{\Sigma}_{*,\textup{haj}}^{2}(t,t')
% + \frac{1}{n} 
% \overline{ \tau}(t,t')^{\top} 
%  {K}_n^{+} 
% \overline{ \tau}(t,t') 
% - \frac{1}{n} 
% \overline{ \tau}(t,t')^{\top} 
%  {K}_n^{-}  
% \overline{ \tau}(t,t') 
% + \frac{1}{n} 
% \hat{\mathbf{\Delta}}_{\textup{haj}}(t,t')^{\top} 
%  {K}_n^{-}  
% \hat{\mathbf{\Delta}}_{\textup{haj}}(t,t')
% + o_\mathbb{P}(1)
\end{align*}
where the second equality holds by Lemma \ref{lemma: beta_haj}, Theorem \ref{thm:Hájek_bias} and definition of $K_n$. 
% \begin{align*}
% & \frac{1}{n} \sum_{i=1}^n \sum_{j=1}^n 
% \frac{{1}_i(t) (Y_i- \hat{\beta}_{\textup{haj}}(t)) }{\pi_i(t)}  
% \frac{{1}_j(t') (Y_j-\hat{\beta}_{\textup{haj}}(t')) }{\pi_j(t')}  
% K_{n}^-(i,j) \\
% &= \frac{1}{n} \sum_{i=1}^n \sum_{j=1}^n
% % \begin{array}{c}
% \frac{1_i(t)(Y_i - \mu(t))}{\pi_i(t)}
% \frac{1_j(t')(Y_j- \mu(t'))}{\pi_j(t')}
% % \end{array} 
% K_{n}^-(i,j)\\
% &+ (\mu(t')- \hat{\beta}_{\textup{haj}}(t')) \frac{1}{n} \sum_{i=1}^n 
% % \begin{array}{c}
% \left(\frac{1_i(t)(\mu(t) - \hat{\beta}_{\textup{haj}}(t))}{\pi_i(t)}   
% + 2 \frac{1_i(t)(Y_i - \mu(t)) }{\pi_i(t)}   
% \right) 
% \sum_{j=1}^n \frac{1_j(t')}{\pi_j(t')}
% K_{n}^-(i,j)
% \end{align*}
Applying the proof of Theorem \ref{thm:Hájek_bias} but replacing Assumption \ref{asu7} with Assumption \ref{asu8}, we can show that 
\begin{align*}
& \frac{1}{n} \sum_{i=1}^n \sum_{j=1}^n 
\frac{{1}_i(t) (Y_i- \hat{\beta}_{\textup{haj}}(t)) }{\pi_i(t)}  
\frac{{1}_j(t') (Y_j-\hat{\beta}_{\textup{haj}}(t')) }{\pi_j(t')}  
K_{n}^-(i,j) \\
=& \frac{1}{n} \sum_{i=1}^n \sum_{j=1}^n
(\mu_i(t) - \mu(t))
(\mu_j(t') - \mu(t')) 
K_{n}^-(i,j) \\
+& \frac{1}{n} \sum_{i=1}^n \sum_{j=1}^n 
% \begin{array}{c}
\left( \frac{{1}_i(t)(Y_i - \mu(t))}{\pi_i(t)}  - M(i,t)  \right)
\left( \frac{{1}_j(t')(Y_j - \mu(t'))}{\pi_j(t')}  - M(j,t') \right)
K_{n}^-(i,j) 
+ o_\mathbb{P}(1).
\end{align*}
Thus, we complete the proof. 
\end{proof}

\subsection{Proofs of the results in Section \ref{sec:covadju}}\label{app:covadj}
\subsubsection{Some useful lemmas}
% Define the following $Q$'s to simplify the presentation: 
% \begin{align*}
% Q_{xx} =& n^{-1} \sum_{i=1}^n x_ix_i^\top, \quad  
% Q_{xx}(t;\pi) = n^{-1} \sum_{i=1}^n 
% x_ix_i^\top \frac{1- \pi_i(t) }{\pi_i(t)}, \\
% Q_{xx}(t;\pi^{-1}) =& n^{-1} \sum_{i=1}^n
% \frac{ x_ix_i^\top }{ \pi_i(t) }, \quad 
% Q_{xx}(t,t';\pi) =  n^{-1}  \sum_{i=1}^n
% x_ix_i^\top \tfrac{ \pi_i(t) + \pi_i(t') }{\pi_i(t)\pi_i(t')}. 
% \end{align*}
Define $Q_{xx} = n^{-1} \sum_{i=1}^n x_ix_i^\top$ and the covariate-adjusted outcome as $Y_i(t;\gamma) = Y_i - x_i^\top {\gamma}(t)$. 
Let
\[
\hat{x}_{\text{ht}}(t) = n^{-1} \sum_{i=1}^n \frac{{1}_i(t)}{\pi_i(t)} x_i
\] 
be a $J\times 1$ Horvitz--Thompson estimator for $\bar{x}=n^{-1}\sum_{i=1}^n x_i=0$ under exposure mapping value $t$ and then combine $\hat{x}_{\text{ht}}(t)$ across all $t\in\mathcal{T}$ to get a $|\mathcal{T}|\times J$ matrix $\hat{x}_{\text{ht}} = (\hat{x}_{\text{ht}}(t): {t\in\mathcal{T}})$. 
\begin{lemma} \label{lemma: x_ht}
Under Assumptions \ref{asu1}--\ref{asu2}, \ref{asu7}(a) and \ref{asu9}, for all $t \in \mathcal{T}$, we have
\begin{itemize}
\item[(i)] $n^{-1}\sum_{i=1}^n {1}_i(t)\pi_i(t)^{-1} x_i = O_{\mathbb{P}}(n^{-1/2})$,
\item[(ii)] $n^{-1}\sum_{i=1}^n {1}_i(t)\pi_i(t)^{-1} x_ix_i^\top = Q_{xx} + O_{\mathbb{P}}(n^{-1/2})$,
\item[(iii)] $n^{-1}\sum_{i=1}^n {1}_i(t)\pi_i(t)^{-1} x_i Y_i = n^{-1}\sum_{i=1}^n x_i \mu_i(t) + O_{\mathbb{P}}(n^{-1/2})$.
\end{itemize}

\end{lemma}
\begin{proof}[Proof of Lemma \ref{lemma: x_ht}]
The results follow from an analogous argument to the proof of Lemma \ref{lemma: beta_haj}.  
\end{proof}

\begin{lemma} \label{lemma: asym}
Consider the Hájek estimator with covariate adjustment: 
\[
\hat{\beta}_{\textup{haj}}(\hat{\gamma} )
=(\hat{\beta}_{\textup{haj}}(t;\hat{\gamma}):{t\in\mathcal{T}})
= \left( \frac{1}{n} \sum_{i=1}^n 
\frac{{1}_i(t)Y_i(t;\hat{\gamma})}{\pi_i(t)} \big/ \hat{1}_{\textup{ht}}(t)  : {t\in\mathcal{T}} \right). 
\]
Under Assumptions \ref{asu1}--\ref{asu6} and \ref{asu9}, with some fixed vector $\gamma$ that satisfies $\hat{\gamma} = {\gamma} + o_\mathbb{P}(1)$,
we have ${\Sigma}^{-1/2}_{\textup{haj}}({\gamma} ) 
\sqrt{n} ( \hat{\beta}_{\textup{haj}} (\hat{\gamma} ) -  {\mu} ) 
\stackrel{\textup{d}}{\rightarrow}  \mathcal{N}(0,{I})$,
% \[
% {\Sigma}^{-1/2}_{\textup{haj}}({\gamma} ) 
% \sqrt{n}\left( \hat{\beta}_{\textup{haj}} (\hat{\gamma} ) -  {\mu} \right) 
% \stackrel{\textup{d}}{\rightarrow}  \mathcal{N}(0,{I}),
% \]
where ${\Sigma}_{\textup{haj}}({\gamma} ) 
= \operatorname{Var} (  n^{-1/2} \sum_{i=1}^n  {1}_i(t)(Y_i(t;\gamma) -\mu(t) )/\pi_i(t)  :{t\in\mathcal{T}} )$.
% \[
% {\Sigma}_{\textup{haj}}({\gamma} ) 
% = \operatorname{Var}\left(  \frac{1}{\sqrt{n}} \sum_{i=1}^n 
% \frac{{1}_i(t)(Y_i(t;\gamma)   -\mu(t) )}{\pi_i(t)}  :{t\in\mathcal{T}} \right).
% \]
% ${\Sigma}_n^2(t,t',{\gamma} )
% = \operatorname{Var}\left(\hat{\tau}_{\text{ht}}\left(t, t^{\prime}, {\gamma} \right)\right)$. 
% with 
% $$ 
% \tilde{\Delta}_i(t,t',{\gamma} ) 
% = \left( \frac{{1}_i(t)}{\pi_i(t)} (Y_i-{\gamma}(t)^\top x_i  -\mu(t) ) - \frac{{1}_i(t')}{\pi_i(t')} (Y_i- {\gamma}(t')^\top x_i  -\mu(t') ) \right). 
% $$
\end{lemma}
\begin{proof}[Proof of Lemma \ref{lemma: asym}]
% Define
% \[
% \hat{\beta}_{\textup{haj}}({\gamma} ) = \hat{Y}_{\textup{haj}}({\gamma}) = \left( \frac{1}{\hat 1_{\text{ht}}(t)}  \frac{1}{n} \sum_{i=1}^n \frac{{1}_i(t)}{\pi_i(t)} (Y_i-x_i^\top {\gamma}(t) ): t\in\mathcal{T} \right). 
% \]
Define $\hat{\beta}_{\textup{haj}}({\gamma} ) =(\hat{\beta}_{\textup{haj}}(t;\gamma):{t\in\mathcal{T}})$ with
\begin{equation}
\hat{\beta}_{\textup{haj}}(t;\gamma) = \frac{1}{\hat{1}_{\text{ht}}(t)} \frac{1}{n}\sum_{i=1}^n 
\frac{{1}_i(t)}{\pi_i(t)} Y_i(t;\gamma) 
= \hat{\beta}_{\textup{haj}}(t;\hat{\gamma}) - \hat{x}_\textup{haj}(t)^\top ({\gamma}(t) - \hat{\gamma}(t)). 
\label{eq: beta_gamma}
\end{equation}
We first apply Theorem \ref{thm:Hájek_asym_n} to the adjusted outcome $Y_i(t;\gamma)$, then the asymptotic normality of $\hat{\beta}_{\textup{haj}}({\gamma})$ follows: ${\Sigma}_{\textup{haj}}^{-1/2}({\gamma}) 
\sqrt{n}( \hat{\beta}_{\textup{haj}}({\gamma}) -  {\mu} ) 
\stackrel{\textup{d}}{\rightarrow}  \mathcal{N}(0,{I})$.
% \[ 
% {\Sigma}_{\textup{haj}}^{-1/2}({\gamma}) 
% \sqrt{n}\left( \hat{\beta}_{\textup{haj}}({\gamma}) -  {\mu} \right) 
% \stackrel{\textup{d}}{\rightarrow}  \mathcal{N}(0,{I}).
% \]
% Define $\hat{x}_{\textup{haj}}(t) = \frac{\frac{1}{n}\sum_{i=1}^n \frac{{1}_i(t)}{\pi_i(t)} x_i}{\hat{1}_{\text{ht}}(t)}$. 
By Slutsky's Theorem, we have
\begin{align*}
\sqrt{n}\left( \hat{\beta}_{\textup{haj}}(t;\hat{\gamma}) - \hat{\beta}_{\textup{haj}}(t;{\gamma}) \right)
% = \sqrt{n}(({\gamma}(t) - \hat{\gamma}(t))^\top \hat{x}_{\textup{haj}}(t) - ({\gamma}(t') - \hat{\gamma}(t'))^\top \hat{x}_{\textup{haj}}(t')) \\
=& - (\hat{\gamma}(t) - {\gamma}(t))^\top \sqrt{n} \hat{x}_{\textup{haj}}(t) \\
=& o_\mathbb{P}(1),
\end{align*}
where the last line holds by $\hat{\gamma}(t) - {\gamma}(t) = o_\mathbb{P}(1)$ and the asymptotic normality of $\sqrt{n} \hat{x}_{\textup{haj}}(t)$ follows from an analogous argument to the proof of Theorem \ref{thm:Hájek_asym_n} for $x_i$. This ensures $\sqrt{n} ( \hat{\beta}_{\textup{haj}}(\hat{\gamma}) - \hat{\beta}_{\textup{haj}}({\gamma}) ) = o_\mathbb{P}(1)$.
Thus, we prove the asymptotic normality of $\hat{\beta}_{\textup{haj}}(\hat{\gamma})$ with asymptotic covariance ${\Sigma}_{\textup{haj}}({\gamma} )$.

% Define $\hat{Y}_{\text{ht}}(t,{\gamma}) = \frac{1}{n} \sum_{i=1}^n 
% \frac{{1}_i(t)}{\pi_i(t)} (Y_i - {\gamma}(t)^\top x_i  - \mu(t))$. The asymptotic normality of the vector $\left\{ \hat{Y}_{\text{ht}}(t,{\gamma}) \right\}_{t\in\mathcal{T}}$ follows from Theorem \ref{thm:Hájek_asym_n}.  
% Since $\hat{\tau}_{\textup{haj}}\left(t, t^{\prime},{\gamma}\right) - \tau(t,t') = \frac{ \hat{Y}_{\text{ht}}(t,{\gamma}) }{ \hat{1}_{\text{ht}}(t) } - \frac{ \hat{Y}_{\text{ht}}(t',{\gamma}) }{ \hat{1}_{\text{ht}}(t') }$, 
% from Slutsky's theorem, we have 
% \[
% {\Sigma}_n(t,t',{\gamma})^{-1} 
% \sqrt{n}\left( \hat{\tau}_{\textup{haj}}\left(t, t^{\prime}, {\gamma}\right) -\tau(t, t^{\prime}))\right) 
% \stackrel{\textup{d}}{\rightarrow}  \mathcal{N}(0,1)
% \]
% where ${\Sigma}_n^2(t,t',{\gamma})=\operatorname{Var}\left( \sqrt{n} \hat{\tau}_{\textup{haj}}\left(t, t^{\prime}, {\gamma}\right) \right)$ with 
% $$ 
% \tilde{\Delta}_i(t,t',{\gamma}) = \left( \frac{{1}_i(t)}{\pi_i(t)} (Y_i-{\gamma}(t)^\top x_i -\mu(t) ) - \frac{{1}_i(t')}{\pi_i(t')} (Y_i- {\gamma}(t')^\top x_i -\mu(t') ) \right). 
% $$
\end{proof}

\subsubsection{Additive regression}
We first show the numerical correspondence between $\hat{\beta}_{\textup{haj},\textsc{f}}$ and $\hat{Y}_{\textup{haj}}$.
Let
\[
\hat{x}_{\textup{haj}}(t) 
= n^{-1}\sum_{i=1}^n \frac{1(T_i=t)}{\pi_i(t)\hat{1}_{\text{ht}}(t)} x_i
\]
% \frac{1(T_i=t)x_i}{\pi_i(t)} \Big/ \hat{1}_{\text{ht}}(t)$
% $\hat{x}_{\textup{haj}}(t) 
% = n^{-1}\sum_{i=1}^n 1(T_i=t)\pi_i(t)^{-1} \Big/ \hat{1}_{\text{ht}}(t)$
% \[
% \hat{x}_{\textup{haj}}(t) 
% = \frac{1}{n}\sum_{i=1}^n \frac{1(T_i=t)x_i}{\pi_i(t)} \Big/ \hat{1}_{\text{ht}}(t)
% \]
be the $J\times 1$ Hájek estimator for $\bar{x}$ under exposure mapping value $t$ and then combine $\hat{x}_{\textup{haj}}(t)$ across all $t\in\mathcal{T}$ to obtain the $|\mathcal{T}|\times J$ matrix $\hat{x}_{\textup{haj}} = (\hat{x}_{\textup{haj}}(t): {t\in\mathcal{T}})$. 
Let $\hat{\gamma}_\textsc{f}$ denote the coefficient vector of $x_i$ from the same WLS fit. 

\begin{proposition}\label{prop:add}
$\hat{\beta}_{\textup{haj},\textsc{f}} 
= \hat{Y}_{\textup{haj}} - \hat{x}_{\textup{haj}} \hat{\gamma}_\textsc{f}$.
\end{proposition}
% $\hat\tau_\textsc{f}(t,t')$
Proposition \ref{prop:add} links the covariate-adjusted $\hat{\beta}_{\textup{haj},\textsc{f}}$ back to the unadjusted $\hat{\beta}_{\textup{haj}}$, and establishes $\hat{\beta}_{\textup{haj},\textsc{f}}$ as the Hájek estimator based on the covariate-adjusted outcome $Y_i - x_i^\top \hat{\gamma}_\textsc{f}$. 
% The correspondence between the WLS fit and the Hájek estimation is preserved in the additive WLS fit in \eqref{eq:add} as well.

\begin{lemma} \label{lemma: gamma_F}
Under Assumptions \ref{asu1}--\ref{asu4}, \ref{asu7}(a) and \ref{asu9}, 
we have $\hat{\gamma}_\textsc{f} = {\gamma}_\textsc{f} + O_{\mathbb{P}}(n^{-1/2})$ with
% ${\gamma}_\textsc{f} = |\mathcal{T}|^{-1}
% (\sum_{i=1}^n x_ix_i^\top)^{-1} \sum_{t\in\mathcal{T}}\sum_{i=1}^n x_i \mu_i(t)$
\[
{\gamma}_\textsc{f} 
= \left(\sum_{i=1}^n x_ix_i^\top\right)^{-1} 
\frac{1}{|\mathcal{T}|}\sum_{t\in\mathcal{T}}\sum_{i=1}^n x_i \mu_i(t).
\]
\end{lemma}
% \begin{proof}
% The probability limit follows from $(\hat{\beta}_{\textup{haj},\textsc{f}}^\top, \hat{\gamma}_\textsc{f}^\top)^\top = G_1^{-1}G_2$ where
% \begin{align*}
% G_1 = \text{diag}\left(I, \frac{|\mathcal{T}|}{n}\sum_{i=1}^n x_ix_i^\top\right) + o_\mathbb{P}(1),  
% \quad 
% G_2 = \left(\mu, ( n^{-1} \sum_{t\in\mathcal{T}}\sum_{i=1}^n x_i \mu_i(t))^\top \right)^\top + o_\mathbb{P}(1)
% \end{align*}
% by Lemma \ref{lemma: x_ht}. 
% \end{proof}

\begin{proof}[Proof of Proposition \ref{prop:add} and Lemma \ref{lemma: gamma_F}]
% Let $W=\text{diag}\{{w}_i, i = 1,\ldots, n\}$. Let $\hat{\gamma}_\textsc{f}$ be the coefficient vectors of $x_i$ from the above weighted least squares fits.
We verify below the numerical result in
Proposition \ref{prop:add} and the probability limit in Lemma \ref{lemma: gamma_F} together.
The first-order condition of the WLS fit in \eqref{eq:add} ensures 
\begin{equation}
G_1 \left(\hat{\beta}_{\textup{haj},\textsc{f}}^\top, \hat{\gamma}_\textsc{f}^\top \right)^\top
% \begin{pmatrix}
% \hat{\beta}_{\textup{haj},\textsc{f}}^\top \\
% \hat{\gamma}_\textsc{f}
% \end{pmatrix} 
= G_2,   \label{eq:F} 
\end{equation}
where by direct algebra, 
\begin{align*}
G_1 
=& n^{-1} C_\textsc{f}^\top W C_\textsc{f}
= n^{-1}
\begin{pmatrix}
{Z}^{\top}  {W}  {Z} &  {Z}^{\top}  {W}  {X} \\
{X}^{\top}  {W}  {Z} &  {X}^{\top}  {W}  {X}
\end{pmatrix}
= 
\begin{pmatrix}
\hat{{1}}_{\text{ht}} & \hat{x}_{\text{ht}} \\
\hat{x}_{\text{ht}}^{\top} & n^{-1}  {X}^{\top}  {W}  {X}
\end{pmatrix}, \\
G_2 
=& n^{-1} C_\textsc{f}^\top W Y
= n^{-1}
\begin{pmatrix}
{Z}^{\top}  {W}  {Y} \\
{X}^{\top}  {W}  {Y} 
\end{pmatrix}
= 
\begin{pmatrix}
\hat{Y}_{\text{ht}} \\
n^{-1}  {X}^{\top}  {W}  {Y} 
\end{pmatrix}.
\end{align*}
By comparing the first row of \eqref{eq:F}, we have 
\[
% \hat{{1}}_{\text{ht}} \hat{\beta}_{\textup{haj},\textsc{f}} + \hat{x}_{\textup{haj}} \hat{\gamma}_\textsc{f} 
% = \hat{Y}_{\text{ht}} 
% \implies 
\hat{\beta}_{\textup{haj},\textsc{f}} = \hat{{1}}_{\text{ht}}^{-1} (\hat{Y}_{\text{ht}} - \hat{x}_{\text{ht}} \hat{\gamma}_\textsc{f} )
= \hat{Y}_{\textup{haj}} - \hat{x}_{\textup{haj}} \hat{\gamma}_\textsc{f}.
% \equiv \hat{Y}_{\textup{haj}}(\hat{\gamma}_\textsc{f})
\]
The probability limit follows from \eqref{eq:F} where
\begin{align*}
G_1 =& \text{diag}\left(I, \frac{|\mathcal{T}|}{n}\sum_{i=1}^n x_ix_i^\top\right) + O_{\mathbb{P}}(n^{-1/2}),  \\
G_2 =& \left(\mu^\top, \left( n^{-1} \sum_{t\in\mathcal{T}}\sum_{i=1}^n x_i \mu_i(t) \right)^\top \right)^\top + O_{\mathbb{P}}(n^{-1/2})
\end{align*}
by Lemma \ref{lemma: x_ht}. 
\end{proof}

% \begin{proof}
% \begin{align*}
% {X}^{\top} W {X}
% =&  \\
% {X}^{\top} W {Y}
% =&  
% \end{align*}
% \end{proof}

\begin{proof}[Proof of Theorem \ref{thm:add_asym_n}]
% \textbf{Asymptotic normality.}
By Lemma \ref{lemma: gamma_F}, we have $\hat{\gamma}_\textsc{f} = {\gamma}_\textsc{f} + o_\mathbb{P}(1)$. The asymptotic normality follows by applying Lemma \ref{lemma: asym} with $\hat{\gamma}=\hat{\gamma}_\textsc{f}$ and ${\gamma}= {\gamma}_\textsc{f}$.

% \textbf{Oracle covariance.}

\end{proof}

\begin{proof}[Proof of Theorem \ref{thm:add_bias}]
We first show the result of ${ {\Sigma}}_{*,\textup{haj}, \textsc{f}} = {\Sigma}_{\textup{haj}, \textsc{f}} +  o_\mathbb{P}(1)$.
% Define $\hat{\tau}_\textup{haj}\left(t, t^{\prime}, {\gamma}_\textsc{f}\right) = \hat{Y}_{\textup{haj}}(t, {\gamma}_\textsc{f}) - \hat{Y}_{\textup{haj}}(t', {\gamma}_\textsc{f})$ with $\hat{Y}_{\textup{haj}}(t, {\gamma}_\textsc{f}) = \frac{1}{n} \sum_{i=1}^n \frac{\frac{{1}_i(t)}{\pi_i(t)}}{\hat 1_{\text{ht}}(t)} 
% (Y_i - x_i^\top {\gamma}_\textsc{f} )$.
Applying Theorem \ref{thm:Hájek_asym_n} to adjusted outcome $Y_i(t;\gamma_\textsc{f}) = Y_i- x_i^\top {\gamma}_\textsc{f}$,
% and by Theorem \ref{thm:Hájek_bias},
we have
\[
{ {\Sigma}}_{*,\textup{haj},\textsc{f}} 
= \operatorname{Var}\left(\sqrt{n} \hat{\beta}_{\textup{haj}}({\gamma_\textsc{f}}) \right) + o_\mathbb{P}(1).
\]
Since $\sqrt{n}( \hat{\beta}_{\textup{haj}}({\gamma_\textsc{f}}) - \hat{\beta}_{\textup{haj},\textsc{f}} ) = o_\mathbb{P}(1)$ by Lemma \ref{lemma: gamma_F}, we have
\[
{ {\Sigma}}_{*,\textup{haj},\textsc{f}} 
= \operatorname{Var}\left(\sqrt{n} \hat{\beta}_{\textup{haj},\textsc{f}} \right) + o_\mathbb{P}(1)
= {\Sigma}_{\textup{haj}, \textsc{f}} +  o_\mathbb{P}(1).
\]
% By Lemma \ref{lemma: gamma_F}, we have $\hat{\gamma}_\textsc{f} = {\gamma}_\textsc{f} + o_\mathbb{P}(1)$. Then the asymptotic normality of $\hat{\tau}_\textup{haj}\left(t, t^{\prime},\hat{\gamma}_\textsc{f} \right)$ follows from Lemma \ref{lemma: asym}. By Proposition 4.1 of \cite{KOJEVNIKOV2021882}, 
% \[
% \left| \hat{ {\Sigma}}_{*,{\textup{haj}, \textsc{f}}}^2(t,t',{\gamma}_\textsc{f}) -\operatorname{Var}\left(\sqrt{n} \hat{\tau}_\textup{haj}\left(t, t^{\prime},\hat{\gamma}_\textsc{f}\right) \right)\right| \stackrel{p}{\rightarrow} 0
% \]
% \end{proof}
% \begin{proof}[Proof of Theorem \ref{thm:add_bias}]
% \textbf{Formula of $\hat{ {V}}_{\text{haj,F}}$.}
Then we show the result of $n \hat{ {V}}_{\textup{haj},\textsc{f}} 
= { {\Sigma}}_{*,\textup{haj}, \textsc{f}} + R_{\textup{haj},\textsc{f}} + o_\mathbb{P}(1)$.
By Proposition \ref{prop:add}, $\hat{\beta}_{\textup{haj},\textsc{f}}$ is the Hájek estimator based on the covariate-adjusted outcome $Y_i(t;\hat{\gamma}_\textsc{f}) = {Y}_i - x_i^\top \hat{\gamma}_\textsc{f} $. 
% Define $\hat{Y}_\textup{haj}(t,\hat{\gamma}_\textsc{f})=\frac{1}{n}\sum_{i=1}^n \frac{\frac{{1}_i(t)}{\pi_i(t)}}{\hat{1}_{\text{ht}}(t)}({Y}_i - \hatx_i^\top {\gamma}_\textsc{f} ) $.
The residual from the additive WLS fit in \eqref{eq:add} is
\[
{e}_{\textsc{f},i} 
% = 
% (Y_i - x_i^\top \hat{\gamma}_\textsc{f} ) 
% - \sum_{t\in \mathcal{T}} \frac{{1}_i(t)}{\hat{1}_{\text{ht}}(t)}  
% \frac{1}{n} \sum_{i=1}^n 
% \frac{{1}_i(t)(Y_i - x_i^\top \hat{\gamma}_\textsc{f}  ) }{\pi_i(t)}
= Y_i - x_i^\top \hat{\gamma}_\textsc{f}   
- \sum_{t\in \mathcal{T}} {1}_i(t) \hat{\beta}_{\textup{haj},\textsc{f}}(t).
\] 
The HAC covariance estimator for $\hat{\beta}_{\textup{haj},\textsc{f}}$ equals the upper-left $|\mathcal{T}|\times |\mathcal{T}|$ submatrix of 
\begin{equation}
% \hat{ {V}}_{\text{haj,\textsc{f}}} = 
( {C}_\textsc{f}^{\top}  {W}  {C}_\textsc{f})^{-1}
( {C}_\textsc{f}^{\top}  {W} 
{e}_{\textup{haj},\textsc{f}}
{K}_n 
{e}_{\textup{haj},\textsc{f}}
{W}  {C}_\textsc{f})
( {C}_\textsc{f}^{\top}  {W}  {C}_\textsc{f})^{-1}. \label{eq:V_F}   
\end{equation}
% \textbf{Bias of $\hat{ {V}}_{\text{haj,F}}$.}
Introduce an intermediate term below for the theoretical analysis:  
\begin{align*}
& \hat{\Omega}_{\textup{haj},\textsc{f}}(\hat{\gamma}_\textsc{f})
= n ( {Z}^{\top}  {W}  {Z})^{-1}
( {Z}^{\top}  {W} 
{e}_{\textup{haj},\textsc{f}}
{K}_n 
{e}_{\textup{haj},\textsc{f}}
{W}  {Z})
( {Z}^{\top}  {W}  {Z})^{-1} 
% = (\hat{\Omega}_{\textup{haj},\textsc{f}}(t,t'; \hat{\gamma}_{\textsc{f}}) )_{t,t'\in\mathcal{T}} 
% =& \left( \frac{ \frac{1}{n} \sum_{i=1}^n \sum_{j=1}^n 
% \frac{{1}_i(t)}{\pi_i(t)}  ( Y_i(t;\hat{\gamma}_\textsc{f})- \hat{\beta}_{\textup{haj},\textsc{f}}(t) ) 
% \frac{{1}_i(t')}{\pi_i(t')} ( Y_j(t';\hat{\gamma}_\textsc{f}) -\hat{\beta}_{\textup{haj},\textsc{f}}(t') ) 
% {1}(\ell_{ {A}}(i,j)\le b_n)  }{ \hat{1}_{\text{ht}}(t) \hat{1}_{\text{ht}}(t') } 
% \right)_{t,t'\in\mathcal{T}}.  
\end{align*}
where the $(t,t')$the entry being
\begin{align*}
\hat{\Omega}_{\textup{haj},\textsc{f}}(t,t'; \hat{\gamma}_{\textsc{f}})
=& \frac{1}{n} \sum_{i=1}^n \sum_{j=1}^n 
\frac{{1}_i(t) ( Y_i(t;\hat{\gamma}_\textsc{f})- \hat{\beta}_{\textup{haj},\textsc{f}}(t) ) }{\pi_i(t) \hat{1}_{\text{ht}}(t) }   
\frac{{1}_j(t') ( Y_j(t';\hat{\gamma}_\textsc{f}) -\hat{\beta}_{\textup{haj},\textsc{f}}(t') )  }{\pi_j(t') \hat{1}_{\text{ht}}(t') } 
K_{n}(i,j).   
\end{align*}
The result on $\hat{ {V}}_{\text{haj,\textsc{f}}}$ holds as long as 
\begin{equation}
\text{(i) } \hat{\Omega}_{\text{haj,\textsc{f}}}(\hat{\gamma}_\textsc{f})
= { {\Sigma}}_{*,\textup{haj}, \textsc{f}} + R_{\textup{haj},\textsc{f}} + o_\mathbb{P}(1)
\text{ and (ii) }
n\hat{ {V}}_{\text{haj,\textsc{f}}} - \hat{\Omega}_{\text{haj,\textsc{f}}}(\hat{\gamma}_\textsc{f}) = o_\mathbb{P}(1). \label{conditionF}
\end{equation}
We verify below these two conditions one by one. 

\noindent\textbf{Condition \ref{conditionF}(ii).}
% We first show that $n\hat{ {V}}_{\text{haj,\textsc{f}}} - \hat{\Omega}_{\text{haj,\textsc{f}}}(\hat{\gamma}_\textsc{f}) = o_\mathbb{P}(1)$. 
Define
\[
G_1 = n^{-1}  {Z}^\top  {W} 
{e}_{\textup{haj},\textsc{f}}  {K}_n 
{e}_{\textup{haj},\textsc{f}}   {W}  {X}
\text{ and }
G_2 = n^{-1} {X}^\top  {W} 
{e}_{\textup{haj},\textsc{f}} 
{K}_n 
{e}_{\textup{haj},\textsc{f}}  {W}  {X}.
\] 
% $G_1 = n^{-1}  {Z}^\top  {W} 
%  {e}_{\textup{haj},\textsc{f}}  {K}_n 
%  {e}_{\textup{haj},\textsc{f}}   {W}  {X}$ 
%  and $G_2 = n^{-1} {X}^\top  {W} 
%  {e}_{\textup{haj},\textsc{f}} 
%  {K}_n 
%  {e}_{\textup{haj},\textsc{f}}  {W}  {X}$.
The ``middle'' part of \eqref{eq:V_F} equals 
\begin{align*}
n^{-1} {C}_\textsc{f}^{\top}  {W}  {e}_{\textup{haj},\textsc{f}}
{K}_n 
{e}_{\textup{haj},\textsc{f}}
{W}  {C}_\textsc{f} 
=&~ n^{-1}  
( {Z},  {X})^\top
{W}  {e}_{\textup{haj},\textsc{f}}
{K}_n  {e}_{\textup{haj},\textsc{f}}
{W} ( {Z},  {X}) \\
% =& n^{-1}
% \begin{pmatrix}
%  {Z}^\top  {W} 
%  {e}_{\textup{haj},\textsc{f}}  {K}_n 
%  {e}_{\textup{haj},\textsc{f}}   {W}  {Z} &  {Z}^\top  {W}  {e}_{\textup{haj},\textsc{f}} 
%  {K}_n 
%  {e}_{\textup{haj},\textsc{f}}   {W}  {X} \\
%  {X}^\top  {W}  {e}_{\textup{haj},\textsc{f}}  {K}_n  {e}_{\textup{haj},\textsc{f}}   {W}  {Z} &  {X}^\top  {W}  {e}_{\textup{haj},\textsc{f}}  {K}_n  {e}_{\textup{haj},\textsc{f}}   {W}  {X} 
% \end{pmatrix} \\
=& \begin{pmatrix}
(n^{-1} {Z}^{\top}  {W}  {Z}) \hat{\Omega}_{\textup{haj},\textsc{f}} (n^{-1} {Z}^{\top}  {W}  {Z}) & G_1 \\
G_1^\top & G_2
\end{pmatrix}.
\end{align*}
The ``bread'' part of \eqref{eq:V_F} equals
\begin{align*}
n^{-1}
{C}_\textsc{f}^{\top}  {W}  {C}_\textsc{f}
% =& 
% n^{-1}
% \begin{pmatrix}
%  {Z}^\top \\
%  {X}^\top
% \end{pmatrix}
%  {W} ( {Z},  {X})
= n^{-1}
\begin{pmatrix}
{Z}^\top  {W}  {Z} &  {Z}^\top  {W}  {X} \\
{X}^\top  {W}  {Z} &  {X}^\top  {W}  {X}
\end{pmatrix}
= \text{diag}\left(I, \frac{|\mathcal{T}|}{n}\sum_{i=1}^n x_ix_i^\top\right) + o_\mathbb{P}(1)
\end{align*}
where the last quality follows from $ n^{-1} {Z}^\top  {W}  {Z} = \hat{1}_{\text{ht}} = I + o_\mathbb{P}(1)$ by Lemma \ref{lemma: beta_haj} and 
\begin{align*}
n^{-1} {Z}^\top  {W} X 
=& n^{-1} \sum_{i=1}^n z_iw_i x_i^\top
= \hat{x}_{\text{ht}}
= o_\mathbb{P}(1) , \\
n^{-1} X^\top {W} X
=& \left( n^{-1}\sum_{i=1}^n \frac{{1}_i(t)}{\pi_i(t)} x_ix_i^\top: t\in\mathcal{T} \right) 
= Q_{xx} + o_\mathbb{P}(1)
\end{align*}
by Lemma \ref{lemma: x_ht}. 
It suffices to show that $G_k = O_\mathbb{P}(1)$ for $k=1,2$. We omit the proof here as it is similar to the proof of Theorem \ref{thm:full_bias}.
\noindent \textbf{Condition \ref{conditionF}(i).}
Recall $\hat{\beta}_{\textup{haj}}(t;\gamma_{\textsc{f}})$ defined in \eqref{eq: beta_gamma} with $\gamma = \gamma_{\textsc{f}}$.
% Define $\hat{\beta}_{\textup{haj},\gamma} =(\hat{\beta}_{\textup{haj}}(t;\gamma):{t\in\mathcal{T}})$ with
% \[
% \hat{\beta}_{\textup{haj}}(t;\gamma) = \frac{1}{\hat{1}_{\text{ht}}(t)} \frac{1}{n}\sum_{i=1}^n 
% \frac{{1}_i(t)}{\pi_i(t)} (Y_i - x_i^\top {\gamma}(t) ) 
% = \hat{\beta}_{\textup{haj},\textsc{f}}(t) - \hat{x}_\textup{haj}(t)({\gamma}(t) - \hat{\gamma}(t)). 
% \]
Let $\hat{\Omega}_{\textup{haj},\textsc{f}}( {\gamma}_{\textsc{f}}) = (\hat{\Omega}_{\textup{haj},\textsc{f}}(t,t'; {\gamma}_{\textsc{f}}) )_{t,t'\in\mathcal{T}}$ where the $(t,t')$th element is
\begin{align*}
&\hat{\Omega}_{\textup{haj},\textsc{f}}(t,t'; {\gamma}_{\textsc{f}}) 
= \frac{1}{n} \sum_{i=1}^n \sum_{j=1}^n 
\frac{{1}_i(t) ( Y_i(t;{\gamma}_\textsc{f}) - \hat{\beta}_\textup{haj}(t; \gamma_\textsc{f}) )}{\pi_i(t) \hat{1}_{\text{ht}}(t)}   
\frac{{1}_j(t') (Y_j(t;{\gamma}_\textsc{f})-\hat{\beta}_\textup{haj}(t'; \gamma_\textsc{f})) }{\pi_j(t') \hat{1}_{\text{ht}}(t')} 
K_{n}(i,j).
% =& \frac{1}{n} \sum_{\ell_{ {A}}(i, j) \leq b_n}
% \frac{\frac{{1}_i(t)}{\pi_i(t)} (Y_i - x_i^\top {\gamma}_\textsc{f}   - \hat{\beta}_\textup{haj}(t; \gamma_\textsc{f})) }{\hat{1}_{\text{ht}}(t)}   
% \frac{\frac{{1}_j(t')}{\pi_j(t')} (Y_j - x_j^\top {\gamma}_\textsc{f}  - \hat{\beta}_\textup{haj}(t'; \gamma_\textsc{f})) }{\hat{1}_{\text{ht}}(t')}. 
\end{align*}
Applying Theorem \ref{thm:Hájek_bias} to adjusted outcome $Y_i(t;{\gamma}_\textsc{f})$, we have 
\begin{equation}
\hat{\Omega}_{\textup{haj},\textsc{f}}( {\gamma}_{\textsc{f}})
= { {\Sigma}}_{*,\textup{haj},\textsc{f}} + R_{\textup{haj},\textsc{f}} 
+ o_\mathbb{P}(1).  \label{OmegaF}
\end{equation}
% with bias 
% \begin{align*}
% & R_\textup{haj}(t,t') 
% = \frac{1}{n} \sum_{i=1}^n \sum_{j=1}^n
% (\tau_i(t,t') - \tau(t,t'))
% (\tau_j(t,t') - \tau(t,t')) 
% {1}(\ell_{ {A}}(i, j) \leq b_n)
% \end{align*}
% Define $\hat{Y}_\textup{haj}(t,{\gamma}_\textsc{f})=\frac{1}{n}\sum_{i=1}^n \frac{\frac{{1}_i(t)}{\pi_i(t)}}{\hat{1}_{\text{ht}}(t)}({Y}_i - x_i^\top {\gamma}_\textsc{f} ) $.
To simplify the notation without loss of generality, we verify that $ \hat{\Omega}_{\textup{haj},\textsc{f}}(t,t'; \hat{\gamma}_{\textsc{f}}) - \hat{\Omega}_{\textup{haj},\textsc{f}}(t,t'; {\gamma}_{\textsc{f}}) = o_\mathbb{P}(1)$ with scalar covariate $x_i$. By direct algebra,  
% With 
% \begin{align*}
%  & \frac{\frac{{1}_i(t)}{\pi_i(t)} (Y_i - \hatx_i^\top {\gamma}_\textsc{f}   - \hat{\beta}_{\textup{haj}, \textsc{f}}(t)) }{\hat{1}_{\text{ht}}(t)}   
% - \frac{\frac{{1}_i(t')}{\pi_i(t')} (Y_i - \hatx_i^\top {\gamma}_\textsc{f}   - \hat{\beta}_{\textup{haj}, \textsc{f}}(t')) }{\hat{1}_{\text{ht}}(t')}  \\
% =& \frac{\frac{{1}_i(t)}{\pi_i(t)} (Y_i - x_i^\top {\gamma}_\textsc{f}   - \hat{\beta}_\textup{haj}(t, \gamma_\textsc{f}) ) }{\hat{1}_{\text{ht}}(t)}  
% - \frac{\frac{{1}_i(t')}{\pi_i(t')} (Y_i - x_i^\top {\gamma}_\textsc{f}   - \hat{\beta}_\textup{haj}(t', \gamma_\textsc{f}) ) }{\hat{1}_{\text{ht}}(t')} \\
% &+ ({\gamma}_\textsc{f} - \hat{\gamma}_\textsc{f})^\top \left(\frac{\frac{{1}_i(t)}{\pi_i(t)}   (x_i - \hat{x}_\textup{haj}(t)) }{\hat{1}_{\text{ht}}(t)} - \frac{\frac{{1}_i(t')}{\pi_i(t')} (x_i - \hat{x}_\textup{haj}(t')) }{\hat{1}_{\text{ht}}(t')}\right)
% \end{align*}
% we have
\begin{align*}
& \hat{\Omega}_{\textup{haj},\textsc{f}}(t,t'; \hat{\gamma}_{\textsc{f}}) - \hat{\Omega}_{\textup{haj},\textsc{f}}(t,t'; {\gamma}_{\textsc{f}}) \\
=& \frac{({\gamma}_\textsc{f} - \hat{\gamma}_\textsc{f})^2}{\hat{1}_{\text{ht}}(t)\hat{1}_{\text{ht}}(t')} 
\frac{1}{n} \sum_{i=1}^n \sum_{j=1}^n
% \left[
% \begin{array}{c}
\frac{{1}_i(t)(x_i - \hat{x}_\textup{haj}(t))}{\pi_i(t)} 
\frac{{1}_j(t')(x_j - \hat{x}_\textup{haj}(t'))}{\pi_j(t')} 
% \end{array} 
K_{n}(i,j)\\
+& \frac{({\gamma}_\textsc{f} - \hat{\gamma}_\textsc{f})}{\hat{1}_{\text{ht}}(t)\hat{1}_{\text{ht}}(t')}
\frac{1}{n} \sum_{i=1}^n \sum_{j=1}^n
% \begin{array}{c}
% \left( 
\frac{{1}_i(t)(x_i - \hat{x}_\textup{haj}(t))}{\pi_i(t)} 
\frac{{1}_j(t')(Y_j - x_j {\gamma}_\textsc{f}   - \hat{\beta}_{\textup{haj}}(t'; {\gamma}_\textsc{f})) }{\pi_j(t')} 
% + \frac{{1}_i(t')(x_i - \hat{x}_\textup{haj}(t'))}{\pi_i(t')} 
% \frac{{1}_j(t)(Y_j - x_j {\gamma}_\textsc{f}  - \hat{\beta}_{\textup{haj}}(t; {\gamma}_\textsc{f}))}{\pi_j(t)} 
% \right)
% \end{array} 
K_{n}(i,j) \\
+& \frac{({\gamma}_\textsc{f} - \hat{\gamma}_\textsc{f})}{\hat{1}_{\text{ht}}(t)\hat{1}_{\text{ht}}(t')}
\frac{1}{n} \sum_{i=1}^n \sum_{j=1}^n
% \begin{array}{c}
% \left( 
\frac{{1}_i(t')(x_i - \hat{x}_\textup{haj}(t'))}{\pi_i(t')} 
\frac{{1}_j(t)(Y_j - x_j {\gamma}_\textsc{f}  - \hat{\beta}_{\textup{haj}}(t; {\gamma}_\textsc{f}))}{\pi_j(t)} 
% \right)
% \end{array} 
K_{n}(i,j).
% =& ({\gamma}_\textsc{f} - \hat{\gamma}_\textsc{f})^2
% \frac{1}{n} \sum_{\ell_{ {A}}(i, j) \leq b_n}
% \begin{array}{c}
% \left( \frac{{1}_i(t)(x_i - \hat{x}_\textup{haj}(t))}{\pi_i(t)}  
% \right) 
% \left( \frac{{1}_j(t')(x_j - \hat{x}_\textup{haj}(t'))}{\pi_j(t')}  \right) 
% \end{array} \\
% &+ ({\gamma}_\textsc{f} - \hat{\gamma}_\textsc{f}) 
% \frac{2}{n} \sum_{\ell_{ {A}}(i, j) \leq b_n} 
% \begin{array}{l}
% \left( \frac{{1}_i(t)(x_i - \hat{x}_\textup{haj}(t))}{\pi_i(t)}  
% \right)
% \left( \frac{{1}_j(t')(Y_j - x_j {\gamma}_\textsc{f}   - \hat{\beta}_{\textup{haj}}(t'; {\gamma}_\textsc{f}))}{\pi_j(t')} \right)
% \end{array}
% + o_\mathbb{P}(1) \\
% =& ({\gamma}_\textsc{f} - \hat{\gamma}_\textsc{f})^2
% \frac{1}{n} \sum_{\ell_{ {A}}(i, j) \leq b_n}
% \frac{{1}_i(t)}{\pi_i(t)}
% \frac{{1}_j(t')}{\pi_j(t')}
% \left(
%   x_i x_j 
% + \hat{x}_\textup{haj}(t) \hat{x}_\textup{haj}(t') 
% -2  x_i\hat{x}_\textup{haj}(t') 
% \right)\\
% +& ({\gamma}_\textsc{f} - \hat{\gamma}_\textsc{f}) 
% \frac{2}{n} \sum_{\ell_{ {A}}(i, j) \leq b_n} 
% \frac{{1}_i(t)}{\pi_i(t)} 
% \frac{{1}_j(t')}{\pi_j(t')} 
% \left[
% \begin{array}{c}
%  x_i
% (Y_j - x_j {\gamma}_\textsc{f} - \mu(t'))  
% +  \hat{x}_\textup{haj}(t) 
% (\hat{\beta}_{\textup{haj}}(t'; {\gamma}_\textsc{f}) - \mu(t'))  \\
% -  x_i
% (\hat{\beta}_{\textup{haj}}(t'; {\gamma}_\textsc{f}) - \mu(t')) 
% - \hat{x}_\textup{haj}(t)
% (Y_j - x_j {\gamma}_\textsc{f} - \mu(t'))
% \end{array}
% \right] 
% + o_\mathbb{P}(1) \\
% =& o_\mathbb{P}(1).
\end{align*}
Under Assumptions \ref{asu2}, \ref{asu3} and \ref{asu9}, ${1}_i(t)\pi_i(t)^{-1}$, $Y_i$ and $x_i$ are uniformly bounded. Then, for some $C>0$ and any $n$, we have
\begin{align*}
\left| \hat{\Omega}_{\textup{haj},\textsc{f}}(t,t'; \hat{\gamma}_{\textsc{f}}) - \hat{\Omega}_{\textup{haj},\textsc{f}}(t,t'; {\gamma}_{\textsc{f}}) \right| 
\le & 
C |{\gamma}_\textsc{f} - \hat{\gamma}_\textsc{f}| 
\frac{1}{n} \sum_{i=1}^n \sum_{j=1}^n
K_{n}(i,j) \\
=& o_\mathbb{P}(1),
\end{align*}
where the last line holds by $\hat{\gamma}_\textsc{f} - {\gamma}_\textsc{f} = O_{\mathbb{P}}(n^{-1/2})$ and Assumption \ref{asu7}(b). 
This, together with \eqref{OmegaF}, ensures Condition \ref{conditionF}(i).
\end{proof}

\begin{proof}[Proof of Theorem \ref{thm:add_adj}] 
We omit the proof as it is analogous to the proof of Theorem \ref{thm:full_adj}.
\end{proof}

\subsubsection{Fully-interacted regression}

We verify in this subsection the results under the fully-interacted WLS fit in \eqref{eq:full}. 
% The proofs of the results under the additive WLS fit in \eqref{eq:add} are similar, so we relegate the proofs to online appendix.
The correspondence between the WLS fit and the Hájek estimation is also preserved in the fully-interacted WLS fit in \eqref{eq:full}. Proposition \ref{prop:full} parallels Proposition \ref{prop:add}, and establishes that $\hat{\beta}_{\textup{haj},\textsc{l}}(t)$ is the Hájek estimator based on the covariate-adjusted outcome $Y_i - x_i^\top \hat{\gamma}_\textsc{l}(t) $. 
A key distinction is that the adjustment is now based on coefficients specific to exposure mapping values.
\begin{proposition} \label{prop:full}
$\hat{\beta}_{\textup{haj},\textsc{l}}(t) = \hat{Y}_{\textup{haj}}(t) - \hat{x}_{\textup{haj}}(t)^{\top} \hat{\gamma}_\textsc{l}(t)$ for all $t\in\mathcal{T}$.
\end{proposition}
% Recall Proposition \ref{prop:add} under the additive WLS fit implies that $\hat{\beta}_{\textup{haj},\textsc{f}}(t) = \hat{Y}_{\textup{haj}}(t) - \hat{x}_{\textup{haj}}(t)^{\top} \hat{\gamma}_\textsc{f}$. 
\begin{proof}[Proof of Proposition \ref{prop:full}]
The numerical equivalence follows from equation \eqref{eq: L}:
\begin{equation*}
G_{1t} \left(\hat{\beta}_{\textup{haj},\textsc{l}}(t), \hat{\gamma}_\textsc{l}(t)^\top \right)^\top
% \begin{pmatrix}
%     \hat{\beta}_{\textup{haj},\textsc{l}}(t) \\
%     \hat{\gamma}_\textsc{l}(t)
% \end{pmatrix}
=  G_{2t}, 
\end{equation*}
where 
\begin{align*}
G_{1t} 
% =& \left( \sum_{i=1}^n {1}_i(t) \right)^{-1} 
% (\iota_{n_t}, {X}_t)^{\top} W_t
% (\iota_{n_t}, {X}_t)
=& \left( \sum_{i=1}^n {1}_i(t) \right)^{-1}
\begin{pmatrix}
\iota_{n_t}^{\top}W_t \iota_{n_t} & \iota_{n_t}^{\top}W_t {X}_t \\
{X}_t^{\top}W_t \iota_{n_t}  & {X}_t^{\top}W_t {X}_t
\end{pmatrix}
= 
\begin{pmatrix}
\hat{1}_{\text{ht}}(t) & \hat{x}_{\text{ht}}^{\top}(t) \\
\hat{x}_{\text{ht}}(t)  & n_t^{-1}{X}_t^{\top}W_t {X}_t
\end{pmatrix}, \\
G_{2t} 
=& \left( \sum_{i=1}^n {1}_i(t) \right)^{-1}
\begin{pmatrix}
\iota_{n_t}^{\top} W_t {Y}_t \\
{X}_t^{\top} W_t {Y}_t 
\end{pmatrix}
= 
\begin{pmatrix}
\hat{Y}_{\text{ht}}(t) \\
n_t^{-1}{X}_t^{\top} W_t {Y}_t
\end{pmatrix}.  
\end{align*}
\end{proof}

% Some useful facts: 
% \begin{align*}
% n^{-1}  {Z}^\top  {W}  {Z}
% =&~ \Lambda_{\text{ht}} \\
% n^{-1}  {Z}^\top  {W} \chi
% =&  \\
% n^{-1}  \chi^\top  {W} \chi
% =&  
% \end{align*}

\begin{lemma}\label{lemma: gamma_L}
Under Assumptions
\ref{asu1}--\ref{asu4}, \ref{asu7}(a) and \ref{asu9}, we have $\hat{\gamma}_\textsc{l}(t) = {\gamma}_\textsc{l}(t) + O_{\mathbb{P}}(n^{-1/2})$ with
${\gamma}_\textsc{l}(t) = 
(\sum_{i=1}^n x_ix_i^\top)^{-1} \sum_{i=1}^n x_i \mu_i(t)$ for all $t\in\mathcal{T}$.
% \[
% {\gamma}_\textsc{l} = 
% (\sum_{i=1}^n x_ix_i^\top)^{-1} \sum_{i=1}^n x_i \mu_i(t).
% \]
\end{lemma}
% \begin{proof}
% The probability limit follows from $(\hat{\beta}_{\textup{haj},\textsc{l}}(t), \hat{\gamma}_\textsc{l}(t)^{\top})^{\top} = G_{1t}^{-1}G_{2t}$ where
% \begin{align*}
% G_{1t} = \text{diag}\left(1, n^{-1}\sum_{i=1}^n x_ix_i^\top\right) + o_\mathbb{P}(1), 
% \quad
% G_{2t} = \left(\mu(t), ( n^{-1} \sum_{i=1}^n x_i \mu_i(t))^\top \right)^\top + o_\mathbb{P}(1)
% \end{align*}
% by Lemma \ref{lemma: x_ht}. 
% \end{proof}

\begin{proof}[Proof of Lemma \ref{lemma: gamma_L}]
% We verify below the numerical results in
% Proposition \ref{prop:full} and the asymptotic results in Lemma \ref{lemma: gamma_L} together.
The inclusion of full interactions ensures that $\hat{\beta}_{\textup{haj},\textsc{l}}(t)$ and $\hat{\gamma}_\textsc{l}(t)$ from the WLS fit in \eqref{eq:full} equal the coefficients of 1 and ${x}_i$ from the following WLS fit:
\begin{equation}
\text{regress }
{Y}_i \sim 1 + {x}_i \text{ with weights }  w_{it}=1/\pi_i(t) \text{ for units that }  {1}_i(t)=1. \label{eq:fullt}
\end{equation}
% with weights $w_{it}=\frac{{1}_i(t)}{\pi_i(t)}$ for units that ${1}_i(t)=1$, respectively. 
Let $W_t=\text{diag}\{{w}_{it}\}_{\{i: {1}_i(t)=1\}}$ and ${n_t} = \sum_{i=1}^n {1}_i(t)$. 
% Let $1_{n_t}$ be an $n_t \times 1$ vector of ones. 
Let ${Y}_t$ and ${X}_t$ be the concatenations of ${Y}_i$ and $x_i$ over $\{i: {1}_i(t)=1\}$, respectively. The first-order condition of WLS fit in \eqref{eq:fullt} ensures 
\begin{equation}
G_{1t} \left(\hat{\beta}_{\textup{haj},\textsc{l}}(t), \hat{\gamma}_\textsc{l}(t)^\top \right)^\top
% \begin{pmatrix}
%     \hat{\beta}_{\textup{haj},\textsc{l}}(t) \\
%     \hat{\gamma}_\textsc{l}(t)
% \end{pmatrix}
=  G_{2t}, \label{eq: L}   
\end{equation}
where 
\begin{align*}
G_{1t} 
% =& \left( \sum_{i=1}^n {1}_i(t) \right)^{-1} 
% (\iota_{n_t}, {X}_t)^{\top} W_t
% (\iota_{n_t}, {X}_t)
=& \left( \sum_{i=1}^n {1}_i(t) \right)^{-1}
\begin{pmatrix}
\iota_{n_t}^{\top}W_t \iota_{n_t} & \iota_{n_t}^{\top}W_t {X}_t \\
{X}_t^{\top}W_t \iota_{n_t}  & {X}_t^{\top}W_t {X}_t
\end{pmatrix}
= 
\begin{pmatrix}
\hat{1}_{\text{ht}}(t) & \hat{x}_{\text{ht}}^{\top}(t) \\
\hat{x}_{\text{ht}}(t)  & n_t^{-1}{X}_t^{\top}W_t {X}_t
\end{pmatrix}, \\
G_{2t} 
=& \left( \sum_{i=1}^n {1}_i(t) \right)^{-1}
\begin{pmatrix}
\iota_{n_t}^{\top} W_t {Y}_t \\
{X}_t^{\top} W_t {Y}_t 
\end{pmatrix}
= 
\begin{pmatrix}
\hat{Y}_{\text{ht}}(t) \\
n_t^{-1}{X}_t^{\top} W_t {Y}_t
\end{pmatrix}.  
\end{align*}
By comparing the first row of \eqref{eq: L}, we have
$\hat{\beta}_{\textup{haj},\textsc{l}}(t)  = \hat{Y}_{\textup{haj}}(t) - \hat{x}_{\textup{haj}}(t)^{\top} \hat{\gamma}_\textsc{l}(t)$.
% \[
% % \hat{1}_{\text{ht}}(t) \hat{\beta}_{\textup{haj},\textsc{l}}(t) + \hat{x}_{\text{ht}}(t) \hat{\gamma}_\textsc{l}(t) = \hat{Y}_{\text{ht}}(t) \implies 
% \hat{\beta}_{\textup{haj},\textsc{l}}(t) 
% = \hat{1}_{\text{ht}}(t)^{-1}(\hat{Y}_{\text{ht}}(t) - \hat{x}_{\text{ht}}(t)^{\top} \hat{\gamma}_\textsc{l}(t))
% = \hat{Y}_{\textup{haj}}(t) - \hat{x}_{\textup{haj}}(t)^{\top} \hat{\gamma}_\textsc{l}(t).
% \]
The probability limit follows from \eqref{eq: L} and by Lemma \ref{lemma: x_ht}
\begin{align*}
G_{1t} = \text{diag}\left(1, n^{-1}\sum_{i=1}^n x_ix_i^\top\right) + O_{\mathbb{P}}(n^{-\frac{1}{2}}), 
\quad 
G_{2t} = \left(\mu(t), ( n^{-1} \sum_{i=1}^n x_i \mu_i(t))^\top \right)^\top + O_{\mathbb{P}}(n^{-\frac{1}{2}}).
\end{align*}
\end{proof}
% Denote by $\hat{\Sigma}^2_{\textup{haj},\textsc{f}}(\hat{\gamma}_\textsc{f})$ the network-robust covariance for $\hat{\tau}_\textup{haj}\left(t, t^{\prime},\hat{\gamma}_\textsc{f}\right)$ from the weighted least squares fit. 
% Theorem 2.1 below established the conservativeness of $\hat{\Sigma}^2_{\textup{haj},\textsc{f}}(\hat{\gamma}_\textsc{f})$ for the asymptotic covariance of $\hat{\tau}_\textup{haj}\left(t, t^{\prime},\hat{\gamma}_\textsc{f}\right)$.

\begin{proof}[Proof of Theorem \ref{thm:full}]
% \textbf{Asymptotic normality.}
The result follows from $\hat{\gamma}_\textsc{l} = {\gamma}_\textsc{l} + O_{\mathbb{P}}(n^{-1/2})$ by Lemma \ref{lemma: gamma_L} and applying Lemma \ref{lemma: asym} with $\hat{\gamma}=\hat{\gamma}_\textsc{l}$ and ${\gamma} = {\gamma}_\textsc{l}$.
\end{proof}

\begin{proof}[Proof of Theorem \ref{thm:full_bias}]
We first show the result of the oracle estimator.
Applying Theorem \ref{thm:Hájek_asym_n} to the covariate-adjusted outcome $Y_i(t;{\gamma}_\textsc{l}) = {Y}_i - {x}_i^\top {\gamma}_\textsc{l}(t)$, we have
\[
{ {\Sigma}}_{*,\textup{haj},\textsc{l}} 
= \operatorname{Var}\left(\sqrt{n} \hat{\beta}_{\textup{haj}}({\gamma_\textsc{l}}) \right)
+ o_\mathbb{P}(1).
\]
Since $\sqrt{n}( \hat{\beta}_{\textup{haj}}({\gamma_\textsc{l}}) - \hat{\beta}_{\textup{haj},\textsc{l}} ) = o_\mathbb{P}(1)$ by Lemma \ref{lemma: gamma_L}, we can show that
\[
{ {\Sigma}}_{*,\textup{haj},\textsc{l}} =
\operatorname{Var}\left(\sqrt{n} \hat{\beta}_{\textup{haj},\textsc{l}} \right)
+ o_\mathbb{P}(1)
= {\Sigma}_{\textup{haj}, \textsc{l}} 
+ o_\mathbb{P}(1).
\]

% By Lemma \ref{lemma: gamma_F}, we have $\hat{\gamma}_\textsc{l} = {\gamma}_\textsc{l} + o_\mathbb{P}(1)$. Then the asymptotic normality of $\hat{\tau}_\textsc{l}\left(t, t^{\prime},\hat{\gamma}_\textsc{l} \right)$ follows from Lemma \ref{lemma: asym}. By Proposition 4.1 of \cite{KOJEVNIKOV2021882}, 
% \[
% \left| \hat{\Sigma}_{*, {\textup{haj},\textsc{l}}}^2\left(t, t^{\prime}, \gamma_{\mathrm{L}}\right) -\operatorname{Var}\left(\sqrt{n} \hat{\tau}_\textsc{l}\left(t, t^{\prime},\hat{\gamma}_\textsc{l} \right) \right)\right| \stackrel{p}{\rightarrow} 0
% \]
% \end{proof}
% \begin{proof}[Proof of Theorem \ref{prop:full}]
% \textbf{Formula of $\hat{ {V}}_{\text{haj,L}}$.}
Then we show the result of the asymptotic bias.
By Proposition \ref{prop:full}, $\hat{\beta}_{\textup{haj},\textsc{l}}$ is the Hájek estimator based on the covariate-adjusted outcome $Y_i(t;\hat{\gamma}_\textsc{l}) = {Y}_i - {x}_i^\top \hat{\gamma}_\textsc{l}(t) $. 
The residual from the WLS fit in \eqref{eq:full} is ${e}_{\textsc{l},i}
= Y_i - {x}_i^\top \hat{\gamma}_\textsc{l}(T_i)
- \hat{\beta}_{\textup{haj},\textsc{l}}(T_i)$.
% \[
% {e}_{\textsc{l},i}
% = Y_i - {x}_i^\top \hat{\gamma}_\textsc{l}(T_i)
% - \hat{\beta}_{\textup{haj},\textsc{l}}(T_i).
% \] 
The HAC covariance estimator for $\hat{\beta}_{\textup{haj},\textsc{l}}$ equals the upper-left $|\mathcal{T}|\times |\mathcal{T}|$ submatrix of
\begin{equation}
( {C}_\textsc{l}^{\top}  {W}  {C}_\textsc{l})^{-1}
( {C}_\textsc{l}^{\top}  {W}  {e}_{\textup{haj},\textsc{l}}  {K}_n  {e}_{\textup{haj},\textsc{l}}^{\top}   {W}  {C}_\textsc{l})
( {C}_\textsc{l}^{\top}  {W}  {C}_\textsc{l})^{-1}.  
\label{eq:V_hajL}  
\end{equation}
Introduce an intermediate term for the theoretical analysis:  
\begin{align*}
& \hat{\Omega}_{\textup{haj},\textsc{l}}(\hat{\gamma}_\textsc{l})
= n ( {Z}^{\top}  {W}  {Z})^{-1}
( {Z}^{\top}  {W}  {e}_{\textup{haj},\textsc{l}}
{K}_n  {e}_{\textup{haj},\textsc{l}}^{\top}   {W}  {Z})
( {Z}^{\top}  {W}  {Z})^{-1} 
% = ( \hat{\Omega}_{\textup{haj},\textsc{l}}(t,t'; \hat{\gamma}_{\textsc{l}}) )_{t,t'\in\mathcal{T}} 
% =& \left( \frac{ \frac{1}{n} \sum_{i=1}^n \sum_{j=1}^n 
% \frac{{1}_i(t)}{\pi_i(t)}  (Y_i(t;\hat{\gamma}_\textsc{l})- \hat{\beta}_{\textup{haj},\textsc{l}}(t) ) 
% \frac{{1}_j(t')}{\pi_j(t')} (Y_j(t';\hat{\gamma}_\textsc{l})-\hat{\beta}_{\textup{haj},\textsc{l}}(t') ) 
% {1}(\ell_{ {A}}(i,j)\le b_n)  }{ \hat{1}_{\text{ht}}(t) \hat{1}_{\text{ht}}(t') } 
% \right)_{t,t'\in\mathcal{T}}.   
\end{align*}
where the $(t,t')$th entry is
\begin{align*}
& \hat{\Omega}_{\textup{haj},\textsc{l}}(t,t'; \hat{\gamma}_{\textsc{l}})
= 
\frac{1}{n} \sum_{i=1}^n \sum_{j=1}^n 
\frac{{1}_i(t)(Y_i(t;\hat{\gamma}_\textsc{l})- \hat{\beta}_{\textup{haj},\textsc{l}}(t) )}{\pi_i(t) \hat{1}_{\text{ht}}(t)}   
\frac{{1}_j(t') (Y_j(t';\hat{\gamma}_\textsc{l})-\hat{\beta}_{\textup{haj},\textsc{l}}(t') )  }{\pi_j(t') \hat{1}_{\text{ht}}(t')}
K_{n}(i,j).   
\end{align*}
The result on $\hat{ {V}}_{\text{haj,\textsc{l}}}$ holds as long as 
\begin{equation}
\text{(i) } \hat{\Omega}_{\text{haj,\textsc{l}}}(\hat{\gamma}_\textsc{l})
= \hat{ {\Sigma}}_{*,\textup{haj}, \textsc{l}} + R_{\textup{haj},\textsc{l}} + o_\mathbb{P}(1)
\text{ and (ii) }
n\hat{ {V}}_{\text{haj,\textsc{l}}} - \hat{\Omega}_{\text{haj,\textsc{l}}}(\hat{\gamma}_\textsc{l}) = o_\mathbb{P}(1). \label{conditionL}
\end{equation}
We verify below these two conditions one by one.
% We first show that $n\hat{ {V}}_{\text{haj,\textsc{l}}} - \hat{\Omega}_{\text{haj,\textsc{l}}}(\hat{\gamma}_\textsc{l}) = o_\mathbb{P}(1)$. 

\noindent \textbf{Condition \ref{conditionL}(ii).}
Define $ \chi = \{z_i \otimes x_i\}_{i=1}^n$. 
Let $G_1 = n^{-1}  {Z}^\top  {W}  {e}_{\textup{haj},\textsc{l}}  {K}_n  {e}_{\textup{haj},\textsc{l}}^{\top}   {W} \chi$ and $G_2 = n^{-1} \chi^\top  {W}  {e}_{\textup{haj},\textsc{l}}  {K}_n  {e}_{\textup{haj},\textsc{l}}^{\top}   {W} \chi$.
The ``middle'' part of \eqref{eq:V_hajL} equals 
\begin{align}
n^{-1}
{C}_\textsc{l}^{\top}  {W}  {e}_{\textup{haj},\textsc{l}}  {K}_n  {e}_{\textup{haj},\textsc{l}}^{\top}   {W}  {C}_\textsc{l}
% =&~ n^{-1}
% ( {Z}, \chi)^\top
%     {W}  {e}_{\textup{haj},\textsc{l}}  {K}_n  {e}_{\textup{haj},\textsc{l}}^{\top}   {W} ( {Z}, \chi) \notag \\
% =& n^{-1} \begin{pmatrix}
%  {Z}^\top  {W}  {e}_{\textup{haj},\textsc{l}}  {K}_n  {e}_{\textup{haj},\textsc{l}}^{\top}   {W}  {Z} &  {Z}^\top  {W}  {e}_{\textup{haj},\textsc{l}}  {K}_n  {e}_{\textup{haj},\textsc{l}}^{\top}   {W} {\chi} \\
% \chi^\top  {W}  {e}_{\textup{haj},\textsc{l}}  {K}_n  {e}_{\textup{haj},\textsc{l}}^{\top}   {W}  {Z} & 
% \chi^\top  {W}  {e}_{\textup{haj},\textsc{l}}  {K}_n  {e}_{\textup{haj},\textsc{l}}^{\top}   {W} \chi 
% \end{pmatrix} \\
=& \begin{pmatrix}
(n^{-1} {Z}^{\top}  {W}  {Z})  
\hat{\Omega}_{\textup{haj},\textsc{l}}(\hat{\gamma}_\textsc{l}) (n^{-1} {Z}^{\top}  {W}  {Z}) ~&~ G_1 \\
G_1^\top & G_2
\end{pmatrix}. \label{eq:meat}
\end{align}
The ``bread'' part of \eqref{eq:V_hajL} equals 
\begin{align*}
n^{-1}
{C}_\textsc{l}^{\top}  {W}  {C}_\textsc{l}
% =& n^{-1}
% ( {Z}, \chi)^\top
%     {W} ( {Z}, \chi)
= n^{-1}
\begin{pmatrix}
{Z}^\top  {W}  {Z} &  {Z}^\top  {W} \chi \\
\chi^\top  {W}  {Z} & \chi^\top  {W} \chi
\end{pmatrix}
= \text{diag}\left(I, I \otimes Q_{xx} \right) + o_\mathbb{P}(1)
\end{align*}
where the last equality follows from $ n^{-1} {Z}^\top  {W}  {Z} = \hat{1}_{\text{ht}} = I + o_\mathbb{P}(1)$ by Lemma \ref{lemma: beta_haj} and 
\begin{align*}
n^{-1} {Z}^\top  {W} \chi 
=& n^{-1} \sum_{i=1}^n z_iw_i (z_i \otimes x_i)^\top
= \text{diag}(\hat{x}_{\text{ht}}(t): t\in\mathcal{T})
= o_\mathbb{P}(1) , \\
n^{-1} \chi^\top {W} \chi
=& \text{diag}\left( n^{-1}\sum_{i=1}^n \frac{{1}_i(t)}{\pi_i(t)} x_ix_i^\top: t\in\mathcal{T} \right) 
= I \otimes Q_{xx} + o_\mathbb{P}(1),
\end{align*}
by Lemma \ref{lemma: x_ht}. 
This, together with \eqref{eq:meat}, ensures that $n\hat{{V}}_{\text{haj,\textsc{l}}} - \hat{\Omega}_{\text{haj,\textsc{l}}}(\hat{\gamma}_\textsc{l}) = o_\mathbb{P}(1)$ holds as long as $G_k = (G_{k}(t,t'))_{t,t'\in\mathcal{T}} = O_\mathbb{P}(1) $ for $k=1,2$.
We verify below $G_{2}(t,t')=O_\mathbb{P}(1)$ for scalar covariate $x\in\mathbb{R}$ for notational simplicity. The
proof for $G_1 = O_\mathbb{P}(1)$ is almost identical to $G_2$ and thus omitted. 
Define $\Delta(t;\hat{\beta}_{\textup{haj},\textsc{l}}) = \hat{\beta}_{\textup{haj},\textsc{l}}(t)-\mu(t)$ and $\Delta(t;\hat{\gamma}_\textsc{l}) = \hat{\gamma}_\textsc{l}(t) - {\gamma}_\textsc{l}(t)$. 
% Define $\Delta_i(t;\hat{\beta}_{\textup{haj},\textsc{l}}, \hat{\gamma}_\textsc{l}) = \hat{\beta}_{\textup{haj},\textsc{l}}(t)-\mu(t) - x_i(\hat{\gamma}_\textsc{l}(t) - {\gamma}_\textsc{l}(t))$.
By direct algebra, we have
\begin{align*}
& G_{2}(t,t')
= \frac{1}{n} \sum_{i=1}^n \sum_{j=1}^n 
% \begin{array}{c}
\frac{{1}_i(t)x_i
(Y_i(t;\hat{\gamma}_\textsc{l})- \hat{\beta}_{\textup{haj},\textsc{l}}(t) )}{\pi_i(t)} 
\frac{{1}_j(t')x_j
(Y_j(t';\hat{\gamma}_\textsc{l})- \hat{\beta}_{\textup{haj},\textsc{l}}(t') )}{\pi_j(t')} 
% \end{array}
K_{n}(i,j) \\
&= \frac{1}{n} \sum_{i,j} 
% \begin{array}{c}
\frac{{1}_i(t)x_i
(Y_i(t;{\gamma}_\textsc{l})- \mu(t) )}{\pi_i(t)} 
\frac{{1}_j(t')x_j
(Y_j(t';\hat{\gamma}_\textsc{l})- \mu(t') )}{\pi_j(t')} 
% \end{array}
K_{n,ij}
- T_1(t,t') - T_1(t',t) + T_2
% &- \frac{1}{n} \sum_{i=1}^n \sum_{j=1}^n 
% % \begin{array}{c}
% \frac{{1}_i(t)x_i
% (\Delta(t;\hat{\beta}_{\textup{haj},\textsc{l}}) - x_i\Delta(t;\hat{\gamma}_{\textsc{l}},{\gamma}_{\textsc{l}})) }{\pi_i(t)} 
% \frac{{1}_j(t')x_j
% (Y_j(t';{\gamma}_\textsc{l})- \mu(t') )}{\pi_j(t')} 
% % \end{array}
% {1}(\ell_{ {A}}(i, j) \leq b_n)\\
% &- \frac{1}{n} \sum_{i=1}^n \sum_{j=1}^n 
% % \begin{array}{c}
% \frac{{1}_i(t)x_i
% (Y_i(t;{\gamma}_\textsc{l}) - \mu(t) ) }{\pi_i(t)} 
% \frac{{1}_j(t')x_j (\Delta(t';\hat{\beta}_{\textup{haj},\textsc{l}}) - x_j\Delta(t';\hat{\gamma}_{\textsc{l}},{\gamma}_{\textsc{l}}))
% }{\pi_j(t')} 
% % \end{array}
% {1}(\ell_{ {A}}(i, j) \leq b_n)\\
% &+ \frac{1}{n} \sum_{i=1}^n \sum_{j=1}^n 
% % \begin{array}{c}
% \frac{{1}_i(t)x_i
% (\Delta(t;\hat{\beta}_{\textup{haj},\textsc{l}}) - x_i\Delta(t;\hat{\gamma}_{\textsc{l}},{\gamma}_{\textsc{l}})) }{\pi_i(t)} 
% \frac{{1}_j(t')x_j (\Delta(t';\hat{\beta}_{\textup{haj},\textsc{l}}) - x_j\Delta(t';\hat{\gamma}_{\textsc{l}},{\gamma}_{\textsc{l}}))
% }{\pi_j(t')} 
% % \end{array}
% {1}(\ell_{ {A}}(i, j) \leq b_n). 
\end{align*}
% $\Delta(t;\hat{\beta}_{\textup{haj},\textsc{l}}) - x_i\Delta(t;\hat{\gamma}_{\textsc{l}},{\gamma}_{\textsc{l}})$
where
\begin{eqnarray*}
&&T_{1}(t,t')
= \frac{1}{n} \sum_{i=1}^n \sum_{j=1}^n 
% \begin{array}{c}
\frac{{1}_i(t)x_i
(\Delta(t;\hat{\beta}_{\textup{haj},\textsc{l}}) - x_i\Delta(t;\hat{\gamma}_{\textsc{l}})) }{\pi_i(t)} 
\frac{{1}_j(t')x_j
(Y_j(t';{\gamma}_\textsc{l})- \mu(t') )}{\pi_j(t')} 
% \end{array}
K_{n,ij}, \\ 
&&T_2
= \frac{1}{n} \sum_{i=1}^n \sum_{j=1}^n 
% \begin{array}{c}
\frac{{1}_i(t)x_i
(\Delta(t;\hat{\beta}_{\textup{haj},\textsc{l}}) - x_i\Delta(t;\hat{\gamma}_{\textsc{l}})) {1}_j(t')x_j (\Delta(t';\hat{\beta}_{\textup{haj},\textsc{l}}) - x_j\Delta(t';\hat{\gamma}_{\textsc{l}}))}{\pi_i(t) \pi_j(t')} 
% \end{array}
K_{n,ij}.
\end{eqnarray*}
% \begin{align*}
% T_{1}(t,t')
% =& \frac{1}{n} \sum_{i=1}^n \sum_{j=1}^n 
% % \begin{array}{c}
% \frac{{1}_i(t)x_i
% (\Delta(t;\hat{\beta}_{\textup{haj},\textsc{l}}) - x_i\Delta(t;\hat{\gamma}_{\textsc{l}})) }{\pi_i(t)} 
% \frac{{1}_j(t')x_j
% (Y_j(t';{\gamma}_\textsc{l})- \mu(t') )}{\pi_j(t')} 
% % \end{array}
% K_{n}(i,j), \\
% % T_2
% % =& \frac{1}{n} \sum_{i=1}^n \sum_{j=1}^n 
% % % \begin{array}{c}
% % \frac{{1}_i(t)x_i
% % (Y_i(t;{\gamma}_\textsc{l}) - \mu(t) ) }{\pi_i(t)} 
% % \frac{{1}_j(t')x_j (\Delta(t';\hat{\beta}_{\textup{haj},\textsc{l}}) - x_j\Delta(t';\hat{\gamma}_{\textsc{l}}))
% % }{\pi_j(t')} 
% % % \end{array}
% % K_{n}(i,j), \\
% T_2
% =& \frac{1}{n} \sum_{i=1}^n \sum_{j=1}^n 
% % \begin{array}{c}
% \frac{{1}_i(t)x_i
% (\Delta(t;\hat{\beta}_{\textup{haj},\textsc{l}}) - x_i\Delta(t;\hat{\gamma}_{\textsc{l}})) {1}_j(t')x_j (\Delta(t';\hat{\beta}_{\textup{haj},\textsc{l}}) - x_j\Delta(t';\hat{\gamma}_{\textsc{l}}))}{\pi_i(t) \pi_j(t')} 
% % \end{array}
% K_{n}(i,j).
% \end{align*}
We first show $T_1(t,t') = o_\mathbb{P}(1)$. Since ${1}_i(t)\pi_i(t)^{-1}$, $x_i$ and $Y_i$ are uniformly bounded by Assumptions \ref{asu2}, \ref{asu3} and \ref{asu9}, then for some $C >0$ and any $n$, we have
\begin{align*}
% & \left| \frac{1}{n} \sum_{i=1}^n \sum_{j=1}^n 
% % \begin{array}{c}
% \frac{{1}_i(t)x_i
% (\Delta(t;\hat{\beta}_{\textup{haj},\textsc{l}}) - x_i\Delta(t;\hat{\gamma}_{\textsc{l}}))}{\pi_i(t)} 
% \frac{{1}_j(t')x_j
% (Y_j(t';{\gamma}_\textsc{l})- \mu(t') )}{\pi_j(t')} 
% % \end{array}
% {1}(\ell_{ {A}}(i, j) \leq b_n)\right| \\
|T_1|
\le & C 
\left(
|\hat{\beta}_{\textup{haj},\textsc{l}}(t)-\mu(t)| + |\hat{\gamma}_\textsc{l}(t) - {\gamma}_\textsc{l}(t)|
\right)
\frac{1}{n} \sum_{i=1}^n \sum_{j=1}^n 
K_{n}(i,j) \\
=& o_\mathbb{P}(1),
\end{align*}
where the last line follows by 
$\hat{\beta}_{\textup{haj},\textsc{l}}(t)-\mu(t) = O_{\mathbb{P}}(n^{-1/2})$,
$\hat{\gamma}_\textsc{l}(t) - {\gamma}_\textsc{l}(t) = O_{\mathbb{P}}(n^{-1/2})$, and Assumption \ref{asu7}(b). The result $T_1(t',t) = o_\mathbb{P}(1)$ follows by symmetry. We finally show $T_2 = o_\mathbb{P}(1)$. Again, since ${1}_i(t)\pi_i(t)^{-1}$, $x_i$ and $Y_i$ are uniformly bounded by Assumptions \ref{asu2}, \ref{asu3} and \ref{asu9}, then for some $C >0$ and any $n$, we have
\begin{align*}
% & \left| \frac{1}{n} \sum_{i=1}^n \sum_{j=1}^n 
% % \begin{array}{c}
% \frac{{1}_i(t)x_i
% (\Delta(t;\hat{\beta}_{\textup{haj},\textsc{l}}) - x_i\Delta(t;\hat{\gamma}_{\textsc{l}})) }{\pi_i(t)} 
% \frac{{1}_j(t')x_j (\Delta(t';\hat{\beta}_{\textup{haj},\textsc{l}}) - x_j\Delta(t';\hat{\gamma}_{\textsc{l}}))
% }{\pi_j(t')} 
% % \end{array}
% {1}(\ell_{ {A}}(i, j) \leq b_n)
% \right| \\
|T_3|
\le &  C
\left[
\begin{array}{c}
|\Delta(t;\hat{\beta}_{\textup{haj},\textsc{l}}) 
\Delta(t';\hat{\beta}_{\textup{haj},\textsc{l}})| 
+ |\Delta(t;\hat{\gamma}_{\textsc{l}})
\Delta(t';\hat{\gamma}_{\textsc{l}}) | \\
+ |\Delta(t;\hat{\gamma}_{\textsc{l}})
\Delta(t';\hat{\beta}_{\textup{haj},\textsc{l}})| 
+ | \Delta(t;\hat{\beta}_{\textup{haj},\textsc{l}})
\Delta(t';\hat{\gamma}_{\textsc{l}})|
\end{array}
\right]
\frac{1}{n} \sum_{i=1}^n \sum_{j=1}^n 
K_{n}(i,j) \\
=& o_\mathbb{P}(1), 
\end{align*}
where the last line holds by $\hat{\beta}_{\textup{haj},\textsc{l}}(t)-\mu(t) = O_{\mathbb{P}}(n^{-1/2})$,
$\hat{\gamma}_\textsc{l}(t) - {\gamma}_\textsc{l}(t) = O_{\mathbb{P}}(n^{-1/2})$ for all $t\in\mathcal{T}$, and Assumption \ref{asu7}(b).
Thus, we have
\begin{align*}
G_{2}(t,t') 
= \frac{1}{n} \sum_{i=1}^n \sum_{j=1}^n 
% \begin{array}{c}
\frac{{1}_i(t)x_i
(Y_i(t;{\gamma}_\textsc{l})- \mu(t))}{\pi_i(t)} 
\frac{{1}_j(t')x_j
(Y_j(t';{\gamma}_\textsc{l})- \mu(t') )}{\pi_j(t')} 
% \end{array}
K_{n}(i,j)
+ o_{\mathbb{P}}(1).
\end{align*}
\noindent \textbf{Condition \ref{conditionL}(i).}
Let $\hat{\Omega}_{\textup{haj},\textsc{l}}( {\gamma}_{\textsc{l}}) = (\hat{\Omega}_{\textup{haj},\textsc{l}}(t,t'; {\gamma}_{\textsc{l}}) )_{t,t'\in\mathcal{T}}$ with the $(t,t')$th element
\begin{align*}
\hat{\Omega}_{\textup{haj},\textsc{l}}(t,t'; {\gamma}_{\textsc{l}}) 
=& \frac{1}{n} \sum_{\ell_{ {A}}(i, j) \leq b_n}
% \left( 
\frac{{1}_i(t) (Y_i(t;{\gamma}_\textsc{l})   - \hat{\beta}_\textup{haj}(t; \gamma_\textsc{l}(t))) }{\pi_i(t) \hat{1}_{\text{ht}}(t)}  \frac{{1}_j(t') (Y_j(t';\hat{\gamma}_\textsc{l})  - \hat{\beta}_\textup{haj}(t'; \gamma_\textsc{l}(t')))}{\pi_j(t') \hat{1}_{\text{ht}}(t')}.
% \right).
% \left( \frac{\frac{{1}_j(t')}{\pi_j(t')} (Y_j - x_j^\top {\gamma}_\textsc{l}(t')  - \hat{\beta}_\textup{haj}(t'; \gamma_\textsc{l}(t'))) }{\hat{1}_{\text{ht}}(t')} 
% \right). 
\end{align*}
Applying Theorem \ref{thm:Hájek_bias} to adjusted outcome $Y_i(t;{\gamma}_\textsc{l})$, we have 
\begin{equation}
\hat{\Omega}_{\textup{haj},\textsc{l}}( {\gamma}_{\textsc{l}})
= { {\Sigma}}_{*,\textup{haj},\textsc{l}} + R_{\textup{haj},\textsc{l}} 
+ o_\mathbb{P}(1).  \label{OmegaL}
\end{equation}
To complete the proof, it suffices to verify that $\hat{\Omega}_{\textup{haj},\textsc{l}}(t,t'; \hat{\gamma}_{\textsc{l}}) - \hat{\Omega}_{\textup{haj},\textsc{l}}(t,t'; {\gamma}_{\textsc{l}}) = o_\mathbb{P}(1)$. To simplify the notation without loss of generality, we verify it with scalar covariate $x_i$: 
\begin{align*}
& \hat{\Omega}_{\textup{haj},\textsc{l}}(t,t'; \hat{\gamma}_{\textsc{l}}) - \hat{\Omega}_{\textup{haj},\textsc{l}}(t,t'; {\gamma}_{\textsc{l}}) \\
=& 
({\gamma}_\textsc{l}(t) - \hat{\gamma}_\textsc{l}(t))
({\gamma}_\textsc{l}(t') - \hat{\gamma}_\textsc{l}(t'))
\frac{1}{n} \sum_{i=1}^n \sum_{j=1}^n 
% \begin{array}{c} 
% \left( 
\frac{{1}_i(t) (x_i - \hat{x}_\textup{haj}(t))  {1}_j(t') (x_j - \hat{x}_\textup{haj}(t')) }{\pi_i(t) \pi_j(t')}
% \right)  
% \left( \frac{}{}    \right)
% \end{array} 
K_{n}(i,j) \\
+& ({\gamma}_\textsc{l}(t') - \hat{\gamma}_\textsc{l}(t'))
\frac{1}{n} \sum_{i=1}^n \sum_{j=1}^n 
% \begin{array}{c}
% \left( 
\frac{{1}_i(t)( Y_i(t;\gamma_\textsc{l})   - \hat{\beta}_{\textup{haj}}(t; {\gamma}_\textsc{l})) {1}_j(t')(x_j - \hat{x}_\textup{haj}(t')) }{\pi_i(t) \pi_j(t')}     
% \right) 
% \left( \frac{}{\pi_j(t')}  \right)
% \end{array} 
K_{n}(i,j) \\
+& ({\gamma}_\textsc{l}(t) - \hat{\gamma}_\textsc{l}(t)) 
\frac{1}{n} \sum_{i=1}^n \sum_{j=1}^n 
% \begin{array}{c}
% \left( 
\frac{ {1}_i(t)(x_i - \hat{x}_\textup{haj}(t)){1}_j(t')( Y_j(t';\gamma_\textsc{l})   - \hat{\beta}_{\textup{haj}}(t'; {\gamma}_\textsc{l}))  }{\pi_i(t) \pi_j(t')}     
% \right) 
% \left( \frac{}{\pi_j(t')}  \right)
% \end{array} 
K_{n}(i,j)
+ o_\mathbb{P}(1),
\end{align*}
by Lemma \ref{lemma: beta_haj}. 
Under Assumptions \ref{asu2}, \ref{asu3} and \ref{asu9}, ${1}_i(t)\pi_i(t)^{-1}$, $x_i$ and $Y_i$ are uniformly bounded, then for some $C >0$ and any $n$, we have
\begin{align*}
& \left| \hat{\Omega}_{\textup{haj},\textsc{l}}(t,t'; \hat{\gamma}_{\textsc{l}}) - \hat{\Omega}_{\textup{haj},\textsc{l}}(t,t'; {\gamma}_{\textsc{l}}) \right| 
\le  C
(|{\gamma}_\textsc{l}(t) - \hat{\gamma}_\textsc{l}(t)|+|{\gamma}_\textsc{l}(t') - \hat{\gamma}_\textsc{l}(t')|) \frac{1}{n} \sum_{i=1}^n \sum_{j=1}^n
K_{n,ij},
% = o_\mathbb{P}(1),
\end{align*}
where the right-hand side term is $o_\mathbb{P}(1)$ by $\hat{\gamma}_\textsc{l}(t) - {\gamma}_\textsc{l}(t) = O_{\mathbb{P}}(n^{-1/2})$ and Assumption \ref{asu7}(b). 
This, together with \eqref{OmegaL}, ensures Condition \ref{conditionL}(i).

\end{proof}

\begin{proof}[Proof of Theorem \ref{thm:full_adj}] 
Let $\hat{ {V}}_{\textup{haj},\textsc{l}}^+(t,t')$ be the $(t,t')$th element of $\hat{ {V}}_{\textup{haj},\textsc{l}}^+$.  
For the adjusted covariance estimator, we have
\begin{align*}
& n \hat{ {V}}_{\textup{haj},\textsc{l}}^+(t,t') \\
=& n \hat{ {V}}_{\textup{haj},\textsc{l}}(t,t') 
+ \frac{1}{n} \sum_{i=1}^n \sum_{j=1}^n
% \begin{array}{c}
\frac{{1}_i(t)( Y_i(t;\hat\gamma_\textsc{l})
- \hat{\beta}_{\textup{haj},\textsc{l}}(t)) {1}_j(t')( Y_j(t';\hat\gamma_\textsc{l})-\hat{\beta}_{\textup{haj},\textsc{l}}(t')) }{\pi_i(t) \pi_j(t') \hat{1}_{\text{ht}}(t) 
\hat{1}_{\text{ht}}(t') }   
K_{n}^-(i,j) \\
=& {\Sigma}_{*,\textup{haj},\textsc{l}}(t,t') + R_{\textup{haj},\textsc{l}}(t,t') \\
+&  \frac{1}{n} \sum_{i=1}^n \sum_{j=1}^n 
\frac{{1}_i(t) ( Y_i(t;\hat\gamma_\textsc{l}) 
- \hat{\beta}_{\textup{haj}}(t)) }{\pi_i(t)}  
\frac{{1}_j(t') ( Y_j(t';\hat\gamma_\textsc{l})
-\hat{\beta}_{\textup{haj}}(t')) }{\pi_j(t')}  
K_{n}^-(i,j) + o_\mathbb{P}(1) \\
=& {\Sigma}_{*,\textup{haj},\textsc{l}}(t,t') 
+ \frac{1}{n} \sum_{i=1}^n \sum_{j=1}^n
M_{\textsc{l}}(i,t)
M_{\textsc{l}}(j,t')
\left(K_{n}^+(i,j) - K_{n}^-(i,j)\right)   \\
% -& \frac{1}{n} \sum_{i=1}^n \sum_{j=1}^n
% \left(\mu_i(t) - \mu(t)- x_i^\top \gamma_\textsc{l}(t)\right)
% \left(\mu_j(t') - \mu(t')- x_j^\top \gamma_\textsc{l}(t')\right) 
% K_{n}^-(i,j) \\
+&  \frac{1}{n} \sum_{i=1}^n \sum_{j=1}^n 
\frac{{1}_i(t) ( Y_i(t;\hat\gamma_\textsc{l})
- \hat{\beta}_{\textup{haj},\textsc{l}}(t)) }{\pi_i(t)} 
\frac{{1}_j(t') ( Y_j(t';\hat\gamma_\textsc{l})
-\hat{\beta}_{\textup{haj},\textsc{l}}(t')) }{\pi_j(t')}  
K_{n}^-(i,j) + o_\mathbb{P}(1), 
\end{align*}
where the first and the third equalities hold by the definition of $K_n^+$, and the second equality holds by Theorem \ref{thm:full_bias}. 
Following the proof of Theorem \ref{thm:full_bias} but replacing Assumption \ref{asu7} with Assumption \ref{asu8}, we have
\begin{align*}
& \frac{1}{n} \sum_{i=1}^n \sum_{j=1}^n 
\frac{{1}_i(t) ( Y_i(t;\hat\gamma_\textsc{l})
- \hat{\beta}_{\textup{haj},\textsc{l}}(t)) }{\pi_i(t)} 
\frac{{1}_j(t') ( Y_j(t';\hat\gamma_\textsc{l})
-\hat{\beta}_{\textup{haj},\textsc{l}}(t')) }{\pi_j(t')}  
K_{n}^-(i,j)  \\
=& \frac{1}{n} \sum_{i=1}^n \sum_{j=1}^n M_{\textsc{l}}(i,t) M_{\textsc{l}}(j,t') K_{n}^-(i,j) 
+ \frac{1}{n} \sum_{i=1}^n \sum_{j=1}^n {\Delta}_{\textup{haj},\textsc{l}}(i,t) {\Delta}_{\textup{haj},\textsc{l}}(j,t')  K_{n}^-(i,j) 
+ o_\mathbb{P}(1).
% =& n^{-1}(M_{\textsc{l},it}:i\in\mathcal{N}_n)^\top K_n^{-}(M_{\textsc{l},it'}:i\in\mathcal{N}_n)
% % \frac{1}{n} \sum_{i=1}^n \sum_{j=1}^n
% % (\mu_i(t) - \mu(t)- x_i^\top \gamma_\textsc{l}(t))
% % (\mu_j(t') - \mu(t')- x_j^\top \gamma_\textsc{l}(t')) 
% % {1}^-\{\ell_{ {A}}(i, j) \leq b_n\} 
% \\
% +& n^{-1}({\Delta}_{\textup{haj},\textsc{l}}(i,t):i\in\mathcal{N}_n)^\top K_n^{-}({\Delta}_{\textup{haj},\textsc{l}}(i,t'):i\in\mathcal{N}_n)
% + o_\mathbb{P}(1).
\end{align*}
Thus, we complete the proof.
\end{proof}

% \end{small}

\end{appendix}

\end{document}